\documentclass[11pt,openany]{book}
\pdfoutput=1
\usepackage{setspace}	
\usepackage{times}

\usepackage{tikz}
\usetikzlibrary{matrix,arrows}
\usetikzlibrary{decorations.pathmorphing}
\usetikzlibrary{decorations.pathreplacing}
\usetikzlibrary{decorations.shapes}
\usetikzlibrary{decorations.markings}
\usetikzlibrary{calc}

\tikzset{l/.style={font=\fontsize{8}{8}\selectfont}}
\definecolor{salmon}{rgb}{1,0.47,0.425}

\usepackage{ifthen}
\usepackage{fancyhdr}
\usepackage[nottoc]{tocbibind}
\usepackage[breaklinks=true]{hyperref}

\usepackage{amsmath,amsfonts,amssymb,amsthm}
\usepackage{slashed}
\usepackage{xy}
\usepackage{float}

\input xy
\xyoption{all}

\newcommand{\maps}{\colon}    
\newcommand{\R}{{\mathbb R}}  
\newcommand{\C}{{\mathbb C}}  
\renewcommand{\H}{{\mathbb H}}  
\renewcommand{\O}{{\mathbb O}}  
\newcommand{\K}{{\mathbb K}}  
\newcommand{\Z}{{\mathbb Z}}  

\renewcommand{\Re}{\mathrm{Re}} 
\renewcommand{\Im}{\mathrm{Im}} 
\newcommand{\Retr}{\Re \; \tr} 
\newcommand{\Mor}{\mathrm{Mor}} 
\newcommand{\h}{\mathfrak{h}} 
\newcommand{\tr}{{\mathrm{tr}}} 

\newcommand{\U}{{\rm U}}    
\newcommand{\SO}{{\rm SO}}    
\newcommand{\SU}{{\rm SU}}    
\newcommand{\Sp}{{\rm Sp}}    
\newcommand{\Spin}{{\rm Spin}}    
\newcommand{\ISO}{\mathrm{ISO}} 
\newcommand{\SISO}{\mathrm{SISO}} 
\newcommand{\String}{{\rm String}}    
\newcommand{\Superstring}{{\rm Superstring}}    
\newcommand{\Brane}{{\rm Brane}}    
\newcommand{\Twobrane}{2\text{-}\mathrm{Brane}} 
 
\newcommand{\n}{{\mathfrak{n}}} 
\newcommand{\so}{{\mathfrak{so}}}  
\newcommand{\gl}{{\mathfrak{gl}}}  
\newcommand{\g}{\mathfrak{g}}  
\newcommand{\T}{\mathcal{T}} 
\renewcommand{\b}{{\mathfrak{b}}}  
\newcommand{\siso}{\mathfrak{siso}} 
\newcommand{\sugra}{\mathfrak{sugra}} 
\newcommand{\brane}{\mathfrak{brane}} 
\newcommand{\twobrane}{2\text{-}\mathfrak{brane}} 
\newcommand{\superstring}{\mathfrak{superstring}} 
\newcommand{\strng}{\mathfrak{string}} 
\newcommand{\tri}{\operatorname{{{\rm tri}}}} 

\newcommand{\End}{{\rm End}} 
\newcommand{\Hom}{{\rm Hom}} 
\newcommand{\Sym}{{\rm Sym}} 
\newcommand{\Cliff}{{\rm Cliff}}    

\newcommand{\hol}{\mathrm{hol}}

\newcommand{\Set}{\mathrm{Set}} 
\newcommand{\Fun}{\mathrm{Fun}} 
\newcommand{\SuperVect}{\mathrm{SuperVect}} 
\newcommand{\SuperAlg}{\mathrm{SuperAlg}} 
\newcommand{\GrAlg}{\mathrm{GrAlg}} 
\newcommand{\Man}{\mathrm{Man}} 
\newcommand{\SuperMan}{\mathrm{SuperMan}} 
\newcommand{\SuperPoints}{\mathrm{SuperPoints}} 

\newcommand{\op}{{\rm op}} 
\newcommand{\Ad}{{\rm Ad}} 
\newcommand{\inv}{\mathrm{inv}} 
\newcommand{\id}{\mathrm{id}} 

\newcommand{\inclusion}{\hookrightarrow}
\newcommand{\iso}{\cong} 
\newcommand{\To}{\Rightarrow}
\newcommand{\tensor}{\otimes} 

\newcommand{\half}{\frac{1}{2}} 
\newcommand{\fourth}{\frac{1}{4}} 

\newcommand{\V}{\mathcal{V}} 
\renewcommand{\S}{\mathcal{S}} 
\newcommand{\A}{\mathcal{A}} 
\newcommand{\B}{\mathcal{B}} 

\newcommand{\Psibar}{\overline{\Psi}} 
\newcommand{\psibar}{\overline{\psi}} 
\newcommand{\phibar}{\overline{\phi}} 
\newcommand{\chibar}{\overline{\chi}} 

\newcommand{\define}[1]{{\bf \boldmath{#1}}}
\newcommand{\arxiv}[1]{\href{http://arxiv.org/abs/#1}{\texttt{arXiv:{#1}}}}

\newtheorem{thm}{Theorem}[chapter]
\newtheorem{cor}{Corollary}[chapter]
\newtheorem{lem}{Lemma}[chapter]
\newtheorem{prop}{Proposition}[chapter]

\newtheorem*{YL}{Yoneda Lemma}

\theoremstyle{definition}

\newtheorem{defn}{Definition}[chapter]

\newcommand{\draftmode}{false}

\newcommand{\thesisTitle}{{Division Algebras, Supersymmetry and Higher Gauge Theory}}

\newcommand{\degreeTitle}{Doctor of Philosophy}
\newcommand{\documentType}{Dissertation}
\newcommand{\documentTypeCAPS}{DISSERTATION}

\newcommand{\authorName}{John Gmerek Huerta}
\newcommand{\degreeMonth}{June}
\newcommand{\degreeYear}{2011}
\newcommand{\committeeChair}{Professor John Baez}
\newcommand{\committeeMemberOne}{Professor Vyjayanthi Chari}
\newcommand{\committeeMemberTwo}{Professor Stefano Vidussi}

\newcommand{\doublespaced}{\ifthenelse{\equal{\draftmode}{true}}{\onehalfspacing}{\doublespacing}}
\newcommand{\singlespaced}{\singlespacing}

\begin{document}

\date{}

\pagestyle{empty}

\begin{centering}
UNIVERSITY OF CALIFORNIA\\
RIVERSIDE\\

\vspace{1in}

\thesisTitle

\vspace{1in}

A \documentType { }submitted in partial satisfaction\\
of the requirements for the degree of\\

\vspace{0.4in}
\doublespaced
\degreeTitle\\
in\\
Mathematics\\
by\\
\authorName\\
\degreeMonth { }\degreeYear\\
\singlespaced
\vspace{1.6in}

\end{centering}

\noindent \documentType { }Committee:\\
\indent\indent\indent \committeeChair, Chairperson\\
\indent\indent\indent \committeeMemberOne\\
\indent\indent\indent \committeeMemberTwo

\newpage
\clearpage

\vspace*{\stretch{1}}
\begin{center}
Copyright by\\
\authorName\\
\degreeYear
\end{center}

\newpage
\clearpage

%
%
%
%
%
%
%
%
%

\pagestyle{plain}
\pagenumbering{roman}

\addtocounter{page}{2}

\begin{center}
Acknowledgements
\end{center}

A dissertation is the capstone to a doctoral program, and the acknowledgements
provide a useful place to thank the countless people who have helped out along
the way, both personally and professionally.

First, of course, I thank my advisor, John Baez. It is hard to imagine a better
advisor, and no one deserves more credit for my mathematical and professional
growth during this program. ``Thanks'' does not seem sufficient, but it is all
I have to give.

Also deserving special mention is John's collaborator, James Dolan. I am
convinced there is no subject in mathematics for which Jim does not have some
deep insight, and I thank him for sharing a few of these insights with me.
Together, John and Jim are an unparalleled team: there are no two better people
with whom to talk about mathematics, and no two people more awake to the joy of
mathematics.

I would also like to thank Geoffrey Dixon, Tevian Dray, Robert Helling, Corinne
Manogue, Chris Rogers, Hisham Sati, James Stasheff, and Riccardo Nicoletti for
helpful conversations and correspondence.  I especially thank Urs Schreiber for
many discussions of higher gauge theory and $L_\infty$-superalgebras, smooth
$\infty$-groups, and supergeometry. 

In addition, I thank Vyjayanthi Chari and Stefano Vidussi for serving on my
committee. I thank Mike Stay for providing some beautiful diagrams formatted in
TikZ. I thank An Huang, Theo Johnson-Freyd and David Speyer for catching some
errors. This work was partially supported by the FQXi grant RFP2-08-04, and by
a 2008--2009 Research Mentorship Fellowship from UC Riverside. I owe FQXi and
UC Riverside thanks for their support.

Many of the people I name above are not just colleagues, but close personal
friends. Besides those above, I would also like to thank Jack Bennett, Alex
Hoffnung, Felix Marhuenda-Donate, Rob Niemeyer, Tim Ridenour and Jonathan
Sarhad, among many others, for making graduate school so much fun. Finally, I
would like to thank my girlfriend, Jen Crowley, for being a wonderful, loving
companion.

\newpage

\null \vfill
{\large
\begin{center}
For my parents, Eliseo and Marion Huerta.
\end{center}}
\vfill \null

\newpage

\singlespacing
\begin{centering}
ABSTRACT OF THE \documentTypeCAPS\\ 

\vspace{0.7in}

\thesisTitle\\

\vspace{0.3in}

by\\

\vspace{0.3in}

\authorName\\

\vspace{0.4in}

\degreeTitle, Graduate Program in Mathematics\\
University of California, Riverside, \degreeMonth { } \degreeYear\\
\committeeChair, Chairperson\\

\end{centering}

\vspace{0.7in}

From the four normed division algebras---the real numbers, complex numbers,
quaternions and octonions, of dimension $k=1$, 2, 4 and 8, respectively---a
systematic procedure gives a 3-cocycle on the Poincar\'e superalgebra in
dimensions $k+2$=3, 4, 6 and 10, and a 4-cocycle on the Poincar\'e superalgebra in
dimensions $k+3$=4, 5, 7 and 11. The existence of these cocycles follow from
spinor identities that hold only in these dimensions, and which are closely
related to the existence of the superstring in dimensions $k+2$, and the
super-2-brane in dimensions $k+3$.

In general, an $(n+1)$-cocycle on a Lie superalgebra yields a `Lie
$n$-superalgebra': that is, roughly, an $n$-term chain complex equipped with a
bracket satisfying the axioms of a Lie superalgebra up to chain homotopy. We
thus obtain Lie 2-superalgebras extending the Poincar\'e superalgebra in
dimensions $k+2$, and Lie 3-superalgebras extending the Poincar\'e superalgebra
in dimensions $k+3$. We present evidence, based on the work of Sati, Schreiber
and Stasheff, that these Lie $n$-superalgebras describe infinitesimal `higher
symmetries' of the superstring and 2-brane.

Generically, integrating a Lie $n$-superalgebra to a Lie $n$-supergroup yields
a `Lie $n$-supergroup' that is hugely infinite-dimensional. However, when the
Lie $n$-superalgebra is obtained from an $(n+1)$-cocycle on a nilpotent Lie
superalgebra, there is a geometric procedure to integrate the cocycle to one on
the corresponding nilpotent Lie supergroup. 

In general, a smooth $(n+1)$-cocycle on a supergroup yields a `Lie
n-supergroup': that is, a weak $n$-group internal to supermanifolds. Using our
geometric procedure to integrate the 3-cocycle in dimensions $k+2$, we obtain a
Lie 2-supergroup extending the Poincar\'e supergroup in those dimensions, and
similarly integrating the 4-cocycle in dimensions $k+3$, we obtain a Lie
3-supergroup extending the Poincar\'e supergroup in those dimensions.

\tableofcontents
\listoftables

\clearpage
\pagenumbering{arabic}
\addtocounter{page}{8}

\setlength\footnotesep{1.5em} 

\chapter{Introduction} \label{ch:intro}

\section{Overview}

There is a deep relationship between supersymmetry, division algebras, and higher
gauge theory. In this thesis, we begin to tell this story: how division
algebras give rise to higher infinitesimal symmetries of strings and membranes,
modeled by a generalization of a Lie algebra called a `Lie $n$-algebra', and
how this infinitesimal picture can be integrated to global one, with higher
symmetries modeled by a `Lie $n$-group'. In this overview, we want to take the
opportunity to explain the big picture, postponing references until the next
section.

From a physical perspective, gauge theory is the geometric language which
allows us to describe how point particles change as they move through
spacetime. \emph{Higher} gauge theory is a generalization which describes how
strings and membranes change as they move through spacetime. 

We can view higher gauge theory as a categorification of gauge theory, which is
intuitively clear from the diagrams we use to describe higher categories: as a
particle moves through spacetime from point $x$ to point $y$, it sweeps out a
worldline $\gamma$ that we can view as a morphism from $x$ to $y$ in a certain
category:
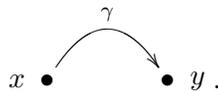
\[ \xy
	(-8,0)*+{\bullet}="4";
	(8,0)*+{\bullet}="6";
	(-12,0)*+{x};
	(12,0)*+{y};
	{\ar@/^1.65pc/^\gamma "4";"6"};
\endxy .
\]
The job of a connection in gauge theory is to assign to $\gamma$ an element
$\hol(\gamma)$ in the gauge group which describes how the state of our particle
changes as it moves along $\gamma$.

Boosting up a dimension, when a string moves through spacetime, it sweeps out a
worldsheet $\Sigma$, which we can view as a 2-morphism in a certain 2-category:
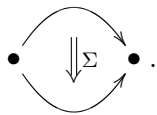
\[
	\xy
	(-8,0)*+{\bullet}="4";
	(8,0)*+{\bullet}="6";
	{\ar@/^1.65pc/ "4";"6"};
	{\ar@/_1.65pc/ "4";"6"};
	{\ar@{=>}^{\Sigma} (0,3)*{};(0,-3)*{}} ;
	\endxy .
\]
The job of a `2-connection' in `higher gauge theory' is to assign to $\Sigma$
an element $\hol(\Sigma)$ in the `higher gauge group' which describes how the
state of our string changes as it moves along $\Sigma$.

In practice, the strings and membranes of interest in physics are
supersymmetric, so they are called superstrings and supermembranes. This also
leads to higher gauge theory, but it goes through the normed division algebras.
There is a mysterious connection between supersymmetry and the four normed
division algebras: the real numbers, complex numbers, quaternions and
octonions.  This can be seen in super-Yang--Mills theory, in superstring
theory, and in theories of supermembranes and supergravity.  Most simply, the
connection is visible from the fact that the normed division algebras have
dimensions 1, 2, 4 and 8, while classical superstring theories and minimal
super-Yang--Mills theories live in spacetimes of dimension \emph{two} higher:
3, 4, 6 and 10.  The simplest classical super-2-brane theories make sense in
spacetimes of dimensions \emph{three} higher: 4, 5, 7 and 11.  Classical
supergravity makes sense in all of these dimensions, but the octonionic cases
are the most important: in 10 dimensions supergravity is a low-energy limit of
superstring theory, while in 11 dimensions it is believed to be a low-energy
limit of `M-theory', which incorporates the 2-brane.

These numerical relationships are far from coincidental.  They arise because we
can use the normed division algebras to construct the spacetimes in question,
as well as their associated spinors.  A certain spinor identity that holds in
dimensions 3, 4, 6 and 10 is an easy consequence of this construction, as is a
related identity that holds in dimensions 4, 5, 7 and 11.  These identities are
fundamental to the physical theories just listed.

Yet these identities have another interpretation: they are cocycle conditions
in Lie superalgebra cohomology for suitably chosen Lie superalgebras. We can
use them to categorify the infinitesimal symmetries of spacetime, or rather its
supersymmetric analog, superspacetime. This gives rise to Lie 2-superalgebras
and Lie 3-superalgebras.

Thanks to work by Hisham Sati, Urs Schreiber and Jim Stasheff, we expect that
generalized connections valued in these Lie 2- and 3-algebras will incorporate
fields of interest to string theory and supergravity. However, these
generalized connections are described in terms of infinitesimal data, because
Lie $n$-superalgebras are infinitesimal objects. We would like to know the global
story, so we want to integrate these to Lie $n$-supergroups.

Given a Lie $n$-algebra, there is a general technique, due to Getzler and
Henriques, to build a Lie $n$-group which integrates it. Usually, these are
hugely infinite-dimensional. For instance, if $\g$ is the finite-dimensional
Lie algebra of a simply-connected, finite-dimensional Lie group $G$, applying
the construction of Getzler and Henriques yields not $G$, but a simplicial
Banach manifold which is infinite-dimensional at almost every level.

Fortuantely, our Lie $n$-algebras are special. The cocycles which define them
are defined on nilpotent Lie subsuperalgebras, and these can be integrated
using a geometric method to smooth cocycles on the corresponding Lie
supergroups. So we obtain Lie $n$-supergroups which are finite-dimensional, and
even algebraic.  We expect that studying the higher gauge theory of these Lie
$n$-supergroups will yield important results for physics.

\section{Introduction} \label{sec:intro}

The relationship between division algebras and supersymmetry can be seen in
super-Yang--Mills theory, in superstring theory, and in theories of
supermembranes and supergravity.  Most simply, the connection is visible from
the fact that the normed division algebras have dimensions 1, 2, 4 and 8, while
classical superstring theories and minimal super-Yang--Mills theories live in
spacetimes of dimension \emph{two} higher: 3, 4, 6 and 10.  The simplest
classical super-2-brane theories make sense in spacetimes of dimensions
\emph{three} higher: 4, 5, 7 and 11.  Classical supergravity makes sense in all
of these dimensions, but the octonionic cases are the most important: in 10
dimensions supergravity is a low-energy limit of superstring theory, while in
11 dimensions it is believed to be a low-energy limit of `M-theory', which
incorporates the 2-brane.

As we noted in our overview, these numerical relationships are far from
coincidental.  They arise because we can use the normed division algebras to
construct the spacetimes in question, as well as their associated spinors. In a
bit more detail, suppose $\K$ is a normed division algebra of dimension $k$.
There are just four examples:
\begin{itemize}
\item the real numbers $\R$ ($k = 1$),
\item the complex numbers $\C$ ($k = 2$),
\item the quaternions $\H$ ($k = 4$),
\item the octonions $\O$ ($k = 8$).
\end{itemize}
Then we can identify vectors in $(k+2)$-dimensional Minkowski
spacetime with $2 \times 2$ hermitian matrices having entries in $\K$.
Similarly, we can identify spinors with elements of $\K^2$.  Matrix
multiplication then gives a way for vectors to act on spinors.  There
is also an operation that takes two spinors $\psi$ and $\phi$ and
forms a vector $[\psi, \phi]$.  Using elementary properties of
normed division algebras, we can prove that
\[               [\psi, \psi] \psi = 0 . \]
Following Schray \cite{Schray}, we call this identity the `3-$\psi$'s rule'.
This identity is an example of a `Fierz identity'---roughly, an identity that
allows one to reorder multilinear expressions made of spinors. This can be made
more visible in the 3-$\psi$'s rule if we polarize the above cubic form to
extract a genuinely trilinear expression:
\[               [\psi , \phi] \chi + [\phi , \chi] \psi + [\chi , \psi] \phi = 0 . \]

In fact, the 3-$\psi$'s rule holds \emph{only} when Minkowski spacetime has
dimension 3, 4, 6 or 10.  Moreover, it is crucial for super-Yang--Mills theory
and superstring theory in these dimensions.  In minimal super-Yang--Mills
theory, we need the 3-$\psi$'s rule to check that the Lagrangian is
supersymmetric, thanks to an argument we will review in Chapter \ref{ch:sym}.
In superstring theory, we need it to check the supersymmetry of the
Green--Schwarz Lagrangian \cite{GreenSchwarz, GreenSchwarzWitten}.  But the
3-$\psi$'s rule also has a deeper significance, which we study here.

This deeper story involves not only the 3-$\psi$'s rule but also the
`4-$\Psi$'s rule', a closely related Fierz identity required for super-2-brane
theories in dimensions 4, 5, 7 and 11.  To help the reader see the forest for
the trees, we present a rough summary of this story in the form of a recipe:

\begin{enumerate}

	\item Spinor identities that come from division algebras are
		\emph{cocycle conditions}.

	\item The corresponding cocycles allow us to extend the Poincar\'e
		Lie superalgebra to a higher structure, a \emph{Lie
		$n$-superalgebra}.

	\item Connections valued in these Lie
		$n$-superalgebras describe the \emph{field content} of
		superstring and super-2-brane theories.

\end{enumerate}

To begin our story in dimensions 3, 4, 6 and 10, let us first introduce some
suggestive terminology: we shall call $[\psi, \phi]$
the \define{bracket of spinors}. This is because this function is symmetric,
and it defines a Lie superalgebra structure on the supervector space
\[ \T = V \oplus S \]
where the even subspace $V$ is the vector representation of
$\Spin(k+1,1)$, while the odd subspace $S$ is a certain spinor
representation.  This Lie superalgebra is called the
\define{supertranslation algebra}.

There is a cohomology theory for Lie superalgebras, sometimes called
Chevalley--Eilenberg cohomology.  The cohomology of $\T$ will play a
central role in what follows.  Why?  First, because the 3-$\psi$'s rule is
really a cocycle condition, for a 3-cocycle $\alpha$ on $\T$ which eats 
two spinors and a vector and produces a number as follows:
\[ \alpha(\psi, \phi, A) = \langle \psi, A \phi \rangle . \]
Here, $\langle -,- \rangle$ is a pairing between spinors.  
Since this 3-cocycle is Lorentz-invariant, it extends to a cocycle
on the Poincar\'e superalgebra
\[    \siso(k+1,1) \iso \so(k+1,1) \ltimes \T  . \]
In fact, we obtain a nonzero element of the third cohomology of
the Poincar\'e superalgebra this way.

Just as 2-cocycles on a Lie superalgebra give ways of extending it to larger
Lie superalgebras, $(n+1)$-cocycles give extensions to \emph{Lie
$n$-superalgebras}.  To understand this, we need to know a bit about
$L_\infty$-algebras \cite{MSS,SS}.  An $L_\infty$-algebra is a chain complex
equipped with a structure like that of a Lie algebra, but where the laws hold
only `up to $d$ of something'.  A Lie $n$-algebra is an $L_\infty$-algebra in
which only the first $n$ terms are nonzero.  All these ideas also have `super'
versions.  In general, an $\h$-valued $(n+1)$-cocycle $\omega$ on $\g$ is a
linear map:
\[ \Lambda^{n+1} \g \to \h \] 
satisfying a certain equation called a `cocycle condition'. We can use an
$\h$-valued $(n+1)$-cocycle $\omega$ on a Lie superalgebra $\g$ to extend $\g$
to a Lie $n$-superalgebra of the following form:
\[ \g \stackrel{d}{\longleftarrow} 0 \stackrel{d}{\longleftarrow} \cdots \stackrel{d}{\longleftarrow} 0 \stackrel{d}{\longleftarrow} \h . \]
Here, $\g$ sits in degree 0 while $\h$ sits in degree $n-1$. We call Lie
$n$-superalgebras of this form `slim Lie $n$-superalgebras', and denote them by
$\brane_\omega(\g,\h)$. 

In particular, we can use the 3-cocycle $\alpha$ to extend $\siso(k+1,1)$ to a
slim Lie 2-superalgebra of the following form:
\[ \xymatrix{ \siso(k+1, 1) & \R \ar[l]_<<<<<d } . \]
We call this the `superstring Lie 2-superalgebra', and denote it as 
$\superstring(k+1,1)$. The superstring Lie 2-superalgebra is an extension of
$\siso(k+1,1)$ by $\b\R$, the Lie 2-algebra with $\R$ in degree 1 and
everything else trivial.  By `extension', we mean that there is a short exact
sequence of Lie 2-superalgebras:
\[         0 \to b\R \to \superstring(k+1,1) \to \siso(k+1,1) \to 0 . \]
To see precisely what this means, let us expand it a bit.  Lie
2-superalgebras are 2-term chain complexes, and writing these
vertically, our short exact sequence looks like this:
\[ 
\def\objectstyle{\scriptstyle}
\xymatrix { 
            0 \ar[r] & \R \ar[r] \ar[d]_d & \R \ar[r] \ar[d]_d    & 0 \ar[r] \ar[d]_d     & 0 \\
            0 \ar[r] & 0 \ar[r]         & \siso(k+1,1) \ar[r] & \siso(k+1,1) \ar[r] & 0 \\
} 
\]
In the middle, we see $\superstring(k+1,1)$.  This Lie 2-superalgebra
is built from two pieces: $\siso(k+1,1)$ in degree $0$ and $\R$ in
degree $1$.  But since the cocycle $\alpha$ is nontrivial, these two
pieces still interact in a nontrivial way.  Namely, the Jacobi
identity for three 0-chains holds only up to $d$ of a 1-chain.  So,
besides its Lie bracket, the Lie 2-superalgebra $\superstring(k+1,1)$
also involves a map that takes three 0-chains and gives a 1-chain.
This map is just $\alpha$.

What is the superstring Lie 2-algebra good for?  The answer lies in a
feature of string theory called the `Kalb--Ramond field', or `$B$
field'.  The $B$ field couples to strings just as the $A$ field in
electromagnetism couples to charged particles.  The $A$ field is
described locally by a 1-form, so we can integrate it over a
particle's worldline to get the interaction term in the Lagrangian for
a charged particle.  Similarly, the $B$ field is described locally by
a 2-form, which we can integrate over the worldsheet of a string.

Gauge theory has taught us that the $A$ field has a beautiful
geometric meaning: it is a connection on a $\U(1)$ bundle over
spacetime.  What is the corresponding meaning of the $B$ field?  It
can be seen as a connection on a `$\U(1)$ gerbe': a gadget like a
$\U(1)$ bundle, but suitable for describing strings instead of point
particles.  Locally, connections on $\U(1)$ gerbes can be identified
with 2-forms.  But globally, they cannot.  The idea that the $B$ field
is a $\U(1)$ gerbe connection is implicit in work going back at least
to the 1986 paper by Gawedzki \cite{Gawedzki}.  More recently, Freed
and Witten \cite{FreedWitten} showed that the subtle difference
between 2-forms and connections on $\U(1)$ gerbes is actually crucial
for understanding anomaly cancellation.  In fact, these authors used
the language of `Deligne cohomology' rather than gerbes.  Later work
made the role of gerbes explicit: see for example Carey, Johnson and
Murray \cite{CareyJohnsonMurray}, and also Gawedzki and Reis
\cite{GawedzkiReis}.

More recently still, work on higher gauge theory has revealed that the
$B$ field can be viewed as part of a larger package.  Just as gauge
theory uses Lie groups, Lie algebras, and connections on bundles to
describe the parallel transport of point particles, higher gauge
theory generalizes all these concepts to describe parallel transport
of extended objects such as strings and membranes
\cite{BaezHuerta:invitation, BaezSchreiber}.  In particular, Schreiber, Sati and
Stasheff \cite{SSS} have developed a theory of `$n$-connections'
suitable for describing parallel transport of objects with
$n$-dimensonal worldvolumes.  In their theory, the Lie algebra of the
gauge group is replaced by a Lie $n$-algebra---or in the
supersymmetric context, a Lie $n$-superalgebra.  Applying their ideas
to $\superstring(k+1,1)$, we get a 2-connection which can be described
locally using the following fields:
\vskip 1em
\begin{center}
\begin{tabular}{cc}
	\hline
	$\superstring(k+1,1)$ & Connection component \\
	\hline
	$\R$                  & $\R$-valued 2-form \\
	$\downarrow$          & \\
	$\siso(k+1,1)$        & $\siso(k+1,1)$-valued 1-form \\
	\hline
\end{tabular}
\end{center}
\vskip 1em
The $\siso(k+1,1)$-valued 1-form consists of three fields which help
define the background geometry on which a superstring propagates: the
Levi-Civita connection $A$, the vielbein $e$, and the gravitino
$\psi$.  But the $\R$-valued 2-form is equally important in the
description of this background geometry: it is the $B$ field!

Alas, this is only part of the story. Rather than building $\superstring(k+1,1)$
with the cocycle $\alpha$, quantum considerations indicate we should really use
a certain linear combination of $\alpha$ and the canonical 3-cocycle on
$\so(k+1,1)$. This canonical 3-cocycle can be defined on any simple Lie algebra.
It comes from combining the Killing form with the bracket:
\[ j = \langle -, [-,-] \rangle . \]
To ensure that certain quantum `anomalies' cancel, we need to replace $\alpha$
with the linear combination:
\[ \frac{1}{32} j + \frac{1}{2} \alpha . \]
These coefficients come from careful analysis of the anomalies associated with
superstring theory. See the paper by Bonora et al.\ and the references therein
\cite{Bonora}.

We choose, however, to focus on $\alpha$. This simplifies our later work, and
because Lie 2-algebras based on $j$ have already been the subject of much
scrutiny \cite{BaezCrans, BCSS, Henriques, SchommerPries}, it should be
possible to combine what we do here with the work of other authors to arrive at
a more complete picture.

Next let us extend these ideas to Minkowski spacetimes one dimension higher:
dimensions 4, 5, 7 and 11.  In this case a certain subspace of $4 \times 4$
matrices with entries in $\K$ will form the vector representation of
$\Spin(k+2, 1)$, while $\K^4$ will form a spinor representation.  As before,
there is a `bracket' operation that takes two spinors $\Psi$ and $\Phi$ and
gives a vector $[\Psi , \Phi]$.  As before, there is an action of vectors on
spinors.  This time the 3- $\psi$'s rule no longer holds:
\[ [\Psi , \Psi] \Psi \neq 0 . \]
However, we show that
\[ [\Psi , [\Psi , \Psi] \Psi] = 0 . \]
We call this the `4-$\Psi$'s rule'.  This identity plays a basic role 
for the super-2-brane, and related theories of supergravity. The explicit
relationship between this identity and the division algebras goes back to Foot
and Joshi \cite{FootJoshi}.

Once again, the bracket of spinors defines a Lie superalgebra structure on the
supervector space
\[ \T = \V \oplus \S \]
where now $\V$ is the vector representation of $\Spin(k+2,1)$, while $\S$ is a
certain spinor representation of this group.  Once again, the cohomology of
$\T$ plays a key role.  The 4-$\Psi$'s rule is a cocycle condition---but this
time for a 4-cocycle $\beta$ which eats two spinors and two vectors and
produces a number as follows:
\[ \beta(\Psi, \Phi, \A, \B) = \langle \Psi, (\A \wedge \B) \Phi \rangle . \]
Here, $\langle -, - \rangle$ denotes the inner product of two spinors, and the
bivector $\A \wedge \B$ acts on $\Phi$ via the usual Clifford action. Since
$\beta$ is Lorentz-invariant, we shall see that it extends to a 4-cocycle on
the Poincar\'e superalgebra $\siso(k+2,1)$.

We can use $\beta$ to extend the Poincar\'e superalgebra to a Lie
3-superalgebra of the following form: 
\[ \xymatrix{ \siso(k+2, 1) & 0 \ar[l]_<<<<<d & \R \ar[l]_d } . \] 
We call this the `2-brane Lie 3-superalgebra', and denote it as
$\mbox{2-}\brane(k+1,1)$.   It is an extension of $\siso(k+2,1)$ by $\b^2\R$,
the Lie 3-algebra with $\R$ in degree 2, and everything else trivial.  In other
words, there is a short exact sequence: 
\[ 0 \to \b^2\R \to \mbox{2-}\brane(k+2,1) \to \siso(k+2,1) \to 0 . \]
Again, to see what this means, let us expand it a bit.  Lie 3-superalgebras are
3-term chain complexes.  Writing out each of these vertically, our short exact
sequence looks like this:
\[ 
\def\objectstyle{\scriptstyle}
\xymatrix { 
            0 \ar[r]  & \R \ar[r] \ar[d]_d & \R \ar[r] \ar[d]_d    & 0 \ar[r] \ar[d]_d     & 0  \\
            0 \ar[r]  & 0 \ar[r] \ar[d]_d  & 0 \ar[r] \ar[d]_d     & 0 \ar[r] \ar[d]_d     & 0  \\
            0 \ar[r]        & 0 \ar[r]         & \siso(k+2,1) \ar[r] & \siso(k+2,1) \ar[r] & 0 \\
} \]
In the middle, we see 2-$\brane(k+2,1)$.

The most interesting Lie 3-algebra of this type, 2-$\brane(10,1)$, plays an
important role in 11-dimensional supergravity.  This idea goes back to the work
of Castellani, D'Auria and Fr\'e \cite {TheCube, DAuriaFre}.  These authors
derived the field content of 11- dimensional supergravity starting from a
differential graded-commutative algebra.  Later, Sati, Schreiber and Stasheff
\cite{SSS} explained that these fields can be reinterpreted as a 3-connection
valued in a Lie 3-algebra which they called `$\sugra(10,1)$'.  This is the Lie
3-algebra we are calling 2-$\brane(10,1)$.  One of our messages here is that
the all-important cocycle needed to construct this Lie 3-algebra arises
naturally from the octonions, and has analogues for the other normed division
algebras.

If we follow these authors and consider a 3-connection valued in
2-$\brane(10,1)$, we find it can be described locally by these fields:
\vskip 1em
\begin{center}
\begin{tabular}{cc}
	\hline
	2-$\brane(k+2,1)$ & Connection component \\
	\hline
	$\R$            & $\R$-valued 3-form \\
	$\downarrow$    & \\
	$0$             & \\
	$\downarrow$    & \\
	$\siso(k+2,1)$  & $\siso(k+2,1)$-valued 1-form \\
	\hline
\end{tabular}
\end{center}
\vskip 1em
Again, a $\siso(k+2,1)$-valued 1-form contains familiar fields: the Levi-Civita
connection, the vielbein, and the gravitino.  But now we also see a 3-form,
called the $C$ field.  This is again something we might expect on physical
grounds, at least in dimension 11.  While the case is less clear than in string
theory, it seems that for the quantum theory of a 2-brane to be consistent, it
must propagate in a background obeying the equations of 11-dimensional
supergravity, in which the $C$ field naturally shows up \cite{Tanii}.  The work
of Diaconescu, Freed, and Moore \cite{DiaconescuFreedMoore}, as well as that of
Aschieri and Jurco \cite{AschieriJurco}, is also relevant here.

So far, we have focused on Lie 2- and 3-algebras and generalized connections
valued in them. This connection data is infinitesimal: it tells us how to
parallel transport strings and 2-branes a little bit. Ultimately, we would like
to understand this parallel transport globally, as we do with particles in
ordinary gauge theory.

To achieve this global description, we will need `Lie $n$-groups' rather than
Lie $n$-algebras. Naively, one expects a Lie 2-supergroup $\Superstring(k+1,1)$
for which the Lie 2-super\-algebra $\superstring(k+1,1)$ is the infinitesimal
approximation, and similarly a Lie 3-supergroup \break $\Twobrane(k+2,1)$ for
which the Lie 3-superalgebra $\twobrane(k+1,1)$ is an infinitesimal
approximation.  In fact, this is precisely what we will construct.

In order to `integrate' Lie $n$-algebras to obtain Lie $n$-groups, we will have
to overcome two obstacles: how does one define a Lie $n$-group? And, how does
one integrate a Lie $n$-algebra to a Lie $n$-group? To answer the former
question, we take a cue from Baez and Lauda's definition of Lie 2-group: it is
a categorified Lie group, a `weak 2-category' with one object with a manifold
of weakly associative and weakly invertible morphisms, a manifold of strictly
associative and strictly invertible 2-morphisms, and all structure maps smooth.
While this definition is known to fall short in important ways, it has the
virtue of being fairly simple. Ultimately, one should use an alternative
definition, like that of Henriques \cite{Henriques} or Schommer-Pries
\cite{SchommerPries}, which weakens the notion of product on a group: rather
than an algebraic operation in which there is a unique product of any two group
elements, `the' product is defined only up to equivalence.

So, roughly speaking, a Lie $n$-group should be a `weak $n$-category' with one
object, a manifold of weakly invertible morphisms, a manifold of weakly
invertible 2-morphisms, and so on, up to a manifold of strictly invertible
$n$-morphisms. To make this precise, however, we need to get very precise about
what a `weak $n$-category' is, which becomes more complicated as $n$ gets
larger. We therefore limit ourselves to the tractable cases of $n=2$ and 3. We
further limit ourselves to what we call a `slim Lie $n$-group', at least for
$n=2$ and 3.

A `slim Lie 2-group' is what Baez and Lauda call a `special Lie 2-group': it is
a skeletal bicategory with one object, a Lie group $G$ of morphisms, a Lie group $G
\ltimes H$ of 2-morphisms, and the group axioms hold strictly \emph{except for
associativity}---there is a nontrivial 2-morphism called the `associator':
\[ a(g_1,g_2,g_3) \maps (g_1g_2)g_3 \Rightarrow g_1(g_2g_3) . \]
The associator, in turn, satisfies the `pentagon identity', which says the
following pentagon commutes:
\[
\xy
 (0,20)*+{(g_1  g_2)  (g_3  g_4)}="1";
 (40,0)*+{g_1  (g_2  (g_3  g_4))}="2";
 (25,-20)*{ \quad g_1  ((g_2  g_3)  g_4)}="3";
 (-25,-20)*+{(g_1  (g_2  g_3))  g_4}="4";
 (-40,0)*+{((g_1  g_2)  g_3)  g_4}="5";
 {\ar@{=>}^{a(g_1,g_2,g_3  g_4)}     "1";"2"}
 {\ar@{=>}_{1_{g_1}  a(g_2,g_3,g_4)}  "3";"2"}
 {\ar@{=>}^{a(g_1,g_2  g_3,g_4)}    "4";"3"}
 {\ar@{=>}_{a(g_1,g_2,g_3)  1_{g_4}}  "5";"4"}
 {\ar@{=>}^{a(g_1  g_2,g_3,g_4)}    "5";"1"}
\endxy
\]
We shall see that this identity forces $a$ to be a 3-cocycle on the Lie group
$G$ of morphisms. We denote the Lie 2-group of this from by $\String_a(G,H)$.

Likewise, a `slim Lie 3-group' is a skeletal tricategory with one object, with
a Lie group $G$ of morphisms, trivial 2-morphisms, and a Lie group $G \ltimes
H$ of 3-morphisms. The associator is necessarily trivial, because it is a
2-morphism:
\[ a(g_1, g_2, g_3) \maps (g_1 g_2) g_3 \stackrel{1}{\Longrightarrow} g_1 (g_2 g_3) , \]
However, it does not satisfy the pentagon identity! There is a nontrivial
3-morphism called the `pentagonator':
\[
\xy
 (0,20)*+{(g_1  g_2)  (g_3  g_4)}="1";
 (40,0)*+{g_1  (g_2  (g_3  g_4))}="2";
 (25,-20)*{ \quad g_1  ((g_2  g_3)  g_4)}="3";
 (-25,-20)*+{(g_1  (g_2  g_3))  g_4}="4";
 (-40,0)*+{((g_1  g_2)  g_3)  g_4}="5";
     {\ar@{=>}^{a(g_1,g_2,g_3 g_4)}     "1";"2"}
     {\ar@{=>}_{1_{g_1} \cdot a_(g_2,g_3,g_4)}  "3";"2"}
     {\ar@{=>}^{a(g_1,g_2 g_3,g_4)}    "4";"3"}
     {\ar@{=>}_{a(g_1,g_2,g_3) \cdot 1_{g_4}}  "5";"4"}
     {\ar@{=>}^{a(g_1 g_2,g_3,g_4)}    "5";"1"}
     {\ar@3{->}^{\pi(g_1,g_2,g_3,g_4)} (-2,5)*{}; (-2,-5)*{} }
\endxy
\]
This 3-morphism satisfies an identity of its own, called the `pentagonator
identity'.  Similar to the case with the slim Lie 2-group $\String_a(G,H)$, the
pentagonator identity forces $\pi$ to be a Lie group 4-cocycle on $G$.

Moreover, we can generalize all of this to obtain Lie 2-supergroups and Lie
3-supergroups from 3- and 4-cocycles on Lie supergroups. In general, we expect
that any supergroup $(n+1)$-cocycle $f$ gives rise to a slim $n$-supergroup,
$\Brane_f(G,H)$, though this cannot be made precise without being more definite
about $n$-categories for higher $n$. 

Nonetheless, the precise examples of Lie 2- and 3-groups suggest a strong
parallel to the way Lie algebra $(n+1)$-cocycles give rise to Lie $n$-algebras.
And this parallel suggests a naive scheme to integrate Lie $n$-algebras. Given
a slim Lie $n$-superalgebra $\brane_\omega(\g,\h)$, we seek a slim Lie
$n$-supergroup $\Brane_f(G,H)$ where:
\begin{itemize}
	\item $G$ is a Lie supergroup with Lie superalgebra $\g$; i.e.\, it is a Lie
		supergroup integrating $\g$,
	\item $H$ is a Lie supergroup with Lie superalgebra $\h$; i.e.\, it is a Lie
		supergroup integrating $\h$,
	\item $f$ is a Lie supergroup $(n+1)$-cocycle on $G$ that, in some suitable sense,
		integrates the Lie superalgebra $(n+1)$-cocycle $\omega$ on $\g$.
\end{itemize}
Admittedly, we only define $\Brane_f(G,H)$ precisely when $n=2$ or $3$, but
that will suffice to handle our cases of interest, $\superstring(k+1,1)$ and
$\twobrane(k+2,1)$.

Unfortunately, this naive scheme fails to work even for well-known examples of
slim Lie 2-algebras, such as the the string Lie 2-algebra $\strng(n)$. In this
case, we can:
\begin{itemize}
	\item integrate $\so(n)$ to $\Spin(n)$ or $\SO(n)$,
	\item integrate $\R$ to $\R$ or $\U(1)$, 
	\item but there is no hope to integrate $\omega$ to a nontrivial
		$(n+1)$-cocycle $f$ on $\SO(n)$ or $\Spin(n)$, because compact
		Lie groups admit \emph{no nontrivial smooth cocycles}.
\end{itemize}
Really, this failure is a symptom of the fact that our definition of Lie
$n$-group is oversimplified. There are more sophisticated approaches to
integrating the string Lie 2-algbera, like those due to Baez, Crans, Schreiber
and Stevenson \cite{BCSS} or Schommer-Pries \cite{SchommerPries}, and a general
technique to integrate any Lie $n$-algebra due to Henriques \cite{Henriques},
which Schreiber \cite{Schreiber} has in turn generalized to handle Lie
$n$-superalgebras and more.  All three techniques involve generalizing the
notion of Lie 2-group (or Lie $n$-group, for Henriques and Schreiber) away from
the world of finite-dimensional manifolds, and the latter two generalize the
notion of 2-group to one in which products are defined only up to equivalence.

Given this history, it is remarkable that the naive scheme we outlined for
integration actually works for the Lie $n$-superalgebras we really care
about---namely, the superstring Lie 2-algbera and the super-2-brane Lie
3-algebra. Moreover, this is not some weird quirk unique to these special
cases, but the result of a beautiful geometric procedure for integrating Lie
algebra cocycles defined on a \emph{nilpotent} Lie algebra. Originally invented
by Houard \cite{Houard}, we generalize this technique to the case of nilpotent
Lie superalgebras and supergroups. 

Finally, we mention another use for the cocycles $\alpha$ and $\beta$.  These
cocycles are also used to build Wess--Zumino--Witten terms for superstrings and
2-branes. For example, in the case of the string, one can extend the string's
worldsheet to be the boundary of a three-dimensional manifold, and then
integrate $\alpha$ over this manifold.  This provides an additional term for
the action of the superstring, a term that is required to give the action
Siegel symmetry, balancing the number of bosonic and fermionic degrees of
freedom. For the 2-brane, the Wess--Zumino--Witten term is constructed in
complete analogy---we just `add one' to all the dimensions in sight
\cite{AchucarroEvans,Duff}.

Indeed, the network of relationships between supergravity, string and 2-brane
theories, and cocycles constructed using normed division algebras is extremely
tight.  The Siegel symmetry of the string or 2-brane action constrains the
background of the theory to be that of supergravity, at least in dimensions 10
and 11 \cite{Tanii}, and without the WZW terms, there would be no Siegel
symmetry.  The WZW terms rely on the cocycles $\alpha$ and $\beta$.  These
cocycles also give rise to the Lie 2- and 3-superalgebras $\superstring(9,1)$
and 2-$\brane(10,1)$.  And these, in turn, describe the field content of
supergravity in these dimensions!

As further grist for this mill, WZW terms can also be viewed in the context of
higher gauge theory.  In string theory, the WZW term is the holonomy of a
connection on a $\U(1)$ gerbe \cite{GawedzkiReis}.  Presumably the WZW term in
a 2-brane theory is the holonomy of a connection on a $\U(1)$ 2-gerbe
\cite{Stevenson}.  This is a tantalizing clue that we are at the beginning of a
larger but ultimately simpler story.

\section{Plan of the thesis}

The focus of this thesis is not on the applications to physics that we sketched
in the Introduction, but on constructing Lie $n$-superalgebras from division
algebras, and integrating these Lie $n$-superalgebras to Lie $n$-supergroups.
We organize the thesis as follows:

\begin{itemize}

	\item In Chapter \ref{ch:geometry}, we give a review of the needed
		facts about normed division algebras, and apply the division
		algebras to construct vectors and spinors in spacetimes of
		certain dimension. We conclude by using these constructions to
		prove certain spinor identities needed for supersymmetric
		physics.

	\item In Chapter \ref{ch:supertranslations}, we introduce the algebra
		underlying supersymmetry: super vector spaces and Lie
		superalgebras. We construct some important examples of Lie
		superalgebras: the supertranslation algebras, $\T$, using
		division algebras. We give a well-known generalization of
		Chevalley--Eilenberg cohomology to Lie superalgebras, and prove
		that the supertranslation algebras admit nontrivial cocycles
		thanks to the spinor identities from the previous chapter.

	\item In Chapter \ref{ch:sym}, we take a break from the larger story to
		discuss super-Yang--Mills theory. We prove the supersymmetry of
		super-Yang--Mills theory in spacetime dimensions 3, 4, 6 and
		10, using the division algebras.

	\item In Chapter \ref{ch:Lie-n-superalgebras}, we describe how a Lie
		superalgebra $(n+1)$-cocycle on $\g$ gives rise to a Lie
		$n$-superalgebra which extends $\g$. We use this general
		construction to build several important examples of Lie 2- and
		3-superalgebras: the well-known string Lie 2-algebra
		$\strng(n)$ extending $\so(n)$, the Heisenberg Lie 2-algebra
		extending the Heisenberg Lie algebra, the superstring Lie
		2-algebra $\superstring(k+1,1)$ and the super-2-brane Lie
		3-algebra $\twobrane(k+2,1)$, both extending the Poincar\'e
		superalgebras $\siso(k+1,1)$ and $\siso(k+2,1)$, respectively.
		
	\item In Chapter \ref{ch:Lie-n-groups}, we describe Lie group
		cohomology based on smooth group cochains. We define Lie
		$n$-groups for $n = 2$ and 3, using bicategories and
		tricategories internal to the category of smooth manifolds. We
		sketch how a Lie group $(n+1)$-cocycle on $G$ gives rise
		to a Lie $n$-group which extends $G$, and give a full
		construction for $n = 2$ and 3.

	\item In Chapter \ref{ch:integrating}, we apply a little-known geometric
		technique to integrate nilpotent Lie $n$-algebras to Lie
		$n$-groups, by integrating Lie algebra $(n+1)$-cocycles to Lie
		group $(n+1)$-cocycles. We compute some examples for 2-step
		nilpotent Lie algebras, and conclude with by constructing the
		Heisenberg Lie 2-group from the Lie 2-algebra.

	\item In Chapter \ref{ch:supergeometry}, we introduce a little
		supergeometry. Specifically, we sketch the definition of supermanifold, and
		discuss the functor of points approach to studying these
		spaces. We describe how to get a supermanifold from any super
		vector space, and show the corresponding functor of points is
		especially simple. We then describe how to integrate a
		nilpotent Lie superalgebra to a supergroup.

	\item In Chapter \ref{ch:Lie-n-supergroups}, we generalize everything
		from Chapter \ref{ch:Lie-n-groups} to the super setting. We
		describe Lie supergroup cohomology, and we define Lie
		$n$-supergroups for $n=2$ and 3, using bicategories and
		tricategories internal to the category of supermanifolds.

	\item In Chapter \ref{ch:integrating2}, we generalize everything from
		Chapter \ref{ch:integrating} to the super setting. We show how
		to integrate nilpotent Lie $n$-superalgebras to Lie
		$n$-supergroups, by integrating Lie superalgebra
		$(n+1)$-cocycles to Lie supergroup $(n+1)$-cocycles. This is
		done using the functor of points.

	\item Finally, in Chapter \ref{ch:finale}, we apply the results of the
		previous chapter to integrate \\ $\superstring(k+1,1)$ and
		$\twobrane(k+2,1)$ to Lie $n$-supergroups,
		$\Superstring(k+1,1)$ and
		$\Twobrane(k+2,1)$. We conclude with some remarks about where
		these results could lead, and the next steps in this research
		program.

\end{itemize}

\section{Prior work}

Portions of this thesis are adapted from two papers coauthored with my advisor,
John Baez, called  ``Supersymmetry and division algebras I and II''
\cite{BaezHuerta:susy1, BaezHuerta:susy2}. Specifically, Sections
\ref{sec:divalg}, \ref{sec:k+2}, \ref{sec:superalgebra}, and Chapter
\ref{ch:sym} are adapted from the first paper \cite{BaezHuerta:susy1}, Sections
\ref{sec:intro}, \ref{sec:k+3}, \ref{sec:cohomology}, the beginning of Chapter
\ref{ch:Lie-n-superalgebras} and Section \ref{sec:n-brane-algs} are adapted
from the second paper \cite{BaezHuerta:susy2}, and Section \ref{sec:spinors}
combines related results from both papers.

\chapter{Spacetime geometry from division algebras} \label{ch:geometry}

In this chapter, we begin to explore the relationship between:
\begin{itemize}
	\item Normed division algebras of dimension $k = 1, 2, 4$ and 8.
	\item Superstring theories in spacetimes of dimension $k+2 = 3, 4, 6$ and 10.
	\item Super-2-brane theories in spacetimes of dimension $k+3 = 4, 5, 7$ and 11.
\end{itemize}
Physically, a supersymmetric theory requires the use of vector representations
of the Lorentz group to describe its bosonic degrees of freedom, and the spinor
representations of the Lorentz group to describe its fermionic degrees of
freedom. In this chapter, we will show that a normed division algebra $\K$ of
dimension $k$ can be used to construct vectors and spinors in $k+2$ and $k+3$
dimensions.

First, let us describe the most general situation. Let $V$ be a real vector
space equipped with a nondegenerate quadratic form, $|\cdot|^2$. The group $\Spin(V)$,
the double-cover of $\SO_0(V)$, acts on $V$ as the symmetries of $|\cdot|^2$. We say
that $V$ is the \define{vector representation} of $\Spin(V)$, and call its
elements \define{vectors}.

We can also construct representations $\Spin(V)$ by considering the Clifford
algebra, $\Cliff(V)$. This is the associative algebra generated by $V$ for
which elements $A \in V$ square to their norm:
\[ \Cliff(V) = TV \big / A^2 \sim |A|^2, \]
where $TV$ denotes the tensor algebra on $V$. Because the \define{Clifford
relation} $A^2 = |A|^2$ respects the parity of the number of vectors in any
expression, the Clifford algebra is $\Z_2$-graded:
\[ \Cliff(V) = \Cliff_0(V) \oplus \Cliff_1(V). \]
We call $\Cliff_0(V)$ and $\Cliff_1(V)$ the \define{even part} and \define{odd
part} of $\Cliff(V)$, respectively. $\Cliff_0(V)$ is the subalgebra of
$\Cliff(V)$ generated by products of pairs of vectors, while $\Cliff_1(V)$ is a
mere subspace of $\Cliff(V)$, spanned by products of odd numbers of vectors. 

It is well-known that $\Spin(V)$ lives inside $\Cliff_0(V)$. This is the group
generated by products of pairs of \define{unit vectors}: vectors $A$ for which
$|A|^2 = \pm 1$. So, we can consider representations of $\Spin(V)$ that come
from modules of $\Cliff_0(V)$. Such a representation is called a
\define{spinor representation} of $\Spin(V)$, and its elements are called
\define{spinors}. The algebra $\Cliff_0(V)$ turns out to be either a matrix
algebra or the sum of two matrix algebras, so there are either two irreducible
spinor representations, $S_+$ and $S_-$, or just one, $S$. In this latter case,
let us define $S_+ = S_- = S$, so that we may use uniform notation throughout.
For a wonderfully clear introduction to Clifford algebras, including a complete
classification, see the text of Porteous \cite{Porteous}.

Since there are many different modules of $\Cliff_0(V)$, there are many
different spinor representations. Physicists distinguish some of them with
special names like `Majorana spinors' or `Weyl spinors', and we will see
some examples of these below. We do not, however, need to define these terms
precisely, because such distinctions are only important for comparing our work
to the literature. Instead, we shall see how to handle all the vectors and
spinors we need in a uniform way using normed division algebras.

So far, we have said nothing that depends on the dimension of the space of
vectors, $V$. In some some dimensions, special phenomena occur, thanks to the
existence of the normed division algebras.  A \define{normed division algebra}
is a real, possibly nonassociative algebra $\K$ with 1, equipped with a norm
$|\cdot|$ satisfying
\[ |ab| = |a||b|. \]
As with the complex numbers, this norm can be expressed using conjugation:
$|a|^2 = aa^* = a^*a$, where $* \maps \K \to \K$ is a suitable involution. By a
classic theorem of Hurwitz~\cite{Hurwitz}, there are only four
finite-dimensional normed division algebras: the real numbers, $\R$, the
complex numbers, $\C$, the quaternions, $\H$, and the octonions, $\O$.  These
algebras have dimension 1, 2, 4, and 8.  Only the octonions are
nonassociative, but mildly so: they are \define{alternative}, meaning that the
subalgebra generated by any two elements is associative.

One can use the theory of Clifford algebras to prove that normed division
algebras can only occur in these dimensions. This is a two-way street, however,
and we will traverse it the other way, using the division algebras to better
understand objects that are usually only studied with Clifford algebras:
vectors and spinors. For a division algebra $\K$ of dimension $k$, we will
mainly be interested in the vectors and spinors in Minkowski spacetime of
dimension $k+2$ or $k+3$, but we can get a taste for how this works just by
considering Euclidean space of dimension $k$. 

In this case, something remarkable happens. Namely, we can identify the vector
and irreducible spinor representations with the division algebra itself:
\[ V = \K, \quad S_+ = \K, \quad S_- = \K. \]
Because each of these representations is just $\K$ in disguise, there is
an obvious way for a vector to act on a spinor: multiplication! We define:
\[
\begin{array}{cccc}
	\cdot \maps & V \tensor S_+ & \to & S_- \\
	            & A \tensor \psi & \mapsto & A\psi 
\end{array}
\]
and 
\[
\begin{array}{cccc}
	\cdot \maps & V \tensor S_- & \to & S_+ \\
	            & A \tensor \phi & \mapsto & A^* \phi .
\end{array}
\]
Because the action of $V$ swaps the spinor spaces, it preserves their direct
sum, $S_+ \oplus S_-$. Acting on this latter space with the same vector twice,
we get:
\begin{eqnarray*}
	A \cdot A \cdot (\psi, \phi) & = & A \cdot (A^* \phi, A \psi) \\
	                             & = & ( A^* A \psi, A A^* \phi) \\
				     & = & |A|^2 (\psi, \phi).
\end{eqnarray*}
Note that nonassociativity poses no problem for us in the above calculation,
thanks to alternativity: everything in sight takes place in the subalgebra
generated by only two elements, $A$ and $\psi$.

Now, the above equation is the Clifford relation: acting twice by $A$ is the
same as multiplying by $|A|^2$. Thus the map $V \tensor (S_+ \oplus S_-) \to S_+
\oplus S_-$ induces a homomorphism:
\[ \Cliff(V) \to \End(S_+ \oplus S_-). \]
In this way, $S_+ \oplus S_-$ becomes a module of $\Cliff(V)$. Because acting
by vectors swaps $S_+$ and $S_-$, both of these subspaces are preserved by the
subalgebra $\Cliff_0(V)$ generated by products of pairs of vectors, and in this
way they become representations of $\Spin(V)$.

We thus see how the vectors and spinors in $k$-dimensional Euclidean space are
both just elements in the division algebra $\K$, albeit with different actions
of $\Spin(V)$. We can view this as a mathematical signpost that supersymmetry
is possible: physically, vectors and spinors are used to describe bosons and
fermions, so the fact that both vectors and spinors lie in division algebra in
dimension $k$ suggests there is a great deal of symmetry between bosons and
fermions in dimension $k$. Such a symmetry is precisely what supersymmetry was
invented to provide.

There is much more to this story even in Euclidean signature. But we are
interested in physics, so having had a brief taste of Euclidean space, we now
turn to Minkowski spacetime. First, in Section \ref{sec:divalg}, we review the basic
facts we need about normed division algebras. Then we develop vectors and
spinors for $(k+2)$-dimensional spacetime in Section \ref{sec:k+2}, and for
$(k+3)$-dimensional spacetime in Section \ref{sec:k+3}.

\section{Normed division algebras} \label{sec:divalg}

As we note above, in 1898 Hurwitz~\cite{Hurwitz} proved there are only four
finite-dimensional normed division algebras: the real numbers, $\R$, the
complex numbers, $\C$, the quaternions, $\H$, and the octonions, $\O$, with
dimensions  1, 2, 4, and 8, respectively. Decades later, in 1960, Urbanik and
Wright \cite{UrbanikWright} removed the finite-dimensionality condition from
this result.  For an overview of this subject, including a Clifford algebra
proof of Hurwitz's theorem, see the review by Baez \cite{Baez:Octonions}. In
this section, we focus on the tools we will need to study vectors and spinors
with division algebras later in this chapter.

Recall, a \define{normed division algebra} $\K$ is a (possibly nonassociative)
real algebra equipped with a multiplicative unit 1 and a norm $| \cdot |$
satisfying:
\[ |ab| = |a| |b|  \]
for all $a, b \in \K$.  Note this implies that $\K$ has no zero divisors.
We will freely identify $\R 1 \subseteq \K$ with $\R$.

In all cases, this norm can be defined using conjugation. Every normed
division algebra has a \define{conjugation} operator---a linear
operator $* \maps \K \to \K$ satisfying
\[ a^{**} = a, \quad (ab)^* = b^* a^* \]
for all $a,b \in \K$.   Conjugation lets us decompose each element of
$\K$ into real and imaginary parts, as follows:
\[ \Re(a) = \frac{a + a^*}{2}, \quad \Im(a) = \frac{a - a^*}{2}. \]
Conjugating changes the sign of the imaginary part and leaves the real part
fixed. We can write the norm as
\[ |a| = \sqrt{a a^*} = \sqrt{a^* a}. \]
This norm can be polarized to give an inner product on $\K$:
\[ (a, b) = \Re(a b^*) = \Re(a^* b). \]

The algebras $\R$, $\C$ and $\H$ are associative.  The octonions $\O$
are not.  Yet they come close: the subalgebra generated by any two
octonions is associative.  Another way to express this fact uses the
\define{associator}:
\[ [ a, b, c ] = (ab)c - a(bc), \]
a trilinear map $\K \tensor \K \tensor \K \to \K$.  A theorem due to
Artin \cite{Schafer} states that for any algebra, the subalgebra
generated by any two elements is associative if and only if the
associator is alternating (that is, completely antisymmetric in its
three arguments).  An algebra with this property is thus called
\define{alternative}.  The octonions $\O$ are alternative, and 
so of course are $\R$, $\C$ and $\H$: for these three the associator
simply vanishes!  

In what follows, our calculations make heavy use of the fact that
all four normed division algebras are alternative.  Besides this, the
properties we require are:

\begin{prop}
The associator changes sign when one of its entries is conjugated.
\end{prop}

\begin{proof}
Since the subalgebra generated by any two elements is associative, and
real elements of $\K$ lie in every subalgebra, $[a,b,c] = 0$ if any
one of $a,b,c$ is real.  It follows that $[a, b, c] = 
[ \Im(a), \Im(b), \Im(c)]$, which yields the desired result.
\end{proof}

\begin{prop}
\label{prop:imaginaryassociator} 
The associator is purely imaginary. 
\end{prop}

\begin{proof}
Since $(ab)^* = b^* a^*$, a calculation shows $[a,b,c]^* = -[c^*,b^*,a^*]$.  
By alternativity this equals $[a^*,b^*,c^*]$, which in
turn equals $-[a,b,c]$ by the above proposition.  So, $[a,b,c]$ is
purely imaginary.
\end{proof}

For any square matrix $A$ with entries in $\K$, we define its
\define{trace} $\tr(A)$ to be the sum of its diagonal entries. This
trace lacks the usual cyclic property, because $\K$ is noncommutative,
so in general $\tr(AB) \neq \tr(BA)$.  Luckily, taking the real part
restores this property:

\begin{prop}
Let $a$, $b$, and $c$ be elements of $\K$. Then
\[ \Re((ab)c) = \Re(a(bc)) \]
and this quantity is invariant under cyclic permutations of $a$, $b$, and
$c$.
\end{prop}

\begin{proof} 
Proposition~\ref{prop:imaginaryassociator} implies
that $\Re((ab)c) = \Re(a(bc))$. For the cyclic property, it 
then suffices to prove $\Re(ab) = \Re(ba)$.  Since $(a,b) = (b,a)$ and 
the inner product is defined by $(a, b) = \Re(ab^*) = \Re(a^* b)$, we see:
\[ \Re(ab^*) = \Re(b^*a). \]
The desired result follows upon substituting $b^*$ for $b$. 
\end{proof}

\begin{prop}
\label{prop:realtrace}
Let $A$, $B$, and $C$ be $k \times \ell$, $\ell \times m$ and
$m \times k$ matrices with entries in $\K$. Then
\[ \Retr((AB)C) = \Retr(A(BC)) \]
and this quantity is invariant under cyclic permutations of $A$, $B$,
and $C$.  We call this quantity the \define{real trace} $\Retr(ABC)$.
\end{prop}

\begin{proof}
This follows from the previous proposition and the definition of 
the trace.
\end{proof}

\section{Spacetime geometry in \emph{k}+2 dimensions} \label{sec:k+2} \label{sec:intertwiners}

We shall now see how to construct vectors and spinors for spacetimes of
dimension $k+2$ from a normed division algebra $\K$ of dimension $k$.
Most of the material for the here is well-known \cite{Baez:Octonions,
ChungSudbery, KugoTownsend,  ManogueSudbery, Sudbery}. We base our approach to
it on the papers of Manogue and Schray \cite{Schray, SchrayManogue}.  The key
facts are that one can describe vectors in $(k+2)$-dimensional Minkowski
spacetime as $2 \times 2$ hermitian matrices with entries in $\K$, and spinors
as elements of $\K^2$.  In fact there are two representations of $\Spin(k+1,1)$
on $\K^2$, which we call $S_+$ and $S_-$.  The nature of these representations
depends on $\K$:
\begin{itemize}
    \item When $\K = \R$, $S_+ \iso S_-$ is the Majorana spinor
        representation of $\Spin(2,1)$.
    \item When $\K = \C$, $S_+ \iso S_-$ is the Majorana spinor
        representation of $\Spin(3,1)$. 
    \item When $\K = \H$, $S_+$ and $S_-$ are the Weyl spinor
        representations of $\Spin(5,1)$.
    \item When $\K = \O$, $S_+$ and $S_-$ are the Majorana--Weyl
        spinor representations of $\Spin(9,1)$.
\end{itemize} 
Of course, these spinor representations are also
representations of the even part of the relevant Clifford algebras:
\vskip 1em
\begin{center}
\renewcommand{\arraystretch}{1.4}
\begin{tabular}{|lcl|}
	\hline
	\multicolumn{3}{|c|}{\textbf{Even parts of Clifford algebras}} \\
	\hline
	$\Cliff_{0}(2,1)$ & $\iso$ & $\R[2]$                \\
	$\Cliff_{0}(3,1)$ & $\iso$ & $\C[2]$                \\
	$\Cliff_{0}(5,1)$ & $\iso$ & $\H[2] \oplus \H[2]$   \\
	$\Cliff_{0}(9,1)$ & $\iso$ & $\R[16] \oplus \R[16]$ \\
	\hline
\end{tabular}
\renewcommand{\arraystretch}{1}
\end{center}
Here we see $\R^2$, $\C^2$, $\H^2$ and $\O^2$ showing up as irreducible
representations of these algebras, albeit with $\O^2$
masquerading as $\R^{16}$. The first two algebras have a unique irreducible
representation.  The last two both have two irreducible representations,
which correspond to left-handed and right-handed spinors.

Our discussion so far has emphasized the differences between the 4
cases.  But the wonderful thing about normed division algebras is that
they allow a unified approach that treats all four cases
simultaneously!  They also give simple formulas for the basic
$\Spin(k+1,1)$-equivariant operators involving vectors, spinors and scalars.

To begin, let $\K[m]$ denote the space of $m \times m$ matrices with
entries in $\K$. Given $A \in \K[m]$, define its \define{hermitian adjoint}
$A^\dagger$ to be its conjugate transpose:
\[ A^\dagger = (A^*)^T. \]
We say such a matrix is \define{hermitian} if $A = A^\dagger$. Now take the
$2 \times 2$ hermitian matrices:
\[ 
\h_2(\K) = \left\{ 
\left( 
\begin{array}{c c} 
	t + x & y     \\
	y^*   & t - x \\
\end{array}
\right) 
\; : \;
t, x \in \R, \; y \in \K
\right\}.
\]
This is an $(k + 2)$-dimensional real vector space. Moreover, the
usual formula for the determinant of a matrix gives the Minkowski norm
on this vector space:
\[ 
-\det 
\left( 
\begin{array}{c c} 
	t + x & y     \\
	y^*   & t - x \\
\end{array}
\right) 
= - t^2 + x^2 + |y|^2.
\]
We insert a minus sign to obtain the signature $(k+1, 1)$. Note this
formula is unambiguous even if $\K$ is noncommutative or nonassociative.   

It follows that the double cover of the Lorentz group, $\Spin(k + 1,
1)$, acts on $\h_2(\K)$ via determinant-preserving linear
transformations.  Since this is the `vector' representation, we will
often call $\h_2(\K)$ simply $V$.  The Minkowski metric
\[           g \maps V \otimes V \to \R  \]
is given by
\[           g(A,A) = -\det(A)  .\]
There is also a nice formula for the inner product of two different
vectors.  This involves the \define{trace reversal} of $A \in \h_2(\K)$,
defined by
\[ \tilde{A} = A - (\tr A) 1. \]
Note we indeed have $\tr(\tilde{A}) = -\tr(A)$.  Also note that
\[ A =
\left( 
\begin{array}{c c} 
	t + x & y     \\
	y^*   & t - x \\
\end{array}
\right)
\qquad \implies \qquad \tilde{A} =
\left( 
\begin{array}{c c} 
	-t + x & y     \\
	y^*    & -t - x \\
\end{array}
\right) ,
\]
so trace reversal is really \emph{time reversal}.  Moreover:

\begin{prop}
\label{prop:metric}
For any vectors $A,B \in V = \h_2(K)$, we have
\[       A \tilde{A} = \tilde{A} A = - \det(A) 1 \]
and
\[  \frac{1}{2} \Retr(A \tilde{B}) = \frac{1}{2} \Retr(\tilde{A} B) = g(A,B) .  \]
\end{prop}

\begin{proof}
We check the first equation by a quick calculation.  Taking the real
trace and dividing by 2 gives
\[      \frac{1}{2} \Retr(A \tilde{A}) = \frac{1}{2} \Retr(\tilde{A}A) = -\det(A) = g(A,A) . \]
Then we use the polarization identity, which says that two symmetric
bilinear forms that give the same quadratic form must be equal.
\end{proof}

Next we consider spinors.  As real vector spaces, the spinor
representations $S_+$ and $S_-$ are both just $\K^2$.
However, they differ as representations of $\Spin(k+1, 1)$. To
construct these representations, we begin by defining ways for vectors
to act on spinors:
\[ \begin{array}{cccl} 
	\gamma \maps & V \tensor S_+ & \to     & S_- \\
	             & A \tensor \psi     & \mapsto & A\psi 
   \end{array} \]
and 
\[ \begin{array}{cccl} 
	\tilde{\gamma} \maps & V \tensor S_- & \to     & S_+ \\
	                     & A \tensor \psi     & \mapsto & \tilde{A}\psi .
\end{array} \]
We have named these maps for definiteness, but we will also write the action of
a vector on a spinor with a dot:
\[ A \cdot \psi, \quad \psi \in S_\pm . \]
We can also think of $\gamma$ and $\tilde{\gamma}$ as maps that send elements
of $V$ to linear operators:
\[ \begin{array}{cccl}
	\gamma \maps         & V & \to & \Hom(S_+, S_-) , \\
	\tilde{\gamma} \maps & V & \to & \Hom(S_-, S_+) .
\end{array} \]
Here a word of caution is needed: since $\K$ may be nonassociative,
$2 \times 2$ matrices with entries in $\K$ cannot be identified with 
linear operators on $\K^2$ in the usual way.  They
certainly induce linear operators via left multiplication:
\[ L_A(\psi) = A \psi. \]
Indeed, this is how $\gamma$ and $\tilde{\gamma}$ turn elements of $V$ into
linear operators:
\[ \begin{array}{ccl} 
	\gamma(A)         & = & L_A, \\
	\tilde{\gamma}(A) & = & L_{\tilde{A}} .
\end{array} \]
However, because of nonassociativity, composing such linear operators is
different from multiplying the matrices:
\[ L_A L_B(\psi) = A( B \psi) \neq (AB) \psi = L_{AB} (\psi) . \]

Since vectors act on elements of $S_+$ to give elements of $S_-$
and vice versa, they map the space $S_+ \oplus S_-$ to itself.
This gives rise to an action of the Clifford algebra 
$\Cliff(V)$ on $S_+ \oplus S_-$:

\begin{prop} \label{prop:n2cliff}
The vectors $V = \h_2(\K)$ act on the spinors $S_+ \oplus S_- = \K^2
\oplus \K^2$ via the map
\[ \Gamma \maps  V \to \End(S_+ \oplus S_-) \]
given by 
\[ \Gamma(A)(\psi, \,\phi) = (\widetilde{A} \phi, \, A \psi)  .\]
Furthermore, $\Gamma(A)$ satisfies the Clifford algebra relation:
\[ \Gamma(A)^2 = g(A,A) 1 \]
and so extends to a homomorphism $\Gamma \maps \Cliff(V) \to 
\End(S_+ \oplus S_-)$, i.e.\ a representation of the Clifford algebra 
$\Cliff(V)$ on $S_+ \oplus S_-$.
\end{prop}

\begin{proof}
Suppose $A \in V$ and $\Psi = (\psi, \phi) \in S_+ \oplus S_-$. 
We need to check that 
\[ \Gamma(A)^2(\Psi) = -\det(A) \Psi . \]
Here we must be mindful of nonassociativity: we have
\[ \Gamma(A)^2(\Psi) = ( \tilde{A}(A \psi), \, A(\tilde{A} \phi)) . \]
Yet it is easy to check that the expressions $\tilde{A}(A \psi)$ 
and $A (\tilde{A} \phi)$ involve multiplying at most two 
different nonreal elements of $\K$.
These associate, since $\K$ is alternative, so in fact
\[ \Gamma(A)^2(\Psi) = ( (\tilde{A}A) \psi, \, (A \tilde{A}) \phi) . \]
To conclude, we use Proposition \ref{prop:metric}.
\end{proof}

The action of a vector swaps $S_+$ and $S_-$, so acting by vectors
twice sends $S_+$ to itself and $S_-$ to itself.  This means that
while $S_+$ and $S_-$ are \emph{not} modules for the Clifford algebra
$\Cliff(V)$, they are both modules for the even part of the Clifford
algebra, generated by products of pairs of vectors. Recalling that 
$\Spin(k+1,1)$ lives in this even part, we see that $S_+$ and $S_-$ 
are both representations of $\Spin(k+1, 1)$.

Now that we have representations of $\Spin(k+1, 1)$ on $V$, $S_+$ and
$S_-$, we need to develop the $\Spin(k+1, 1)$-equivariant maps that
relate them.  Ultimately, we need:
\begin{itemize}
	\item An invariant pairing: 
		\[ \langle -, - \rangle \maps S_+ \tensor S_- \to \R. \]
	\item An equivariant map that turns pairs of spinors into vectors:
		\[ [-,-] \maps S_\pm \tensor S_\pm \to V. \]
\end{itemize}
Another name for an equivariant map between group representations
is an `intertwining operator'.  As a first step, we show that the 
action of vectors on spinors is itself an intertwining operator:

\begin{prop} 
The maps 
\[ \begin{array}{cccl}
	\gamma \maps & V \tensor S_+ & \to     & S_- \\
                     & A \tensor \psi     & \mapsto & A \psi
\end{array} \]
and
\[ \begin{array}{cccl}
	\tilde{\gamma} \maps & V \tensor S_- & \to     & S_+ \\
	                     & A \tensor \psi    & \mapsto & \tilde{A} \psi
\end{array} \]
are equivariant with respect to the action of $\Spin(k+1, 1)$.
\end{prop}

\begin{proof}
Both $\gamma$ and $\tilde{\gamma}$ are restrictions of the map \[ \Gamma \maps
V \tensor (S_+ \oplus S_-) \to S_+ \oplus S_- ,\] so it suffices to check that
$\Gamma$ is equivariant. Indeed, an element $g \in \Spin(k+1,1)$ acts on $V$ by
conjugation on $V \subseteq \Cliff(V)$, and it acts on $S_+ \oplus S_-$ by
$\Gamma(g)$. Thus, we compute:
\[ \Gamma(gAg^{-1}) \Gamma(g)\Psi = \Gamma(g) (\Gamma(A) \Psi), \]
for any $\Psi \in S_+ \oplus S_-$. Here it is important to note that the
conjugation $gAg^{-1}$ is taking place in the associative algebra $\Cliff(V)$,
not in the algebra of matrices.  This equation says that $\Gamma$ is indeed
$\Spin(k+1,1)$-equivariant, as claimed.  
\end{proof}

Now we exhibit the key tool: the pairing between $S_+$ and $S_-$:

\begin{prop}
The pairing
\[ \begin{array}{cccl}
\langle -, - \rangle \maps & S_+ \tensor S_- & \to     & \R \\
 & \psi \tensor \phi & \mapsto & \Re(\psi^\dagger \phi)
\end{array} \]
is invariant under the action of $\Spin(k+1,1)$.
\end{prop}

\begin{proof}
Given $A \in V$, we use the fact that the associator is purely
imaginary to show that 
\[ 
\Re \left( ( \tilde{A} \phi )^\dagger ( A \psi ) \right) = 
\Re \left( (\phi^\dagger \tilde{A}) (A \psi) \right) = 
\Re \left( \phi^\dagger (\tilde{A} (A \psi)) \right).
\]
As in the proof of the Clifford relation, it is easy to check that the column
vector $\tilde{A}(A \psi)$ involves at most two nonreal elements of $\K$ and
equals $g(A,A) \psi$. So:
\[ 
\langle \tilde{\gamma}(A) \phi, \gamma(A) \psi \rangle = g(A,A) \langle \psi, \phi \rangle.
 \]
In particular when $A$ is a unit vector, acting by $A$ swaps the order of
$\psi$ and $\phi$ and changes the sign at most.  The pairing $\langle -, -
\rangle$ is thus invariant under the group in $\Cliff(V)$ generated by products
of \emph{pairs} of unit vectors, which is $\Spin(k+1, 1)$.
\end{proof}

With this pairing in hand, there is a manifestly equivariant way to turn a pair
of spinors into a vector.  Given $\psi, \phi \in S_+$, there is a unique vector
$[\psi, \phi]$ whose inner product with any vector $A$ is given by
\[ g([\psi, \phi], A) = \langle \psi, \gamma(A) \phi \rangle .\]
Similarly, given $\psi, \phi \in S_-$, we define $[\psi , \phi] \in V$ by
demanding
\[ g([\psi, \phi], A) = \langle \tilde{\gamma}(A) \psi, \phi \rangle \]
for all $A \in V$.  This gives us maps
\[ S_\pm \tensor S_\pm \to V \]
which are manifestly equivariant.

On the other hand, because $S_\pm = \K^2$ and $V = \h_2(\K)$, there is
also a naive way to turn a pair of spinors into a vector using matrix
operations: just multiply the column vector $\psi$ by the row vector
$\phi^\dagger$ and then take the hermitian part:
\[ \psi \phi^\dagger + \phi \psi^\dagger \in \h_2(\K), \]
or perhaps its trace reversal:
\[ \widetilde{\psi \phi^\dagger + \phi \psi^\dagger} \in \h_2(\K). \]
In fact, these naive guesses match the manifestly equivariant approach
described above:

\begin{prop}
The maps $[-,-] \maps S_\pm \tensor S_\pm \to V$ are given by:
\[ \begin{array}{cccl}
	[-,-] \maps & S_+ \tensor S_+   & \to     & V \\
                    & \psi \tensor \phi & \mapsto & \widetilde{\psi \phi^\dagger + \phi \psi^\dagger}
\end{array} \]
\[ \begin{array}{cccl}
	[-,-] \maps & S_- \tensor S_-   & \to     & V \\
                    & \psi \tensor \phi & \mapsto & \psi \phi^\dagger + \phi \psi^\dagger .
\end{array} \]
These maps are equivariant with respect to the action of $\Spin(k+1, 1)$.
\end{prop}

\begin{proof}
First suppose $\psi,\phi \in S_+$.  We have already seen that the map $[-,-]
\maps S_+ \tensor S_+ \to V$ is equivariant.  We only need to show that this
map has the desired form.  We start by using some definitions:
\[ g([\psi, \phi], A) = \langle \psi, \gamma(A) \phi \rangle = \Re(\psi^\dagger (A \phi)) = \Retr( \psi^\dagger A \phi) . \]
We thus have
\[     g([\psi, \phi], A) = \Retr( \psi^\dagger A \phi) = \Retr( \phi^\dagger A \psi) ,
\]
where in the last step we took the adjoint of the inside.  Applying the cyclic
property of the real trace, we obtain
\[ 
g([\psi, \phi], A) = \Retr( \phi \psi^\dagger A ) = \Retr( \psi \phi^\dagger A ). 
\]
Averaging gives
\[ 
g([\psi, \phi], A) = \half \Retr( ( \psi \phi^\dagger + \phi \psi^\dagger ) A ). 
\]
On the other hand, Proposition \ref{prop:metric} implies that
\[ g([\psi, \phi], A) = \frac{1}{2} \Retr(\widetilde{[\psi, \phi]}A) .\]
Since both these equations hold for all $A$, we must have 
\[   \widetilde{[\psi, \phi]} =  \psi \phi^\dagger + \phi \psi^\dagger. \]
Doing trace reversal twice gets us back where we started, so
\[   [\psi, \phi] = \widetilde{\psi \phi^\dagger + \phi \psi^\dagger} \]
as desired.  A similar calculation shows that if $\psi, \phi \in S_-$, then
$[\psi, \phi] = \psi \phi^\dagger + \phi \psi^\dagger$.
\end{proof}

\begin{table}[H]
	\begin{center}
\renewcommand{\arraystretch}{1.3}
		\begin{tabular}{r@{$\maps$}c@{$\tensor$}c@{$\to$}ccc}
			\hline
			\multicolumn{4}{c}{Map}                                  & Division algebra notation                           & Index notation\\
			\hline   
			$g$                     & $\, V$   & $\, V$   & $\, \R$  & $\frac{1}{2} \Retr(A \tilde{B})$                    & $A^\mu B_\mu$ \\
			$\gamma$                & $\, V$   & $\, S_+$ & $\, S_-$ & $A\psi$                                             & $A^\mu \gamma_\mu \psi$ \\
			$\tilde{\gamma}$        & $\, V$   & $\, S_-$ & $\, S_+$ & $\tilde{A}\psi$                                     & $A^\mu \tilde{\gamma}_\mu \psi$ \\
			$[-,-]$                 & $\, S_+$ & $\, S_+$ & $\, V$   & $\widetilde{\psi \phi^\dagger + \phi \psi^\dagger}$ & $\psibar \gamma^\mu \phi$ \\
			$[-,-]$                 & $\, S_-$ & $\, S_-$ & $\, V$   & $\psi \phi^\dagger + \phi \psi^\dagger$             & $\psibar \tilde{\gamma}^\mu \phi$ \\
			$\langle - , - \rangle$ & $\, S_+$ & $\, S_-$ & $\, \R$  & $\Re(\psi^\dagger \phi)$                            & $\psibar \phi$ \\
			\hline
		\end{tabular}
\caption{\label{tab:intertwiners}
Division algebra notation vs.\ index notation}
	\end{center}
\end{table}
\renewcommand{\arraystretch}{1}
\noindent 
We can summarize our work so far with a table of the basic bilinear
maps involving vectors, spinors and scalars.  Table 1 shows how to
translate between division algebra notation and something more closely
resembling standard physics notation.  In this table the adjoint
spinor $\psibar$ denotes the spinor dual to $\psi$ under the pairing
$\langle -, - \rangle$.  The gamma matrix $\gamma^\mu$ denotes a
Clifford algebra generator acting on $S_+$, while $\tilde{\gamma}^\mu$
denotes the same element acting on $S_-$.  Of course $\tilde{\gamma}$
is not standard physics notation; the standard notation for this
depends on which of the four cases we are considering: $\R$, $\C$, $\H$ or 
$\O$.

\section{Spacetime geometry in \emph{k}+3 dimensions} \label{sec:k+3}

In the last section we recalled how to describe spinors and vectors in
$(k+2)$-dimensional Minkowski spacetime using a division algebra $\K$
of dimension $k$.  Here we show how to boost this up one dimension,
and give a division algebra description of vectors and spinors in
$(k+3)$-dimensional Minkowski spacetime.

We shall see that vectors in $(k+3)$-dimensional Minkowski spacetime can
be identified with $4 \times 4$ $\K$-valued matrices of 
this particular form:
\[ \left( 
\begin{array}{cc} 
	a & \tilde{A} \\
	A & -a
\end{array}
\right)
\]
where $a$ is a real multiple of the $2 \times 2$ identity matrix
and $A$ is a $2 \times 2$ hermitian matrix with entries in $\K$.
Moreover, $\Spin(k+2, 1)$ has a representation on $\K^4$,
which we call $\S$.  Depending on $\K$, this gives the following
types of spinors: 
\begin{itemize}
    \item When $\K = \R$, $\S$ is the Majorana spinor representation of
	    $\Spin(3,1)$.
    \item When $\K = \C$, $\S$ is the Dirac spinor representation of
	    $\Spin(4,1)$. 
    \item When $\K = \H$, $\S$ is the Dirac spinor representation of
	    $\Spin(6,1)$.
    \item When $\K = \O$, $\S$ is the Majorana spinor representation of
	    $\Spin(10,1)$.
\end{itemize} 
Again, these spinor representations are also representations of the even part
of the relevant Clifford algebra:
\vskip 1em
\begin{center}
\renewcommand{\arraystretch}{1.4}
\begin{tabular}{|lcl|}
	\hline
	\multicolumn{3}{|c|}{\textbf{Even parts of Clifford algebras}} \\
	\hline
	$\Cliff_{0}(3,1)$  & $\iso$ & $\C[2]$                 \\
	$\Cliff_{0}(4,1)$  & $\iso$ & $\H[2]$                 \\
	$\Cliff_{0}(6,1)$  & $\iso$ & $\H[4]$    \\
	$\Cliff_{0}(10,1)$ & $\iso$ & $\R[32]$ \\
	\hline
\end{tabular}
\renewcommand{\arraystretch}{1}
\end{center}
These algebras have irreducible representations on $\R^4 \iso \C^2$, 
$\C^4 \iso \H^2$, $\H^4$ and $\O^4 \iso \R^{32}$, respectively.  

The details can be described in a uniform way for all four cases.  We
take as our space of `vectors' the following $(k+3)$-dimensional
subspace of $\K[4]$:
\[ \V = 
\left\{ 
\left( 
\begin{array}{cc} 
	a & \tilde{A} \\
	A & -a
\end{array}
\right)
: a \in \R, \quad A \in \h_2(\K)
\right\} .
\]
In the last section, we defined vectors in $k+2$ dimensions to be $V =
\h_2(\K)$. That space has an obvious embedding into $\V$, given by
\[ \begin{array}{rcl}
	V & \inclusion & \V \\
	A & \mapsto    & \left( \begin{matrix} 0 & \tilde{A} \\ A & 0 \end{matrix} \right) .
\end{array} \]
The Minkowski metric 
\[ h \maps \V \tensor \V \to \R \]
is given by extending the Minkowski metric $g$ on $V$:
\[ h \left( \left( \begin{smallmatrix} a & \tilde{A} \\ A & -a \end{smallmatrix} \right), \left( \begin{smallmatrix} a & \tilde{A} \\ A & -a \end{smallmatrix} \right) \right) = g(A,A) + a^2 . \]
From our formulas for $g$, we can derive formulas for $h$:

\begin{prop} \label{prop:metric2}
	For any vectors $\A,\B \in \V \subseteq \K[4]$, we have
	\[ \A^2 = h(\A, \A)1 \]
	and
	\[ \fourth \Retr(\A \B) = h(\A, \B). \]
\end{prop}
\begin{proof}
	For $\A = \left( \begin{smallmatrix} a & \tilde{A} \\ A & -a
	\end{smallmatrix} \right)$, it is easy to check:
	\[ \A^2 = \left( \begin{matrix} a^2 + \tilde{A}A & 0 \\ 0 & A\tilde{A} + a^2 \end{matrix} \right) . \]
	By Proposition \ref{prop:metric}, we have $A \tilde{A} = \tilde{A} A =
	g(A,A)1$, and substituting this in establishes the first formula. The
	second formula follows from polarizing and taking the real trace of
	both sides.
\end{proof}

Define a space of `spinors' by $\S = S_+ \oplus S_- = \K^4$. To distinguish 
elements of $\V$ from elements of $\h_2(\K)$, we will denote them with
calligraphic letters such as $\A$ and $\B$.  Similarly, to distinguish
elements of $\S$ from $S_\pm$, we will denote them with capital Greek letters
such as $\Psi$ and $\Phi$.

Elements of $\V$ act on $\S$ by left multiplication:
\[ \begin{array}{rcl}
	\V \tensor \S & \to     & \S \\
	\A \tensor \Psi        & \mapsto & \A \Psi .
\end{array} \]
We can dualize this to get a map:
\[ \begin{array}{cccl}
	\Gamma \maps & \V & \to     & \End(\S) \\
	             & \A & \mapsto & L_\A .
\end{array} \]
This induces the Clifford action of $\Cliff(\V)$ on $\S$.
Note that this $\Gamma$ is the same as the map in Proposition
\ref{prop:n2cliff} when we restrict to $V \subseteq \V$.

\begin{prop}
The vectors $\V \subseteq \K[4]$ act on the spinors $\S = \K^4$ via the map
\[ \Gamma \maps \V \to \End(\S) \]
given by 
\[   \Gamma(\A)\Psi = \A \Psi. \]
Furthermore, $\Gamma(\A)$ satisfies the Clifford algebra relation:
\[ \Gamma(\A)^2 = h(\A,\A) 1 \]
and so extends to a homomorphism $\Gamma \maps \Cliff(\V) \to \End(\S)$, i.e.\ a
representation of the Clifford algebra $\Cliff(\V)$ on $\S$.
\end{prop}
\begin{proof}
Here, we must be mindful of nonassociativity. For $\Psi = (\psi, \phi) \in
\S$ and $\A = \left( \begin{smallmatrix} a & \tilde{A} \\ A &
-a \end{smallmatrix} \right) \in \V$, we have:
\[ \Gamma(\A)^2 \Psi = \A ( \A \Psi) \]
which works out to be:
\[ \Gamma(\A)^2 \Psi = \left( \begin{array}{c} a^2 \psi + \tilde{A}(A\psi) \\ A(\tilde{A}\phi) + a^2 \phi \end{array} \right). \]
A quick calculation shows that the expressions $\tilde{A}(A\psi)$ and
$A(\tilde{A} \phi)$ involve at most two nonreal elements of $\K$, so everything
associates and we can write: 
\[ \Gamma(\A)^2 \Psi = \A^2 \Psi \]
By Proposition \ref{prop:metric2}, we are done.
\end{proof}

This tells us how $\S$ is a module of $\Cliff(\V)$, and thus a representation
of $\Spin(\V)$, the subgroup of $\Cliff(\V)$ generated by products of pairs of
unit vectors. 

In the last section, we saw how to construct a $\Spin(V)$-invariant pairing
\[ \langle -,- \rangle \maps S_+ \tensor S_- \to \R. \]
We can use this to build up to a $\Spin(\V)$-invariant pairing on
$\S$:
\[ \langle (\psi, \phi), (\chi, \theta) \rangle = \langle \psi, \theta \rangle - \langle \chi, \phi \rangle \]
To see this, let
\[ \Gamma^0 = \left( \begin{matrix} 0 & -1 \\ 1 & 0 \end{matrix} \right) \]
Then, because $\langle \psi, \phi \rangle = \Re(\psi^\dagger \phi)$, it is easy
to check that:
\[ \langle \psi, \theta \rangle - \langle \chi, \phi \rangle = \Re \left( \left( \begin{matrix} \psi \\ \phi \end{matrix} \right)^\dagger \Gamma^0 \left( \begin{matrix} \chi \\ \theta \end{matrix} \right) \right). \]
We can show this last expression is invariant by explicit calculation.
\begin{prop} \label{prop:bilinearform}
	Define the nondegenerate skew-symmetric bilinear form
	\[ \langle -,- \rangle \maps \S \tensor \S \to \R \]
	by
	\[ \langle \Psi, \Phi \rangle = \Re(\Psi^\dagger \Gamma^0 \Phi). \]
	This form is invariant under $\Spin(\V)$.
\end{prop}
\begin{proof}
	It is easy to see that, for any spinors $\Psi, \Phi \in \S$ and vectors
	$\A \in \V$, we have
	\[ \langle \A \Psi, \A \Phi \rangle = \Re \left( ( \Psi^\dagger \A^\dagger ) \Gamma^0 ( \A \Phi ) \right) = \Re \left( \Psi^\dagger ( \A^\dagger \Gamma^0 ( \A \Phi ) ) \right) ,\]
	where in the last step we have used Proposition \ref{prop:realtrace}.
	Now, given that
	\[ \A = \left( \begin{matrix} a & \tilde{A} \\ A & -a \end{matrix} \right) , \]
	a quick calculation shows:
	\[ \A^\dagger \Gamma^0 = -\Gamma^0 \A. \]
	So, this last expression becomes:
	\[ -\Re \left( \Psi^\dagger ( \Gamma^0 \A ( \A \Phi ) ) \right) = -\Re \left( \Psi^\dagger ( \Gamma^0 \Gamma(\A)^2 \Phi ) ) \right) = -|\A|^2 \Re \left( \Psi^\dagger \Gamma^0 \Phi \right) , \]
	where in the last step we have used the Clifford relation. Summing up,
	we have shown: 
	\[ \langle \A \Psi, \A \Phi \rangle = -|\A|^2 \langle \Psi, \Phi \rangle . \]
	In particular, when $\A$ is a unit vector, acting by
	$\A$ changes the sign at most. Thus, $\langle -,- \rangle$ is
	invariant under the group generated by products of pairs of unit
	vectors, which is $\Spin(\V)$. It is easy to see that it is
	nondegenerate, and it is skew-symmetric because $\Gamma^0$ is.
\end{proof}

With the form $\langle -,- \rangle$ in hand, there is a manifestly equivariant
way to turn a pair of spinors into a vector.  Given $\Psi, \Phi \in \S$, there
is a unique vector $[\Psi, \Phi]$ whose inner product with any vector $\A$ is
given by
\[ h([\Psi, \Phi], \A) = \langle \Psi, \Gamma(\A) \Phi \rangle .\]

It will be useful to have an explicit formula for this operation:

\begin{prop} 
\label{prop:bracket}
Given $\Psi = (\psi_1, \psi_2)$ and $\Phi = (\phi_1, \phi_2)$ in
$\S = S_+ \oplus S_-$, we have:
\[ [\Psi, \Phi] = \left( 
\begin{array}{cc}
	\langle \psi_1, \phi_2 \rangle + \langle \phi_1, \psi_2 \rangle & -\widetilde{[\psi_1 , \psi_2]} + \widetilde{[\phi_1, \phi_2]} \\
	-[\psi_1 , \psi_2] + [\phi_1, \phi_2]                           & -\langle \psi_1, \phi_2 \rangle - \langle \phi_1, \psi_2 \rangle \\
\end{array} \right) .
\]
\end{prop}

\begin{proof}
	Decompose $\V$ into orthogonal subspaces:
	\[ \V = \left\{ \left( \begin{matrix} 0 & \tilde{A} \\ A & 0 \end{matrix} \right) : A \in V \right\} \oplus \left\{ \left( \begin{matrix} a & 0 \\ 0 & -a \end{matrix} \right) : a \in \R \right\} \]
	The first of these is just a copy of $V$, a $(k+2)$-dimensional
	Minkowski spacetime. The second is the single extra spatial dimension
	in our $(k+3)$-dimensional Minkowski spacetime, $\V$.

	Now, use the definition of $[\Psi, \Phi]$, but restricted to $V$. It is
	easy to see that, for any vector $A \in V$, we have:
	\[ h([\Psi, \Phi], A) = -\langle \psi_1, \gamma(A) \phi_1 \rangle + \langle \tilde{\gamma}(A) \psi_2, \phi_2 \rangle \]
	Letting $B$ be the component of $[\Psi, \Phi]$ which lies in $V$, this
	becomes:
	\[ g(B, A) =  -\langle \psi_1, \gamma(A) \phi_1 \rangle + \langle \tilde{\gamma}(A) \psi_2, \phi_2 \rangle. \]
	Note that we have switched to the metric $g$ on $V$, to which $h$
	restricts. By definition, this is the same as:
	\[ g(B, A) =  g(-[\psi_1, \phi_1] + [\psi_2, \phi_2], A). \]
	Since this holds for all $A$, we must have $B = -[\psi_1, \phi_1] +
	[\psi_2, \phi_2]$.

	It remains to find the component of $[\Psi, \Phi]$ orthogonal to
	$B$. Since $\left\{ \left( \begin{smallmatrix} a & 0 \\ 0 & -a
	\end{smallmatrix} \right) : a \in \R \right\}$ is 1-dimensional, this
	is merely a number. Specifically, it is the constant of proportionality
	in the expression:
	\[ h \left([\Psi, \Phi], \left( \begin{smallmatrix} a & 0 \\ 0 & -a \end{smallmatrix} \right) \right) = a ( \langle \psi_1, \phi_2 \rangle + \langle \phi_1, \psi_2 \rangle ) . \]
	Thus, this component is $\langle \psi_1, \phi_2 \rangle + \langle
	\phi_1, \psi_2 \rangle$. Putting everything together, we get 
	\[ [\Psi, \Phi] = \left( 
	\begin{array}{cc}
		\langle \psi_1, \phi_2 \rangle + \langle \phi_1, \psi_2 \rangle & -\widetilde{[\psi_1 , \psi_2]} + \widetilde{[\phi_1, \phi_2]} \\
		-[\psi_1, \psi_2] + [\phi_1, \phi_2]                            & -\langle \psi_1, \phi_2 \rangle - \langle \phi_1, \psi_2 \rangle \\
	\end{array} \right) .
	\]
\end{proof}

\section{The spinor identities} \label{sec:spinors}

We now prove crucial identities involving spinors in spacetimes of dimension
$k+2$ and $k+3$. In a sense, this solves our puzzle concerning how division
algebras are related to string theory and 2-brane theory: the spinor identities
allow the construction of WZW terms for these theories, thus guaranteeing they
have Siegel symmetry. Siegel symmetry forces the bosonic and fermionic degrees
of freedom to match, so it is essential for supersymmetry. In dimensions 10 and
11, Siegel symmetry also constrains the background of spacetime to be that of
supergravity. Yet, in solving the puzzle, we uncover new questions. What, for
instance, is the significance of these spinor identities? We shall see, in the
remainder of this thesis, that these identities lead the way to higher gauge
theory.

The first identity we shall prove holds in spacetimes of dimension 3, 4, 6 and
10. This identity appears in several guises in the physics literature. Besides
the role it plays in string theory, we shall see in Chapter \ref{ch:sym} that
it implies the supersymmetry of super-Yang--Mills theories in these dimensions.

Let us see the various forms this identity can take. In dimensions 3, 4, 6 and
10, we have what Schray \cite{Schray} has dubbed the \define{3-$\psi$'s rule}:
\[ [\psi, \psi] \cdot \psi = 0 , \]
for all spinors $\psi \in S_+$, where the bracket is the bilinear map defined
in Proposition \ref{prop:bracket}. That is, a spinor squared to a vector and
then acting on itself vanishes. It is also common to see this cubic form
polarized to obtain a trilinear in three spinors:
\[ [\psi, \phi] \cdot \chi + [\chi, \psi] \cdot \phi + [\phi, \chi] \cdot \psi = 0. \]
A more geometric interpretation, emphasized by Deligne and Freed
\cite{DeligneFreed}, is that spinors square to null vectors in these
special dimensions:
\[ |[\psi,\psi]|^2 = 0 , \]
where $|A|^2 = g(A,A)$ is the quadratic form associated with the Minkowski
inner product.  On the other hand, physicists prefer to write all of these
expressions using gamma matrices. Referring to Table \ref{tab:intertwiners}, we
write components of the vector $[\psi, \psi]$ as $\psibar \gamma^\mu \psi$ when
$\psi \in S_+$.  These identities then become, respectively:
\[ (\psibar \gamma^\mu \psi) \gamma_\mu \psi = 0, \]
\[ (\psibar \gamma^\mu \phi) \gamma_\mu \chi + (\chibar \gamma^\mu \psi) \gamma_\mu \phi + (\phibar \gamma^\mu \chi) \gamma_\mu \psi = 0, \]
and
\[ (\psibar \gamma^\mu \psi)(\psibar \gamma_\mu \psi) = 0. \]
Finally, it also common for the spinors to be removed from the second identity,
to obtain an equivalent expression in terms of gamma matrices alone. We now
establish that these are all equivalent. In fact, this is a consequence of the
following symmetries:
\begin{prop} \label{prop:symmetries}
	For any spinors $\psi, \phi, \chi, \theta \in S_+$, the 4-linear
	expression $\langle \theta, [\psi, \phi] \chi \rangle$ is symmetric
	under the exchange of the unbracketed spinors:
	\[ \langle \theta, [\psi, \phi] \cdot \chi \rangle = \langle \chi, [\psi, \phi] \cdot \theta \rangle \]
	and under the exchange of the bracketed and unbracketed spinors:
	\[ \langle \theta, [\psi, \phi] \cdot \chi \rangle = \langle \psi, [\theta, \chi] \cdot \phi \rangle \]
	Similarly, for any spinors $\psi, \phi, \chi, \theta \in S_-$, the 4-linear
	expression $\langle [\psi, \phi] \cdot \theta, \chi \rangle$ is symmetric
	under the exchange of the unbracketed spinors:
	\[ \langle [\psi, \phi] \cdot \theta,  \chi \rangle = \langle [\psi, \phi] \cdot \chi, \theta \rangle \]
	and under the exchange of the bracketed and unbracketed spinors:
	\[ \langle [\psi, \phi] \cdot \theta, \chi \rangle = \langle [\theta, \chi] \cdot \psi, \phi \rangle \]

\end{prop}
\begin{proof}
	We shall prove this in the case where all spinors are in $S_+$, the
	case of $S_-$ being almost identical. Using the definition of bracket,
	write:
	\[ \langle \theta, [\psi, \phi] \cdot \chi \rangle = g([\theta, \chi], [\psi,\phi]). \]
	The first formula then follows from the symmetry of the bracket
	$[\theta, \chi]$. The second formula follows from the definition of the
	bracket $[\psi, \phi]$.
\end{proof}
\begin{prop} \label{prop:variants}
	The following are equivalent:
	\begin{enumerate}
		\item $[\psi, \psi] \cdot \psi = 0$ for all $\psi \in S_\pm$.
		\item $[\psi, \phi] \cdot \chi + [\chi, \psi] \cdot \phi + [\phi, \chi] \cdot \psi = 0$ for all $\psi, \phi, \chi \in S_\pm$
		\item $|[\psi, \psi]|^2 = 0$ for all $\psi \in S_\pm$.
	\end{enumerate}
\end{prop}
\begin{proof}
Again, we shall prove this in the case where all spinors are in $S_+$, the case
of $S_-$ being virtually identical. Because the bracket is symmetric, the
trilinear expression 
\[ [\psi, \phi] \cdot \chi + [\chi, \psi] \cdot \phi + [\phi, \chi] \cdot \psi \]
is totally symmetric in its three arguments. Just as a symmetric bilinear
vanishes if and only if the associated quadratic form vanishes, a symmetric
trilinear vanishes if and only if the associated cubic form does. In this case,
that cubic form is, up to a numerical factor:
\[ [\psi, \psi] \cdot \psi. \]
So (1) holds if and only if (2) does.

On the other hand, we can use the symmetries of Proposition
\ref{prop:symmetries} to show that the following expression is symmetric in all
four spinors:
\[ \langle \theta, [\psi, \phi] \cdot \chi + [\chi, \psi] \cdot \phi + [\phi, \chi] \cdot \psi \rangle. \]
Statement (2) holds if and only if this expression vanishes for all $\theta$,
but this totally symmetric 4-linear vanishes if and only if the associated
quartic form vanishes. In this case, that quartic form is, up to a multiplicative
factor:
\[ \langle \psi, [\psi, \psi] \cdot \psi \rangle = g([\psi, \psi], [\psi, \psi]). \] 
Thus, (2) holds if and only if (3) holds.
\end{proof}

We now prove the 3-$\psi$'s rule. Note that it is really the alternative law,
rather than any division algebra axioms, that does the work. Our proof is based
on an argument in the appendix to a paper by Dray, Janesky and Manogue
\cite{DrayJaneskyManogue}.
\begin{thm}  
\label{thm:3psirule}
Suppose $\psi \in S_\pm$.  Then $[\psi, \psi] \cdot \psi = 0$. In other words,
$[\psi, \psi] \psi = 0$ for $\psi \in S_+$, and $\widetilde{[\psi, \psi]}
\psi = 0$ for $\psi \in S_-$.
\end{thm}
\begin{proof}
Let $\psi \in S_+$. By definition,
\[     [\psi, \psi] \psi =  2(\widetilde{\psi \psi^\dagger}) \psi 
= 2(\psi \psi^\dagger - \tr(\psi \psi^\dagger) 1) \psi .\]
It is easy to check that $\tr (\psi \psi^\dagger) = \psi^\dagger \psi$,
so 
\[     [\psi, \psi] \psi =  
2((\psi \psi^\dagger) \psi - (\psi^\dagger \psi) \psi ).\]
Since $\psi^\dagger \psi$ is a real number, it commutes with $\psi$:
\[     [\psi, \psi] \psi =  
2((\psi \psi^\dagger) \psi - \psi (\psi^\dagger \psi) ) .\]
Since $\K$ is alternative, every subalgebra of $\K$ generated by two elements
is associative.  The spinor $\psi \in \K^2$ is built from just two elements of
$\K$, so the right-hand side vanishes. The proof for the second case is
similar.
\end{proof}

Similarly, spinors in dimension 4, 5, 7 and 11 satisfy a related identity,
written in gamma matrix notation as follows:
\[ \Psibar \Gamma_{ab} \Psi \Psibar \Gamma^b \Psi = 0 \]
This identity shows up in two prominent places in the physics literature.
First, it is required for the existence of 2-brane theories in these dimensions
\cite{AchucarroEvans, Duff}. This is because it allows the construction of a
Wess--Zumino--Witten term for these theories, which give these theories
Siegel symmetry.  

Yet it is known that 2-branes in 11 dimensions are intimately connected to
supergravity. Indeed, the Siegel symmetry imposed by the WZW term constrains
the 2-brane background to be that of 11-dimensional supergravity \cite{Tanii}.
So it should come as no surprise that this spinor identity also plays a
crucial role in supergravity, most visibly in the work of D'Auria and 
Fr\'e \cite{DAuriaFre} and subsequent work by Sati, Schreiber and Stasheff 
\cite{SSS}.  

This identity is equivalent to the \define{4-$\Psi$'s rule}:
\[ [\Psi, [\Psi , \Psi] \Psi] = 0 . \]
To see this, note that we can turn a pair of spinors $\Psi$ and $\Phi$ into a
2-form, $\Psi * \Phi$.  This comes from the fact that we can embed bivectors inside
the Clifford algebra $\Cliff(\V)$ via the map
\[ \A \wedge \B \mapsto \A\B - \B\A \in \Cliff(\V). \]
These can then act on spinors using the Clifford action. Thus, define:
\begin{equation}
\label{eq:star}
 (\Psi * \Phi) (\A, \B) = \langle \Psi, (\A \wedge \B) \Phi \rangle. 
\end{equation}
But when $\Psi = \Phi$, we can simplify this using the Clifford relation:
\begin{eqnarray*}
	(\Psi * \Psi) (\A, \B) & = & \langle \Psi, (\A \B - \B \A) \Psi \rangle \\
	                       & = & \langle \Psi, 2\A \B \Psi \rangle - \langle \Psi, \Psi \rangle h(\A, \B) \\
			       & = & 2 \langle \Psi, \A \B \Psi \rangle
\end{eqnarray*}
where we have used the skew-symmetry of the form. The index-ridden identity
above merely says that inserting the vector $[\Psi, \Psi]$ into one slot
of the 2-form $\Psi * \Psi$ is zero, no matter what goes into the other slot:
\[ (\Psi * \Psi) (\A, [\Psi, \Psi]) = 2 \langle \Psi, \A[\Psi, \Psi] \Psi \rangle  = 0\]
for all $\A$. By the definition of the bracket, this is the same as
\[ 2 h( [\Psi, [\Psi, \Psi] \Psi], \A) = 0 \]
for all $\A$. Thus, the index-ridden identity is equivalent to:
\[ [\Psi, [\Psi, \Psi] \Psi] = 0 \]
as required.

Now, let us prove this. As we shall see, we use both the 3-$\psi$'s rule and
the cyclic property of the real trace in the proof, both of which follow from
our use of normed division algebras. For another derivation of the 4-$\Psi$'s
rule using division algebras, see the paper of Foot and Joshi \cite{FootJoshi}.
\begin{thm}
Suppose $\Psi \in \S$.  Then $[\Psi, [\Psi, \Psi] \Psi] = 0$. 
\end{thm}

\begin{proof}
Let $\Psi = (\psi, \phi)$.  By Proposition \ref{prop:bracket},
\[ [\Psi, \Psi] = \left( \begin{array}{cc} 2 \langle \psi, \phi \rangle & -\widetilde{[\psi, \psi]} + \widetilde{[\phi, \phi]} \\
					   -[\psi, \psi] + [\phi, \phi] & -2 \langle \psi, \phi \rangle  \\
				      \end{array} \right) \]
and thus
\[ [\Psi, \Psi] \Psi = \left( \begin{array}{c} 2 \langle \psi, \phi \rangle \psi - \widetilde{[\psi, \psi]} \phi + \widetilde{[\phi, \phi]} \phi \\
						    -[\psi, \psi]\psi + [\phi, \phi]\psi - 2\langle \psi, \phi \rangle \phi  \\
				      \end{array} \right). \]
Both $[\psi, \psi]\psi = 0$ and $\widetilde{[\phi, \phi]} \phi = 0$ by
the 3-$\psi$'s rule, Theorem \ref{thm:3psirule}.  So:
\[ [\Psi, \Psi] \Psi = \left( \begin{array}{c} 2 \langle \psi, \phi \rangle \psi - \widetilde{[\psi, \psi]} \phi \\ 
					       \ [\phi, \phi] \psi - 2 \langle \psi, \phi \rangle \phi \\ 
				      \end{array} \right). \]
The resulting matrix for $[\Psi, [\Psi, \Psi] \Psi]$ is large
and unwieldy, so we shall avoid writing it out. Fortunately, all
we really need is the $(1,1)$ entry. Recall, this is the component of the
vector $[\Psi, [\Psi, \Psi] \Psi]$ that is orthogonal to the
subspace $V \subset \V$. Call this component $a$.  A calculation shows:
\begin{eqnarray*}
a & = & \langle \psi, [\phi, \phi] \psi \rangle - \langle \widetilde{[\psi, \psi]} \phi, \phi \rangle \\
  & = & \Retr (\psi^\dagger (2 \phi \phi^\dagger) \psi)  - \Retr (\phi^\dagger (2 \psi \psi^\dagger) \phi ) \\
  & = & 0 ,
\end{eqnarray*}
where the two terms cancel by the cyclic property of the real trace, 
Proposition \ref{prop:realtrace}.   Thus, this
component of the vector $[\Psi, [\Psi, \Psi] \Psi]$ vanishes.
But since the map $\Psi \mapsto [\Psi, [\Psi, \Psi] \Psi]$ is 
equivariant with respect to the action of $\Spin(\V)$, and $\V$ is an 
irreducible representation of this group, it follows that all components of 
this vector must vanish.
\end{proof}

\chapter{Supertranslation algebras and their cohomology} \label{ch:supertranslations}

\section{Superalgebra} \label{sec:superalgebra}

So far we have used normed division algebras to construct a number of algebraic
structures: vectors as elements of $\h_2(\K)$ or $\K[4]$, spinors as elements
of $\K^2$ or $\K^4$, and the various bilinear maps involving vectors, spinors,
and scalars.  However, to describe supersymmetry, we also need superalgebra.
Specifically, we need anticommuting spinors.  Physically, this is because
spinors are fermions, so we need them to satisfy anticommutation relations.
Mathematically, this means that we will do our algebra in the category of
`super vector spaces', SuperVect, rather than the category of vector spaces,
Vect.

A \define{super vector space} is a $\Z_2$-graded vector space $V = V_0
\oplus V_1$ where $V_0$ is called the \define{even} or
\define{bosonic} part, and $V_1$ is called the \define{odd} or
\define{fermionic} part. Like Vect, SuperVect is a symmetric monoidal 
category \cite{BaezStay}.  It has:
\begin{itemize}
	\item $\Z_2$-graded vector spaces as objects;
	\item Grade-preserving linear maps as morphisms;
	\item A tensor product $\tensor$ that has the following grading: if $V
		= V_0 \oplus V_1$ and $W = W_0 \oplus W_1$, then $(V \tensor
		W)_0 = (V_0 \tensor W_0) \oplus (V_1 \tensor W_1)$ and 
              $(V \tensor W)_1 = (V_0 \tensor W_1) \oplus (V_1 \tensor W_0)$;
	\item A braiding
		\[ B_{V,W} \maps V \tensor W \to W \tensor V \]
		defined as follows: $v \in V$ and $w \in W$ 
              are of grade $p$ and $q$, then
		\[ B_{V,W}(v \tensor w) = (-1)^{pq} w \tensor v. \]
\end{itemize}
The braiding encodes the `the rule of signs': in any calculation, when
two odd elements are interchanged, we introduce a minus sign.  

There is an obvious notion of direct sums for super vector spaces, with
\[         (V \oplus W)_0 = V_0 \oplus W_0  , \qquad 
           (V \oplus W)_1 = V_1 \oplus W_1 \]
and also an obvious notion of duals, with
\[         (V^*)_0 = (V_0)^*, \qquad (V^*)_1 = (V_1)^*  .\]
We say a super vector space $V$ is \define{even} if it equals its 
even part ($V = V_0$), and \define{odd} if it equals its odd part 
($V = V_1$).  Any subspace $U \subseteq V$ of an even (resp.\ odd) 
super vector space becomes a super vector space which is again even
(resp.\ odd). 

It is noteworthy that treating division algebras as \emph{odd} is compatible
with the physical applications of this thesis. For instance, this turns out to
force the spinor representations $S_\pm$ to be odd and the vector
representation $V$ to be even, as follows.

We treat the spinor representations $S_\pm$ as super vector spaces
using the fact that they are the direct sum of two copies of $\K$.
Since $\K$ is odd, so are $S_+$ and $S_-$.  Since $\K^2$ is odd, so is
its dual.  This in turn forces the space of linear maps from $\K^2$ to
itself, $\End(\K^2) = \K^2 \tensor (\K^2)^*$, to be even. This even
space contains the $2 \times 2$ matrices $\K[2]$ as the subspace of
maps realized by left multiplication:
\[ \begin{array}{rcl}
\K[2] & \inclusion & \End(\K^2) \\
A     & \mapsto         & L_A .
\end{array} \]
$\K[2]$ is thus even. Finally, this forces the subspace of hermitian
$2 \times 2$ matrices, $\h_2(\K)$, to be even.  So, the vector
representation $V$ is even.  All this matches the usual rules in
physics, where spinors are fermionic and vectors are bosonic.

\section{Cohomology of Lie superalgebras} \label{sec:cohomology}

We now fuse the vectors and spinors we described with division algebras into a
single structure.  In any dimension, a symmetric bilinear intertwining operator
that eats two spinors and spits out a vector gives rise to a `super-Minkowski
spacetime' \cite{Deligne}.  The infinitesimal translation symmetries of this
object form a Lie superalgebra, called the `supertranslation algebra', $\T$.  The
cohomology of this Lie superalgebra is interesting and apparently rather
subtle \cite{Brandt, MSX}.  We shall see that its 3rd cohomology is
nontrivial in dimensions $k+2=3$, 4, 6 and 10, thanks to the 3-$\psi$'s rule.
algebras.  Similarly, its 4th cohomology is nontrivial in dimensions $k+3 = 4$, 5, 7
and 11, thanks to the 4-$\Psi$'s rule.

For arbitrary superspacetimes, the cohomology of $\T$ is not explicitly known.
Techniques to compute it have been described by Brandt \cite{Brandt}, who
applied them in dimension 5 and below. Movshev, Schwarz and Xu \cite{MSX}
showed how to augment these techniques using the computer algebra system LiE
\cite{LiE}, and fully describe the cohomology in dimension 6 and 10 in this
way.

Based on the work of these authors, it seems likely that the 3rd and 4th
cohomology of $\T$ is nontrivial in sufficiently large dimensions. We
conjecture, however, that dimensions $k+2$ and $k+3$ are the only ones with
\emph{Lorentz invariant} 3- and 4-cocycles. Exploratory calculations with LiE
bare this conjecture out, but the general answer appears to be unknown.

Now for some definitions. Briefly, a \define{Lie superalgebra} is a Lie algebra
in the category of super vector spaces. More concretely, it is a super vector
space $\g$ equipped with a graded-antisymmetric bracket:
\[ [-,-] \maps \Lambda^2 \g \to \g, \]
that satisfies the Jacobi identity up to some signs:
\[ [X,[Y,Z]] = [ [X,Y], Z] + (-1)^{|X||Y|} [Y,[X,Z]], \]
for homogeneous $X, Y, Z \in \g$. Here, $\Lambda^2 \g$ is the exterior square
of $\g = \g_0 \oplus \g_1$ as a super vector space. As an ordinary vector
space, 
\[ \Lambda^2 \g \iso \Lambda^2 \g_0 \, \oplus \, \g_0 \tensor \g_1 \, \oplus \, \Sym^2 \g_1 , \]
thanks to the rule of signs.

We will be concerned with several Lie superalgebras in this thesis. However,
one of the most important is also one of the most simple. Take $V$ to be the
space of vectors in Minkowski spacetime in any dimension, and take $S$ to be
any spinor representation in this dimension.  Suppose that there is a symmetric
equivariant bilinear map:
\[ [-,-] \maps S \tensor S \to V. \]
Form a super vector space $\T$ with 
\[ \T_0 = V , \qquad \T_1 = S .\]
We make $\T$ into a Lie superalgebra, the \define{supertranslation algebra}, by
giving it a suitable bracket operation.  This bracket will be zero except when
we bracket a spinor with a spinor, in which case it is simply $[-,-]$.  Since
this is symmetric and spinors are odd, the bracket operation is
graded-antisymmetric overall. Furthermore, the Jacobi identity holds trivially,
thanks to the near triviality of the bracket.  Thus $\T$ is indeed, a Lie
superalgebra.

Despite the fact that $\T$ is nearly trivial, its cohomology is not.  To see
this, we must first recall how to generalize Chevalley--Eilenberg cohomology
\cite{AzcarragaIzquierdo, ChevalleyEilenberg} from Lie algebras to Lie
superalgebras \cite{Leites}. Suppose $\g$ is a Lie superalgebra and $R$ is a
representation of $\g$. That is, $R$ is a supervector space equipped with a Lie
superalgebra homomorphism $\rho \maps \g \to \gl(R)$.  We now define the
cohomology groups of $\g$ with values in $R$. 

First, of course, we need a cochain complex. We define the \define{Lie
superalgebra cochain complex} $C^\bullet(\g, R)$ to consist of
graded-antisymmetric $p$-linear maps at level $p$:
\[ C^p(\g,R) = \left\{ \omega \maps \Lambda^p \g \to R \right\} . \]
In fact, the $p$-cochains $C^p(\g, R)$ are a super vector space, in which
grade-preserving elements are even, while grade-reversing elements are odd.
When $R = \R$, the trivial representation, we typically omit it from the
cochain complex and all associated groups, such as the cohomology groups. Thus,
we write $C^\bullet(\g)$ for $C^\bullet(\g,\R)$.

Next, we define the coboundary operator $d \maps C^p(\g, R) \to C^{p+1}(\g,
R)$. Let $\omega$ be a homogeneous $p$-cochain and let $X_1, \dots, X_{p+1}$ be
homogeneous elements of $\g$. Now define:
\begin{eqnarray*}
& & d\omega(X_1, \dots, X_{n+1}) = \\ 
& & \sum^{p+1}_{i=1} (-1)^{i+1} (-1)^{|X_i||\omega|} \epsilon^{i-1}_1(i) \rho(X_i) \omega(X_1, \dots, \hat{X}_i, \dots, X_{p+1}) \\
& & + \sum_{i < j} (-1)^{i+j} (-1)^{|X_i||X_j|} \epsilon^{i-1}_1(i) \epsilon^{j-1}_1(j) \omega([X_i, X_j], X_1, \dots, \hat{X}_i, \dots, \hat{X}_j, \dots X_{p+1}) .
\end{eqnarray*}
Here, $\epsilon^j_i(k)$ is shorthand for the sign one obtains by moving $X_k$
through $X_i, X_{i+1}, \dots, X_j$. In other words,
\[ \epsilon^j_i(k) = (-1)^{|X_k|(|X_i| + |X_{i+1}| + \dots + |X_j|)}. \]

Following the usual argument for Lie algebras, one can check that:

\begin{prop}
	The Lie superalgebra coboundary operator $d$ satisfies $d^2 = 0$.
\end{prop}

\noindent
We thus say a $R$-valued $p$-cochain $\omega$ on $\g$ is an
\define{$p$-cocycle} or \define{closed} when $d \omega = 0$, and an
\define{$p$-coboundary} or \define{exact} if there exists an $(p-1)$-cochain
$\theta$ such that $\omega = d \theta$.   Every $p$-coboundary is an
$p$-cocycle, and we say an $p$-cocycle is \define{trivial} if it is a
coboundary.  We denote the super vector spaces of $p$-cocycles and
$p$-coboundaries by $Z^p(\g,V)$ and $B^p(\g,V)$ respectively.  The $p$th \define{ 
Lie superalgebra cohomology of $\g$ with coefficients in $R$}, denoted
$H^p(\g,R)$ is defined by 
\[ H^p(\g,R) = Z^p(\g,R)/B^p(\g,R). \]
This super vector space is nonzero if and only if there is a nontrivial
$p$-cocycle. In what follows, we shall be especially concerned with the even
part of this super vector space, which is nonzero if and only if there is a
nontrivial even $p$-cocycle. Our motivation for looking for even cocycles is
simple: these parity-preserving maps can regarded as morphisms in the category
of super vector spaces, which is crucial for the construction in Theorem
\ref{trivd} and everything following it.

Now consider Minkowski spacetimes of dimensions 3, 4, 6 and 10.  Here
Minkowski spacetime can be written as $V = \h_2(\K)$, and we can take our
spinors to be $S_+ = \K^2$.  Since from Section \ref{sec:k+2} we know there is
a symmetric bilinear intertwiner $[-,-] \maps S_+ \otimes S_+ \to V$, we obtain
the supertranslation algebra $\T = V \oplus S_+$.  We can decompose the space
of $n$-cochains with on $\T$ into summands by counting how many of the
arguments are vectors and how many are spinors:
\[     C^n(\T) \iso  
 \bigoplus_{p + q = n}   (\Lambda^p(V) \otimes \Sym^q(S_+))^* . \]
We call an element of $(\Lambda^p(V) \otimes \Sym^q(S_+))^*$ a 
{\bf {\boldmath$(p,q)$-form}}.  Since the bracket of two spinors is
a vector, and all other brackets are zero, $d$ of a $(p,q)$-form is
a $(p-1,q+2)$-form.  

Using the 3-$\psi$'s rule we can show:

\begin{thm}
\label{thm:3-cocycle}
In dimensions 3, 4, 6 and 10, the supertranslation algebra $\T$ has a
nontrivial, Lorentz-invariant even 3-cocycle taking values in the trivial
representation $\R$, namely the unique $(1,2)$-form with 
\[ \alpha(\psi, \phi, A) = g([\psi, \phi], A) \]
for spinors $\psi, \phi \in S_+$ and vectors $A \in V$.
    \end{thm}

\begin{proof}
	First, note that $\alpha$ has the right symmetry to be a linear map on
	$\Lambda^3(V \oplus S_+)$. Second, note that $\alpha$ is a
	$(1,2)$-form, eating one vector and two spinors. Thus $d\alpha$ is a
	$(0,4)$-form.

	Because spinors are odd, $d\alpha$ is a symmetric function of four
	spinors. By the definition of $d$, $d\alpha(\psi, \phi, \chi, \theta)$
	is the totally symmetric part of $\alpha([\psi , \phi], \chi, \theta)
	= \alpha(\chi, \theta, [\psi, \phi]) = g([\chi, \theta], [\psi ,
	\phi])$.   But any symmetric 4-linear form can be obtained from
	polarizing a quartic form. In this, we polarize $g([\psi, \psi],
	[\psi, \psi])$ to get $d\alpha$. Thus:
	\[ d\alpha(\psi, \psi, \psi, \psi) = g([\psi, \psi], [\psi, \psi]) = \langle \psi, [\psi, \psi] \psi \rangle \]
	where we have used the definition of the bracket to obtain the last
	expression, which vanishes due to the 3-$\psi$ rule. Thus $\alpha$ is
	closed.

	It remains to show $\alpha$ is not exact. So suppose it is exact, and
	that
	\[ \alpha = d\omega .\]
	By our remarks above we may assume $\omega$ is a $(2,0)$-form: that is,
	an antisymmetric bilinear function of two vectors. By the definition of
	$d$, this last equation says:
	\[ g([\psi, \phi], A) = -\omega([\psi, \phi], A). \]
	But since $S_+ \tensor S_+ \to V$ is onto, this implies
	\[ g = -\omega, \]
	a contradiction, since $g$ is symmetric while $\omega$ is
	antisymmetric.	
\end{proof}

Next consider Minkowski spacetimes of dimensions 4, 5, 7 and 11.  In this case
Minkowski spacetime can be written as a subspace $\V$ of the $4 \times 4$
matrices valued in $\K$, and we can take our spinors to be $\S = \K^4$.  Since
from Section \ref{sec:k+3} we know there is a symmetric bilinear intertwiner
$[-,-] \maps \S \otimes \S \to \V$, we obtain a supertranslation algebra $\T =
\V \oplus \S$.  As before, we can uniquely 
decompose any $n$-cochain in $C^n(\T, \R)$ into a sum of 
$(p,q)$-forms, where now a {\bf {\boldmath$(p,q)$-form}} is an 
an element of $(\Lambda^p(\V) \otimes \Sym^q(\S))^*$. 
As before, $d$ of a $(p,q)$-form is a $(p-1,q+2)$-form.  
And using the 4-$\Psi$'s rule, we can show:

\begin{thm}
\label{thm:4-cocycle}
	In dimensions 4, 5, 7 and 11, the supertranslation algebra $\T$ has a
	nontrivial, Lorentz-invariant even 4-cocycle, namely the unique
	$(2,2)$-form with
	\[ \beta(\Psi, \Phi, \A, \B) = \langle \Psi, (\A\B - \B\A) \Phi \rangle \]
	for spinors $\Psi, \Phi \in \S$ and vectors $\A, \B \in \V$.  Here the
	commutator $\A\B - \B\A$ is taken in the Clifford algebra of $\V$.
\end{thm} 

\begin{proof}
	First, to see that $\beta$ has the right symmetry to be a map on
	$\Lambda^4(\V \oplus \S)$, we note that it is antisymmetric on vectors.
	Recalling from Proposition \ref{prop:bilinearform} that the
	skew-symmetric bilinear form $\langle -, - \rangle$ is defined by:
	\[ \langle \Psi, \Phi \rangle = \Re(\Psi^\dagger \Gamma^0 \Phi) , \]
	and that because
	\[ \Gamma^0 \A = -\A^\dagger \Gamma^0, \]
	we have:
	\[ \Gamma^0 \A \B = \A^\dagger \B^\dagger \Gamma^0. \]
	Thus:
	\[ \langle \Psi, \A \B \Phi \rangle = \langle \B \A \Psi, \Phi \rangle = -\langle \Phi, \B \A \Psi \rangle, \]
	so we have:
	\[ \langle \Psi, (\A\B - \B\A) \Phi \rangle = \langle \Phi, (\A\B - \B\A) \Psi \rangle. \]
	Thus, $\beta$ is symmetric on spinors.

	Next note that $d\beta$ is a $(1,4)$-form, symmetric on its four spinor
	inputs. It is thus proportional to the polarization of
	\[ \alpha(\Psi, \Psi, [\Psi, \Psi], \A) = 
         \Psi * \Psi( [\Psi, \Psi], \A) \]
	We encountered this object in Section~\ref{sec:spinors}, where we showed
	that it is proportional to 
	\[ h([\Psi, [\Psi, \Psi] \Psi], \A). \]
	Moreover, this last expression vanishes by the 4-$\Psi$'s rule.  So,
	$\beta$ is closed.

	Furthermore, $\beta$ is not exact. To see this, consider the unit
	vector $\left( \begin{smallmatrix} 1 & 0 \\ 0 & -1 \end{smallmatrix}
	\right)$ orthogonal to $V \subseteq \V$. Taking the interior product of
	$\beta$ with this vector, a quick calculation shows:
	\[ \beta(\Psi, \Phi, \left( \begin{smallmatrix} 1 & 0 \\ 0 & -1 \end{smallmatrix} \right), \A) = 2\langle \psi_1, \gamma(A) \phi_1 \rangle + 2\langle \tilde{\gamma}(A) \psi_2, \phi_2 \rangle, \]
	where we have decomposed $\Psi = (\psi_1,\psi_2)$ and $\Phi = (\phi_1,
	\phi_2)$ into their components in $\S = S_+ \oplus S_-$, and $A$ is the
	component of $\A$ in $V$. Restricting to the subalgebra $V \oplus S_+
	\subseteq \V \oplus \S$, we see this is just $\alpha$, up to a factor.

	So, it suffices to check that interior product with $X = \left(
	\begin{smallmatrix} 1 & 0 \\ 0 & -1 \end{smallmatrix} \right)$
	preserves exactness. For then, if $\beta$ were exact, it would
	contradict that fact that $\alpha$ is not. Indeed, let $\omega$ be an
	$n$-cochain on $\T$, and let $X_1, \dots, X_n \in \T$. Then, by our
	formula for the coboundary operator, we have:
	\begin{eqnarray*}
	& & d\omega(X, X_1, \dots, X_n)  = \\
	& & \sum_{i < j} -(-1)^{i+j} (-1)^{|X_i||X_j|} \epsilon^{i-1}_1(i) \epsilon^{j-1}_1(j) \omega(X, [X_i, X_j], X_1, \dots, \hat{X}_i, \dots, \hat{X}_j, \dots X_{n}) \\
	& & + \sum_{i=1}^n (-1)^{1+i} \epsilon_1^{i-1}(i) \omega([X,X_i], X_1, \dots, \hat{X}_i, \dots, X_n),
	\end{eqnarray*}
	where, taking care with signs, we have collected terms involving
	bracketing with $X$ into the second summation. But $X$ is a vector, so
	all brackets with it vanish, and the second summation is zero.

	If we write $i_X \omega$ for the operation of taking the interior
	product of $\omega$ with $X$, we have just shown:
	\[ i_X d \omega = -d i_X \omega \]
	for any $\omega$. In particular, if $\omega = d \theta$ then $i_X
	\omega = d( -i_X \theta )$, and so interior product with $X$ preserves
	exactness, as claimed.
\end{proof}

\chapter{An application: super-Yang--Mills theory} \label{ch:sym}

For the moment, let us set aside our quest to understand division algebras,
supersymmetry and higher gauge theory and focus on a special case: the
connection between division algebras and super-Yang--Mills theories. Such
theories are the low energy limit of superstring theories in a fixed background
\cite{GreenSchwarzWitten}, so it is not surprising that they also occur only in
spacetimes of dimension 3, 4, 6 and 10. 

The minimal supersymmetric extension of pure Yang--Mills theory has the
Lagrangian:
\[ L = 
-\fourth \langle F, F \rangle + 
\half \langle \psi, \slashed{D}_A \psi \rangle. \]
Here $A$ is a connection on a bundle with semisimple gauge group $G$, $F$ is
the curvature of $A$, $\psi$ is a $\g$-valued spinor field, and $\slashed{D}_A$
is the covariant Dirac operator associated with $A$.  In the physics
literature, it is well-known that this theory is supersymmetric if and only if
the dimension of spacetime is $3,4,6,$ or $10$.  Our goal in this section is to
present a self-contained proof of the `if' part of this result, based on the
3-$\psi$'s rule. Along the way, we shall give a division algebra interpretation
of the Lagrangian, $L$.

The proof that $L$ is supersymmetric goes back to the work of Brink, Schwarz
and Sherk~\cite{BrinkSchwarzScherk}, and others. The book by Green, Schwarz and
Witten~\cite{GreenSchwarzWitten} contains a standard proof based on the
properties of Clifford algebras in various dimensions.  But Evans~\cite{Evans}
has shown that the supersymmetry of $L$ in dimension $k+2$ implies the
existence of a normed division algebra of dimension $k$.  Conversely, Kugo and
Townsend~\cite{KugoTownsend} showed how spinors in dimension 3, 4, 6 and 10
derive special properties from the normed division algebras $\R$, $\C$, $\H$
and $\O$.  They formulated a supersymmetric model in 6 dimensions using the
quaternions, $\H$. They also speculated about a similar formalism in 10
dimensions using the octonions, $\O$.

Shortly after Kugo and Townsend's work, Sudbery~\cite{Sudbery} used division
algebras to construct vectors, spinors and Lorentz groups in Minkowski
spacetimes of dimensions 3, 4, 6 and 10.  He then refined his construction
with Chung~\cite{ChungSudbery}, and with Manogue~\cite{ManogueSudbery} he used
these ideas to give an octonionic proof of the supersymmetry of the above
Lagrangian in dimension 10.  This proof was later simplified by Manogue, Dray
and Janesky~\cite{DrayJaneskyManogue}. In the meantime, Schray~\cite{Schray}
applied the same tools to the superparticle.

All this work has made it quite clear that normed division algebras
explain why the above theory is supersymmetric in dimensions 3, 4, 6
and 10.  Technically, what we need to check for supersymmetry is that
$\delta L$ is a total divergence with respect to the supersymmetry
transformation
\begin{eqnarray*}
    \delta A    & = & [ \epsilon, \psi] \\
    \delta \psi & = & \textstyle{\half} F \epsilon
\end{eqnarray*}
for any constant spinor field $\epsilon$. A calculation that works in any
dimension shows that
\[ \delta L = \tri \psi + \mbox{divergence} \]
where $\tri \psi$ is a certain expression depending in a trilinear 
way on $\psi$ and linearly on $\epsilon$.  

So, the marvelous fact that needs to be understood is that $\tri \psi = 0$ in
dimensions 3, 4, 6 and 10, thanks to special properties of the normed division
algebras $\R$, $\C$, $\H$ and $\O$.  Indeed, we shall show that it is a
consequence of the 3-$\psi$'s rule. Yet this rule is a direct consequence of the
fact that $\R$, $\C$, $\H$ and $\O$ are alternative, so one could say that the
vanishing of $\tri \psi$ is a direct consequence of the total antisymmetry of a
simpler trilinear: the associator $[a,b,c]$. 

Let's get to work. For simplicity, we shall work over Minkowski spacetime, $M$.
This allows us to treat all bundles as trivial, sections as functions, and
connections as $\g$-valued 1-forms.  At the outset, we fix an invariant inner
product on $\g$, the Lie algebra of a semisimple Lie group $G$. We shall use
the following standard tools from differential geometry to construct $L$, none
of which need involve spinors or division algebra technology:
\begin{itemize}
	\item A connection $A$ on a principal $G$-bundle over $M$.
	      Since the bundle is trivial we think of this connection 
             as a $\g$-valued 1-form.
	\item The exterior covariant derivative $d_A = d + [A, -]$ on 
             $\g$-valued $p$-forms.
	\item The curvature $F = dA + \frac{1}{2} [A,A]$,  
             which is a $\g$-valued 2-form.
	\item The usual pointwise inner product $\langle F, F \rangle$ on
		$\g$-valued 2-forms, defined using the Minkowski metric on $M$
		and the invariant inner product on $\g$.
\end{itemize}
We also need the following spinorial tools. Because spinors describe
fermions, we assume $S_+$ and $S_-$ are odd objects in SuperVect.
So, whenever we switch two spinors, we introduce a minus sign.
\begin{itemize}
	\item A $\g$-valued section $\psi$ of a spin bundle over $M$. Note that
		this is, in fact, just a function:
		\[ \psi \maps M \to S_\pm \tensor \g. \]
		We call the collection of all such functions $\Gamma(S_\pm
		\tensor \g)$.
	\item The covariant Dirac operator $\slashed{D}_A$ derived from $D_A$.
		Of course,
		\[ \slashed{D}_A \maps \Gamma(S_\pm \tensor \g) \to \Gamma(S_\mp \tensor \g) \]
		and in fact,
		\[ \slashed{D}_A = \slashed{\partial} + A. \]
	\item A bilinear pairing 
		\[ \langle -,- \rangle \maps \Gamma(S_+ \tensor \g) \tensor \Gamma(S_- \tensor \g) \to C^\infty(M) \]
		built pointwise using our pairing
		\[ \langle -,- \rangle \maps S_+ \tensor S_- \to \R \]
		and the invariant inner product on $\g$.
\end{itemize}

The basic fields in our theory are a
connection on a principal $G$-bundle, which we think of as a
$\g$-valued 1-form:
\[ A \maps M \to V^* \tensor \g. \]
and a $\g$-valued spinor field, which we think of as a
$S_+ \tensor \g$-valued function on $M$:
\[ \psi \maps M \to S_+ \tensor \g .\]
All our arguments would work just as well with
$S_-$ replacing $S_+$.

To show that $L$ is supersymmetric, we need to show $\delta L$ is
a total divergence when $\delta$ is the following supersymmetry 
transformation:
\begin{eqnarray*}
    \delta A    & = & [\epsilon, \psi] \\
    \delta \psi & = & \half F \epsilon
\end{eqnarray*}
where $\epsilon$ is an arbitrary constant spinor field, treated as 
odd, but not $\g$-valued.  By a \define{supersymmetry transformation} we mean
that computationally we treat $\delta$ as a derivation on the algebra of
functions on spacetime. So, it is linear: 
\[      \delta (\alpha f + \beta g) = \alpha \delta f + \beta \delta g \]
where $\alpha, \beta \in \R$, and it satisfies the product rule:
\[         \delta (f g) = \delta(f) g + f \delta g. \]
For a more precise discussion of `supersymmetry transformations', see
Deligne and Freed \cite{DeligneFreed}.

The above equations require further explanation.  The bracket
$[\epsilon, \psi]$ denotes an operation that combines 
the spinor $\epsilon$ with the $\g$-valued spinor $\psi$ 
to produce a $\g$-valued 1-form.  We build this from our basic
intertwiner
\[ [-,-] \maps S_+ \tensor S_+ \to V. \]
We identify $V$ with $V^*$ using the Minkowski inner product $g$, obtaining
\[ [-,-] \maps S_+ \tensor S_+ \to V^*. \]
Then we tensor both sides with $\g$. This gives us a way to act by a 
spinor field on a $\g$-valued spinor field to obtain a $\g$-valued 1-form.
We take the liberty of also denoting this with by the bracket:
\[ [-,-] \maps \Gamma(S_+) \tensor \Gamma(S_+ \tensor \g) 
\to \Omega^1(M,\g). \]

We also need to explain how the 2-form $F$ acts on the constant spinor 
field $\epsilon$.  Using the Minkowski metric, we can identify differential
forms on $M$ with sections of the Clifford algebra bundle over $M$:
\[ \Omega^\bullet (M) \iso \Cliff(M). \]
Using this, differential forms act on spinor fields.  Tensoring
with $\g$, we obtain a way for $\g$-valued differential forms like $F$ 
to act on spinor fields like $\epsilon$ to give $\g$-valued spinor fields
like $F \epsilon$.  

Let us now apply the supersymmetry transformation
to each term in the Lagrangian.  First, the bosonic term:

\begin{prop}
The bosonic term has:
\[
 \delta \langle F, F \rangle 
= 2 (-1)^{k+1} \, \langle \psi, ({\star d_A \star} \, F) \epsilon \rangle + 
\rm{divergence}. 
\]
\end{prop}

\begin{proof}
By the symmetry of the inner product, we get:
\[ \delta \langle F, F \rangle = 2 \langle F, \delta F \rangle. \]
Using the handy formula $\delta F = d_A \delta A$, we have:
\[ \langle F, \delta F \rangle = \langle F, d_A \delta A \rangle. \]
Now the adjoint of the operator $d_A$ is $\star d_A \star$, up
to a pesky sign: if $\nu$ is a $\g$-valued $(p-1)$-form and $\mu$ is
a $\g$-valued $p$-form, we have 
\[   \langle \mu, d_A \nu \rangle =
(-1)^{dp + d + 1 + s} \langle {\star d_A \star} \, \mu, \nu \rangle 
+ \mbox{divergence} \]
where $d$ is the dimension of spacetime and $s$ is the signature,
i.e., the number of minus signs in the diagonalized metric.  It
follows that
\[ \langle F, \delta F \rangle = \langle F, d_A \delta A \rangle =
(-1)^k \,
\langle {\star d_A \star} \, F, \delta A \rangle + \mbox{divergence} \]
where $k$ is the dimension of $\K$.  
By the definition of $\delta A$, we get
\[ \langle {\star d_A \star} \, F, \delta A \rangle
=
\langle  {\star d_A \star} \, F, [\epsilon, \psi] \rangle. 
\]
Now we can use division algebra technology to show:
\[
\langle {\star d_A \star} \, F , [\epsilon, \psi] \rangle = \half
\Retr\left( ({\star d_A \star} \, F) (\epsilon \psi^\dagger + \psi
\epsilon^\dagger) \right) = -\langle \psi, ({\star d_A \star} \, F)
\epsilon \rangle, 
\] 
using the cyclic property of the real trace in the last step, and
introducing a minus sign in accordance with the sign rule.  Putting
everything together, we obtain the desired result.
\end{proof}

Even though this proposition involved the bosonic term only, division algebra
technology was still a useful tool in its proof. Ironically, division algebra
technology is absent from the proof of the next proposition, which deals with
the fermionic term:

\begin{prop}
The fermionic term has:
\[ \delta \langle \psi, \slashed{D}_A \psi \rangle = \langle \psi, 
\slashed{D}_A(F \epsilon) \rangle + \tri \psi + \rm{divergence} \]
where
\[ \tri \psi = \langle \psi, [\epsilon, \psi] \psi \rangle. \]
\end{prop}

\begin{proof}
It is easy to compute:
\[ \delta \langle \psi, \slashed{D}_A \psi \rangle = \langle \delta \psi, \slashed{D}_A \psi \rangle + \langle \psi, \delta{\slashed{D}_A} \psi \rangle + \langle \psi, \slashed{D}_A \delta \psi \rangle. \]
Now we insert $\delta \slashed{D}_A = \delta A = [\epsilon, \psi]$, and thus
see that the penultimate term is the trilinear one:
\[ \tri \psi = \langle \psi, [\epsilon, \psi] \psi \rangle. \]
So, let us concern ourselves with the remaining terms:
\[ \langle \delta \psi, \slashed{D}_A \psi \rangle + \langle \psi, \slashed{D}_A \delta \psi \rangle. \]
A computation using the product rule shows that the divergence of the 
1-form $[\psi, \phi]$ is given by $-\langle \phi, \slashed{D}_A \psi \rangle
+ \langle \psi, \slashed{D}_A \phi \rangle$, where the minus sign on the first
term arises from using the sign rule with these odd spinors. In the terms under
consideration, we can use this identity to move $\slashed{D}_A$ onto $\delta
\psi$:
\[ \langle \delta \psi, \slashed{D}_A \psi \rangle + \langle \psi, \slashed{D}_A \delta \psi \rangle = 2 \langle \psi, \slashed{D}_A \delta \psi \rangle + \mbox{divergence}. \]
Substituting $\delta \psi = \half F \epsilon$, we obtain the desired
result.
\end{proof}

Using these two propositions, it is immediate that
\begin{eqnarray*}
	\delta L & = & 
-\fourth \delta \langle F,F \rangle + 
\half \delta \langle \psi, \slashed{D}_A \psi \rangle \\
        & = & \half (-1)^k 
\langle \psi, ({\star d_A \star}\, F) \epsilon \rangle + 
\half \langle \psi, \slashed{D}_A(F \epsilon) \rangle + 
\half \tri \psi + \rm{divergence} 
\end{eqnarray*}
All that remains to show is that $\slashed{D}_A(F \epsilon) = (-1)^{k+1}
({\star d_A \star}F) \, \epsilon$.  Indeed, Snygg shows (Eq.\ 7.6
in~\cite{Snygg}) that for an ordinary, non-$\g$-valued $p$-form $F$ 
\[ \slashed{\partial} (F \epsilon) =
(d F) \epsilon + (-1)^{d + dp + s} ({\star d \star} \, F) \epsilon  \] 
where $d$ is the dimension of spacetime and $s$ is the signature.  This is
easily generalized to covariant derivatives and $\g$-valued $p$-forms:
\[ \slashed{D}_A (F \epsilon) = 
(d_A F) \epsilon + (-1)^{d + dp + s} ({\star d_A \star} \, F) \epsilon .\] 
In particular, when $F$ is the curvature 2-form, the first term vanishes by the
Bianchi identity $d_A F = 0$, and we are left with:
\[ \slashed{D}_A (F \epsilon) = (-1)^{k+1} ({\star d_A \star} \, F) \epsilon \]
where $k$ is the dimension of $\K$. We have thus shown:

\begin{prop}
\label{prop:variation}
Under supersymmetry transformations, the Lagrangian $L$ has: 
\[   \delta L = \half \tri \psi + \rm{divergence}.   \]
\end{prop}

The above result actually holds in every dimension, though our proof used
division algebras and was thus adapted to the dimensions of interest: 3, 4, 6
and 10.  The next result is where division algebra technology becomes really
crucial:

\begin{prop}
\label{prop:variation2}
For Minkowski spacetimes of dimensions 3, 4, 6 and 10, $\tri \psi = 0$.  
\end{prop}

\begin{proof} 
At each point of $M$, we can write 
\[ \psi = \sum \psi^a \tensor g_a, \]
where $\psi^a \in S_+$ and $g_a \in \g$. When we insert this into $\tri
\psi$, we see that
\[ \tri \psi = 
\sum \langle \psi^a, [\epsilon, \psi^b] \psi^c \rangle 
\, \langle g_a, [g_b, g_c] \rangle. \]
Since $\langle g_a, [g_b, g_c] \rangle$ is totally antisymmetric, this implies
$\tri \psi = 0$ for all $\epsilon$ if and only if the part of $\langle \psi^a,
[\epsilon, \psi^b] \psi^c \rangle$ that is antisymmetric in $a$, $b$ and $c$
vanishes for all $\epsilon$. Yet these spinors are odd; for even spinors, we
require the part of $\langle \psi^a, [\epsilon, \psi^b] \psi^c \rangle$ that is
\emph{symmetric} in $a$, $b$ and $c$ to vanish for all $\epsilon$.

Now let us remove our dependence on $\epsilon$. While we do this, let us
replace $\psi^a$ with $\psi$, $\psi^b$ with $\phi$, and $\psi^c$ with $\chi$ to
lessen the clutter of indices.
By the second formula in Proposition \ref{prop:symmetries}, we have:
\[ \langle \psi, [\epsilon, \phi] \chi \rangle =  \langle \epsilon, [\psi, \chi] \phi \rangle, \]
So, if we seek to show that the part of $\langle \psi, [\epsilon, \phi] \chi
\rangle$ that is totally symmetric in $\psi$, $\phi$ and $\chi$ vanishes for
all $\epsilon$, it is equivalent to show the totally symmetric part of $[\phi,
\chi] \psi$ vanishes. But this happens for all $\psi, \phi$ and $\chi$ in $S_+$
if and only if $[\psi,\psi]\psi = 0$ for all $\psi$ in $S_+$. This is the
3-$\psi$'s rule, Theorem \ref{thm:3psirule}, so we are done.  
\end{proof}

\chapter{Lie \emph{n}-superalgebras from Lie superalgebra cohomology} \label{ch:Lie-n-superalgebras}

In Section \ref{sec:cohomology}, we saw that the 3-$\psi$'s and 4-$\Psi$'s
rules are cocycle conditions for the cocycles $\alpha$ and $\beta$. This sheds
some light on the meaning of these rules, but it prompts an obvious followup
question: what are these cocycles good for?  

There is a very general answer to this question: a cocycle on a Lie
superalgebra lets us extend it to an `$L_\infty$-superalgebra'.  As we touched
on in the Introduction, this is a chain complex equipped with structure like
that of a Lie superalgebra, but where all the laws hold only `up to chain
homotopy'.  We give the precise definition below.

It is well known that that the 2nd cohomology of a Lie algebra $\g$
with coefficients in some representation $R$ classifies `central
extensions' of $\g$ by $R$ \cite{AzcarragaIzquierdo, ChevalleyEilenberg}.
These are short exact sequences of Lie algebras:
\[         0 \to R \to \tilde{\g} \to \g \to 0  \]
where $\tilde{g}$ is arbitrary and $R$ is treated as an abelian Lie algebra
whose image lies in the center of $\tilde{g}$.
The same sort of result is true for Lie superalgebras.  But this is just a
special case of an even more general fact.

Suppose $\g$ is a Lie superalgebra with a representation on a
supervector space $R$.  Then we shall prove that an even $R$-valued
$(n+2)$-cocycle $\omega$ on $\g$ lets us construct an 
$L_\infty$-superalgebra, called $\brane_\omega(\g,R)$, of the following form:
\[  
\g \stackrel{d}{\longleftarrow} 0 
\stackrel{d}{\longleftarrow} \dots 
\stackrel{d}{\longleftarrow} 0 
\stackrel{d}{\longleftarrow} R .
\]
where only the 0th and and $n$th grades are nonzero.   Moreover,
$\brane_\omega(\g,R)$ is an \define{extension} of $\g$: there
is a short exact sequence of $L_\infty$-superalgebras 
\[  0 \to \b^n R \to \brane_\omega(\g,R) \to \g \to 0 . \]
Here $\b^n R$ is the abelian $L_\infty$-superalgebra
with $R$ as its $n$th grade and all the rest zero:
\[             
0 
\stackrel{d}{\longleftarrow} 0 
\stackrel{d}{\longleftarrow} \cdots
\stackrel{d}{\longleftarrow} 0 
\stackrel{d}{\longleftarrow} R   \]
Note that when $n = 0$ and our vector spaces are all purely even, 
we are back to the familiar construction of Lie algebra extensions 
from 2-cocycles.

Technically, we should be more general than this in defining extensions.
Maps between $L_\infty$-algebras admit homotopies among themselves, and
this allows us to introduce a weakened notion of `short exact sequence': 
a fibration sequence in the $(\infty,1)$-category of $L_\infty$-algebras.
In general, these fibration sequences give the right concept of 
extension for $L_\infty$-algebras.  However, for the very special 
extensions we consider here, ordinary short exact sequences are all we 
need. 

It is useful to have a special name for $L_\infty$-superalgebras
whose nonzero terms are all of degree $< n$: we call them
\define{Lie $n$-superalgebras}.  In this
language, the 3-cocycle $\alpha$ defined in Theorem \ref{thm:3-cocycle}
gives rise to a Lie 2-superalgebra
\[              \T \stackrel{d}{\longleftarrow} \R   \]
extending the supertranslation algebra $\T$ in dimensions 3, 4, 6 and 10.
Similarly, the 4-cocycle $\beta$ defined in Theorem \ref{thm:4-cocycle}
gives a Lie 3-superalgebra 
\[              \T \stackrel{d}{\longleftarrow} 0 
                   \stackrel{d}{\longleftarrow} \R   \]
extending the supertranslation algebra in dimensions $4,5,7$ and $11$.

Now let us make all of these ideas precise. In what follows, we shall use
\define{super chain complexes}, which are chain complexes in the category
SuperVect of $\Z_2$-graded vector spaces:
\[  V_0 \stackrel{d}{\longleftarrow}
    V_1 \stackrel{d}{\longleftarrow}
    V_2 \stackrel{d}{\longleftarrow} \cdots \]
Thus each $V_p$ is $\Z_2$-graded and $d$ preserves this grading.

There are thus two gradings in play: the $\Z$-grading by
\define{degree}, and the $\Z_2$-grading on each vector space, which we
call the \define{parity}. We shall require a sign convention to
establish how these gradings interact. If we consider an object of odd
parity and odd degree, is it in fact even overall?  By convention, we
assume that it is. That is, whenever we interchange something of
parity $p$ and degree $q$ with something of parity $p'$ and degree
$q'$, we introduce the sign $(-1)^{(p+q)(p'+q')}$. We shall call the
sum $p+q$ of parity and degree the \define{overall grade}, or when it
will not cause confusion, simply the grade. We denote the overall
grade of $X$ by $|X|$.

We require a compressed notation for signs. If $x_{1}, \ldots, x_{n}$ are
graded, $\sigma \in S_{n}$ a permutation, we define the \define{Koszul sign}
$\epsilon (\sigma) = \epsilon(\sigma; x_{1}, \dots, x_{n})$ by 
\[ x_{1} \cdots x_{n} = \epsilon(\sigma; x_{1}, \ldots, x_{n}) \cdot x_{\sigma(1)} \cdots x_{\sigma(n)}, \]
the sign we would introduce in the free graded-commutative algebra generated by
$x_{1}, \ldots, x_{n}$. Thus, $\epsilon(\sigma)$ encodes all the sign changes
that arise from permuting graded elements. Now define:
\[ \chi(\sigma) = \chi(\sigma; x_{1}, \dots, x_{n}) := \textrm{sgn} (\sigma) \cdot \epsilon(\sigma; x_{1}, \dots, x_{n}). \]
Thus, $\chi(\sigma)$ is the sign we would introduce in the free
graded-anticommutative algebra generated by $x_1, \dots, x_n$.

Yet we shall only be concerned with particular permutations. If $n$ is a
natural number and $1 \leq j \leq n-1$ we say that $\sigma \in S_{n}$ is an
\define{$(j,n-j)$-unshuffle} if
\[ \sigma(1) \leq\sigma(2) \leq \cdots \leq \sigma(j) \hspace{.2in} \textrm{and} \hspace{.2in} \sigma(j+1) \leq \sigma(j+2) \leq \cdots \leq \sigma(n). \] 
Readers familiar with shuffles will recognize unshuffles as their inverses. A
\emph{shuffle} of two ordered sets (such as a deck of cards) is a permutation
of the ordered union preserving the order of each of the given subsets. An
\emph{unshuffle} reverses this process. We denote the collection of all
$(j,n-j)$ unshuffles by $S_{(j,n-j)}$.

The following definition of an $L_{\infty}$-algebra was formulated by
Schlessinger and Stasheff in 1985 \cite{SS}:

\begin{defn} \label{L-alg} An
\define{$L_{\infty}$-superalgebra} is a graded vector space $V$
equipped with a system $\{l_{k}| 1 \leq k < \infty\}$ of linear maps
$l_{k} \maps V^{\otimes k} \rightarrow V$ with $\deg(l_{k}) = k-2$
which are totally antisymmetric in the sense that
\begin{eqnarray}
   l_{k}(x_{\sigma(1)}, \dots,x_{\sigma(k)}) =
   \chi(\sigma)l_{k}(x_{1}, \dots, x_{n})
\label{antisymmetry}
\end{eqnarray}
for all $\sigma \in S_{n}$ and $x_{1}, \dots, x_{n} \in V,$ and,
moreover, the following generalized form of the Jacobi identity
holds for $0 \le n < \infty :$
\begin{eqnarray}
   \displaystyle{\sum_{i+j = n+1}
   \sum_{\sigma \in S_{(i,n-i)}}
   \chi(\sigma)(-1)^{i(j-1)} l_{j}
   (l_{i}(x_{\sigma(1)}, \dots, x_{\sigma(i)}), x_{\sigma(i+1)},
   \ldots, x_{\sigma(n)}) =0,}
\label{megajacobi}
\end{eqnarray}
where the inner summation is taken over all $(i,n-i)$-unshuffles with $i
\geq 1.$
\end{defn}

The following result shows how to construct $L_{\infty}$-superalgebras from Lie
superalgebra cocycles.  This is the `super' version of a result due to Crans
\cite{BaezCrans}. In this result, we require our cocycle to be even so we can
consider it as a morphism in the category of super vector spaces.

\begin{thm} \label{trivd}
There is a one-to-one correspondence between $L_{\infty}$-superalgebras
consisting of only two nonzero terms $V_{0}$ and $V_{n}$, with $d=0$, and
quadruples $(\g, V, \rho, l_{n+2})$ where $\g$ is a Lie superalgebra, $V$ is a
super vector space, $\rho$ is a representation of $\g$ on $V$, and $l_{n+2}$ is
an even $(n+2)$-cocycle on $\g$ with values in $V$.
\end{thm}

\begin{proof}

Given such an $L_{\infty}$-superalgebra we set $\g=
V_0$.  $V_0$ comes equipped with a bracket as part of the
$L_{\infty}$-structure, and since $d$ is trivial, this bracket satisfies the
Jacobi identity on the nose, making $\g$ into a Lie superalgebra. We define $V
= V_{n}$, and note that the bracket also gives a map $\rho \maps \g \tensor V
\to V$, defined by $\rho(x)f = [x,f]$ for $x \in \g, f \in V$. We have
\begin{eqnarray*}
  \rho ([x,y])f &=& [[x,y],f] \\
  		&=& (-1)^{|y||f|}[[x,f],y] + [x,[y,f]] \; \; \; \;
                    \textrm{by $(3)$ of Definition \ref{L-alg}} \\
                &=& (-1)^{|f||y|}[\rho(x)f, y] + [x, \rho(y) f] \\
                &=& -(-1)^{|x||y|}\rho(y)\rho(x)f + \rho(x)\rho(y) f \\
		&=& [\rho(x), \rho(y)]f
\end{eqnarray*}

\noindent
for all $x,y \in \g$ and $f \in V$, so that $\rho$ is indeed a representation.
Finally, the $L_{\infty}$ structure gives a map $l_{n+2} \maps
\Lambda^{n+2} \g \to V$ which is in fact an $(n+2)$-cocycle.  To see this, note
that
\[ 0 = \sum_{i+j = n+4} \sum_{\sigma} \chi(\sigma) (-1)^{i(j-1)} l_{j}(l_{i}(x_{\sigma(1)}, \ldots, x_{\sigma(i)}), x_{\sigma(i+1)}, \ldots, x_{\sigma(n+3)}) , \]
where we sum over $(i, (n+3)-i)$-unshuffles $\sigma \in S_{n+3}$.  However, the
only choices for $i$ and $j$ that lead to nonzero $l_{i}$ and $l_{j}$ are
$i=n+2, j=2$ and $i=2, j=n+2$. Thus, the above becomes, with $\sigma$ a
$(n+2, 1)$-unshuffle and $\tau$ a $(2, n+1)$-unshuffle:
\begin{eqnarray*}
0 & = & \sum_{\sigma} \chi(\sigma) (-1)^{n+2} [l_{n+2}(x_{\sigma(1)}, \dots, x_{\sigma(n+2)}), x_{\sigma(n+3)}] \\
  & + & \sum_{\tau} \chi(\tau) l_{n+2}([x_{\tau(1)}, x_{\tau(2)}], x_{\tau(3)}, \dots, x_{\tau(n+3)}) \\
  & = & \sum_{i=1}^{n+3} (-1)^{n+3-i}(-1)^{n+2} \epsilon^{n+2}_{i+1}(i) [l_{n+2}(x_1, \dots, \hat{x}_i, \dots, x_{n+3}), x_{i}] \\
  & + & \sum_{1 \leq i < j \leq n+3} (-1)^{i+j+1} (-1)^{|x_i||x_j|} \epsilon^{i-1}_1(i) \epsilon^{j-1}_1(j) l_{n+2}([x_{i}, x_{j}], x_{1}, \dots,\hat{x}_i, \dots, \hat{x}_j, \dots, x_{n+3}) .
\end{eqnarray*}
On the second line, we have explicitly specified the unshuffles and unwrapped
the signs encoded by $\chi$. Since $l_{n+2}$ is a morphism in SuperVect, it
preserves parity, and thus the element 
\[ l_{n+2}(x_1, \dots, \hat{x}_i, \dots, x_{n+2}) \] 
has parity $|x_1| + \dots + |x_{i-1}| + |x_{i+1}| + \dots + |x_{n+2}|$. So, we
can reorder the bracket in the first term, at the cost of a sign:
\begin{eqnarray*}
0 & = & \sum_{i=1}^{n+3} -(-1)^{i+1} \epsilon^{i-1}_1(i) [x_i, l_{n+2}(x_{1}, \dots, \hat{x}_i, \dots, x_{n+3})] \\
  & + & \sum_{1 \leq i < j \leq n+3} -(-1)^{i+j} (-1)^{|x_i||x_j|} \epsilon^{i-1}_1(i) \epsilon^{j-1}_1(j) l_{n+2}([x_{i}, x_{j}], x_{1}, \dots, \hat{x}_i, \dots, \hat{x}_j, \dots, x_{n+3}) \\
  & = & -dl_{n+2} .
\end{eqnarray*}
Here, we have used the fact that $\epsilon_{i+1}^{n+2}(i) (-1)^{|x_i|(|x_1| +
\dots + |x_{i-1}| + |x_{i+1}| + \dots + |x_{n+2}|)} = \epsilon_1^{i-1}(i)$.
Thus, $l_{n+2}$ is indeed a cocycle.

Conversely, given a Lie superalgebra $\g$, a representation $\rho$ of $\g$ on a
vector space $V$, and an even $(n+2)$-cocycle $l_{n+2}$ on $\g$ with values in
$V$, we define our $L_{\infty}$-superalgebra $V$ by setting
$V_{0} = \g$, $V_{n} = V$, $V_i = \{0\}$ for $i \ne 0,n$,
and $d=0$.  It remains to define the system of linear maps
$l_{k}$, which we do as follows: Since $\g$ is a Lie
superalgebra, we have a bracket defined on $V_{0}$. We extend this
bracket to define the map $l_2$, denoted by $[\cdot, \cdot] \maps
V_{i} \otimes V_{j} \rightarrow V_{i+j}$ where $i,j=0,n,$ as
follows:
\[ [x,f] = \rho(x) f, \quad [f,y] = (-1)^{|y||f|} \rho(y)f, \quad [f,g] = 0 \]
for $x,y \in V_0$ and $f,g \in V_n$.
With this definition, the map $[\cdot, \cdot]$ satisfies condition $(1)$ of
Definition \ref{L-alg}. We define $l_{k}=0$ for $3 \leq k \leq n+1$ and $k>
n+2$, and take $l_{n+2}$ to be the given $(n+2)$ cocycle, which satisfies
conditions $(1)$ and $(2)$ of Definition \ref{L-alg} by the cocycle condition. 
\end{proof}

This theorem tells us how to take a Lie superalgebra $(n+1)$-cocycle $\omega$,
and construct a Lie $n$-superalgebra with $d=0$, concentrated in degrees 0 and
$n-1$. We call such a Lie $n$-superalgebra a \define{slim Lie
$n$-superalgebra}, and denote it by $\brane_\omega(\g,R)$. When $n=2$, we will
also write $\strng_\omega(\g,R)$ for the same object, and when $R$ is the
trivial representation $\R$, we omit it. In the next section, we give some
examples of these objects.

\section{Examples of slim Lie \emph{n}-superalgebras}

\subsection{The string Lie 2-algebra} \label{sec:string2alg}

For $n \geq 3$, consider the Lie algebra $\so(n)$ of infinitesimal rotations of
$n$-dimensional Euclidean space. This matrix Lie algebra has Killing form given
by the trace, $\langle X , Y \rangle = \tr(XY)$, and an easy calculation shows
that 
\[ j = \langle - , [-, -] \rangle \] 
is a 3-cocycle on $\so(n)$. We call $j$ the \define{canonical 3-cocycle} on
$\so(n)$. Using $j$, we get a Lie 2-algebra $\strng_j(\so(n))$, which we denote
simply by $\strng(n)$. We call this the \define{string Lie 2-algebra}. First
defined by Baez--Crans \cite{BaezCrans}, it is so-named because it turned out
to be intimately related to the string group, $\String(n)$, the topological
group obtained from $\SO(n)$ by killing the 1st and 3rd homotopy groups. For a
description of this relationship, as well as the construction of Lie 2-groups
which integrate $\strng(n)$, see the papers of
Baez--Crans--Schreiber--Stevenson \cite{BCSS}, Henriques \cite{Henriques}, and
Schommer-Pries \cite{SchommerPries}.

\subsection{The Heisenberg Lie 2-algebra} \label{sec:sec:heisenberg2alg}

As we mentioned earlier, central extensions of Lie algebras are classified by
second cohomology. A famous example of this is the `Heisenberg Lie algebra', so
named because it mimics the canonical commutation relations in quantum
mechanics. Here we present a Lie 2-algebra generalization: the `Heisenberg Lie
2-algebra'.

Consider the abelian Lie algebra of translations in position-momentum space:
\[ \R^2 = \mathrm{span}(p,q). \]
Here, $p$ and $q$ are our names for the standard basis, the usual letters for
momentum and position in physics. Up to rescaling, this Lie algebra has a
single, nontrivial 2-cocycle:
\[ p^* \wedge q^* \in \Lambda^2(\R^2), \]
where $p^*$ and $q^*$ comprise the dual basis. Thus it has a nontrivial central
extension:
\[ 0 \to \R \to \mathfrak{H} \to \R^2 \to 0. \]

This central extension is called the \define{Heisenberg Lie algebra}. As a
vector space, $\mathfrak{H} = \R^3$, and we call the basis vectors $p,
q$ and $z$, where $z$ is central. When chosen with suitable normalization,
they satisfy the relations:
\[ [p,q] = z, \quad [p,z] = 0, \quad [q,z] = 0. \]
These are the same as the canonical commutation relations in quantum mechanics,
except that the generator $z$ would usually be a number, $-i \hbar$. It is from
this parallel that the Heisenberg Lie algebra derives its physical
applications: a representation of $\mathfrak{H}$ is exactly a way of choosing
linear operators $p$, $q$ and $z$ on a Hilbert space that satisfy the canonical
commutation relations.

With Lie 2-algebras, we can repeat the process that yielded the Heisenberg Lie
algebra to obtain a higher structure. Before we needed a 2-cocycle, but now we
need a 3-cocycle. Indeed, letting $p^*$, $q^*$ and $z^*$ be the dual basis of
$\mathfrak{H}^*$, it is easy to check that $\gamma = p^* \wedge q^* \wedge z^*$
is a nontrivial 3-cocycle on $\mathfrak{H}$. Thus there is a Lie 2-algebra
$\strng_\gamma(\mathfrak{H})$, the \define{Heisenberg Lie 2-algebra}, which we
denote by $\mathfrak{Heisenberg}$.  Later, in Chapter \ref{ch:integrating}, we
will see how to integrate this Lie 2-algebra to a Lie 2-group.

We suspect the Heisenberg Lie 2-algebra, like its Lie algebra cousin, is also
important for physics. We also suspect that the pattern continues: the
Heisenberg Lie 2-algebra may admit a `4-cocycle', and a central extension to a
Lie 3-algebra. However, since we have not defined the cohomology of Lie
$n$-algebras \cite{Penkava}, we do not pursue this here.

\subsection{The supertranslation Lie \emph{n}-superalgebras} \label{sec:n-trans-algs}

Some exceptional cocycles arise on the supertranslation algebras in certain
dimensions. Recall from Section \ref{sec:cohomology} that a supertranslation
algebra is a Lie superalgebra of the form:
\[ \T = V \oplus S, \]
where the even part $V$ is a vector space with a nondegerate quadratic form,
the odd part $S$ is a spinor representation of $\Spin(V)$, and the bracket
comes from a symmetric, $\Spin(V)$-equivariant map that takes pairs of spinors
to vectors:
\[ [-,-] \maps \Sym^2 S \to V. \]

In spacetime dimensions 3, 4, 6 and 10, we proved in Theorem \ref{thm:3-cocycle}
that there is a 3-cocycle $\alpha$, which is nonzero only when given two
spinors and a vector:
\[ \begin{array}{cccc}
	\alpha \maps & \Lambda^3(\T)              & \to     &  \R \\
	             &  A \wedge \psi \wedge \phi & \mapsto & \langle \psi, A \phi \rangle . \\
\end{array}
\]
There is thus a Lie 2-superalgebra, the \define{supertranslation Lie
2-superalgebra}, $\strng_\alpha(\T)$.

Likewise, in spacetime dimensions 4, 5, 7 and 11, we proved in Theorem
\ref{thm:4-cocycle} that there is a 4-cocycle $\beta$, which is nonzero only
when given two spinors and two vectors:
\[ \begin{array}{cccc}
	\beta \maps & \Lambda^4(\T)                        & \to     & \R \\
	            & \A \wedge \B \wedge \Psi \wedge \Phi & \mapsto & \langle \Psi, (\A \wedge \B) \Phi \rangle . \\
   \end{array}
\]
There is thus a Lie 3-superalgebra, the \define{supertranslation Lie
3-superalgebra}, $\brane_{\beta}(\T)$.

There is much more that one can do with the cocycles $\alpha$ and $\beta$,
however. We can use them to extend not just the supertranslations $\T$ to a Lie
$n$-superalgebra, but the full Poincar\'e superalgebra, $\so(V) \ltimes \T$. We
turn to this now.

\subsection{Superstring Lie 2-superalgebras, 2-brane Lie 3-superalgebras} \label{sec:n-brane-algs}

One of the principal themes of theoretical physics over the last century has
been the search for the underlying symmetries of nature. This began with
special relativity, which could be summarized as the discovery that the laws of
physics are invariant under the action of the Poincar\'e group:
\[ \ISO(V) = \Spin(V) \ltimes V. \]
Here, $V$ is the set of vectors in Minkowski spacetime and acts on Minkowski
spacetime by translation, while $\Spin(V)$ is the \define{Lorentz group}: the
double cover of $\SO_0(V)$, the connected component of the group of symmetries
of the Minkowski norm. Much of the progress in physics since special relativity
has been associated with the discovery of additional symmetries, like the
$\U(1) \times \SU(2) \times \SU(3)$ symmetries of the Standard Model of
particle physics \cite{BaezHuerta:guts}.

Today, `supersymmetry' could be summarized as the hypothesis that the laws of
physics are invariant under the `Poincar\'e supergroup', which is larger than
the Poincar\'e group:
\[ \SISO(V) = \Spin(V) \ltimes T. \]
Here, $V$ is again the set of vectors in Minkowski spacetime and $\Spin(V)$ is
the Lorentz group, but $T$ is the supergroup of translations on Minkowski
`superspacetime'. Though we have not yet learned enough supergeometry to talk
about $T$ precisely, we have already met its infinitesimal approximation: the
superstranslation algebra, $\T = V \oplus S$.  We think of the spinor
representation $S$ as giving extra, supersymmetric translations, or
`supersymmetries'.

In this thesis, we show how to further extend the Poincar\'e group to include
higher symmetries, thanks to the normed division algebras. That is, we will
show that in dimensions $k+2 = 3$, 4, 6 and 10, one can extend the Poincar\'e
supergroup $\SISO(k+1,1)$ to a `Lie 2-supergroup' we call
$\Superstring(k+1,1)$. Similarly, in dimensions $k+3 = 4$, 5, 7 and 11, one can
extend the Poincar\'e supergroup $\SISO(k+2,1)$ to a `Lie 3-supergroup' we call
$\Twobrane(k+2,1)$. 

We begin this construction in this section by working at the infinitesimal
level.  We construct a Lie 2-superalgebra, 
\[ \superstring(k+1,1) , \]
which extends the Poincar\'e superalgebra in dimension $k+2$:
\[ \siso(k+1,1 = \so(k+1,1) \ltimes \T \]
Then we construct a Lie 3-superalgebra,
\[ \twobrane(k+2,1) , \] 
which extends the Poincar\'e superalgebra in dimension $k+3$:
\[ \siso(k+2,1)= \so(k+2,1) \ltimes \T. \] 
We do this construction using the cocycles $\alpha$ and $\beta$. This is
possible because both $\alpha$ and $\beta$ are invariant under the action of
the corresponding Lorentz algebra: $\so(k+1,1)$ in the case of $\alpha$, and
$\so(k+2,1)$ for $\beta$.  This is manifestly true, because $\alpha$ and
$\beta$ are built from equivariant maps.

As we shall see, this invariance implies that $\alpha$ and $\beta$ are
cocycles, not merely on the supertranslations, but on the full Poincar\'e
superalgebra---$\siso(k+1,1)$ in the case of $\alpha$, and $\siso(k+2,1)$ in
the case of $\beta$.  We can extend $\alpha$ and $\beta$ to these larger
algebras in a trivial way: define the unique extension which vanishes unless
all of its arguments come from $\T$. Doing this, $\alpha$ and $\beta$ remain
cocycles, even though the Lie bracket (and thus $d$) has changed. Moreover,
they remain nontrivial. All of this is contained in the following proposition:

\begin{prop} Let $\g$ and $\h$ be Lie superalgebras such that $\g$ acts on
	$\h$, and let $R$ be a representation of $\g \ltimes \h$. Given any
	$R$-valued $n$-cochain $\omega$ on $\h$, we can uniquely extend it to
	an $n$-cochain $\tilde{\omega}$ on $\g \ltimes \h$ that takes the value
	of $\omega$ on $\h$ and vanishes on $\g$. When $\omega$ is even, we
	have:
	\begin{enumerate}
		\item $\tilde{\omega}$ is closed if and only if $\omega$ is
			closed and $\g$-equivariant.
		\item $\tilde{\omega}$ is exact if and only if $\omega =
			d\theta$, for $\theta$ a $\g$-equivariant
			$(n-1)$-cochain on $\h$.
	\end{enumerate}
\end{prop}
\begin{proof}
	As a vector space, $\g \ltimes \h = \g \oplus \h$, so that 
	\[ \Lambda^n(\g \ltimes \h) \iso \bigoplus_{p+q=n} \Lambda^p \g \tensor \Lambda^q \h, \]
	as a vector space. Thanks to this decomposition, we can uniquely
	decompose $n$-cochains on $\g \ltimes \h$ by restricting to the summands.
	In keeping with our prior terminology, we call an $n$-cochain supported
	on $\Lambda^p \g \tensor \Lambda^q \h$ a $(p,q)$-form. Note that
	$\tilde{\omega}$ is just the $n$-cochain $\omega$ regarded as a
	$(0,n)$-form on $\g \ltimes \h$. We shall denote the space of
	$(p,q)$-forms by $C^{p,q}$. 
	
	We have two actions to distinguish: the action of $\g \ltimes \h$ on
	$R$, which we denote by $\rho$, and the action of $\g$ on $\h$, which
	we shall denote simply by the bracket, $[-,-]$. Inspecting the formula
	for the differential:
	\begin{eqnarray*}
		& & d\tilde{\omega}(X_1, \dots, X_{n+1}) = \\
		& & \sum^{n+1}_{i=1} (-1)^{i+1} (-1)^{|X_i||\tilde{\omega}|} \epsilon^{i-1}_1(i) \rho(X_i) \tilde{\omega}(X_1, \dots, \hat{X}_i, \dots, X_{n+1}) \\
		& & + \sum_{i < j} (-1)^{i+j} (-1)^{|X_i||X_j|} \epsilon^{i-1}_1(i) \epsilon^{j-1}_1(j) \tilde{\omega}([X_i, X_j], X_1, \dots, \hat{X}_i, \dots, \hat{X}_j, \dots X_{n+1})
	\end{eqnarray*}
	it is easy to see that
	\[ d \maps C^{p,q} \to C^{p,q+1} \oplus C^{p+1,q}. \]
	In particular:
	\[ d \maps C^{0,n} \to C^{0,n+1} \oplus C^{1,n}. \]
	Given an $n$-cochain $\omega$ on $\h$, it is easy to see that the
	part of $d\tilde{\omega}$ which lies in $C^{0,n+1}$ is just
	$\widetilde{d\omega}$, the extension of the $(n+1)$-cochain $d\omega$
	to $\g \ltimes \h$. 
	
	Let $e\omega$ denote the $(1,n)$-form part of $d\tilde{\omega}$.  To
	express this explicitly, choose $Y_1 \in \g$ and $X_2, \dots, X_{n+1}
	\in \h$. By definition $e\omega(Y_1, X_2, \dots, X_{n+1}) =
	d\tilde{\omega}(Y_1, X_2, \dots, X_{n+1})$, and inspecting the formula
	for the differential once more, we see this consists of only two
	nonzero terms:
	\begin{eqnarray*}
		e\omega(Y_1, X_2, \dots, X_{n+1}) 
		& = & (-1)^{|\tilde{\omega}||Y_1|} \rho(Y_1) \tilde{\omega}(X_2, \dots, X_{n+1}) \\
		&   & + \sum_{i=2}^{n+1} (-1)^{i+1} \epsilon_2^{i-1}(i) \tilde{\omega}([Y_1, X_i], X_2, \dots, \hat{X}_i, \dots, X_{n+1}) \\
		& = & (-1)^{|\omega||Y_1|} \rho(Y_1) \omega(X_2 \wedge \dots \wedge X_{n+1}) - \omega([Y_1, X_2 \wedge \dots \wedge X_{n+1}]) .
	\end{eqnarray*}
	In particular, note that for even $\omega$,  $e\omega = 0$ if and only
	if $\omega$ is $\g$-equivariant.

	To summarize, for any $n$-cochain $\omega$, we have that
	\[ d\tilde{\omega} = \widetilde{d\omega} + e\omega, \]
	where the first $d$ is defined on $\g \ltimes \h$, while the second is
	only defined on $\h$. The proof of 1 is now immediate: for even
	$\omega$, $d\tilde{\omega} = 0$ if and only if $\widetilde{d\omega} =
	0$ and $e\omega = 0$, which happens if and only $d\omega = 0$ and
	$\omega$ is $\g$-equivariant.

	To prove 2, suppose $\omega$ is even. Assume $\tilde{\omega} = d\chi$,
	for some $(n-1)$-cochain $\chi$ on $\g \ltimes \h$. Because $d\chi$ is
	an even $(0,n)$-form, we may assume $\chi$ is an even $(0,n-1)$-form,
	as any other part of $\chi$ is closed and does not contribute to
	$d\chi$.  Thus $\chi$ is the extension of an even $(n-1)$-cochain
	$\theta$ on $\h$. By our prior formula, we have:
	\[ \tilde{\omega} = d\tilde{\theta} = \widetilde{d\theta} + e\theta \]
	The left-hand side is a $(0,n)$-form, and thus the $(1,n-1)$-form part
	of the right-hand side, $e\theta$, vanishes. Thus $\theta$ is
	$\g$-equivariant, and $\tilde{\omega} = \widetilde{d\theta}$, which
	implies $\omega = d\theta$.  On the other hand, if $\omega = d\theta$
	and $\theta$ is $\g$-equivariant, then $e\theta = 0$ and thus
	$\tilde{\omega} = d\tilde{\theta}$.
\end{proof}

Thus we can extend $\alpha$ and $\beta$ to nonexact cocycles on the Poincar\'e
Lie superalgebra, simply by defining $\alpha$ and $\beta$ to vanish outside of
the supertranslation algebra. Thanks to Theorem~\ref{trivd}, we know that
$\alpha$ lets us extend $\siso(k+1,1)$ to a Lie 2-superalgebra:

\begin{thm} \label{thm:superstring}
	In dimensions 3, 4, 6 and 10, there exists a Lie 2-superalgebra formed
	by extending the Poincar\'e superalgebra $\siso(k+1,1)$ by the
	3-cocycle $\alpha$, which we call we the \define{superstring Lie
	2-superalgebra, $\superstring(k+1,1)$}.
\end{thm}
\noindent
Likewise, in dimensions one higher, $\beta$ lets us extend $\siso(k+2,1)$ to a
Lie 3-superalgebra. In the 11-dimensional case, this coincides with
the Lie 3-superalgebra which Sati, Schreiber and Stasheff call $\sugra(10,1)$
\cite{SSS}, which is the Koszul dual of an algebra defined by D'Auria and
Fr\'e~\cite{DAuriaFre}.
\begin{thm} \label{thm:2brane}
	In dimensions 4, 5, 7 and 11, there exists a Lie 3-superalgebra formed
	by extending the Poincar\'e superalgebra $\siso(k+2,1)$ by the
	4-cocycle $\beta$, which we call the \define{2-brane Lie
	3-superalgebra, $\twobrane(k+2,1)$}.
\end{thm}

\chapter{Lie \emph{n}-groups from group cohomology} \label{ch:Lie-n-groups}

Having constructed Lie $n$-algebras from Lie algebra $(n+1)$-cocycles, we now
turn to a parallel construction of Lie $n$-groups. Roughly speaking, an
`$n$-group' is a weak $n$-groupoid with one object---an $n$-category with one
object in which all morphisms are weakly invertible, up to higher-dimensional
morphisms. This definition is a rough one because there are many possible
definitions to use for `weak $n$-category', but despite this ambiguity, it can
still serve to motivate us. 

The richness of weak $n$-categories, no matter what definition we apply, makes
$n$-groups a complicated subject. In the midst of this complexity, we seek to
define a class of $n$-groups that have a simple description, and which are
straightforward to internalize, so that we may easily construct Lie $n$-groups
and Lie $n$-supergroups, as we shall do later in this thesis. The motivating
example for this is what Baez and Lauda \cite{BaezLauda} call a `special
2-group', which has a concrete description using group cohomology. Since Baez
and Lauda prove that all 2-groups are equivalent to special ones, group
cohomology also serves to classify 2-groups.

So, we will define `slim Lie $n$-groups', at least for $n \leq 3$.  This is an
Lie $n$-group which is skeletal (every weakly isomorphic pair of objects are
equal), and almost trivial: all $k$-morphisms are the identity for $1 < k < n$.
Slim Lie $n$-groups are useful because they can be completely classified by Lie
group cohomology. They are also easy to `superize', and their super versions
can be completely classified using Lie supergroup cohomology, as we shall see
in Chapter \ref{ch:Lie-n-supergroups}. Finally, we note that we could equally
well define `slim $n$-groups', working in the category of sets rather than the
category of smooth manifolds. The results in this section would hold in this
case as well, but are of less use to us in this thesis.

We should stress that the definition of Lie $n$-group we sketch here (and make
precise for $n \leq 3$), while it is good enough for our needs, is known to be
too naive in some important respects. For instance, it does not seem possible
to integrate every Lie $n$-algebra to a Lie $n$-group of this type, while
Henriques's definition of Lie $n$-group does make this possible
\cite{Henriques}. 

First we need to review the cohomology of Lie groups, as originally defined by
van Est \cite{vanEst}, who was working in parallel with the definition of group
cohomology given by Eilenberg and MacLane. Fix a Lie group $G$, an abelian Lie
group $H$, and a smooth action of $G$ on $H$ which respects addition in $H$.
That is, for any $g \in G$ and $h, h' \in H$, we have:
\[ g(h + h') = gh + gh'. \]
Then \define{the cohomology of $G$ with coefficients in $H$} is given by the
\define{Lie group cochain complex}, $C^{\bullet}(G,H)$. At level $p$, this
consists of the smooth functions from $G^p$ to $H$:
\[ C^p(G,H) = \left\{ f \maps G^p \to H \right\}. \]
We call elements of this set \define{$H$-valued $p$-cochains on $G$}. The
boundary operator is the same as the one defined by Eilenberg--MacLane. On a
$p$-cochain $f$, it is given by the formula:
\begin{eqnarray*}
	df(g_1, \dots, g_{p+1}) & = & g_1 f(g_2, \dots, g_{p+1}) \\
	                        &   & + \sum_{i=1}^p (-1)^i f(g_1, \dots, g_{i-1}, g_i g_{i+1}, g_{i+2}, \dots, g_{p+1}) \\
				&   & + (-1)^{p+1} f(g_1, \dots, g_p) .
\end{eqnarray*}
The proof that $d^2 = 0$ is routine. All the usual terminology applies: a
$p$-cochain $f$ for which $df = 0$ is called \define{closed}, or a
\define{cocycle}, a $p$-cochain $f = dg$ for some $(p-1)$-cochain $g$ is called
\define{exact}, or a \define{coboundary}. A $p$-cochain is said to be
\define{normalized} if it vanishes when any of its entries is 1. Every
cohomology class can be represented by a normalized cocycle. Finally, when $H =
\R$ with trivial $G$ action, we omit it when writing the complex
$C^\bullet(G)$, and we call real-valued cochains, cocycles, or coboundaries,
simply cochains, cocycles or coboundaries, respectively.

This last choice, that $\R$ will be our default coefficient group, may seem
innocuous, but there is another one-dimensional abelian Lie group we might have
chosen: $\U(1)$, the group of phases. This would have been an equally valid
choice, and perhaps better for some physical applications, but we have chosen
$\R$ because it simplifies our formulas slightly.

We now sketch how to build a slim Lie $n$-group from an $(n+1)$-cocycle. In
essence, given a normalized $H$-valued $(n+1)$-cocycle $a$ on a Lie group $G$,
we want to construct a Lie $n$-group $\Brane_a(G,H)$, which is the smooth, weak
$n$-groupoid with:
\begin{itemize}
	\item One object. We can depict this with a dot, or `0-cell': $\bullet$

	\item For each element $g \in G$, a 1-automorphism of the one object,
		which we depict as an arrow, or `1-cell':
		\[ \xymatrix{ \bullet \ar[r]^g & \bullet }, \quad g \in G. \]
		Composition corresponds to multiplication in the group:
		\[ \xymatrix{ \bullet \ar[r]^g & \bullet \ar[r]^{g'} & \bullet } = \xymatrix{ \bullet \ar[r]^{gg'} & \bullet }. \]
	\item Trivial $k$-morphisms for $1 < k < n$. If we depict 2-morphisms
		with 2-cells, 3-morphisms with 3-cells, then we are saying
		there is just one of each of these (the identity) up to level
		$n-1$:
		
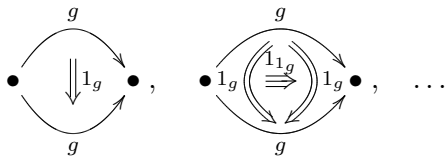
\[
		\xy
		(-8,0)*+{\bullet}="4";
		(8,0)*+{\bullet}="6";
		{\ar@/^1.65pc/^g "4";"6"};
		{\ar@/_1.65pc/_g "4";"6"};
		{\ar@{=>}^{1_g} (0,3)*{};(0,-3)*{}} ;
		\endxy
		, \quad
		\xy 
		(-10,0)*+{\bullet}="1";
		(10,0)*+{\bullet}="2";
		{\ar@/^1.65pc/^g "1";"2"};
		{\ar@/_1.65pc/_g "1";"2"};
		(0,5)*+{}="A";
		(0,-5)*+{}="B";
		{\ar@{=>}@/_.75pc/ "A"+(-1.33,0) ; "B"+(-.66,-.55)};
		{\ar@{=}@/_.75pc/ "A"+(-1.33,0) ; "B"+(-1.33,0)};
		{\ar@{=>}@/^.75pc/ "A"+(1.33,0) ; "B"+(.66,-.55)};
		{\ar@{=}@/^.75pc/ "A"+(1.33,0) ; "B"+(1.33,0)};
		{\ar@3{->} (-2,0)*{}; (2,0)*{}};
		(0,2.5)*{\scriptstyle 1_{1_g}};
		(-7,0)*{\scriptstyle 1_g};
		(7,0)*{\scriptstyle 1_g};
		\endxy
		, \quad \dots
		\] 

	\item For each element $h \in H$, an $n$-automorphism on the identity
		of the identity of \dots the identity of the 1-morphism $g$,
		and no $n$-morphisms which are not $n$-automorphisms. For
		example, when $n = 3$, we have:
		
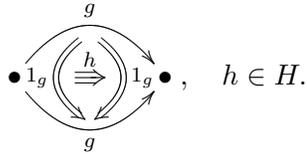
\[
		\xy 
		(-10,0)*+{\bullet}="1";
		(10,0)*+{\bullet}="2";
		{\ar@/^1.65pc/^g "1";"2"};
		{\ar@/_1.65pc/_g "1";"2"};
		(0,5)*+{}="A";
		(0,-5)*+{}="B";
		{\ar@{=>}@/_.75pc/ "A"+(-1.33,0) ; "B"+(-.66,-.55)};
		{\ar@{=}@/_.75pc/ "A"+(-1.33,0) ; "B"+(-1.33,0)};
		{\ar@{=>}@/^.75pc/ "A"+(1.33,0) ; "B"+(.66,-.55)};
		{\ar@{=}@/^.75pc/ "A"+(1.33,0) ; "B"+(1.33,0)};
		{\ar@3{->} (-2,0)*{}; (2,0)*{}};
		(0,2.5)*{\scriptstyle h};
		(-7,0)*{\scriptstyle 1_g};
		(7,0)*{\scriptstyle 1_g};
		\endxy
		, \quad h \in H.
		\]

	\item There are $n$ ways of composing $n$-morphisms, given by different
		ways of sticking $n$-cells together. For example, when $n = 3$,
		we can glue two 3-cells along a 2-cell, which should just
		correspond to addition in $H$:
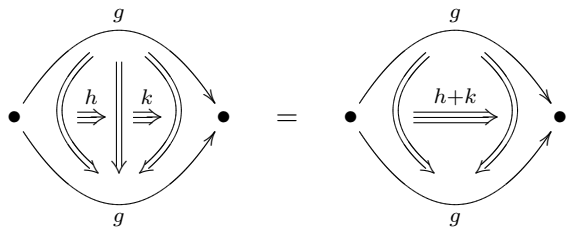
		\[
		\xy 0;/r.22pc/:
		(0,15)*{};
		(0,-15)*{};
		(0,8)*{}="A";
		(0,-8)*{}="B";
		{\ar@{=>} "A" ; "B"};
		{\ar@{=>}@/_1pc/ "A"+(-4,1) ; "B"+(-3,0)};
		{\ar@{=}@/_1pc/ "A"+(-4,1) ; "B"+(-4,1)};
		{\ar@{=>}@/^1pc/ "A"+(4,1) ; "B"+(3,0)};
		{\ar@{=}@/^1pc/ "A"+(4,1) ; "B"+(4,1)};
		{\ar@3{->} (-6,0)*{} ; (-2,0)*+{}};
		(-4,3)*{\scriptstyle h};
		{\ar@3{->} (2,0)*{} ; (6,0)*+{}};
		(4,3)*{\scriptstyle k};
		(-15,0)*+{\bullet}="1";
		(15,0)*+{\bullet}="2";
		{\ar@/^2.75pc/^g "1";"2"};
		{\ar@/_2.75pc/_g "1";"2"};
		\endxy 
		\quad = \quad
		\xy 0;/r.22pc/:
		(0,15)*{};
		(0,-15)*{};
		(0,8)*{}="A";
		(0,-8)*{}="B";
		{\ar@{=>}@/_1pc/ "A"+(-4,1) ; "B"+(-3,0)};
		{\ar@{=}@/_1pc/ "A"+(-4,1) ; "B"+(-4,1)};
		{\ar@{=>}@/^1pc/ "A"+(4,1) ; "B"+(3,0)};
		{\ar@{=}@/^1pc/ "A"+(4,1) ; "B"+(4,1)};
		{\ar@3{->} (-6,0)*{} ; (6,0)*+{}};
		(0,3)*{\scriptstyle h + k};
		(-15,0)*+{\bullet}="1";
		(15,0)*+{\bullet}="2";
		{\ar@/^2.75pc/^g "1";"2"};
		{\ar@/_2.75pc/_g "1";"2"};
		\endxy .
		\]
		We also can glue two 3-cells along a 1-cell, which should again
		just be addition in $H$:
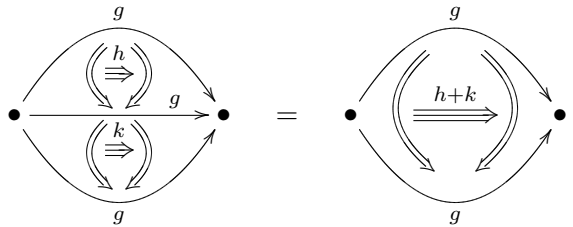
		\[	
		\xy 0;/r.22pc/:
		(0,15)*{};
		(0,-15)*{};
		(0,9)*{}="A";
		(0,1)*{}="B";
		{\ar@{=>}@/_.5pc/ "A"+(-2,1) ; "B"+(-1,0)};
		{\ar@{=}@/_.5pc/ "A"+(-2,1) ; "B"+(-2,1)};
		{\ar@{=>}@/^.5pc/ "A"+(2,1) ; "B"+(1,0)};
		{\ar@{=}@/^.5pc/ "A"+(2,1) ; "B"+(2,1)};
		{\ar@3{->} (-2,6)*{} ; (2,6)*+{}};
		(0,9)*{\scriptstyle h};
		(0,-1)*{}="A";
		(0,-9)*{}="B";
		{\ar@{=>}@/_.5pc/ "A"+(-2,-1) ; "B"+(-1,-1.5)};
		{\ar@{=}@/_.5pc/ "A"+(-2,0) ; "B"+(-2,-.7)};
		{\ar@{=>}@/^.5pc/ "A"+(2,-1) ; "B"+(1,-1.5)};
		{\ar@{=}@/^.5pc/ "A"+(2,0) ; "B"+(2,-.7)};
		{\ar@3{->} (-2,-5)*{} ; (2,-5)*+{}};
		(0,-2)*{\scriptstyle k};
		(-15,0)*+{\bullet}="1";
		(15,0)*+{\bullet}="2";
		{\ar@/^2.75pc/^g "1";"2"};
		{\ar@/_2.75pc/_g "1";"2"};
		{\ar "1";"2"};
		(8,2)*{\scriptstyle g};
		\endxy 
		\quad = \quad
		\xy 0;/r.22pc/:
		(0,15)*{};
		(0,-15)*{};
		(0,8)*{}="A";
		(0,-8)*{}="B";
		{\ar@{=>}@/_1pc/ "A"+(-4,1) ; "B"+(-3,0)};
		{\ar@{=}@/_1pc/ "A"+(-4,1) ; "B"+(-4,1)};
		{\ar@{=>}@/^1pc/ "A"+(4,1) ; "B"+(3,0)};
		{\ar@{=}@/^1pc/ "A"+(4,1) ; "B"+(4,1)};
		{\ar@3{->} (-6,0)*{} ; (6,0)*+{}};
		(0,3)*{\scriptstyle h + k};
		(-15,0)*+{\bullet}="1";
		(15,0)*+{\bullet}="2";
		{\ar@/^2.75pc/^g "1";"2"};
		{\ar@/_2.75pc/_g "1";"2"};
		\endxy  .
		\]
		And finally, we can glue two 3-cells at the 0-cell, the object
		$\bullet$.  This is the only composition of $n$-morphisms where
		the attached 1-morphisms can be distinct, which distinguishes
		it from the first two cases. It should be addition
		\emph{twisted by the action of $G$}:
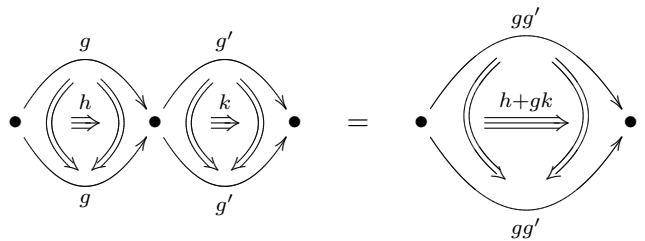
		\[
		\xy 0;/r.22pc/:
		(0,15)*{};
		(0,-15)*{};
		(-20,0)*+{\bullet}="1";
		(0,0)*+{\bullet}="2";
		{\ar@/^2pc/^g "1";"2"};
		{\ar@/_2pc/_g "1";"2"};
		(20,0)*+{\bullet}="3";
		{\ar@/^2pc/^{g'} "2";"3"};
		{\ar@/_2pc/_{g'} "2";"3"};
		(-10,6)*+{}="A";
		(-10,-6)*+{}="B";
		{\ar@{=>}@/_.7pc/ "A"+(-2,0) ; "B"+(-1,-.8)};
		{\ar@{=}@/_.7pc/ "A"+(-2,0) ; "B"+(-2,0)};
		{\ar@{=>}@/^.7pc/ "A"+(2,0) ; "B"+(1,-.8)};
		{\ar@{=}@/^.7pc/ "A"+(2,0) ; "B"+(2,0)};
		(10,6)*+{}="A";
		(10,-6)*+{}="B";
		{\ar@{=>}@/_.7pc/ "A"+(-2,0) ; "B"+(-1,-.8)};
		{\ar@{=}@/_.7pc/ "A"+(-2,0) ; "B"+(-2,0)};
		{\ar@{=>}@/^.7pc/ "A"+(2,0) ; "B"+(1,-.8)};
		{\ar@{=}@/^.7pc/ "A"+(2,0) ; "B"+(2,0)};
		{\ar@3{->} (-12,0)*{}; (-8,0)*{}};
		(-10,3)*{\scriptstyle h};
		{\ar@3{->} (8,0)*{}; (12,0)*{}};
		(10,3)*{\scriptstyle k};
		\endxy 
		\quad = \quad 
		\xy 0;/r.22pc/:
		(0,15)*{};
		(0,-15)*{};
		(0,8)*{}="A";
		(0,-8)*{}="B";
		{\ar@{=>}@/_1pc/ "A"+(-4,1) ; "B"+(-3,0)};
		{\ar@{=}@/_1pc/ "A"+(-4,1) ; "B"+(-4,1)};
		{\ar@{=>}@/^1pc/ "A"+(4,1) ; "B"+(3,0)};
		{\ar@{=}@/^1pc/ "A"+(4,1) ; "B"+(4,1)};
		{\ar@3{->} (-6,0)*{} ; (6,0)*+{}};
		(0,3)*{\scriptstyle h + g k};
		(-15,0)*+{\bullet}="1";
		(15,0)*+{\bullet}="2";
		{\ar@/^2.75pc/^{gg'} "1";"2"};
		{\ar@/_2.75pc/_{gg'} "1";"2"};
		\endxy .
		\]
		For arbitary $n$, we define all $n$ compositions to be addition
		in $H$, except for gluing at the object, where it is addition
		twisted by the action.

	\item For any $(n+1)$-tuple of 1-morphisms, an $n$-automorphism $a(g_1,
		g_2, \dots, g_{n+1})$ on the identity of the identity of \dots
		the identity of the 1-morphism $g_1 g_2 \dots g_{n+1}$. We call
		$a$ the \define{$n$-associator}.

	\item $a$ satisfies an equation corresponding to the $n$-dimensional
		associahedron, which is equivalent to the cocycle condition.
\end{itemize}
In principle, it should be possible to take a globular definition of
$n$-category, such as that of Batanin or Trimble, and fill out this sketch to
make it a real definition of an $n$-group. Doing this here, however, would lead us
too far afield from our goal, for which we only need 2- and 3-groups.  So let
us flesh out these cases. The reader interested in learning more about the
various definitions of $n$-categories should consult Leinster's survey
\cite{Leinster:ncat} or Cheng and Lauda's guidebook \cite{ChengLauda}. 

\section{Lie 2-groups}

Speaking precisely, a \define{2-group} is a bicategory with one object in which
all 1-morphisms and 2-morphisms are weakly invertible. Rather than plain
2-groups, we are interested in \emph{Lie} 2-groups, where all the structure in
sight is smooth. So, we really need a bicategory `internal to the category of
smooth manifolds', or a `smooth bicategory'. To this end, we will give an
especially long and unfamiliar definition of bicategory, isolating each
operation and piece of data so that we can indicate its smoothness.  Readers
not familiar with bicategories are encouraged to read the introduction by
Leinster \cite{Leinster:bicat}. 

Before we give this definition, let us review the idea of a `bicategory', so
that its basic simplicity is not obscured in technicalities. A bicategory has
objects:
\[ x \, \bullet, \]
morphisms going between objects,
\[ \xymatrix{ x \, \bullet \ar[r]^f & \bullet \, y}, \]
and 2-morphisms going between morphisms:
\[
\xy
(-10,0)*+{x};
(-8,0)*+{\bullet}="4";
(8,0)*+{\bullet}="6";
(10,0)*+{y};
{\ar@/^1.65pc/^f "4";"6"};
{\ar@/_1.65pc/_g "4";"6"};
{\ar@{=>}^{\scriptstyle \alpha} (0,3)*{};(0,-3)*{}} ;
\endxy .
\] 
Morphisms in a bicategory can be composed just as morphisms in a category:
\[ \xymatrix{ x \ar[r]^f & y \ar[r]^g & z } \quad = \quad \xymatrix{ x \ar[r]^{f \cdot g} & z } . \]
But there are two ways to compose 2-morphisms---vertically:
\[
\xy
(-8,0)*+{x}="4";
(8,0)*+{y}="6";
{\ar^g "4";"6"};
{\ar@/^1.75pc/^{f} "4";"6"};
{\ar@/_1.75pc/_{h} "4";"6"};
{\ar@{=>}^<<{\scriptstyle \alpha} (0,6)*{};(0,1)*{}} ;
{\ar@{=>}^<<{\scriptstyle \beta} (0,-1)*{};(0,-6)*{}} ;
\endxy
\quad = \quad \xy
(-8,0)*+{x}="4";
(8,0)*+{y}="6";
{\ar@/^1.65pc/^f "4";"6"};
{\ar@/_1.65pc/_h "4";"6"};
{\ar@{=>}^{\scriptstyle \alpha \circ \beta } (0,3)*{};(0,-3)*{}} ;
\endxy
\] 
and horizontally:
\[
\xy
(-16,0)*+{x}="4";
(0,0)*+{y}="6";
{\ar@/^1.65pc/^{f} "4";"6"};
{\ar@/_1.65pc/_{g} "4";"6"};
{\ar@{=>}^<<<{\scriptstyle \alpha} (-8,3)*{};(-8,-3)*{}} ;
(0,0)*+{y}="4";
(16,0)*+{z}="6";
{\ar@/^1.65pc/^{f'} "4";"6"};
{\ar@/_1.65pc/_{g'} "4";"6"};
{\ar@{=>}^<<<{\scriptstyle \beta} (8,3)*{};(8,-3)*{}} ;
\endxy
\quad = \quad \xy
(-10,0)*+{x}="4";
(10,0)*+{z}="6";
{\ar@/^1.65pc/^{f \cdot f'} "4";"6"};
{\ar@/_1.65pc/_{g \cdot g'} "4";"6"};
{\ar@{=>}^{\alpha \cdot \beta} (0,3)*{};(0,-3)*{}} ;
\endxy .
\] 
Unlike a category, composition of morphisms need not be associative or have
left and right units. The presence of 2-morphisms allow us to \emph{weaken the
axioms}.  Rather than demanding $(f \cdot g) \cdot h = f \cdot (g \cdot h)$,
for composable morphisms $f, g$ and $h$, the presence of 2-morphisms allow for
the weaker condition that these two expressions are merely isomorphic:
 \[ a(f,g,h) \maps (f \cdot g) \cdot h \Rightarrow f \cdot (g \cdot h), \]
where $a(f,g,h)$ is an 2-isomorphism called the \define{associator}. In the
same vein, rather than demanding that:
\[ 1_x \cdot f = f = f \cdot 1_y, \]
for $f \maps x \to y$, and identities $1_x \maps x \to x$ and $1_y \maps y \to
y$, the presence of 2-morphisms allow us to weaken these equations to
isomorphisms:
\[ l(f) \maps 1_x \cdot f \Rightarrow f, \quad r(f) \maps f \cdot 1_y \Rightarrow f. \]
Here, $l(f)$ and $r(f)$ are 2-isomorphisms called the \define{left and right
unitors}.

Of course, these 2-isomorphisms obey rules of their own. The associator
satisfies its own axiom, called the \define{pentagon identity}, which says that
this pentagon commutes:
\[
\xy
 (0,20)*+{(f g) (h k)}="1";
 (40,0)*+{f (g (h k))}="2";
 (25,-20)*{ \quad f ((g h) k)}="3";
 (-25,-20)*+{(f (g h)) k}="4";
 (-40,0)*+{((f g) h) k}="5";
 {\ar@{=>}^{a(f,g,h k)}     "1";"2"}
 {\ar@{=>}_{1_f \cdot a_(g,h,k)}  "3";"2"}
 {\ar@{=>}^{a(f,g h,k)}    "4";"3"}
 {\ar@{=>}_{a(f,g,h) \cdot 1_k}  "5";"4"}
 {\ar@{=>}^{a(fg,h,k)}    "5";"1"}
\endxy
\]
Finally, the associator and left and right unitors satisfy the \define{triangle
identity}, which says the following triangle commutes:
\[ 
\xy
(-20,10)*+{(f 1) g}="1";
(20,10)*+{f (1 g)}="2";
(0,-10)*+{f g}="3";
{\ar@{=>}^{a(f,1,g)}	"1";"2"}
{\ar@{=>}_{r(f) \cdot 1_g}	"1";"3"}
{\ar@{=>}^{1_f \cdot l(g)} "2";"3"}
\endxy
\]

A word of caution is needed here before we proceed: \emph{in this chapter 
only}, we are bucking standard mathematical practice by writing the result of
doing first $\alpha$ and then $\beta$ as $\alpha \circ \beta$ rather than
$\beta \circ \alpha$, as one would do in most contexts where $\circ$ denotes
composition of \emph{functions}. This has the effect of changing how we read
commutative diagrams. For instance, the commutative triangle:
\[ \xymatrix{ f \ar[r]^\alpha \ar[rd]_\gamma & g \ar[d]^\beta \\
		         & h \\
}
\]
reads $\gamma = \alpha \circ \beta$ rather than $\gamma = \beta \circ \alpha$.

We shall now give the full definition, not of a bicategory, but of a `smooth
bicategory'. To do this, we use the idea of internalization. Dating back to
Ehresmann \cite{Ehresmann} in the 1960s, internalization has become a standard tool
of the working category theorist. The idea is based on a familiar one: any mathematical
structure that can be defined using sets, functions, and equations between
functions can be defined in categories other than Set. For instance, a group in
the category of smooth manifolds is a Lie group. To perform internalization, we
apply this idea to the definition of category itself. We recall the essentials
here to define `smooth categories'. More generally, one can define a `category
in $K$' for many categories $K$, though here we will work exclusively with the
example where $K$ is the category of smooth manifolds. For a readable treatment
of internalization, see Borceux's handbook \cite{Borceux}.

\begin{defn}  A \define{smooth category} $C$ consists of
\begin{itemize}
	\item a \define{smooth manifold of objects} $C_{0}$;
	\item a \define{smooth manifold of morphisms} $C_1$;
\end{itemize}
together with
\begin{itemize}
	\item smooth {\bf source} and {\bf target} maps $s,t \maps C_{1} \rightarrow C_{0}$, 
	\item a smooth {\bf identity-assigning} map $i \maps C_{0} \rightarrow C_{1}$, 
	\item a smooth {\bf composition} map $\circ \maps C_{1} \times _{C_{0}}
		C_{1} \rightarrow C_{1}$, where $C_1 \times_{C_0} C_1$ is the
		pullback of the source and target maps:
		\[ C_1 \times_{C_0} C_1 = \left\{ (f,g) \in C_1 \times C_1 : t(f) = s(g) \right\}, \]
		and is assumed to be a smooth manifold.
\end{itemize}
such that the following diagrams commute, expressing the usual category laws:
\begin{itemize}
	\item laws specifying the source and target of identity morphisms:
	\[
	\xymatrix{
 C_{0}
   \ar[r]^{i}
   \ar[dr]_{1}
   & C_{1}
   \ar[d]^{s} \\
  & C_{0} }
\hspace{.2in} \xymatrix{
   C_{0}
   \ar[r]^{i}
   \ar[dr]_{1}
   & C_{1}
   \ar[d]^{t} \\
  & C_{0}}
\]
	\item laws specifying the source and target of composite
	morphisms:
	\[
\xymatrix{ C_{1} \times _{C_{0}} C_{1}
  \ar[rr]^{\circ}
  \ar[dd]_{p_{1}}
  && C_{1}
  \ar[dd]^{s} \\ \\
C_{1}
  \ar[rr]^{s}
  && C_{0} }
  \hspace{.2in}
\xymatrix{ C_{1} \times_{C_{0}} C_{1}
  \ar[rr]^{\circ}
  \ar[dd]_{p_{2}}
   && C_{1}
  \ar[dd]^{t} \\ \\
   C_{1}
  \ar[rr]^{t}
   && C_{0} }
\]
	\item the associative law for composition of morphisms:
	\[
	\xymatrix{ C_{1} \times _{C_{0}} C_{1} \times _{C_{0}} C_{1}
  \ar[rr]^{\circ \times_{C_{0}} 1}
  \ar[dd]_{1 \times_{C_{0}} \circ}
   && C_{1} \times_{C_{0}} C_{1}
  \ar[dd]^{\circ} \\ \\
   C_{1} \times _{C_{0}} C_{1}
  \ar[rr]^{\circ}
   && C_{1} }
\]
	\item the left and right unit laws for composition of morphisms:
	\[
	\xymatrix{ C_{0} \times _{C_{0}} C_{1}
  \ar[r]^{i \times 1}
  \ar[ddr]_{p_2}
   & C_{1} \times _{C_{0}} C_{1}
  \ar[dd]^{\circ}
   & C_{1} \times_{C_{0}} C_{0}
  \ar[l]_{1 \times i}
  \ar[ddl]^{p_1} \\ \\
   & C_{1} }
\]
\end{itemize}
\end{defn}

The existence of pullbacks in the category of smooth manifolds is a delicate
issue. When working with categories internal to some category $K$, it is
customary to assume $K$ contains all pullbacks, but this is merely a
convenience. All the definitions still work as long as the existence of each
required pullback is implicit. 

To define smooth bicategories, we must first define smooth functors and natural
transformations:

\begin{defn} 
Given smooth categories $C$ and $C'$, a {\bf smooth functor} $F \maps C \to C'$
consists of:
\begin{itemize}
	\item a smooth map on objects, $F_{0} \maps C_{0} \to C_{0}'$; 
	\item a smooth map on morphisms, $F_{1} \maps C_{1} \rightarrow C_{1}'$;
\end{itemize}
such that the following diagrams commute, corresponding to the usual laws satisfied by a functor:
\begin{itemize}
\item preservation of source and target:
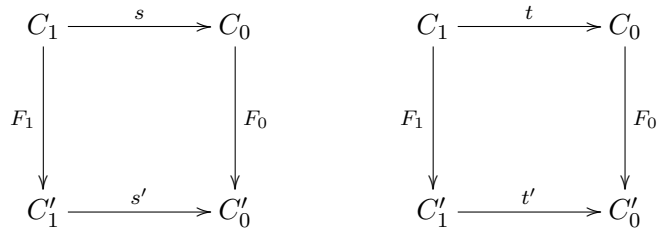
\[
\xymatrix{ C_{1} \ar[rr]^{s} \ar[dd]_{F_{1}}
 && C_{0}
\ar[dd]^{F_{0}} \\ \\
 C_{1}'
\ar[rr]^{s'}
 && C_{0}' }
\qquad \qquad \xymatrix{ C_{1} \ar[rr]^{t} \ar[dd]_{F_{1}}
 && C_{0}
\ar[dd]^{F_{0}} \\ \\
 C_{1}'
\ar[rr]^{t'}
 && C_{0}' }
\]
\item preservation of identity morphisms:
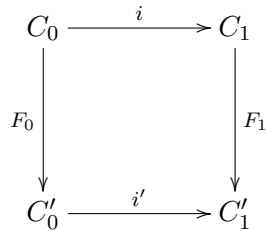
\[
\xymatrix{
 C_{0}
\ar[rr]^{i} \ar[dd]_{F_{0}}
 && C_{1}
\ar[dd]^{F_{1}} \\ \\
 C_{0}'
\ar[rr]^{i'}
 && C_{1}' }
\]
\item preservation of composite morphisms:
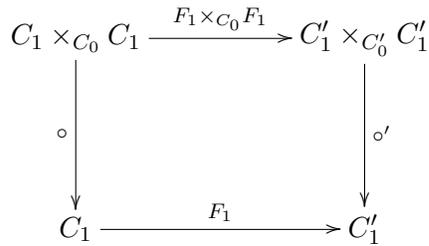
\[
\xymatrix{ C_{1} \times _{C_{0}} C_{1}
 \ar[rr]^{F_{1} \times_{C_0} F_{1}}
 \ar[dd]_{\circ}
  && C_{1}' \times_{C_{0}'} C_{1}'
 \ar[dd]^{\circ'} \\ \\
  C_{1}
 \ar[rr]^{F_{1}}
  && C_{1}' }
\]
\end{itemize}
\end{defn}

\begin{defn}  Given categories smooth categories $C$ and $C'$, and smooth
	functors $F,G \maps C \to C'$, a {\bf smooth natural transformation}
	$\theta \maps F \To G$ is a smooth map $\theta \maps C_0 \to C'_1$ for
	which the following diagrams commute, expressing the usual laws
	satisfied by a natural transformation: 
	\begin{itemize}
		\item laws specifying the source and target of the natural
		transformation:
		\[
		 \xymatrix{C_0 \ar[dr]^F \ar[d]_{\theta} \\ C'_1 \ar[r]_s & C'_0 }
		 \qquad \qquad
		 \xymatrix{C_0 \ar[dr]^G \ar[d]_{\theta} \\ C'_1 \ar[r]_t & C'_0 }
		\]
		\item the commutative square law:
		\[  \xymatrix{
		C_1
		 \ar[rr]^{\Delta (s\theta \times G)}
		 \ar[dd]_{\Delta (F \times t\theta)}
		  && C'_1 \times_{C_0} C'_1
		 \ar[dd]^{\circ'} \\ \\
		  C'_1 \times_{C_0} C'_1
		 \ar[rr]^{\circ'}
		  && C'_1
		}
		\]
	\end{itemize}
\end{defn}

Now we know enough about smooth category theory to bootstrap the definition of
smooth bicategories.  We do this in a somewhat nonstandard way: we make use of
the fact that the morphisms and 2-morphisms of a bicategory form an ordinary
category under vertical composition. Generalizing this, the morphisms and
2-morphisms in a smooth bicategory should form, by themselves, a smooth
category.  We can then define horizontal composition as a smooth functor, and
introduce the associator and left and right unitors as smooth natural
transformations between certain functors.  In detail: 
\begin{defn} \label{def:smoothbicat}
	A \define{smooth bicategory} $B$ consists of 
\begin{itemize}
	\item a \define{manifold of objects} $B_0$;
	\item a \define{manifold of morphisms} $B_1$;
	\item a \define{manifold of 2-morphisms} $B_2$;
\end{itemize}
equipped with:
\begin{itemize}
	\item a smooth category structure on $\underline{\Mor} B$, with
		\begin{itemize}
			\item $B_1$ as the smooth manifold of objects;
			\item $B_2$ as the smooth manifold of morphisms;
		\end{itemize}
		The composition in $\underline{\Mor} B$ is called
		\define{vertical composition} and denoted $\circ$.
	\item smooth \define{source} and \define{target maps}:
		\[ s, t \maps B_1 \to B_0. \]
	\item a smooth \define{identity-assigning map}:
		\[ i \maps B_0 \to B_1. \]
	\item a smooth \define{horizontal composition} functor:
		\[ \cdot \maps \underline{\Mor} B \times_{B_0} \underline{\Mor} B \to \underline{\Mor} B . \]
		That is, a pair of smooth maps:
		\[ \cdot \maps B_1 \times_{B_0} B_1 \to B_1  \]
		\[ \cdot \maps B_2 \times_{B_0} B_2 \to B_2, \]
		satisfying the axioms for a functor.
	\item a smooth natural transformation, the \define{associator}:
		\[ a(f,g,h) \maps (f \cdot g) \cdot h \To f \cdot (g \cdot h). \]
	\item smooth natural transformations, the \define{left} and
		\define{right unitors}, which are both trivial in the
		bicategories we consider:
		\[ l(f) \maps 1 \cdot f \To f, \quad r(f) \maps f \cdot 1 \To f. \]
\end{itemize}
such that the following diagrams commute, expressing the same laws regarding
sources, targets and identities as with a smooth category, and two new laws
expressing the compatibility of the various source and target maps:
\begin{itemize}
\item laws specifying the source and target of identity morphisms:
\[
\xymatrix{
 B_{0}
   \ar[r]^{i}
   \ar[dr]_{1}
   & B_{1}
   \ar[d]^{s} \\
  & B_{0} }
\hspace{.2in} \xymatrix{
   B_{0}
   \ar[r]^{i}
   \ar[dr]_{1}
   & B_{1}
   \ar[d]^{t} \\
  & B_{0}}
\]
\item laws specifying the source and target of the horizontal composite
of 1-morphisms:
\[
\xymatrix{ B_1 \times _{B_0} B_1
  \ar[rr]^{\cdot}
  \ar[dd]_{p_{1}}
  && B_1
  \ar[dd]^{t} \\ \\
B_1
  \ar[rr]^{t}
  && B_0 }
  \hspace{.2in}
\xymatrix{ B_1 \times_{B_0} B_1
  \ar[rr]^{\cdot}
  \ar[dd]_{p_{2}}
   && B_1
  \ar[dd]^{s} \\ \\
   B_1
  \ar[rr]^{s}
   && B_0 }
\]

\item laws expressing the compatibility of source and target maps:
\[ \xymatrix{
		B_2 \ar[rr]^s \ar[dd]_t & &  B_1 \ar[dd]^s \\
		                        & & \\
		B_1 \ar[rr]_s           & & B_0 \\
	}
	\hspace{.2in}
\xymatrix{
		B_2 \ar[rr]^t \ar[dd]_s & &  B_1 \ar[dd]^t \\
		                        & & \\
		B_1 \ar[rr]_t           & & B_0 \\
	}
\]
\end{itemize}
Finally, associator and left and right unitors satisfy some laws of their
own---the following diagrams commute:
\begin{itemize}
 \item the {\bf pentagon identity} for the associator:
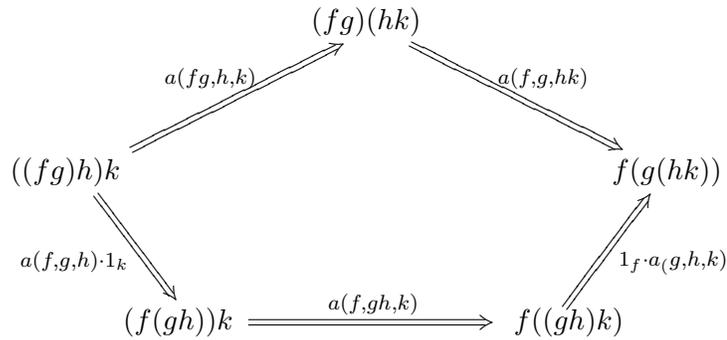
\[
\xy
 (0,20)*+{(f g) (h k)}="1";
 (40,0)*+{f (g (h k))}="2";
 (25,-20)*{ \quad f ((g h) k)}="3";
 (-25,-20)*+{(f (g h)) k}="4";
 (-40,0)*+{((f g) h) k}="5";
 {\ar@{=>}^{a(f,g,h k)}     "1";"2"}
 {\ar@{=>}_{1_f \cdot a_(g,h,k)}  "3";"2"}
 {\ar@{=>}^{a(f,g h,k)}    "4";"3"}
 {\ar@{=>}_{a(f,g,h) \cdot 1_k}  "5";"4"}
 {\ar@{=>}^{a(fg,h,k)}    "5";"1"}
\endxy
\]
for any four composable morphisms $f$, $g$, $h$ and $k$.
\item the {\bf triangle identity} for the left and right unit
laws:
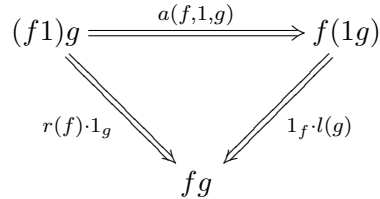
\[ 
\xy
(-20,10)*+{(f 1) g}="1";
(20,10)*+{f (1 g)}="2";
(0,-10)*+{f g}="3";
{\ar@{=>}^{a(f,1,g)}	"1";"2"}
{\ar@{=>}_{r(f) \cdot 1_g}	"1";"3"}
{\ar@{=>}^{1_f \cdot l(g)} "2";"3"}
\endxy
\]
for any two composable morphisms $f$ and $g$.
\end{itemize}
\end{defn}

This definition of smooth bicategory may seem so long that checking it is
utterly intimidating, but we shall see an example in a moment where this is
easy. This will be an example of a \define{Lie 2-group}, a smooth bicategory
with one object where all morphisms are weakly invertible, and all 2-morphisms
are strictly invertible.

Secretly, the pentagon identity is a cocycle condition, as we shall now see.
Given a normalized $H$-valued 3-cocycle $a$ on a Lie group $G$, we can
construct a Lie 2-group $\String_a(G,H)$ with:
\begin{itemize}
	\item One object, $\bullet$, regarded as a manifold in the trivial way.
	\item For each element $g \in G$, an automorphism of the one object:
		\[ \bullet \stackrel{g}{\longrightarrow} \bullet . \]
		Horizontal composition given by multiplication in the group:
		\[ \cdot \maps G \times G \to G. \]
		Note that source and target maps are necessarily trivial. The
		identity-assigning map takes the one object to $1 \in G$.
	\item For each $h \in H$, a 3-automorphism of the 2-morphism $1_g$,
		and no 3-morphisms between distinct 2-morphisms:
		\[
		\xy
		(-8,0)*+{\bullet}="4";
		(8,0)*+{\bullet}="6";
		{\ar@/^1.65pc/^g "4";"6"};
		{\ar@/_1.65pc/_g "4";"6"};
		{\ar@{=>}^h (0,3)*{};(0,-3)*{}} ;
		\endxy, \quad h \in H .
		\] 
		Thus the space of all 2-morphisms is $G \times H$, and
		the source and target maps are projection onto the first
		factor. The identity-assigning map takes each element of $G$ to
		$0 \in H$.
	\item Two kinds of composition of 2-morphisms: given a pair of
		2-morphisms on the same morphism, vertical compostion is given
		by addition in $H$:
		\[
		\xy
		(-8,0)*+{\bullet}="4";
		(8,0)*+{\bullet}="6";
		{\ar "4";"6"};
		{\ar@/^1.75pc/^{g} "4";"6"};
		{\ar@/_1.75pc/_{g} "4";"6"};
		{\ar@{=>}^<<{h} (0,6)*{};(0,1)*{}} ;
		{\ar@{=>}^<<{h'} (0,-1)*{};(0,-6)*{}} ;
		\endxy
		\quad = \quad \xy
		(-8,0)*+{\bullet}="4";
		(8,0)*+{\bullet}="6";
		{\ar@/^1.65pc/^g "4";"6"};
		{\ar@/_1.65pc/_g "4";"6"};
		{\ar@{=>}^{h + h'} (0,3)*{};(0,-3)*{}} ;
		\endxy .
		\] 
		That is, vertical composition is just the map:
		\[ \circ = 1 \times + \maps G \times H \times H \to G \times H. \]
		where we have used the fact that the pullback of 2-morphisms
		over the one object is trivially:
		\[ (G \times H) \times_\bullet (G \times H) \iso G \times H \times H. \]
		Given a pair of 2-morphisms on different morphisms, horizontal
		composition is addition \emph{twisted by the action of $G$}:
		\[
		\xy
		(-16,0)*+{\bullet}="4";
		(0,0)*+{\bullet}="6";
		{\ar@/^1.65pc/^{g} "4";"6"};
		{\ar@/_1.65pc/_{g} "4";"6"};
		{\ar@{=>}^<<<{h} (-8,3)*{};(-8,-3)*{}} ;
		(0,0)*+{\bullet}="4";
		(16,0)*+{\bullet}="6";
		{\ar@/^1.65pc/^{g'} "4";"6"};
		{\ar@/_1.65pc/_{g'} "4";"6"};
		{\ar@{=>}^<<<{h'} (8,3)*{};(8,-3)*{}} ;
		\endxy
		\quad = \quad \xy
		(-12,0)*+{\bullet}="4";
		(12,0)*+{\bullet}="6";
		{\ar@/^1.65pc/^{g g'} "4";"6"};
		{\ar@/_1.65pc/_{g g'} "4";"6"};
		{\ar@{=>}^{h + gh'} (0,3)*{};(0,-3)*{}} ;
		\endxy .
		\] 
		Or, in terms of a map, this is the multiplication on the
		semidirect product, $G \ltimes H$:
		\[ \cdot \maps (G \ltimes H) \times (G \ltimes H) \to G \ltimes H. \]
	\item For any triple of morphisms, a 2-isomorphism, the
		associator:
		\[ a(g_1,g_2,g_3) \maps g_1 g_2 g_3 \to g_1 g_2 g_3, \]
		given by the 3-cocycle $a \maps G^3 \to H$, where by a slight
		abuse of definitions we think of this 2-isomorphism as living
		in $H$ rather than $G \times H$, because the source (and
		target) are understood to be $g_1 g_2 g_3$.
	\item The left and right unitors are trivial.
\end{itemize}
A \define{slim Lie 2-group} is one of this form. When $H = \R$, we write simply
$\String_a(G)$ for the above Lie 2-group. It remains to check that this is, in
fact, a Lie 2-group:

\begin{prop} \label{prop:Lie2group}
	$\String_a(G,H)$ is a Lie 2-group: a smooth bicategory with one
	object in which all 1-morphisms and 2-morphisms are weakly invertible.
\end{prop}
In brief, we prove this by showing that the 3-cocycle condition implies the one
nontrivial axiom for this bicategory: the pentagon identity.
\begin{proof}
	For $\String_a(G,H)$, the left and right unitors are the identity, and
	thus the triangle identity just says $a(g_1,1,g_2) = 1$.  Or, written
	additively, $a(g_1,1,g_2) = 0$ Since $a$ is normalized, this is
	automatic.

	To check that $\String_a(G,H)$ is really a bicategory, it
	therefore remains to check the pentagon identity. This says that the
	following automorphisms of $g_1 g_2 g_3 g_4$ are equal:
	\[ a(g_1, g_2, g_3 g_4) \circ a(g_1 g_2, g_3, g_4) = (1_{g_1} \cdot a(g_2, g_3, g_4)) \circ a(g_1, g_2 g_3, g_4) \circ (a(g_1,g_2,g_3) \cdot 1_{g_4}) \]
	Or, using the definition of vertical composition:
	\[ a(g_1, g_2, g_3 g_4) + a(g_1 g_2, g_3, g_4) = (1_{g_1} \cdot a(g_2, g_3, g_4)) + a(g_1, g_2 g_3, g_4) + (a(g_1,g_2,g_3) \cdot 1_{g_4}) \]
	Finally, use the definition of the dot operation for 2-morphisms, as
	the semidirect product:
	\[ a(g_1, g_2, g_3 g_4) + a(g_1 g_2, g_3, g_4) = g_1 a(g_2, g_3, g_4)) + a(g_1, g_2 g_3, g_4) + a(g_1,g_2,g_3). \] 
	This is the 3-cocycle condition---it holds because $a$ is a 3-cocycle.

	So, $\String_a(G,H)$ is a bicategory. It is smooth because everything
	in sight is smooth: $G$, $H$, the source, target, identity-assigning,
	and composition maps, and the associator $a \maps G^3 \to H$. And it is
	a Lie 2-group: the morphisms in $G$ and 2-morphisms in $H$ are all
	strictly invertible, and thus of course they are weakly invertible.
\end{proof}

In fact, we can say something a bit stronger about $\String_a(G,H)$, if we let
$a$ be any normalized $H$-valued 3-cochain, rather requiring it to be a
cocycle. In this case, $\String_a(G,H)$ is a Lie 2-group if and only if $a$ is
a 3-cocycle, because $a$ satisfies the pentagon identity if and only if it is a
cocycle.

\section{Lie 3-groups}

We now sketch the construction of slim Lie 3-groups from a normalized 4-cocycle
$\pi$. In a sense, this is a straightforward generalization of what we have
done above, but the details must be checked against a specific definition of
3-category. We choose to use tricategories, originally defined by Gordon, Power
and Street \cite{GPS}, but extensively studied by Gurski. We use the definition
from his thesis \cite{Gurski}.

We saw in the last section that a smooth bicategory $B$ is consists of a
smooth manifold of objects, $B_0$, a smooth manifold of morphisms, $B_1$, and a
smooth manifold of 2-morphisms, $B_2$, such that:
\begin{itemize}
	\item $B_1$ and $B_2$ fit together to form a smooth category;
	\item horizontal composition is a smooth functor;
	\item satisfying associativity and left and right unit laws up to natural transformations, the associator and left and right unitors;
	\item satisfying the pentagon and triangle identities.
\end{itemize}

Here, we will define a `smooth tricategory' $T$ to consist of a smooth manifold
of objects, $T_0$, a smooth manifold of morphisms, $T_1$, a smooth manifold of
2-morphisms, $T_2$, and a smooth manifold of 3-morphisms, $T_3$, such that:
\begin{itemize}
	\item $T_1$, $T_2$ and $T_3$ fit together to form a smooth bicategory;
	\item horizontal composition is a smooth `pseudofunctor';
	\item satisfying associativity and left and right unit laws up to smooth `pseudonatural transformations', the associator and left and right unitors;
	\item satisfying the pentagon and triangle identities up to smooth `modifications';
	\item satisfying some identities of their own.
\end{itemize}
Each of the above quoted terms---pseudofunctor, pseudonatural transformations,
modification---would usually need to be defined completely in order to understand
tricategories. But we really only need modifications, because our functors and
natural transformations will not be `pseudo'. Nonetheless, so it is clear what
we leave out, let us discuss each of these terms briefly.

\begin{itemize}
	\item `Pseudofunctor' is to `bicategory' as `functor' is to `category':
		it is a map $F \maps B \to B'$ between bicategories $B$ and
		$B'$, preserving all structure in sight \emph{except}
		horizontal composition and identities, which are only preserved
		up to specified 2-isomorphisms:
		\[ F(f \cdot g) \To F(f) \cdot F(g), \quad F(1_x) \To 1_{F(x)} . \] 
		For the tricategories we construct, all pseudofunctors will be
		strict: the above 2-isomorphisms are identities.

	\item `Pseudonatural transformation' is to `pseudofunctor' as `natural
		transformation' is to 'functor': given two pseudofunctors 
		
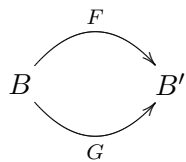
\[
		\xy 
		(-10,0)*+{B}="1";
		(10,0)*+{B'}="2";
		{\ar@/^1.65pc/^F "1";"2"};
		{\ar@/_1.65pc/_G "1";"2"};
		(0,5)*+{}="A";
		(0,-5)*+{}="B";
		\endxy
		\]
		a pseudonatural transformation is a map:
		
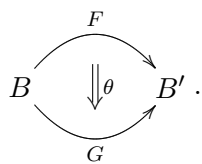
\[
		\xy 
		(-10,0)*+{B}="1";
		(10,0)*+{B'}="2";
		{\ar@/^1.65pc/^{F} "1";"2"};
		{\ar@/_1.65pc/_{G} "1";"2"};
		(0,5)*+{}="A";
		{\ar@{=>}^{\scriptstyle \theta} (0,3)*{};(0,-3)*{}} ;
		(0,-5)*+{}="B";
		\endxy .
		\]
		Like a natural transformation, this consists of a morphism for
		each object $x$ in $B$:
		\[ \theta(x) \maps F(x) \to G(x). \]
		Unlike a natural transformation, it is only natural up to a
		specified 2-isomorphism. That is, the naturality square:
		\[ \xymatrix{ F(x) \ar[r]^{F(f)} \ar[d]_{\theta(x)} & F(y) \ar[d]^{\theta(y)} \\
			G(x) \ar[r]^{G(f)} & G(y) \\
		}
		\]
		does \emph{not} commute. It is replaced with a 2-isomorphism:
		\[ \xymatrix{ F(x) \ar[r]^{F(f)} \ar[d]_{\theta(x)} & F(y) \ar[d]^{\theta(y)} \\
		\ar@{=>}[ur]_{\theta(f)} G(x) \ar[r]_{G(f)} & G(y) \\
			}
		\]
		that satisfies some equations of its own. For the tricategories
		we construct, all pseudonatural transformations will be strict:
		the 2-isomorphism above is the identity.
	\item Finally, a `modification' is something new: it is a map between
		pseudonatural transformations. Given two pseudonatural
		transformations:
		\[
		\xy 
		(-12,0)*+{B}="1";
		(12,0)*+{B'}="2";
		{\ar@/^1.65pc/^F "1";"2"};
		{\ar@/_1.65pc/_G "1";"2"};
		(0,5)*+{}="A";
		(0,-5)*+{}="B";
		{\ar@{=>}@/_.75pc/ "A"+(-1.33,0) ; "B"+(-.66,-.55)};
		(-6.5,0)*{\scriptstyle \theta};
		{\ar@{=}@/_.75pc/ "A"+(-1.33,0) ; "B"+(-1.33,0)};
		{\ar@{=>}@/^.75pc/ "A"+(1.33,0) ; "B"+(.66,-.55)};
		(6.5,0)*{\scriptstyle \eta};
		{\ar@{=}@/^.75pc/ "A"+(1.33,0) ; "B"+(1.33,0)};
		\endxy
		\]
		a modification $\Gamma$ is a map:
		\[
		\xy 
		(-12,0)*+{B}="1";
		(12,0)*+{B'}="2";
		{\ar@/^1.65pc/^F "1";"2"};
		{\ar@/_1.65pc/_G "1";"2"};
		(0,5)*+{}="A";
		(0,-5)*+{}="B";
		{\ar@{=>}@/_.75pc/ "A"+(-1.33,0) ; "B"+(-.66,-.55)};
		(-6.5,0)*{\scriptstyle \theta};
		{\ar@{=}@/_.75pc/ "A"+(-1.33,0) ; "B"+(-1.33,0)};
		{\ar@{=>}@/^.75pc/ "A"+(1.33,0) ; "B"+(.66,-.55)};
		(6.5,0)*{\scriptstyle \eta};
		{\ar@{=}@/^.75pc/ "A"+(1.33,0) ; "B"+(1.33,0)};
		{\ar@3{->} (-2,0)*{}; (2,0)*{}};
		(0,2.5)*{\scriptstyle \Gamma};
		\endxy .
		\]
		Just as a pseudonatural transformation consists of a morphism for
		each object $x$ in $B$, a modification consists of a 2-morphism
		for each object $x$ in $B$:
		\[
		\xy 
		(-12,0)*+{F(x)}="1";
		(12,0)*+{G(x)}="2";
		{\ar@/^1.65pc/^{\theta(x)} "1";"2"};
		{\ar@/_1.65pc/_{\eta(x)} "1";"2"};
		(0,5)*+{}="A";
		{\ar@{=>}^{\scriptstyle \Gamma(x)} (0,3)*{};(0,-3)*{}} ;
		(0,-5)*+{}="B";
		\endxy .
		\]
		These 2-morphisms satisfy an equation that will hold trivially
		in the tricategories we consider, so we omit it. See Leinster
		\cite{Leinster:bicat} for more details.
\end{itemize}

With these preliminaries in mind, we can now sketch the definition of a smooth
tricategory. 
\begin{defn} \label{def:smoothtricat}
	A \define{smooth tricategory} $T$ consists of: 
\begin{itemize}
	\item a \define{manifold of objects}, $T_0$;
	\item a \define{manifold of morphisms}, $T_1$;
	\item a \define{manifold of 2-morphisms}, $T_2$;
	\item a \define{manifold of 3-morphisms}, $T_3$;
\end{itemize}
equipped with:
\begin{itemize}
	\item a smooth bicategory structure on $\underline{\Mor} \, T$, with
		\begin{itemize}
			\item $T_1$ as the smooth manifold of objects;
			\item $T_2$ as the smooth manifold of morphisms;
			\item $T_3$ as the smooth manifold of 2-morphisms;
		\end{itemize}
		We call the vertical composition in $\underline{\Mor} \, T$
		\define{composition at a 2-cell}, and the horizontal
		composition in $\underline{\Mor} \, T$ \define{composition at a
		1-cell}.
	\item smooth \define{source} and \define{target maps}:
		\[ s, t \maps T_1 \to T_0. \]
	\item a smooth \define{identity-assigning map}:
		\[ i \maps T_0 \to T_1. \]
	\item a smooth \define{composition} pseudofunctor, called
		\define{composition at a 0-cell}, which is strict in the
		tricategories we consider:
		\[ \cdot \maps \underline{\Mor} \, T \times_{T_0} \underline{\Mor} \, T \to \underline{\Mor} \, T. \]
		That is, three smooth maps:
		\[ \cdot \maps T_1 \times_{T_0} T_1 \to T_1 \]
		\[ \cdot \maps T_2 \times_{T_0} T_2 \to T_2 \]
		\[ \cdot \maps T_3 \times_{T_0} T_3 \to T_3 \]
		satisfying the axioms of a strict functor.
	\item a smooth pseudonatural transformation, the \define{associator},
		which is trivial in the tricategories we consider:
		\[ a(f,g,h) \maps (f \cdot g) \cdot h \To f \cdot (g \cdot h). \]
	\item smooth pseudonatural transformations, the \define{left} and
		\define{right unitors}, all trivial in the tricategories we
		consider:
		\[ l(f) \maps 1 \cdot f \To f, \quad r(f) \maps f \cdot 1 \To f. \]
	\item a smooth modification called the \define{pentagonator}:
		
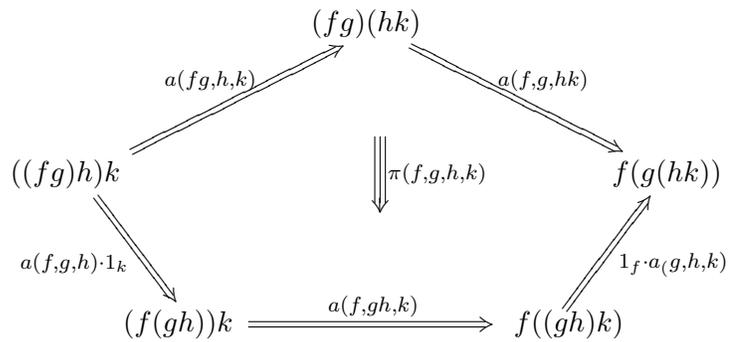
\[
		\xy
		 (0,20)*+{(f g) (h k)}="1";
		 (40,0)*+{f (g (h k))}="2";
		 (25,-20)*{ \quad f ((g h) k)}="3";
		 (-25,-20)*+{(f (g h)) k}="4";
		 (-40,0)*+{((f g) h) k}="5";
		 {\ar@{=>}^{a(f,g,h k)}     "1";"2"}
		 {\ar@{=>}_{1_f \cdot a_(g,h,k)}  "3";"2"}
		 {\ar@{=>}^{a(f,g h,k)}    "4";"3"}
		 {\ar@{=>}_{a(f,g,h) \cdot 1_k}  "5";"4"}
		 {\ar@{=>}^{a(f g,h,k)}    "5";"1"}
		 {\ar@3{->}^{\pi(f,g,h,k)} (-2,5)*{}; (-2,-5)*{} }
		\endxy
		\]
	\item smooth modifications called the \define{middle, left} and
		\define{right triangulators}, all trivial in the tricategories
		we consider:
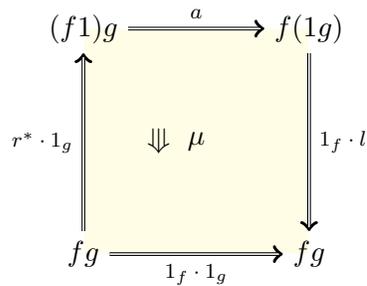
\begin{center}
  \begin{tikzpicture}
    \filldraw[white,fill=yellow,fill opacity=0.1](0,0)--(3,0)--(3,3)--(0,3)--cycle;
    \node (AB1) at (0,0) {$fg$};
    \node (AB2) at (3,0) {$fg$}
      edge [<-, double] node [l, below] {$1_f \cdot 1_g$} (AB1);
      \node (AIB2) at (3,3) {$f(1g)$}
      edge [->, double] node [l, right] {$1_f \cdot l$} (AB2);
      \node (AIB1) at (0,3) {$(f1)g$}
      edge [->, double] node [l, above] {$a$} (AIB2)
      edge [<-, double] node [l, left] {$r^* \cdot 1_g$} (AB1);
    \node at (1,1.5) {\tikz\node [rotate=-90] {$\Rrightarrow$};};
    \node at (1.5,1.5) {$\mu$};
  \end{tikzpicture}
\end{center}
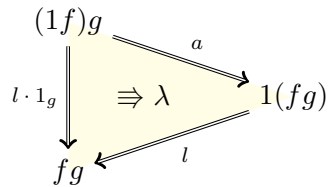
\begin{center}
  \begin{tikzpicture}
    \filldraw[white,fill=yellow,fill opacity=0.1](0,2)--(3,1)--(0,0)--cycle;
    \node (IAB1) at (0,2) {$(1f)g$};
    \node (IAB2) at (3,1) {$1(fg)$}
      edge [<-, double] node [l, above right] {$a$} (IAB1);
    \node (AB) at (0,0) {$fg$}
      edge [<-, double] node [l, below right] {$l$} (IAB2)
      edge [<-, double] node [l, left] {$l \cdot 1_g $} (IAB1);
    \node at (1,1) {$\Rrightarrow \lambda$};
  \end{tikzpicture}
\end{center}
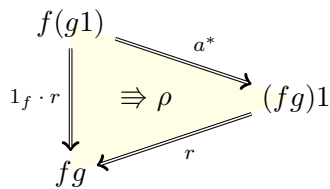
\begin{center}
  \begin{tikzpicture}
    \filldraw[white,fill=yellow,fill opacity=0.1](0,2)--(3,1)--(0,0)--cycle;
    \node (ABI1) at (0,2) {$f(g1)$};
    \node (ABI2) at (3,1) {$(fg)1$}
      edge [<-, double] node [l, above right] {$a^*$} (ABI1);
    \node (AB) at (0,0) {$fg$}
      edge [<-, double] node [l, below right] {$r$} (ABI2)
      edge [<-, double] node [l, left] {$1_f \cdot r$} (ABI1);
    \node at (1,1) {$\Rrightarrow \rho$};
  \end{tikzpicture}
\end{center}
\end{itemize}

These smooth modifications all satisfy their own axioms. When $\lambda$, $\rho$
and $\mu$ are trivial, their axioms boil down to the statement that $\pi$ is
trivial whenever one its arguments is trivial. We therefore omit them. The one
axiom we need to consider is the \define{pentagonator identity}:

\newpage
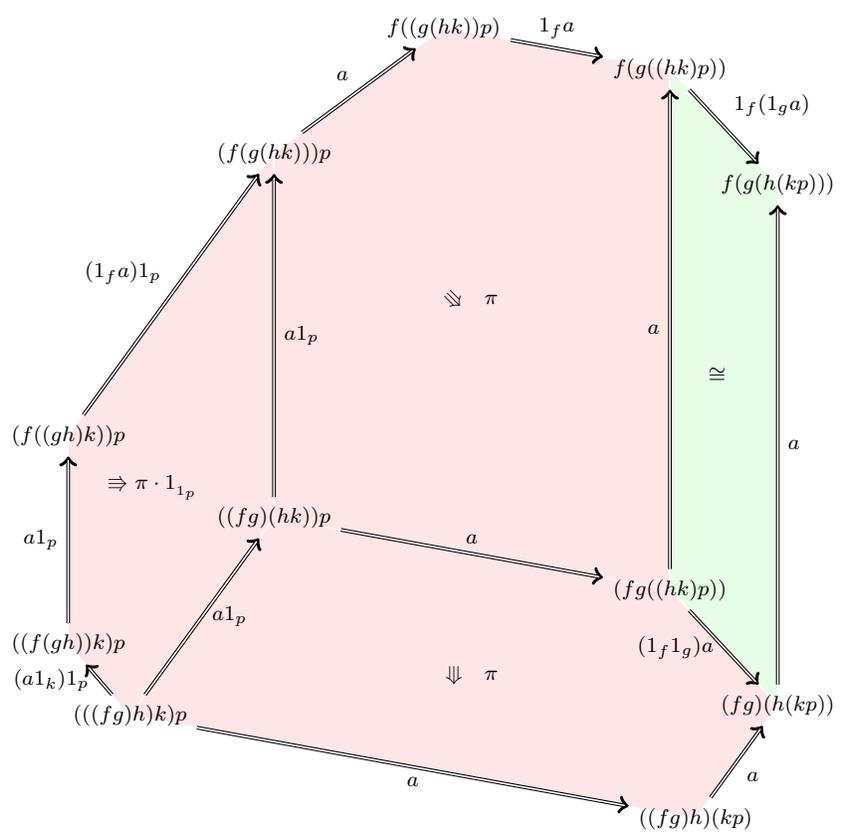
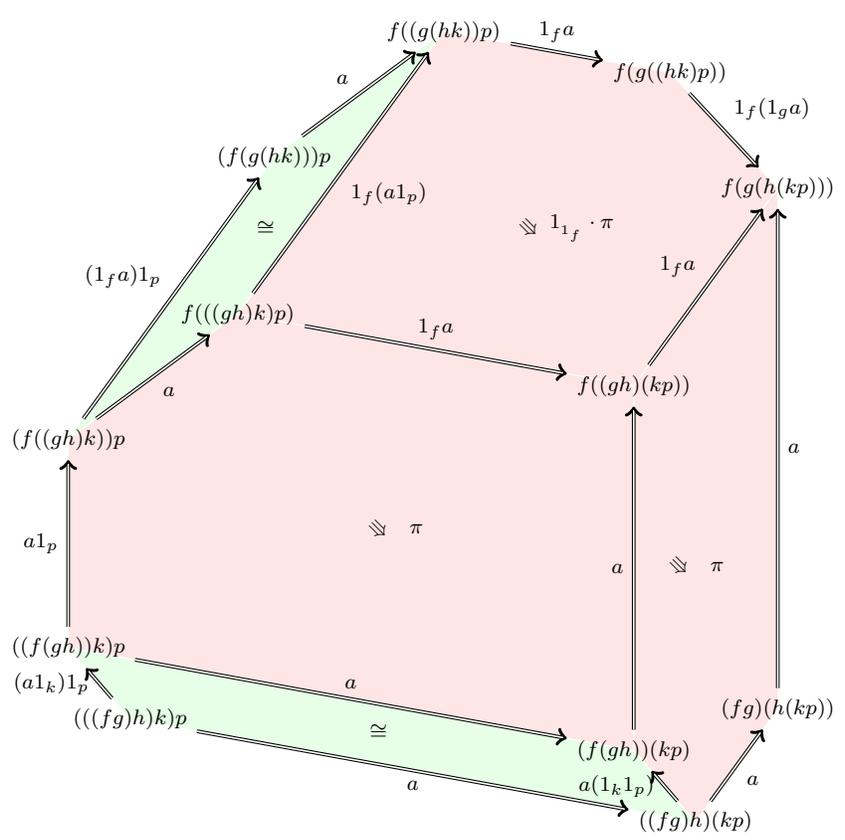
\thispagestyle{empty}

\begin{figure}[H]
    \begin{center}
      \begin{tikzpicture}[line join=round]
        \filldraw[white,fill=red,fill opacity=0.1](-4.306,-3.532)--(-2.391,-.901)--(-2.391,3.949)--(-5.127,.19)--(-5.127,-2.581)--cycle;
        \filldraw[white,fill=red,fill opacity=0.1](-4.306,-3.532)--(-2.391,-.901)--(2.872,-1.858)--(4.306,-3.396)--(3.212,-4.9)--cycle;
        \filldraw[white,fill=red,fill opacity=0.1](2.872,-1.858)--(2.872,5.07)--(-.135,5.617)--(-2.391,3.949)--(-2.391,-.901)--cycle;
        \filldraw[white,fill=green,fill opacity=0.1](4.306,-3.396)--(4.306,3.532)--(2.872,5.07)--(2.872,-1.858)--cycle;
        \begin{scope}[font=\fontsize{8}{8}\selectfont]
          \node (A) at (-2.391,3.949) {$(f(g(hk)))p$};
          \node (B) at (-5.127,.19) {$(f((gh)k))p$}
	  edge [->, double] node [l, above left] {$(1_{f}  a)1_{p}$} (A);
          \node (C) at (-5.127,-2.581) {$((f(gh))k)p$}
	  edge [->, double] node [l, left] {$a  1_{p}$} (B);
          \node (D) at (-4.306,-3.532) {$(((fg)h)k)p$}
	  edge [->, double] node [l, left] {$(a  1_{k})  1_{p}$} (C);
          \node (E) at (3.212,-4.9) {$((fg)h)(kp)$}
            edge [<-, double] node [l, below] {$a$} (D);
          \node (F) at (4.306,-3.396) {$(fg)(h(kp))$}
            edge [<-, double] node [l, below right] {$a$} (E);
          \node (G) at (4.306,3.532) {$f(g(h(kp)))$}
            edge [<-, double] node [l, right] {$a$} (F);
          \node (H) at (2.872,5.07) {$f(g((hk)p))$}
            edge [->, double] node [l, above right] {$1_{f}(1_{g}a)$} (G);
          \node (I) at (-.135,5.617) {$f((g(hk))p)$}
            edge [->, double] node [l, above] {$1_{f}a$} (H)
            edge [<-, double] node [l, above left] {$a$} (A);
          \node (M) at (-2.391,-.901) {$((fg)(hk))p$}
            edge [<-, double] node [l, right] {$a1_{p}$} (D)
            edge [->, double] node [l, right] {$a1_{p}$} (A);
          \node (N) at (2.872,-1.858) {$(fg((hk)p))$}
            edge [<-, double] node [l, above] {$a$} (M)
            edge [->, double] node [l, left] {$a$} (H)
            edge [->, double] node [l, left] {$(1_{f}1_{g})a$} (F);
	    \node at (-4,-.5) {$\Rrightarrow \pi \cdot 1_{_{1_{p}}}$};
          \node at (0,-3) {\tikz\node [rotate=-90] {$\Rrightarrow$};};
          \node at (0.5,-3) {$\pi$};
          \node at (0,2) {\tikz\node [rotate=-45] {$\Rrightarrow$};};
          \node at (0.5,2) {$\pi$};
          \node at (3.5,1) {$\cong$};
        \end{scope}
      \end{tikzpicture}
      \[ = \]
      \begin{tikzpicture}[line join=round]
        \filldraw[white,fill=green,fill opacity=0.1](-2.872,1.858)--(-.135,5.617)--(-2.391,3.949)--(-5.127,.19)--cycle;
        \filldraw[white,fill=red,fill opacity=0.1](3.212,-4.9)--(4.306,-3.396)--(4.306,3.532)--(2.391,.901)--(2.391,-3.949)--cycle;
        \filldraw[white,fill=green,fill opacity=0.1](-4.306,-3.532)--(3.212,-4.9)--(2.391,-3.949)--(-5.127,-2.581)--cycle;
        \filldraw[white,fill=red,fill opacity=0.1](-2.872,1.858)--(2.391,.901)--(4.306,3.532)--(2.872,5.07)--(-.135,5.617)--cycle;
        \filldraw[white,fill=red,fill opacity=0.1](-5.127,-2.581)--(-5.127,.19)--(-2.872,1.858)--(2.391,.901)--(2.391,-3.949)--cycle;
        \begin{scope}[font=\fontsize{8}{8}\selectfont]
          \node (A) at (-2.391,3.949) {$(f(g(h k)))p$};
          \node (B) at (-5.127,.19) {$(f((gh)k))p$}
            edge [->, double] node [l, above left] {$(1_{f}a)1_{p}$} (A);
          \node (C) at (-5.127,-2.581) {$((f(gh))k)p$}
            edge [->, double] node [l, left] {$a1_{p}$} (B);
          \node (D) at (-4.306,-3.532) {$(((fg)h)k)p$}
            edge [->, double] node [l, left] {$(a1_{k})1_{p}$} (C);
          \node (E) at (3.212,-4.9) {$((fg)h)(kp)$}
            edge [<-, double] node [l, below] {$a$} (D);
          \node (F) at (4.306,-3.396) {$(fg)(h(kp))$}
            edge [<-, double] node [l, below right] {$a$} (E);
          \node (G) at (4.306,3.532) {$f(g(h(kp)))$}
            edge [<-, double] node [l, right] {$a$} (F);
          \node (H) at (2.872,5.07) {$f(g((hk)p))$}
            edge [->, double] node [l, above right] {$1_{f}(1_{g}a)$} (G);
          \node (I) at (-.135,5.617) {$f((g(hk))p)$}
            edge [->, double] node [l, above] {$1_{f}a$} (H)
            edge [<-, double] node [l, above left] {$a$} (A);
          \node (J) at (-2.872,1.858) {$f(((gh)k)p)$}
            edge [->, double] node [l, below right] {$1_{f}(a1_{p})$} (I)
            edge [<-, double] node [l, below right] {$a$} (B);
          \node (K) at (2.391,-3.949) {$(f(gh))(kp)$}
            edge [<-, double] node [l, left] {$a(1_{k}1_{p})$} (E)
            edge [<-, double] node [l, above] {$a$} (C);
          \node (L) at (2.391,.901) {$f((gh)(kp))$}
            edge [<-, double] node [l, left] {$a$} (K)
            edge [<-, double] node [l, above] {$1_{f}a$} (J)
            edge [->, double] node [l, above left] {$1_{f}a$} (G);
          \node at (-1,-1) {\tikz\node [rotate=-45] {$\Rrightarrow$};};
	  \node at (-.5,-1) {$\pi$};
          \node at (1,3) {\tikz\node [rotate=-45] {$\Rrightarrow$};};
	  \node at (1.7,3) {$1_{_{1_{f}}} \cdot \pi$};
          \node at (3,-1.5) {\tikz\node [rotate=-45] {$\Rrightarrow$};};
          \node at (3.5,-1.5) {$\pi$};
          \node at (-1,-3.7) {$\cong$};
          \node at (-2.5,3) {$\cong$};
        \end{scope}
      \end{tikzpicture}
    \end{center}
    \end{figure}

    \clearpage

\end{defn}

In the pentagonator identity, we have omitted $\cdot$ everywhere except the
faces to save space.  This identity comes from a 3-dimensional solid called the
\define{associahedron}. This is the polyhedron where:
\begin{itemize}
	\item vertices are parenthesized lists of five morphisms, e.g.\ $(((fg)h)k)p$;
	\item edges connect any two vertices related by an application of the associator, e.g.\ 
		\[ (((fg)h)k)p \Rightarrow ((fg)(hk))p . \]
\end{itemize}

In fact, the pentagonator identity gives us a picture of the associahedron.
Regarding the left-hand side of the equation as the back and the right-hand
side as the front, we assemble the following polyhedron:
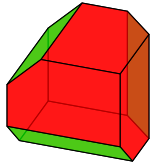
\begin{center}
  \begin{tikzpicture}[line join=round]
    \begin{scope}[scale=.2]
      \filldraw[fill=red,fill opacity=0.7](-4.306,-3.532)--(-2.391,-.901)--(-2.391,3.949)--(-5.127,.19)--(-5.127,-2.581)--cycle;
      \filldraw[fill=red,fill opacity=0.7](-4.306,-3.532)--(-2.391,-.901)--(2.872,-1.858)--(4.306,-3.396)--(3.212,-4.9)--cycle;
      \filldraw[fill=red,fill opacity=0.7](2.872,-1.858)--(2.872,5.07)--(-.135,5.617)--(-2.391,3.949)--(-2.391,-.901)--cycle;
      \filldraw[fill=green,fill opacity=0.7](4.306,-3.396)--(4.306,3.532)--(2.872,5.07)--(2.872,-1.858)--cycle;
      \filldraw[fill=green,fill opacity=0.7](-2.872,1.858)--(-.135,5.617)--(-2.391,3.949)--(-5.127,.19)--cycle;
      \filldraw[fill=red,fill opacity=0.7](3.212,-4.9)--(4.306,-3.396)--(4.306,3.532)--(2.391,.901)--(2.391,-3.949)--cycle;
      \filldraw[fill=green,fill opacity=0.7](-4.306,-3.532)--(3.212,-4.9)--(2.391,-3.949)--(-5.127,-2.581)--cycle;
      \filldraw[fill=red,fill opacity=0.7](-2.872,1.858)--(2.391,.901)--(4.306,3.532)--(2.872,5.07)--(-.135,5.617)--cycle;
      \filldraw[fill=red,fill opacity=0.7](-5.127,-2.581)--(-5.127,.19)--(-2.872,1.858)--(2.391,.901)--(2.391,-3.949)--cycle;
    \end{scope}
  \end{tikzpicture}
\end{center}
Identifying the vertices, edges and faces of this polyhedron with the
corresponding morphisms, 2-morphisms and 3-morphisms from the pentagonator
identity, we see the identity just says that the associahedron commutes. 

A \define{Lie 3-group} is a smooth tricategory with one object where all
morphisms, 2-morphisms and 3-morphisms are weakly invertible. Though it looks
quite complex, the pentagonator identity is secretly a cocycle condition. In
fact, given a normalized $H$-valued 4-cocycle $\pi$ on a Lie group $G$, we can
construct a Lie 3-group $\Brane_\pi(G,H)$ with:
\begin{itemize}
	\item One object, $\bullet$, regarded as a manifold in the trivial way.
	\item For each element $g \in G$, an automorphism of the one object:
		\[ \bullet \stackrel{g}{\longrightarrow} \bullet \]
		Composition at a 0-cell given by multiplication in the group:
		\[ \cdot \maps G \times G \to G. \]
		The source and target maps are trivial, and identity-assigning
		map takes the one object to $1 \in G$.
	\item Only the identity 2-morphism on any 1-morphism, and no
		2-morphisms between distinct 1-morphisms:
		\[
		\xy
		(-8,0)*+{\bullet}="4";
		(8,0)*+{\bullet}="6";
		{\ar@/^1.65pc/^g "4";"6"};
		{\ar@/_1.65pc/_g "4";"6"};
		{\ar@{=>}^{1_g} (0,3)*{};(0,-3)*{}} ;
		\endxy
		, \quad g \in G.
		\]
		So the space of 2-morphisms is also $G$. The source, target and
		identity-assigning maps are all the identity on $G$.
		Composition at a 1-cell is trivial, while composition at a
		0-cell is again multiplication in $G$.
	\item For each $h \in H$, a 3-automorphism of the 2-morphism $1_g$,
		and no 3-morphisms between distinct 2-morphisms:
		
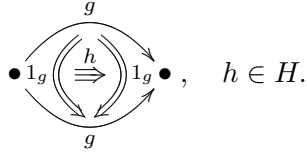
\[
		\xy 
		(-10,0)*+{\bullet}="1";
		(10,0)*+{\bullet}="2";
		{\ar@/^1.65pc/^g "1";"2"};
		{\ar@/_1.65pc/_g "1";"2"};
		(0,5)*+{}="A";
		(0,-5)*+{}="B";
		{\ar@{=>}@/_.75pc/ "A"+(-1.33,0) ; "B"+(-.66,-.55)};
		{\ar@{=}@/_.75pc/ "A"+(-1.33,0) ; "B"+(-1.33,0)};
		{\ar@{=>}@/^.75pc/ "A"+(1.33,0) ; "B"+(.66,-.55)};
		{\ar@{=}@/^.75pc/ "A"+(1.33,0) ; "B"+(1.33,0)};
		{\ar@3{->} (-2,0)*{}; (2,0)*{}};
		(0,2.5)*{\scriptstyle h};
		(-7,0)*{\scriptstyle 1_g};
		(7,0)*{\scriptstyle 1_g};
		\endxy
		, \quad h \in H.
		\]
		Thus the space of 3-morphisms is $G \times H$. The source and
		target maps are projection onto $G$, and the identity assigning
		map takes $1_g$ to $0 \in H$, for all $g \in G$.
	\item Three kinds of composition of 3-morphisms: given a pair of
		3-morphisms on the same 2-morphism, we can compose them at at a
		2-cell, which we take to be addition in $H$:
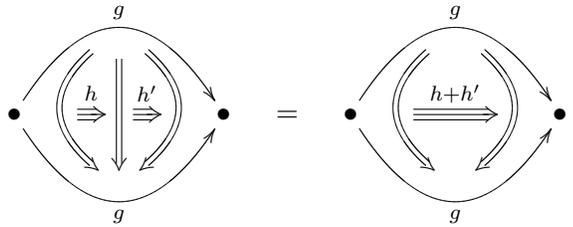
		\[
		\xy 0;/r.22pc/:
		(0,15)*{};
		(0,-15)*{};
		(0,8)*{}="A";
		(0,-8)*{}="B";
		{\ar@{=>} "A" ; "B"};
		{\ar@{=>}@/_1pc/ "A"+(-4,1) ; "B"+(-3,0)};
		{\ar@{=}@/_1pc/ "A"+(-4,1) ; "B"+(-4,1)};
		{\ar@{=>}@/^1pc/ "A"+(4,1) ; "B"+(3,0)};
		{\ar@{=}@/^1pc/ "A"+(4,1) ; "B"+(4,1)};
		{\ar@3{->} (-6,0)*{} ; (-2,0)*+{}};
		(-4,3)*{\scriptstyle h};
		{\ar@3{->} (2,0)*{} ; (6,0)*+{}};
		(4,3)*{\scriptstyle h'};
		(-15,0)*+{\bullet}="1";
		(15,0)*+{\bullet}="2";
		{\ar@/^2.75pc/^g "1";"2"};
		{\ar@/_2.75pc/_g "1";"2"};
		\endxy 
		\quad = \quad
		\xy 0;/r.22pc/:
		(0,15)*{};
		(0,-15)*{};
		(0,8)*{}="A";
		(0,-8)*{}="B";
		{\ar@{=>}@/_1pc/ "A"+(-4,1) ; "B"+(-3,0)};
		{\ar@{=}@/_1pc/ "A"+(-4,1) ; "B"+(-4,1)};
		{\ar@{=>}@/^1pc/ "A"+(4,1) ; "B"+(3,0)};
		{\ar@{=}@/^1pc/ "A"+(4,1) ; "B"+(4,1)};
		{\ar@3{->} (-6,0)*{} ; (6,0)*+{}};
		(0,3)*{\scriptstyle h + h'};
		(-15,0)*+{\bullet}="1";
		(15,0)*+{\bullet}="2";
		{\ar@/^2.75pc/^g "1";"2"};
		{\ar@/_2.75pc/_g "1";"2"};
		\endxy .
		\]
		We can also compose two 3-morphisms at a 1-cell, which we again
		take to be composition in $H$:
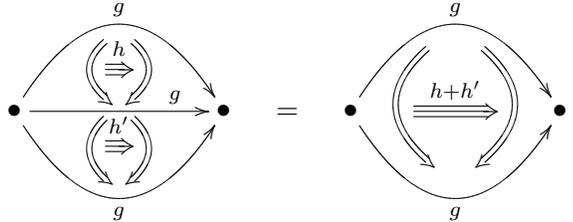
		\[	
		\xy 0;/r.22pc/:
		(0,15)*{};
		(0,-15)*{};
		(0,9)*{}="A";
		(0,1)*{}="B";
		{\ar@{=>}@/_.5pc/ "A"+(-2,1) ; "B"+(-1,0)};
		{\ar@{=}@/_.5pc/ "A"+(-2,1) ; "B"+(-2,1)};
		{\ar@{=>}@/^.5pc/ "A"+(2,1) ; "B"+(1,0)};
		{\ar@{=}@/^.5pc/ "A"+(2,1) ; "B"+(2,1)};
		{\ar@3{->} (-2,6)*{} ; (2,6)*+{}};
		(0,9)*{\scriptstyle h};
		(0,-1)*{}="A";
		(0,-9)*{}="B";
		{\ar@{=>}@/_.5pc/ "A"+(-2,-1) ; "B"+(-1,-1.5)};
		{\ar@{=}@/_.5pc/ "A"+(-2,0) ; "B"+(-2,-.7)};
		{\ar@{=>}@/^.5pc/ "A"+(2,-1) ; "B"+(1,-1.5)};
		{\ar@{=}@/^.5pc/ "A"+(2,0) ; "B"+(2,-.7)};
		{\ar@3{->} (-2,-5)*{} ; (2,-5)*+{}};
		(0,-2)*{\scriptstyle h'};
		(-15,0)*+{\bullet}="1";
		(15,0)*+{\bullet}="2";
		{\ar@/^2.75pc/^g "1";"2"};
		{\ar@/_2.75pc/_g "1";"2"};
		{\ar "1";"2"};
		(8,2)*{\scriptstyle g};
		\endxy 
		\quad = \quad
		\xy 0;/r.22pc/:
		(0,15)*{};
		(0,-15)*{};
		(0,8)*{}="A";
		(0,-8)*{}="B";
		{\ar@{=>}@/_1pc/ "A"+(-4,1) ; "B"+(-3,0)};
		{\ar@{=}@/_1pc/ "A"+(-4,1) ; "B"+(-4,1)};
		{\ar@{=>}@/^1pc/ "A"+(4,1) ; "B"+(3,0)};
		{\ar@{=}@/^1pc/ "A"+(4,1) ; "B"+(4,1)};
		{\ar@3{->} (-6,0)*{} ; (6,0)*+{}};
		(0,3)*{\scriptstyle h + h'};
		(-15,0)*+{\bullet}="1";
		(15,0)*+{\bullet}="2";
		{\ar@/^2.75pc/^g "1";"2"};
		{\ar@/_2.75pc/_g "1";"2"};
		\endxy .
		\]
		In terms of maps, both of these compositions are just:
		\[ 1 \times + \maps G \times H \times H \to G \times H. \]
		And finally, we can glue two 3-cells at the 0-cell, the object.
		We call this \define{composition at a 0-cell}, and define it to
		be addition \emph{twisted by the action of $G$}:
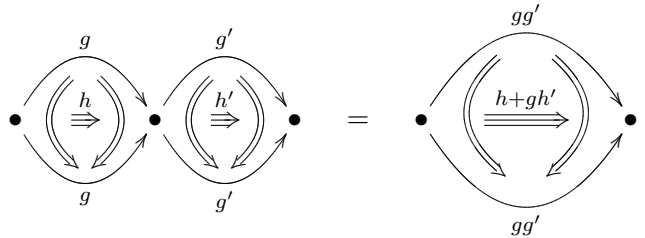
		\[
		\xy 0;/r.22pc/:
		(0,15)*{};
		(0,-15)*{};
		(-20,0)*+{\bullet}="1";
		(0,0)*+{\bullet}="2";
		{\ar@/^2pc/^g "1";"2"};
		{\ar@/_2pc/_g "1";"2"};
		(20,0)*+{\bullet}="3";
		{\ar@/^2pc/^{g'} "2";"3"};
		{\ar@/_2pc/_{g'} "2";"3"};
		(-10,6)*+{}="A";
		(-10,-6)*+{}="B";
		{\ar@{=>}@/_.7pc/ "A"+(-2,0) ; "B"+(-1,-.8)};
		{\ar@{=}@/_.7pc/ "A"+(-2,0) ; "B"+(-2,0)};
		{\ar@{=>}@/^.7pc/ "A"+(2,0) ; "B"+(1,-.8)};
		{\ar@{=}@/^.7pc/ "A"+(2,0) ; "B"+(2,0)};
		(10,6)*+{}="A";
		(10,-6)*+{}="B";
		{\ar@{=>}@/_.7pc/ "A"+(-2,0) ; "B"+(-1,-.8)};
		{\ar@{=}@/_.7pc/ "A"+(-2,0) ; "B"+(-2,0)};
		{\ar@{=>}@/^.7pc/ "A"+(2,0) ; "B"+(1,-.8)};
		{\ar@{=}@/^.7pc/ "A"+(2,0) ; "B"+(2,0)};
		{\ar@3{->} (-12,0)*{}; (-8,0)*{}};
		(-10,3)*{\scriptstyle h};
		{\ar@3{->} (8,0)*{}; (12,0)*{}};
		(10,3)*{\scriptstyle h'};
		\endxy 
		\quad = \quad 
		\xy 0;/r.22pc/:
		(0,15)*{};
		(0,-15)*{};
		(0,8)*{}="A";
		(0,-8)*{}="B";
		{\ar@{=>}@/_1pc/ "A"+(-4,1) ; "B"+(-3,0)};
		{\ar@{=}@/_1pc/ "A"+(-4,1) ; "B"+(-4,1)};
		{\ar@{=>}@/^1pc/ "A"+(4,1) ; "B"+(3,0)};
		{\ar@{=}@/^1pc/ "A"+(4,1) ; "B"+(4,1)};
		{\ar@3{->} (-6,0)*{} ; (6,0)*+{}};
		(0,3)*{\scriptstyle h + gh'};
		(-15,0)*+{\bullet}="1";
		(15,0)*+{\bullet}="2";
		{\ar@/^2.75pc/^{gg'} "1";"2"};
		{\ar@/_2.75pc/_{gg'} "1";"2"};
		\endxy .
		\]
		In terms of a map, $\cdot$ is just given by multiplication on
		the semidirect product:
		\[ \cdot \maps (G \ltimes H) \times (G \ltimes H) \to G \ltimes H. \]

	\item The associator, left and right unitors are automatically
		trivial, because all 2-morphisms are trivial.

	\item For each quadruple of 1-morphisms, a specified 3-isomorphism,
		the \define{2-associator} or \define{pentagonator}:
		\[ \pi(g_1, g_2, g_3, g_4) \maps 1_{g_1 g_2 g_3 g_4} \to 1_{g_1 g_2 g_3 g_4}. \]
		given by the 4-cocycle $\pi \maps G^4 \to H$, which we think of
		as element of $H$ because the source (and target) are
		understood to be $1_{g_1 g_2 g_3 g_4}$.

	\item The three other specified 3-isomorphisms are trivial.
\end{itemize}
A \define{slim Lie 3-group} is one of this form. As before, it remains to check
that it is, in fact, a Lie 3-group. We claim:
\begin{prop} \label{prop:Lie3group}
	$\Brane_\pi(G,H)$ is a Lie 3-group: a smooth tricategory with one
	object and all morphisms weakly invertible.
\end{prop}
Once again, we prove this by showing that the 4-cocycle condition implies the
one nontrivial axiom for this tricategory: the pentagonator identity.
\begin{proof}
	As noted above, the triangulators $\lambda$, $\rho$ and
	$\mu$ are trivial. The axioms they satisfy are automatic because $\pi$
	is normalized.

	So to check that $\Brane_\pi(G,H)$ is a tricategory, it remains to
	check that $\pi$ satisfies the pentagonator identity. Since the 3-cells
	of the pentagonator identity commute (they represent elements of $H$),
	and all the faces not involving a $\pi$ are trivial, the first half reads:
	\[ \pi(g_1,g_2,g_3,g_4) \cdot 0_{g_5} + \pi(g_1 g_2, g_3, g_4, g_5) + \pi(g_1,g_2,g_3 g_4,g_5) . \]
	Here, we write $0_{g_5}$ to denote 3-morphism which is the identity on
	the identity of the 1-morphism $g_5$. We do not need to be worried
	about order of terms, since composition of 3-morphisms is addition in
	$H$. The second half of the pentagonator identity reads:
	\[ 0_{g_1} \cdot \pi(g_2,g_3,g_4,g_5) + \pi(g_1,g_2 g_3, g_4, g_5) + \pi(g_1,g_2,g_3,g_4 g_5). \]
	Here, we write $0_{g_1}$ for to denote the 3-morphism which is the
	identity on the identity of the 1-morphism $g_1$. Applying the
	definition of $\cdot$ as the semidirect product, we see the equality of
	the first half with the second half is just the cocycle condition on
	$\pi$:
	\begin{eqnarray*} 
		\pi(g_1,g_2,g_3,g_4) + \pi(g_1 g_2, g_3, g_4, g_5) + \pi(g_1,g_2,g_3 g_4,g_5)   \\
		= g_1 \pi(g_2,g_3,g_4,g_5) + \pi(g_1,g_2 g_3, g_4, g_5) + \pi(g_1,g_2,g_3,g_4 g_5) .
	\end{eqnarray*} 

	So, $\Brane_\pi(G,H)$ is a tricategory. It is smooth because
	everything in sight is smooth: $G$, $H$, and the map $\pi \maps G^4
	\to H$. And it is a Lie 3-group: the 1-morphisms $G$, the trivial
	2-morphisms, and the 3-morphisms $H$ are all strictly invertible, and
	thus of course they are weakly invertible.
\end{proof}

Once again, we can say something a bit stronger about $\Brane_\pi(G,H)$, if we
let $\pi$ be any normalized $H$-valued 4-cochain, rather requiring it to be a
cocycle. In this case, $\Brane_{\pi}(G,H)$ is a Lie 3-group if and only if
$\pi$ is a 4-cocycle, because $\pi$ satisfies the pentagon identity if and only
if it is a cocycle.

\chapter{Integrating nilpotent Lie \emph{n}-algebras} \label{ch:integrating}

Any mathematician worth her salt knows that we can easily construct Lie
algebras as the infinitesimal versions of Lie groups, and that a more
challenging inverse construction exists: we can `integrate' Lie algebras to get
Lie groups. By analogy, we expect that the same is true of Lie $n$-algebras and
Lie $n$-groups: that we can construct Lie $n$-algebras as the infinitesimal
versions of Lie $n$-groups, and we can `integrate' Lie $n$-algebras to obtain
Lie $n$-groups. 

In fact, it is easy to see how to obtain slim Lie $n$-algebras from slim Lie
$n$-groups. As we saw in Chapter~\ref{ch:Lie-n-superalgebras}, slim Lie
$n$-algebras are built from $(n+1)$-cocycles in Lie algebra cohomology.
Remember, $p$-cochains on the Lie algebra $\g$ are linear maps:
\[ C^p(\g,\h) = \left\{ \omega \maps \Lambda^p \g \to \h \right\} , \]
where $\h$ is a representation of $\g$, though we shall restrict ourselves to
the trivial representation $\h = \R$ in this chapter.

On the other hand, in Chapter~\ref{ch:Lie-n-groups}, we saw that slim Lie $n$-groups
are built from $(n+1)$-cocycles in Lie group cohomology, at least for $n=2$ and
3. Remember, $p$-cochains on $G$ are smooth maps:
\[ C^p(G,H) = \left\{ f \maps G^p \to H \right\} , \]
where $H$ is an abelian group on which $G$ acts by automorphism, though we
shall restrict ourselves to $H = \R$ with trivial action in this chapter.

Thus, to derive a Lie $n$-algebra from a Lie $n$-group, just differentiate
the defining Lie group $(n+1)$-cocycle at the identity to obtain a Lie algebra
$(n+1)$-cocycle. In other words, for every Lie group $G$ with Lie algebra $\g$,
there is a cochain map:
\[ D \maps C^\bullet(G) \to C^\bullet(\g) , \]
given by differentiation. Here, we have omitted reference to the coefficients
$H$ and $\h$ because both are assumed to be $\R$. We continue this practice for
the rest of the chapter.

Going the other way, however, is challenging---integrating a Lie $n$-algebra is
harder, even when the Lie $n$-algebra in question is slim.  Nonetheless, this
challenge has been met. Building on the earlier work of Getzler \cite{Getzler}
on integrating nilpotent Lie $n$-algebras, Henriques \cite{Henriques} has shown
that any Lie $n$-algebra can be integrated to a `Lie $n$-group', which
Henriques defines as a sort of smooth Kan complex in the category of Banach
manifolds.  More recently, Schreiber \cite{Schreiber} has generalized this
integration procedure to a setting much more general than that of Banach
manifolds, including both supermanifolds and manifolds with infinitesimals. For
both Henriques and Schreiber, the definition of Lie $n$-group is weaker
than the one we sketched in Chapter \ref{ch:Lie-n-groups}---it weakens the
notion of multiplication so that the product of two group `elements' is only
defined up to equivalence. This level of generality seems essential for the
construction to work for \emph{every} Lie $n$-algebra.

However, for \emph{some} Lie $n$-algebras, we can integrate them using the more
naive idea of Lie $n$-group we prefer in this thesis: a smooth $n$-category
with one object in which every $k$-morphism is weakly invertible, for all $1
\leq k \leq n$. We shall see that, for some slim Lie $n$-algebras, we can
integrate the defining Lie algebra $(n+1)$-cocycle to obtain a Lie group
$(n+1)$-cocycle. In other words, for certain Lie groups $G$ with Lie algebra
$\g$, there is a cochain map:
\[ \smallint \maps C^\bullet(\g) \to C^\bullet(G) . \]
which is a chain homotopy inverse to differentiation.

When is this possible? We can always differentiate Lie group cochains to obtain
Lie algebra cochains, but if we can also integrate Lie algebra cochains to obtain
Lie group cochains, the cohomology of the Lie group and its Lie algebra
will coincide:
\[ H^\bullet(\g) \iso H^\bullet(G) . \]
By a theorem of van Est \cite{vanEst}, this happens when all the homology
groups of $G$, as a topological space, vanish.

Thus, we should look to Lie groups with vanishing homology for our examples.
How bad can things be when the Lie group is not homologically trivial? To get a
sense for this, recall that any semisimple Lie group $G$ is diffeomorphic to
the product of its maximal compact subgroup $K$ and a contractible space $C$:
\[ G \approx K \times C . \]
When $K$ is a point, $G$ is contractible, and certainly has vanishing homology.
At the other extreme, when $C$ is a point, $G$ is compact. And indeed, in this
case there is no hope of obtaining a nontrivial cochain map from Lie algebra
cochains to Lie group cochains:
\[ \smallint \maps C^\bullet(\g) \to C^\bullet(G) \]
because \emph{every smooth cochain on a compact group is trivial}.

This fact provided an obstacle to early attempts to integrate Lie 2-algebras.
For instance, consider the string Lie 2-algebra $\strng(n)$ we described in
Section \ref{sec:string2alg}. Recall that it is the slim Lie 2-algebra
$\strng_j(\so(n))$, where $j$ is the canonical 3-cocycle on $\so(n)$, given by
combining the Killing form with the bracket: 
\[ j = \langle -, [-,-] \rangle . \]
One could attempt to integrate $\strng(n)$ to a slim Lie 2-group
$\String_{\smallint j}(\SO(n))$, where $\smallint j$ is a Lie group 3-cocycle
on $\SO(n)$ which somehow integrates $j$, but because the compact group
$\SO(n)$ admits no nontrivial smooth Lie group cocycles, this idea fails.

The real lesson of the string Lie 2-algebra is that, once again, our notion of
Lie 2-group is not general enough. By generalizing the concept of Lie 2-group,
various authors, like Baez--Crans--Schreiber--Stevenson \cite{BCSS}, Henriques
\cite{Henriques}, and Schommer-Pries \cite{SchommerPries}, were successful in
integrating $\strng(n)$. 

Nonetheless, there is a large class of Lie $n$-algebras for which our Lie
$n$-groups \emph{are} general enough. In particular, when $G$ is an
`exponential' Lie group, the story is completely different. A Lie group or Lie
algebra is called \define{exponential} if the exponential map 
\[ \exp \maps \g \to G \]
is a diffeomorphism.  For instance, all simply-connected nilpotent Lie groups
are exponential, though the reverse is not true. Certainly, all exponential Lie
groups have vanishing homology, because $\g$ is contractible. We caution the
reader that some authors use the term `exponential' merely to indicate that
$\exp$ is surjective.

When $G$ is an exponential Lie group with Lie algebra $\g$, we can use a
geometric technique developed by Houard \cite{Houard} to construct a cochain
map:
\[ \smallint \maps C^\bullet(\g) \to C^\bullet(G). \]
The basic idea behind this construction is simple, a natural outgrowth of a
familiar concept from the cohomology of Lie algebras. Because a Lie algebra
$p$-cochain is a linear map:
\[ \omega \maps \Lambda^p \g \to \R, \]
using left translation, we can view $\omega$ as defining a $p$-form on the
Lie group $G$. So, we can integrate this $p$-form over $p$-simplices in $G$.
Thus we can define a smooth function:
\[ \smallint \omega \maps G^p \to \R, \]
by viewing the integral of $\omega$ as a function of the vertices of a
$p$-simplex:
\[ \smallint \omega(g_1, g_2, \dots, g_p) = \int_{[1, g_1, g_1 g_2, \dots, g_1 g_2 \cdots g_p]} \omega . \]
For the right-hand side to truly be a function of the $p$-tuple $(g_1, g_2,
\dots, g_p)$, we will need a standard way to `fill out' the $p$-simplex $[1,
g_1, g_1 g_2, \dots, g_1 g_2 \cdots g_p]$, based only on its vertices. It is
here that the fact that $G$ is exponential is key: in an exponential group, we
can use the exponential map to define a unique path from the identity $1$ to
any group element. We think of this path as giving a 1-simplex, $[1,g]$, and we
can extend this idea to higher dimensional $p$-simplices. 

Therefore, when $G$ is exponential, we can construct $\smallint$. Using this
cochain map, it is possible to integrate the slim Lie $n$-algebra
$\brane_\omega(\g)$ to the slim Lie $n$-group $\Brane_{\smallint \omega}(G)$. 

We proceed as follows. In Section \ref{sec:integratingcochains}, we construct
$\smallint$ and show that, along with $D$, it gives a homotopy equivalence
between the complexes $C^\bullet(\g)$ and $C^\bullet(G)$. In Section
\ref{sec:examples}, we give explicit formulas for the $p$-cochain $\smallint
\omega$ in terms of $\omega$, for $p = 0$, 1, 2, and 3. Finally, in Section
\ref{sec:heisenberg2group}, we use $\smallint$ to integrate the Heisenberg Lie
2-algebra of Section \ref{sec:sec:heisenberg2alg}.  Later, in Chapter
\ref{ch:integrating2}, we shall see that this construction can be `superized',
and integrate Lie $n$-superalgebras to $n$-supergroups.

\section{Integrating Lie algebra cochains} \label{sec:integratingcochains}

In what follows, we shall see that for an exponential Lie group $G$, we can
construct simplices in $G$ that get along with the action of $G$ on itself.
Since we can treat any $p$-cochain $\omega$ on $\g$ as a left-invariant
$p$-form on $G$, we can integrate $\omega$ over a $p$-simplex in $S$ in $G$.
Regarding $\int_S \omega$ as a function of the vertices of $S$, we will see
that it defines a Lie group $p$-cochain. The fact that this is cochain map is
purely geometric: it follows automatically from Stokes' theorem.

Let us begin by replacing the cohomology of $\g$ with the cohomology of
left-invariant differential forms on $G$.  Recall that the cohomology of the
Lie algebra $\g$ is given by the Lie algebra cochain complex,
$C^{\bullet}(\g)$, which at level $p$ consists of $p$-linear maps from $\g$ to
$\R$:
\[ C^p(\g) = \left\{ \omega \maps \Lambda^n \g \to \R \right\}. \]
We already defined this for Lie superalgebras in Section \ref{sec:cohomology}.
In that section, we saw that the coboundary map $d$ on this complex is usually
defined by a rather lengthy formula, but here we shall substitute an
equivalent, more geometric definition. Since we can think of the Lie algebra
$\g$ as the tangent space $T_1 G$, we can think of a $p$-cochain on $\g$ as
giving a $p$-form on this tangent space. Using left translation on the group,
we can translate this $p$-form over $G$ to define a $p$-form on all of $G$.
This $p$-form is left-invariant, and it is easy to see that any left-invariant
$p$-form on $G$ can be defined in this way.

So, in fact, we could just as well define 
\[ C^p(\g) = \left\{ \mbox{left-invariant $p$-forms on $G$} \right\}. \] 
It is well-known that the de Rham differential of a left-invariant $p$-form
$\omega$ is again left-invariant, and remarkably, the formula for $d \omega_1$
involves only the Lie bracket on $\g$. This formula is Chevalley and
Eilenberg's original definition of $d$ \cite{ChevalleyEilenberg}, the one we gave
in Section \ref{sec:cohomology}, albeit adapted for Lie superalgebras. In this
section, we may forget about this messy formula, and use the de Rham
differential instead.

The cohomology of the Lie group $G$ is given by the Lie group cochain complex,
$C^{\bullet}(G)$, which at level $p$ is given by the set of smooth functions
from $G^p$ to $\R$:
\[ C^p(G) = \left\{ f \maps G^p \to \R \right\}. \]
We have already discussed this in Chapter~\ref{ch:Lie-n-groups}. The coboundary
map $d$ on this complex is usually defined by a complicated formula we gave in
that chapter, but we can give it a more geometric description just as we did in
the case of Lie algebras. 

Since we are going to construct a cochain map by integrating $p$-forms over
$p$-simplices, it would be best to view Lie group cohomology in terms of
simplices now. To this end, let us define a \define{combinatorial $p$-simplex}
in the group $G$ to be an $(p+1)$-tuple of elements of $G$, which we call the
vertices in this context.  Of course, $G$ acts on the set of combinatorial
$p$-simplices by left multiplication of the vertices.

Now, we would like to think of Lie group $p$-cochains as `smooth,
homogeneous, $\R$-valued cochains' on the free abelian group on
combinatorial $p$-simplices. Of course, we need to say what this means. We say
an $\R$-valued $p$-cochain $F$ is \define{homogeneous} if it is
invariant under the action of $G$, and that it is \define{smooth} if
the corresponding map 
\[ F \maps G^{p+1} \to \R \]
is smooth. Now if
\[ C^p_H(G) = \left\{ \mbox{smooth homogeneous $p$-cochains} \right\}. \]
denotes the abelian group of all smooth, homogeneous $p$-cochains, there is a
standard way to make $C^\bullet_H(G)$ into a cochain complex. Just take the
coboundary operator to be:
\[ dF = F \circ \partial, \]
where $\partial$ is the usual boundary operator on $p$-chains. It is automatic
that $d^2 = 0$.

In fact, this cochain complex is isomorphic to the original one, which we
distinguish as the \define{inhomogeneous cochains}:
\[ C^p_I(G) = \left\{ f \maps G^p \to \R \right\} . \]
To see this, note that any inhomogeneous cochain:
\[ f \maps G^p \to \R \]
gives rise to a unique, smooth, homogeneous $p$-cochain $F$, by defining:
\[ F(g_0, \dots, g_p) = f(g^{-1}_0 g_1, g^{-1}_1 g_2,  \dots, g^{-1}_{p-1} g_p) \]
for each combinatorial $p$-simplex $(g_0, \dots, g_p)$. Conversely, every
smooth, homogeneous $p$-cochain $F$ gives a unique inhomogeneous $p$-cochain $f
\maps G^p \to \R$, by defining:
\[ f(g_1, \dots, g_p) = F(1, g_1, g_1 g_2, \dots, g_1 g_2 \dots g_p). \]
Finally, note that these isomorphisms commute with the coboundary operators on
$C^\bullet_H(G)$ and $C^\bullet_I(G)$. Henceforth, we will write $C^\bullet(G)$
to mean either complex.

These simplicial notions will permit us to define a cochain map from the Lie
algebra complex to the Lie group complex:
\[ C^{\bullet}(\g) \to C^{\bullet}(G). \]
For $\omega \in C^p(\g)$, the idea is to define an element $\smallint \omega
\in C^p(G)$ by \emph{integrating} the left-invariant $p$-form $\omega$ over a
$p$-simplex $S$ in the group $G$. In other words, the value which $\smallint
\omega$ assigns to $S$ is defined to be:
\[ (\smallint \omega)(S) = \int_S \omega. \]
This is nice because Stokes' theorem will tell us it is a cochain map:
\[ (\smallint d \omega)(S) = \int_S d \omega = \int_{\partial S} \omega = d (\smallint \omega)(S) \]
The only hard part is defining $p$-simplices in $G$ in such a way that $\smallint
\omega$ is actually a smooth, homogeneous $p$-cochain. It is here that the
fact that $G$ is \emph{exponential} is key. 

Note that, up until this point, we have only discussed combinatorial
$p$-simplices, which have no relationship to the Lie group structure of
$G$---they are just $(p+1)$-tuples of vertices. We now wish to `fill out'
the combinatorial simplices. That is, we want to create a rule that to any
$(p+1)$-tuple $(g_0, \dots, g_p)$ of vertices in $G$ assigns a filled
$p$-simplex in $G$, which we denote
\[ [g_0, \dots, g_p ]. \]
In order to prove that $\smallint \omega$ is smooth, we need smoothness
conditions for this rule, and in order to prove $\smallint \omega$ is
homogeneous, we shall require the left-translate of a $p$-simplex to again be a
$p$-simplex. In other words, we need:
\[ g [g_0, \dots, g_p] = [g g_0, \dots, g g_p]. \]
We make this precise as follows.

\begin{defn} \label{def:simplices}
Let $\Delta^p$ denote $\{(x_0, \dots, x_p) \in \R^{p+1} : \sum x_i = 1, x_i
\geq 0 \}$,  the standard $p$-simplex in $\R^{p+1}$.  Given a collection of
smooth maps
\[ \varphi_p \maps \Delta^p \times G^{p+1} \to G \]
for each $p \geq 0$, we say this collection defines a \define{left-invariant notion
of simplices} in $G$ if it satisfies:
\begin{enumerate}
	\item \textbf{The vertex property.} For any $(p+1)$-tuple, the
		restriction 
		\[ \varphi_p \maps \Delta^p \times \{(g_0, \dots, g_p)\} \to G \] 
		sends the vertices of $\Delta^p$ to $g_0, \dots, g_p$, in that
		order. We denote this restriction by 
		\[ [g_0, \dots, g_p]. \] 
		We call this map a \define{\boldmath{$p$}-simplex}, and regard
		it as a map from $\Delta^p$ to $G$.
	      
	\item \textbf{Left-invariance.} For any $p$-simplex $[g_0, \dots, g_p]$ and any $g \in G$, we
		have:
		\[ g [g_0, \dots, g_p] = [g g_0, \dots, g g_p ]. \]

	\item \textbf{The face property.} For any $p$-simplex
		\[ [g_0, \dots, g_p] \maps \Delta^p \to G \]
		the restriction to a face of $\Delta^p$ is a $(p-1)$-simplex.
\end{enumerate}

\end{defn}

Note that the second condition just says that the map
\[ \varphi_p \maps \Delta^p \times G^{p+1} \to G \]
is equivariant with respect to the left action of $G$, where we take $G$ to act
trivially on $\Delta^p$. 

On any group equipped with a left-invariant notion of simplex, we have the
following result:
\begin{prop} \label{prop:integral}
	Let $G$ be a Lie group equipped with a left-invariant notion of
	simplices, and let $\g$ be its Lie algebra. Then there is a cochain map
	from the Lie algebra cochain complex to the Lie group cochain complex
	\[ \smallint \maps C^{\bullet}(\g) \to C^{\bullet}(G) \]
	given by integration---that is, if $\omega$ is a left-invariant
	$p$-form on $G$, and $S$ is a $p$-simplex in $G$, then define:
	\[ (\smallint \omega)(S) = \int_S \omega. \]
\end{prop}
\begin{proof}
	Let $\omega \in C^p(\g)$. We have already noted that Stokes' theorem
	\[ \int_S d \omega = \int_{\partial S} \omega \]
	implies that this map is a cochain map. We only need to check that
	$\smallint \omega$ really lands in $C^p(G)$. That is, that it is smooth
	and homogeneous. Because $G$ acts trivially on the coefficient group
	$\R$, homogeneity means that $(\smallint \omega)(S)$ is invariant of
	the left action of $G$ on $S$.

	Indeed, note that we can pull the smooth, left-invariant $p$-form
	$\omega$ back along
	\[ \varphi_p \maps \Delta^p \times G^{p+1} \to G. \]
	The result, $\varphi^*_p \omega$, is a smooth $p$-form on $\Delta^p
	\times G^{p+1}$, still invariant under the action of $G$. Integrating
	out the dependence on $\Delta^p$, we see this results in a smooth,
	invariant map:
	\[ \smallint \omega \maps G^{p+1} \to \R, \]
	which is precisely what we wanted to prove.
\end{proof}

We would now like to show that any exponential Lie group $G$ comes with a
left-invariant notion of simplices. Our essential tool for this is our ability
to use the exponential map to connect any element of $G$ to the identity by a
uniquely-defined path. If $h = \exp(X) \in G$ is such an element, we can then
define the `based' 1-simplex $[1,h]$ to be swept out by the path $\exp(tX)$,
left translate this to define the general 1-simplex $[g, gh]$ as that swept out
by the path $g \exp(tX)$, and proceed to define higher-dimensional simplices
with the help of the exponential map and induction, using what we call the
\define{apex-base construction}: given a definition of $(p-1)$-simplex, we
define the $p$-simplex
\[ [1, g_1, \dots, g_p] \]
by using the exponential map to sweep out a path from 1, the \define{apex}, to
each point of the already defined $(p-1)$-simplex, the \define{base}:
\[ [g_1, \dots, g_p]. \]
Having done this, we can then use left translation to define the general
$p$-simplex:
\[ [g_0, g_1, \dots, g_p] = g_0 [1, g^{-1}_0 g_1, \dots, g^{-1}_0 g_p ]. \]
In fact, this construction also covers the 1-simplex case. All we need to kick
off our induction is to define 0-simplices to be points in $G$.

To make all this precise, we must use it to define smooth maps
\[ \varphi_p \maps \Delta^p \times G^{p+1} \to G, \]
for each $p \geq 0$. To overcome some analytic technicalities in constructing
$\varphi_p$, we will also need to fix a smooth increasing function:
\[ \ell \maps [0,1] \to [0,1] \]
which is 0 on a \emph{neighborhood} of 0, and then monotonically increases to 1
at 1. We shall call $\ell$ the \define{smoothing factor}. We shall see latter
that our choice of smoothing factor is immaterial: $\varphi_p$ depends on
$\ell$, but integrals over simplices do not. 

Let us begin by defining 0-simplices as points. That is, we define
\[ \varphi_0 \maps \Delta^0 \times G \to G \]
as the obvious projection.

Now, assume that we have defined $(p-1)$-simplices, so we have:
\[ \varphi_{p-1} \maps \Delta^{p-1} \times G^p \to G. \]
Using this, we wish to define:
\[ \varphi_p \maps \Delta^p \times G^{p+1} \to G. \]
But since we want this to be $G$-equivariant, we might as well define it for
\define{based \boldmath{$p$}-simplices}: a simplex whose first vertex is
$1$. So first, we will give a map:
\[ \tilde{f}_p \maps \Delta^p \times G^p \to G \]
which we think of as giving us the based $p$-simplex
\[ [1, g_1, \dots, g_p] \]
for any $p$-tuple. We do this using the apex-base construction. First, the map
$\varphi_{p-1} \maps \Delta^{p-1} \times G^p \to G$ can be extended to a map
\[ f_p \maps [0,1] \times \Delta^{p-1} \times G^p \to G \]
by defining $f_p$ to be $\varphi_{p-1}$ on $\{1\} \times \Delta^{p-1} \times G^p$, to
be $1$ on $\{0\} \times \Delta^{p-1} \times G^p$, and using the exponential map
to interpolate in between. Since $[0,1] \times \Delta^p$ is a kind of
generalized prism, we take the liberty of calling $\{0\} \times \Delta^p$ the
\define{0 face}, and $\{1\} \times \Delta^p$ the \define{1 face}.

Here, the requirement for smoothness complicates things slightly,
because we shall actually need $f_p$ to be $1$ on a \emph{neighborhood} of
the 0 face. So, to be precise, for $(t,x,g_1,\dots,g_p)) \in [0,1] \times
\Delta^{p-1} \times G^p$, we have that $\varphi_{p-1}(x,g_0,\dots,g_p)$ is a
point of $G$, say $\exp(X)$. Define:
\[ f_p(t, x, g_0, \dots, g_p) = \exp(\ell(t)X). \]
where $\ell$ is the smoothing factor we mention above: a smooth increasing
function which is 0 on a \emph{neighborhood} of 0, and then monotonically
increases to 1 at 1. This guarantees $f_p$ will be $1$ on a neighborhood of the
0 face, and will match $\varphi_{p-1}$ on the 1 face. 

Since $f_p$ is smooth and is constant on a neighborhood of the 0 face of the
prism, $[0,1] \times \Delta^{p-1}$, we can quotient by this face and obtain a
smooth map:
\[ \tilde{f}_p \maps \Delta^p \times G^p \to G. \]
For definiteness, we can use the smooth quotient map defined by:
\[ 
\begin{array}{lccc}
	q_p \maps & [0,1] \times \Delta^{p-1} & \to     & \Delta^p \\
	          & (t, x)                     & \mapsto & (1-t, tx) 
\end{array}
\]
which we note sends the 0 face to the 0th vertex of $\Delta^p$, and sends the
vertices of $\Delta^{p-1}$ to the remaining vertices of $\Delta^p$, in order.
Finally, to define the nonbased $p$-simplices, we extend by the left action of
$G$---for any $g \in G$ and any $(x, g_1, \dots, g_p) \in \Delta^p \times G^p$,
set:
\[ \varphi_p(x, g, g g_1, \dots, g g_p) = g \tilde{f}_p(x, g_1, \dots, g_p). \]
This defines
\[ \varphi_p \maps \Delta^p \times G^{p+1} \to G. \]
It just remains to check that:
\begin{prop} \label{prop:standard}
	This defines a left-invariant notion of simplices on $G$, which we call
	the \define{standard left-invariant notion of simplices with smoothing
	factor $\ell$}.
\end{prop}

\begin{proof}
	By construction, the $\varphi_p$ are all smooth and $G$-equivariant, so
	we only need to check the vertex property and the face property. We do
	this inductively.

	For 0-simplices, the vertex property is trivial. Assume it holds for
	$(p-1)$-simplices. In particular, the map
	\[ [g_1, \dots, g_p] \maps \Delta^{p-1} \to G \]
	sends the vertices of $\Delta^{p-1}$ to $g_1, \dots, g_p$, in that order.
	By construction, the based $p$-simplex
	\[ [1, g_1, \dots, g_p] \maps \Delta^p \to G \]
	sends the 0th vertex to 1 and the rest of the vertices to $g_1, \dots,
	g_p$, since the $(p-1)$-simplex $[g_1, \dots, g_p]$ has the vertex
	property and is defined to be the base of this $p$-simplex in the
	apex-base construction. By $G$-equivariance, this extends to all
	$p$-simplices.

	For 0-simplices, the face property holds vacuously, and for 1-simplices
	it is the same as the vertex property. Now take $p \geq 2$, and assume
	the face property holds for all $k$-simplices with $k < p$. By
	$G$-equivariance, the face property will hold for all $p$-simplices as
	long as it holds for all based $p$-simplices, for instance:
	\[ [1, g_1, \dots, g_p]. \]
	By the apex-base construction, the $(p-1)$-simplex $[g_1, \dots, g_p]$
	is the 0th face of $[1, g_1, \dots, g_p]$, since it was chosen as the
	base. For any other face, say the $i$th face, the apex-base
	construction gives the $(p-1)$-simplex
	\[ [1, g_1, \dots, \hat{g_i}, \dots, g_p] \maps \Delta^{p-1} \to G \]
	with 1 as apex, and the $(p-2)$-simplex $[g_1, \dots, \hat{g_i}, \dots,
	g_p]$ as base. Thus, the face property holds for the $p$-simplex $[1,
	g_1, \dots, g_p]$.
\end{proof}

While the existence of \emph{any} left-invariant notion of simplices in $G$ suffices
to integrate Lie algebra cochains, we have found an almost overwhelming wealth
of these notions---one for each smoothing factor $\ell$. In fact, for the
moment we will indicate the dependence of the standard notion of left-invariant
simplices on $\ell$ with a superscript:
\[ \varphi_p^\ell \maps \Delta^p \times G^{p+1} \to G . \]
Of course, the dependence of $\varphi_p^\ell$ on $\ell$ passes to the
individual simplices, so we give them a superscript as well:
\[ [g_0, \dots, g_p]^\ell \maps \Delta^p \to G . \]
Fortunately, however, the cochain map:
\[ \smallint \maps C^\bullet(\g) \to C^\bullet(G) \]
is \emph{independent} of $\ell$. That is, if $\ell'$ is another smoothing
factor, we have:
\[ \int_{[g_0, \dots, g_p]^\ell} \omega = \int_{[g_0, \dots, g_p]^{\ell'}} \omega , \]
for any left-invariant $p$-form $\omega$. 

We shall prove this not by comparing the integrals for two smoothing factors,
but rather computing the integral in a way that is manifestly independent of
smoothing factor. We do this by showing that the role of the smoothing factor
is basically to allow us to smoothly quotient the $p$-dimensional cube
$[0,1]^p$ down to the standard $p$-simplex $\Delta^p$. Had we parameterized our
$p$-simplices with cubes to begin with, we would have had no need for a smoothing
factor. As a trade off, however, our proof that integration gives a chain map
would have required more care when analyzing the boundary.

Now we get to work. Rather than parameterizing the $p$-simplex on the domain
$\Delta^p$:
\[ [g_0, g_1, \dots, g_p]^\ell \maps \Delta^p \to G , \]
we shall show how to parameterize it on the $p$-dimensional cube:
\[ \langle g_0, g_1, \dots, g_p \rangle \maps [0,1]^p \to G , \]
That is, these two functions have the same images---a $p$-simplex in $G$ with
vertices $g_0, \dots, g_p \in G$, they induce the same orientations on
their images, and both traverse it the image precisely once. So, as we shall
prove, the integral over either simplex is the same. But, as we shall also see,
the latter parameterization does not depend on the smoothing factor $\ell$.

How do we discover the parameterization $\langle g_0, \dots, g_p \rangle$? We
just repeat the apex-base construction, but we avoid quotienting to down to
$\Delta^p$! Begin by defining the 0-simplices to map the 0-dimensional cube to
the indicated vertex:
\[ \langle g_0 \rangle \maps \{0\} \to G \]
Define a 1-simplex by using the exponential map to sweep out a path from $g_0$
to $g_1$:
\[ \langle g_0, g_1 \rangle \maps [0,1] \to G , \] 
by defining:
\[ \langle g_0, g_1 \rangle (t_1) = g_0 \exp(t_1 X_1), \quad t_1 \in [0,1] . \]
where $g_0^{-1} g_1 = \exp(X_1)$. Now, define a 2-simplex using the exponential map to sweep out paths from $g_0$
to the 1-simplex $\langle g_1, g_2 \rangle$. That is, define:
\[ \langle g_0, g_1, g_2 \rangle \maps [0,1]^2 \to G , \]
to be given by:
\[ \langle g_0, g_1, g_2 \rangle (t_1, t_2) = g_0 \exp(t_1 Z(X_1, t_2 X_2)) , \]
where $g_0^{-1} g_1 = \exp(X_1)$, $g_1^{-1} g_2 = \exp(X_2)$, and $Z$ denotes the
Baker--Campbell--Hausdorff series:
\[ g_1 g_2 = \exp(Z(X_1, X_2) = \exp(X_1 + X_2 + \half[X_1, X_2] + \cdots ) . \]
Continuing in this manner, with a bit of work one can see that the $p$-simplex:
\[ \langle g_0, g_1, g_2, \dots, g_{p-1}, g_p \rangle \maps [0,1]^p \to G \]
is given by the horrendous formula:
\[ 
\langle g_0, \dots , g_p \rangle (t_1, \dots, t_p) = g_0 \exp(t_1 Z(X_1, t_2 Z(X_2, \dots, t_{p-1} Z(X_{p-1} , t_p X_p) \dots ))) ,
\]
where $g_0^{-1} g_1 = \exp(X_1)$, $g_1^{-1} g_2 = \exp(X_2)$, \dots,
$g_{p-1}^{-1} g_p = \exp(X_p)$. While horrendous, this formula is at least
independent of the smoothing factor $\ell$, and this forms the basis of the
following proposition:

\begin{prop}
	Let $G$ be an exponential Lie group with Lie algebra $\g$, let $\ell$
	be a smoothing factor, and equip $G$ with the standard left-invariant
	notion of simplices with smoothing factor $\ell$. For any $p$-simplex
	\[ [g_0, \dots, g_p ]^\ell \maps \Delta^p \to G , \]
	depending on $\ell$ and parameterized on the domain $\Delta^p$, there
	is a $p$-simplex:
	\[ \langle g_0, \dots, g_p \rangle \maps [0,1]^p \to G \]
	given by the formula:
	\[ 
	\langle g_0, \dots , g_p \rangle (t_1, \dots, t_p) = g_0 \exp(t_1 Z(X_1, t_2 Z(X_2, \dots, t_{p-1} Z(X_{p-1} , t_p X_p) \dots ))) ,
	\]
	where $g_0^{-1} g_1 = \exp(X_1)$, $g_1^{-1} g_2 = \exp(X_2)$, \dots,
	$g_{p-1}^{-1} g_p = \exp(X_p)$.  Then $\langle g_0, \dots, g_p \rangle$
	is independent of $\ell$, parameterized on the domain $[0,1]^p$, and
	has the same image and orientation as $[g_0, \dots, g_p]^\ell$.
	Furthermore, for any $p$-form $\omega$ on $G$, the integral of $\omega$
	is the same over either simplex:
	\[ \int_{[g_0, \dots, g_p]^\ell} \omega = \int_{\langle g_0, \dots, g_p \rangle} \omega . \]
\end{prop}
\begin{proof}[Sketch of proof]
	Equality of images and orientations follows from the apex-base
	construction, and equality of the integrals follows from
	reparameterization invariance---specifically the change of variables
	formula for multiple integrals, for which the monotonicity of $\ell$
	becomes crucial.
\end{proof}

\begin{cor}
	Let $G$ be an exponential Lie group with Lie algebra $\g$, let $\ell$
	be a smoothing factor, and equip $G$ with the standard left-invariant
	notion of simplices with smoothing factor $\ell$. Let 
	\[ \smallint \maps C^\bullet(\g) \to C^\bullet(G) \]
	be the cochain map from Lie algebra cochains to Lie group cochains
	given by integration over simplices. Then the cochain map $\smallint$
	is independent of the smoothing factor $\ell$.
\end{cor}
\begin{proof}
	Recall that if $\omega$ is a left-invariant $p$-form on $G$, and $[g_0,
	\dots, g_p]^\ell$ is a $p$-simplex in $G$, the cochain map $\smallint$
	is defined by:
	\[ (\smallint \omega)(g_0, \dots, g_p) = \int_{[g_0, \dots, g_p]^\ell} \omega \]
	By the previous proposition, this integral is equal to 
	\[ \int_{\langle g_0, \dots, g_p \rangle} \omega , \]
	where $\langle g_0, \dots, g_p \rangle \maps [0,1]^p \to G$ is given as
	above, and is independent of $\ell$. Thus, $\smallint$ is also
	independent of $\ell$.
\end{proof}

Having proven that the cochain map $\smallint$ is independent of smoothing
factor, we will now allow the smoothing factor to recede into the background.
Henceforth, we abuse terminology somewhat and speak of \emph{the} standard
left-invariant notion of simplices to mean the standard notion with some
implicit choice of smoothing factor.

The hard work of integrating Lie algebra cochains is now done. We would now
like to go the other way, and show how to get a Lie algebra cochain from a Lie
group cochain. This direction is much easier: in essence, we differentiate the
Lie group cochain at the identity, and antisymmetrize the result. To do this,
we make use of the fact that any element of the Lie algebra can be viewed as a
directional derivative at the identity. The following result, due to van Est
(c.f.\ \cite{vanEst}, Formula 46) just says this map defines a cochain map:
\begin{prop}
	Let $G$ be a Lie group with Lie algebra $\g$. Then there is a cochain map
	from the van Est complex to the Chevalley--Eilenberg complex:
	\[ D \maps C^{\bullet}(G) \to C^{\bullet}(\g) \]
	given by differentiation---that is, if $F$ is a homogeneous
	$p$-cochain on $G$, and $X_1, \dots, X_p \in \g$, then we can define:
	\[ DF(X_1, \dots, X_p) = \frac{1}{p!} \sum_{\sigma \in S_p} \mathrm{sgn}(\sigma) X_{\sigma(1)}^1 \dots X_{\sigma(p)}^p F(1, g_1, g_1 g_2, \dots, g_1 g_2 \dots g_p), \]
	where by $X_i^j$ we indicate that the operator $X_i$ differentiates
	only the $j$th variable, $g_j.$
\end{prop}
\begin{proof}
	See Houard \cite{Houard}, p.\ 224, Lemma 1.
\end{proof}

Having now defined cochain maps
\[ \smallint \maps C^{\bullet}(\g) \to C^{\bullet}(G) \]
and 
\[ D \maps C^{\bullet}(G) \to C^{\bullet}(\g), \]
the obvious next question is whether or not this defines a homotopy equivalence
of cochain complexes. Indeed, as proved by Houard, they do:
\begin{thm}
	Let $G$ be a Lie group equipped with a left-invariant notion of
	simplices, and $\g$ its Lie algebra. The cochain map
	\[ D \smallint \maps C^{\bullet}(\g) \to C^{\bullet}(\g), \]
	is the identity, whereas the cochain map
	\[ \smallint D \maps C^{\bullet}(G) \to C^{\bullet}(G) \]
	is cochain-homotopic to the identity. Therefore the Lie algebra cochain
	complex $C^\bullet(\g)$ and the Lie group cochain complex
	$C^\bullet(G)$ are homotopy equivalent and thus have isomorphic
	cohomology.
\end{thm}
\begin{proof}
	See Houard \cite{Houard}, p.\ 234, Proposition 2.
\end{proof}

\section{Examples: Explicitly integrating 0-, 1-, 2- and 3-cochains} \label{sec:examples}

In this section, in order to get a feel for the integration procedure given in
Proposition \ref{prop:integral}, we shall explicitly calculate some Lie group
cochains from Lie algebra cochains. The resulting formulas are polynomials on
the Lie group, at least in the nilpotent case. It is important to note,
however, that you do not need to know these explicit formulas in what follows.
It is enough to understand that they exist, and have the properties described
in Section \ref{sec:integratingcochains}. We nonetheless suspect that explicit
formulas will prove useful in future work, so we collect some here.

To facilitate this calculation, we shall also have to explicitly construct some
low-dimensional left-invariant simplices. For 0-cochains and 1-cochains, we
will find the task very easy---we only need our Lie group $G$ to be
exponential. On the other hand, for 2- and 3-cochains, the construction gets
much harder. This complexity shows just how powerful the abstract approach of
the previous section actually is---imagine having to prove Proposition
\ref{prop:integral} through an explicit integration such as those we present
here!

So, for 2- and 3-cochains, we simplify the problem by assuming our Lie algebra
$\g$ to be \define{2-step nilpotent}: all brackets of brackets are zero. This
allows us to use a simplified form of the Baker--Campbell--Hausdorff formula:
\[ \exp(X) \exp(Y) = \exp(X + Y + \half [X,Y]) \]
and the Zassenhaus formula:
\begin{equation} \label{eqn:zassenhaus} 
	\exp(X + Y) = \exp(X) \exp(Y) \exp(-\half [X,Y]) = \exp(X) \exp(Y - \half[X,Y]). 
\end{equation}
Partially, this nilpotentcy assumption just makes our calculations tenable, but
secretly it is because our main application of these ideas will be to 2-step
nilpotent Lie \emph{superalgebras}.

\subsection*{0-cochains} \label{sec:0-cochains}

Let $\omega$ be a Lie algebra 0-cochain: that is, a real number. Then
$\smallint \omega = \omega$ is a Lie group 0-cochain. We can view it as the
integral of $\omega$ over the 0-simplex $[1]$.

\subsection*{1-cochains} \label{sec:1-cochains}

Let $\omega$ be a Lie algebra 1-cochain: that is, a linear map 
\[ \omega \maps \g \to \R, \]
which we extend to a 1-form on $G$ by left translation. We define a Lie group
1-cochain $\smallint \omega$ by integrating $\omega$ over 1-simplices in $G$.
In particular
\[ \smallint \omega(g) = \int_{[1,g]} \omega. \]
Since $G$ is exponential, it has a standard left-invariant notion of 1-simplex,
given by exponentiation. So, if $g = \exp(X)$, then the 1-simplex $[1,g]$ is
given by
\[ [1,g](t) = \exp(tX), \quad 0 \leq t \leq 1.\]
We denote this map by $\varphi$ for brevity. So:
\[ \smallint \omega(g) = \int_0^1 \omega(\dot{\varphi}(t)) \, dt \]
Noting that the derivative of $\varphi$ is 
\[ \dot{\varphi}(t) = \exp(tX) X \]
we have
\[ \smallint \omega(g) = \int_0^1 \omega( \exp(tX) X) \, dt =  \int_0^1 \omega(X) \, dt = \omega(X), \]
where we have used the left invariance of $\omega$. In summary, 
\[ \smallint \omega(g) = \omega(X), \]
for $g = \exp(X)$.

As a check on this, note that because we have proved $\smallint$ is a cochain
map, $\smallint \omega$ should be a cocycle whenever $\omega$ is. So let us
verify this. Assume $d \omega = 0$. That is, for all $X$ and $Y \in \g$, we
have:
\[ d \omega(X,Y) = - \omega([X,Y]) = 0. \]
So the cocycle condition merely says that $\omega$ must vanish on brackets. 

Now compute the coboundary of $\smallint \omega$. As a function of a single
group element, $\smallint \omega$ is an inhomogeneous Lie group 2-cochain. So
we must use the coboundary formula from Chapter \ref{ch:Lie-n-groups}, which we
recall here for convenience: when $f \maps G^p \to \R$ is an inhomogeneous Lie
group $p$-cochain, its coboundary is:
\begin{eqnarray*}
	df(g_1, \dots, g_{p+1}) & = & \sum_{i=1}^p (-1)^i f(g_1, \dots, g_{i-1}, g_i g_{i+1}, g_{i+2}, \dots, g_{p+1}) \\
				&   & + (-1)^{p+1} f(g_1, \dots, g_p) .
\end{eqnarray*}
In the case of $\smallint \omega$, this becomes simply:
\[ d \smallint \omega(g,h) = \smallint \omega(h) - \smallint \omega(gh) + \smallint \omega(g). \]

Finally, check that this coboundary is zero when $\omega$ is a cocycle, and
hence that $\smallint \omega$ is a cocycle whenever $\omega$ is. If $g =
\exp(X)$ and $h = \exp(Y)$, we have 
\[ gh = \exp(X) \exp(Y) = \exp(X + Y + \half[X,Y] + \cdots) \] 
by the Baker--Campbell--Hausdorff formula, and thus:
\[ d \smallint \omega(g, h) = \omega(Y) - \omega(X + Y + \half[X,Y] + \cdots) + \omega(X) = 0 \]
where we have used $\omega$'s linearity along with the cocycle condition that
$\omega$ must vanish on brackets.

\subsection*{2-cochains} \label{sec:2-cochains}

As we have just seen, 0-cochains and 1-cochains are easily integrated on any
exponential Lie group, and the result is always a polynomial Lie group cochain.
Unfortunately, even for 2-cochains, the integration is much more complicated,
and no longer polynomial unless $\g$ is nilpotent. So, at this point, we will
simplify matters by assuming $\g$ to be 2-step nilpotent. To hint at this with
our notation, we will now call our Lie algebra $\n$ and the corresponding
simply-connected Lie group $N$. 

Let $\omega$ be a Lie algebra 2-cochain: that is, a left-invariant 2-form. We
define a Lie group 2-cochain $\smallint \omega$ by integrating $\omega$ over
2-simplices in $N$. In particular:
\[ \smallint \omega(g, h) = \int_{[1,g,gh]} \omega. \]
Now suppose $g = \exp(X)$ and $h = \exp(Y)$. Recall we that obtain the
2-simplex $[1, g, gh]$ using the apex-base construction: we connect each point
of the base $[g,gh] = g[1,h]$ to 1 by the exponential map. Since $[1,h](t) =
\exp(tY)$, the base is parameterized by 
\[ [g,gh](t) = g \exp(tY) = \exp(X + tY + \frac{t}{2} [X,Y]) \]
by the Baker--Campbell--Hausdorff formula. Now let us construct $[1,g,gh]$ by
first constructing a map from the square
\[ \varphi \maps [0,1] \times [0,1] \to N \]
given by 
\[ \varphi(s,t) = \exp(s(X + tY + \frac{t}{2} [X,Y])). \]
At this stage in our general construction, since this map is 1 on the $\{0\}
\times [0,1]$ edge of the square, we would typically quotient the square out by
this edge to obtain a map from the standard 2-simplex. But in practice, we do
not need to do this. Since the integral $\int_{[1,g,gh]} \omega$ is invariant
under reparameterization, we might as well parameterize our 2-simplex
$[1,g,gh]$ with $\varphi$ and integrate over the square to obtain:
\[ \smallint \omega(g,h) = \int_0^1 \int_0^1 \omega(\frac{\partial \varphi}{\partial s}, \frac{\partial \varphi}{\partial t}) \, ds \, dt. \]
Our task has essentially been reduced to computing the partial derivatives of
$\varphi$. Thanks to the left invariance of $\omega$, we may as well
left translate these partials back to 1 once we have them, since:
\[ \omega(\frac{\partial \varphi}{\partial s}, \frac{\partial \varphi}{\partial t}) = \omega(\varphi^{-1} \frac{\partial \varphi}{\partial s}, \varphi^{-1} \frac{\partial \varphi}{\partial t}). \]

Let us begin with $\frac{\partial \varphi}{\partial s}$. Since the exponent of
$\varphi(s,t) = \exp(s(X + tY + \frac{t}{2} [X,Y]))$ is linear in $s$, this is
simply:
\[ \frac{\partial \varphi}{\varphi s}(s,t) = \varphi(s,t)(X + tY \frac{t}{2} [X, Y]). \]
This is a tangent vector at $\varphi(s,t)$. We can left translate it back to 1 to obtain:
\[ \varphi^{-1} \frac{\partial \varphi}{\partial s} = X + tY \frac{t}{2} [X, Y]). \]

The partial with respect to $t$ is slightly harder, because the exponent is not
linear in $t$. To compute this, we need the Zassenhaus formula, Formula
\ref{eqn:zassenhaus}, to separate the terms linear in $t$ from those that are
not. Applying this, we obtain
\[ \varphi(s,t) = \exp(sX) \exp(stY + \frac{st}{2} [X,Y] -\frac{s^2 t}{2}[X,Y]). \]
Differentiating this with respect to $t$ and left translating the result to 1, we get:
\[ \varphi^{-1} \frac{\partial \varphi}{\partial t} = sY + \frac{s - s^2}{2} [X,Y]. \]
Substituting these partial derivatives into the integral, our problem becomes:
\[ \smallint \omega(g, h) = \int_0^1 \int_0^1 \omega(X + tY + \frac{t}{2} [X,Y], \, sY + \frac{s - s^2}{2} [X,Y]) \, ds \, dt. \]
It is now easy enough, using $\omega$'s bilinearity and antisymmetry, to bring
all the polynomial coefficients out and integrate them, obtaining an expression
which is the sum of three terms:
\[ \smallint \omega(g, h) = \half \omega(X,Y) + \frac{1}{12} \omega(X, [X,Y]) - \frac{1}{12} \omega(Y, [X,Y]). \]

Nevertheless, we would like to do this calculation explicitly. In essence, we
use $\omega$'s bilinearity and antisymmetry to our advantage, to write these
coefficients as integrals of various determinants. To wit, the coefficent of
$\omega(X,Y)$ is the integral of the determinant
\[ \left| \begin{array}{cc} 
	1 & t \\ 
	0 & s
\end{array} \right|
= s,
\]
which we obtain from reading off the coefficients of $X$ and $Y$ in the integrand:
\[ \omega(X + tY + \frac{t}{2} [X,Y], \, sY + \frac{s - s^2}{2} [X,Y]). \]
So the coefficient of $\omega(X,Y)$ is $\int_0^1 \int_0^1 s \, ds \, dt =
\half$. We can use this idea to obtain the other two coefficients as well---the
coefficient of $\omega(X,[X,Y])$ is the integral of the determinant
\[ \left| \begin{array}{cc} 
	1 & \frac{t}{2} \\ 
	0 & \frac{s - s^2}{2} 
\end{array} \right|
= \frac{s - s^2}{2},
\]
which is $\frac{1}{12}$, and the coefficient of $\omega(Y, [X,Y])$ is the
integral of the determinant
\[ \left| \begin{array}{cc} 
	t & \frac{t}{2} \\ 
	s & \frac{s - s^2}{2} 
\end{array} \right|
= -\frac{s^2 t}{2}, 
\]
which is $-\frac{1}{12}$.

As a final check on this calculation, let us again show that when $\omega$ is a
cocycle, so is $\smallint \omega$. We know this must be true by Proposition
\ref{prop:integral}, of course, but when checking it explicitly the cocycle
condition seems almost miraculous. Since this final computation is a bit of a
workout, we tuck it into the proof of the following proposition. It is only a
check, and understanding the calculation is not necessary for what follows.

\begin{prop}
	Let $N$ be a simply-connected Lie group whose Lie algebra $\n$ is
	2-step nilpotent. If $\omega$ is a Lie algebra 2-cocycle on $\n$, then
	the Lie group 2-cochain on $N$ defined by 
	\[ \smallint \omega(g, h) = \half \omega(X, Y) + \frac{1}{12} \omega(X - Y, [X,Y]), \]
	where $g = \exp(X)$ and $h = \exp(Y)$, is also a cocycle.
\end{prop}

\begin{proof}

As already noted, this fact is immediate from Proposition \ref{prop:integral},
but we want to ignore this and check it explicitly. To do this, we
repeatedly use the Baker--Campbell--Hausdorff formula, the assumption that $\n$
is 2-step nilpotent, and the cocycle condition on $\omega$. This latter
condition reads:
\[ d \omega(X,Y,Z) = -\omega([X,Y], Z) + \omega([X,Z],Y) - \omega([Y,Z],X) = 0. \]
Note how this resembles the Jacobi identity. We prefer to write it as follows:
\[ \omega(X, [Y,Z]) = \omega([X,Y],Z) + \omega(Y, [X,Z]). \]

To begin, the coboundary of the inhomogeneous Lie group 3-cochain $\smallint
\omega$ is given by:
\[ d \smallint \omega(g, h, k) = \smallint \omega(h, k) - \smallint \omega(gh, k) + \smallint \omega(g, hk) - \smallint \omega(g, h) . \]
Let us assume that
\[ g = \exp(X), \quad h = \exp(Y), \quad k = \exp(Z), \]
so that 
\[ gh = \exp(X + Y + \half [X,Y]), \quad hk = \exp(Y + Z + \half[Y,Z]). \]
Now we repeatedly insert the expression for our Lie group 2-cochain, so the
coboundary of $\smallint \omega$ becomes:
\begin{eqnarray*}
	d \smallint \omega(g, h, k) & = & \half \omega(Y, Z) + \frac{1}{12} \omega(Y - Z, [Y,Z]) \\
				    &   & - \half \omega(X + Y + \half[X,Y], Z) - \frac{1}{12} \omega(X + Y + \half [X,Y] - Z, [X + Y, Z]) \\
				    &   & + \half \omega(X, Y + Z + \half [Y,Z] ) + \frac{1}{12} \omega(X - Y - Z - \half [Y,Z] , [X,Y + Z]) \\
				    &   & - \half \omega(X, Y) - \frac{1}{12} \omega(X - Y, [X,Y]), 
\end{eqnarray*}

Note that the cocycle condition combined with nilpotency implies that any term
in which $\omega$ eats two brackets vanishes. In general,
\[ \omega([X,Y], [Z,W]) = \omega([ [X,Y], Z], W) + \omega(Z, [ [X, Y], W]) = 0, \]
thanks to the fact that brackets of brackets vanish. So, in the expression for
$d \smallint \omega$, we can simplify the fourth term:
\begin{eqnarray*}
	\omega(X + Y + \half [X,Y] - Z, [X + Y, Z]) & = & \omega(X + Y - Z, [X + Y, Z]) + \half \omega([X,Y], [X + Y, Z]) \\
	                                            & = & \omega(X + Y - Z, [X + Y, Z]).
\end{eqnarray*}
Similarly for the sixth term:
\[ \omega(X - Y - Z - \half [Y,Z] , [X,Y + Z]) = \omega(X - Y - Z, [X,Y + Z]). \]
This leaves us with:
\begin{eqnarray*}
	d \smallint \omega(g, h, k) & = & \half \omega(Y, Z) + \frac{1}{12} \omega(Y - Z, [Y,Z]) \\
				    &   & - \half \omega(X + Y + \half[X,Y], Z) - \frac{1}{12} \omega(X + Y - Z, [X + Y, Z]) \\
				    &   & + \half \omega(X, Y + Z + \half [Y,Z] ) + \frac{1}{12} \omega(X - Y - Z, [X,Y + Z]) \\
				    &   & - \half \omega(X, Y) - \frac{1}{12} \omega(X - Y, [X,Y]), 
\end{eqnarray*}
Expanding this using bilinearity, we obtain, after many cancellations:
\begin{eqnarray*}
	d \smallint \omega(g, h, k) & = & - \frac{1}{4} \omega([X,Y],Z) - \frac{1}{12} \omega(X, [Y,Z]) - \frac{1}{12} \omega(Y,[X,Z])  \\
                                    &   & + \frac{1}{4} \omega(X,[Y,Z]) - \frac{1}{12} \omega(Y, [X,Z]) - \frac{1}{12} \omega(Z,[X,Y])  .
\end{eqnarray*}
We combine the two terms with coefficient $1/4$ using the cocycle condition:
\[ -\omega([X,Y],Z) + \omega(X, [Y,Z]) = \omega([Y,X],Z) + \omega(X,[Y,Z]) = \omega(Y,[X,Z]). \]
Similarly, for the first and fourth terms with coefficent $1/12$, we apply the
cocycle condition to get:
\[ \omega(X,[Y,Z]) + \omega(Z,[X,Y]) = \omega(Y, [X, Z]). \]
So, substituting these in, we finally obtain:
\[ d \smallint \omega(g, h, k) = \frac{1}{4} \omega(Y, [X,Z]) - \frac{1}{12} \omega(Y,[X,Z]) - \frac{1}{12} \omega(Y, [X,Z]) - \frac{1}{12} \omega(Y,[X,Z]) = 0, \]
as desired.

\end{proof}

As a corollary, note that we could equally well have said:

\begin{cor}
	Let $N$ be a simply-connected Lie group whose Lie algebra $\n$ is
	2-step nilpotent. If $\omega$ is a Lie algebra 2-cocycle on $\n$, then
	the Lie group 2-cochain on $N$ defined by 
	\[ \smallint \omega(g, h) = \int_0^1 \int_0^1 \omega(X + tY + \frac{t}{2} [X,Y], \, sY + \frac{s - s^2}{2} [X,Y]) \, ds \, dt, \]
	where $g = \exp(X)$ and $h = \exp(Y)$, is also a cocycle.
\end{cor}

\begin{proof}
	By our calculation in this section,
	\[ \smallint \omega(g, h) = \half \omega(X, Y) + \frac{1}{12} \omega(X - Y, [X,Y]), \]
	so the result is immediate.
\end{proof}

\subsection*{3-cochains} \label{sec:3-cochains}

Let $\omega$ be a 3-cochain on the Lie algebra: that is, a left-invariant
3-form. Judging by our experience in the last section, the complexity of
integrating $\omega$ to a Lie group 3-cochain may be quite high. Indeed, we
shall ultimately avoid writing down $\smallint \omega$, except as an integral.
Nonetheless, we can make this integral quite explicit.

We define the Lie group 3-cochain $\smallint \omega$ to be the integral of
$\omega$ over a 3-simplex in $N$. In particular:
\[ \smallint \omega(g, h, k) = \int_{[1, g, gh, ghk]} \omega. \]
Now assume that $g = \exp(X)$, $h = \exp(Y)$ and $k = \exp(Z)$. Recall we that obtain the
3-simplex $[1, g, gh, ghk]$ using the apex-base construction: we connect each point
of the base $[g,gh,ghk] = g[1,h, hk]$ to 1 by the exponential map. In the last
section, we saw that $[1,h,hk](t,u) = \exp(t(Y + uZ + \frac{u}{2} [Y,Z]))$, so
the base is parameterized by 
\begin{eqnarray*}
	[g,gh,ghk](t,u) & = & g \exp(t(Y + uZ + \frac{u}{2} [Y,Z]) \\
	                & = & \exp(X + tY + tuZ + \frac{tu}{2} [Y,Z] + \half[X, tY + tuZ]),
\end{eqnarray*}
by the Baker--Campbell--Hausdorff formula. Now let us construct $[1,g,gh,ghk]$
by first constructing a map from the cube
\[ \varphi \maps [0,1] \times [0,1] \times [0,1] \to N \]
given by 
\begin{eqnarray*}
	\varphi(s,t,u) & = & \exp(s(X + tY + tuZ + \frac{tu}{2} [Y,Z] + \half[X, tY + tuZ])) \\
	               & = & \exp(sX + stY + stuZ + \frac{st}{2}[X, Y] + \frac{stu}{2} [Y,Z] + \frac{stu}{2} [X,Z]).
\end{eqnarray*}
At this stage in our general construction, since this map is 1 on the $\{0\}
\times [0,1] \times [0,1]$ face of the cube and on the lines $\{s\} \times
\{0\} \times [0,1]$ of constant $s$ on the $[0,1] \times {0} \times [0,1]$ face
of the cube, we could quotient the cube out by these sets to obtain a map from
the standard 3-simplex. But in practice, we do not need to do this. Since the
integral $\int_{[1,g,gh,ghk]} \omega$ is invariant under reparameterization, we
might as well parameterize our 3-simplex $[1,g,gh,ghk]$ with $\varphi$ and
integrate over the cube to obtain:
\[ \smallint \omega (g,h,k) = \int_0^1 \int_0^1 \int_0^1 \omega(\frac{\partial \varphi}{\partial s}, \frac{\partial \varphi}{\partial t}, \frac{\partial \varphi}{\partial u}) \, ds \, dt \, du . \]
Once again, our task has essentially reduced to computing the partial
derivatives of $\varphi$, and once again, thanks to the left invariance of
$\varphi$, we may as well left translate these partials back to 1 once we have
them, since:
\[ \omega(\frac{\partial \varphi}{\partial s}, \frac{\partial \varphi}{\partial t}, \frac{\partial \varphi}{\partial u}) = \omega(\varphi^{-1} \frac{\partial \varphi}{\partial s}, \varphi^{-1} \frac{\partial \varphi}{\partial t}, \varphi^{-1} \frac{\partial \varphi}{\partial u}). \]

Let us begin with $\frac{\partial \varphi}{\partial s}$. Since the exponent of
$\varphi(s,t,u)$ is linear in $s$, this is simply:
\[ \frac{\partial \varphi}{\varphi s}(s,t,u) = \varphi(s,t,u)(X + tY + tuZ + \frac{t}{2} [X, Y] + \frac{tu}{2} [Y,Z] + \frac{tu}{2} [X,Z]). \]
This is a tangent vector at $\varphi(s,t,u)$. We can left translate it back to 1 to obtain:
\[ \varphi^{-1} \frac{\partial \varphi}{\varphi s} = X + tY + tuZ + \frac{t}{2} [X, Y] + \frac{tu}{2} [Y,Z] + \frac{tu}{2} [X,Z]. \]

The partial with respect to $t$ is slightly harder, because the exponent is not
linear in $t$. To compute this, we again need the Zassenhaus formula, Formula
\ref{eqn:zassenhaus}, to separate the terms linear in $t$ from those that are
not. Applying this, we obtain
\[ \varphi(s,t,u) = \exp(sX) \exp(stY + stuZ + \frac{st}{2} [X,Y] + \frac{stu}{2}[Y,Z] + \frac{stu}{2}[X,Z] - \half [sX, stY + stuZ]). \]
Differentiating this with respect to $t$ and left translating the result to 1, we get:
\[ \varphi^{-1} \frac{\partial \varphi}{\partial t} = sY + suZ + \frac{s}{2} [X,Y] + \frac{su}{2}[Y,Z] + \frac{su}{2}[X,Z] - \half [sX, sY + suZ], \]
which we can simplify by combining like terms:
\[ \varphi^{-1} \frac{\partial \varphi}{\partial t} = sY + suZ + \frac{s-s^2}{2} [X,Y] + \frac{su}{2}[Y,Z] + \frac{su - s^2 u}{2}[X,Z] . \]

Finally, the partial with respect to $u$ requires that we separate out the
terms linear in $u$, again using the Zassenhaus formula:
\[ \varphi(s,t,u) = \exp(sX +stY + \frac{st}{2}[X,Y]) \exp(stuZ + \frac{stu}{2}[Y,Z] + \frac{stu}{2}[X,Z] - \half [sX + stY, stuZ]). \]
Differentiating this with respect to $u$ and left translating the result to 1, we get:
\[ \varphi^{-1} \frac{\partial \varphi}{\partial u} = stZ + \frac{st}{2}[Y,Z] + \frac{st}{2}[X,Z] - \half [sX + stY, stZ], \]
which we can again simplify by combining like terms:
\[ \varphi^{-1} \frac{\partial \varphi}{\partial u} = stZ + \frac{st - s^2 t^2}{2}[Y,Z] + \frac{st - s^2 t}{2}[X,Z] . \]

Substituting these partial derivatives into the integral, our problem becomes:
\[
\begin{array}{rcrl}
	\smallint \omega(g, h, k) & = & \displaystyle \int_0^1 \int_0^1 \int_0^1 \omega( & \displaystyle X + tY + tuZ + \frac{t}{2} [X, Y] + \frac{tu}{2} [Y,Z] + \frac{tu}{2} [X,Z], \\
	                                                                             & & & \displaystyle sY + suZ + \frac{s-s^2}{2} [X,Y] + \frac{su}{2}[Y,Z] + \frac{su - s^2 u}{2}[X,Z], \\
	                                                                             & & & \displaystyle stZ + \frac{st - s^2 t^2}{2}[Y,Z] + \frac{st - s^2 t}{2}[X,Z] \quad ) \, ds \, dt \, du . \\
\end{array}
\]
This integral is bad enough. Further evaluating this integral is quite a chore
(the answer involves 17 nonzero terms!), so we stop here. We would only like to
give a hint as to how the evaluation could be done. As in Section
\ref{sec:2-cochains}, thanks to $\omega$'s trilinearity and antisymmetry, the
coefficients of the terms in $\smallint \omega(g, h, k)$ are integrals of
various determinants. For instance, the coefficient of $\omega(X,Y,Z)$ is the
integral of the $3 \times 3$ determinant
\[ \left| \begin{array}{ccc} 
	1 & t & tu \\
	0 & s & su \\
	0 & 0 & st \\
\end{array} \right|
= s^2 t,
\]
which we obtain from reading off the coefficients of $X$, $Y$ and $Z$ in the
integrand. So the coefficient of $\omega(X, Y, Z)$ in $\smallint \omega(g, h,
k)$ is $\int_0^1 \int_0^1 \int_0^1 s^2 t \, ds \, dt \, du = \frac{1}{6}$. The
other terms may be computed similarly.

Just as we shall not attempt to evaluate the integral for $\smallint \omega(g,
h, k)$, we also do not attempt to demonstrate that it gives a Lie group cocycle
when $\omega$ is a Lie algebra cocycle. After all,
Proposition~\ref{prop:integral} does this for us:

\begin{prop}
	Let $N$ be a simply-connected Lie group whose Lie algebra $\n$ is
	2-step nilpotent. If $\omega$ is a Lie algebra 3-cocycle on $\n$, then
	the Lie group 3-cochain on $N$ given by
	\[
	\begin{array}{rcrl}
		\smallint \omega(g, h, k) & = & \displaystyle \int_0^1 \int_0^1 \int_0^1 \omega( & \displaystyle X + tY + tuZ + \frac{t}{2} [X, Y] + \frac{tu}{2} [Y,Z] + \frac{tu}{2} [X,Z], \\
		                                                                             & & & \displaystyle sY + suZ + \frac{s-s^2}{2} [X,Y] + \frac{su}{2}[Y,Z] + \frac{su - s^2 u}{2}[X,Z], \\
		                                                                             & & & \displaystyle stZ + \frac{st - s^2 t^2}{2}[Y,Z] + \frac{st - s^2 t}{2}[X,Z] \quad ) \, ds \, dt \, du , \\
	\end{array}
	\]
	where $g = \exp(X)$, $h = \exp(Y)$ and $k = \exp(Z)$, is also a
	cocycle.
\end{prop}

\begin{proof}
	This is immediate upon combining Proposition \ref{prop:integral} with
	the above discussion.
\end{proof}

\section{The Heisenberg Lie 2-group} \label{sec:heisenberg2group}

In Section \ref{sec:sec:heisenberg2alg}, we met the Heisenberg Lie algebra,
$\mathfrak{H} = \mathrm{span}(p,q,z)$. This is the 3-dimensional Lie algebra
where the generators $p$, $q$ and $z$ satisfy relations which mimic the
canonical commutation relations from quantum mechanics:
\[ [p,q] = z, \quad [p,z] = 0, \quad [q,z] = 0 . \]
As one can see from the above relations, $\mathfrak{H}$ is 2-step nilpotent:
brackets of brackets are zero. 

We then met the Lie 2-algebra generalization, the Heisenberg Lie 2-algebra:
\[ \mathfrak{Heisenberg} = \strng_\gamma(\mathfrak{H}), \]
built by extending $\mathfrak{H}$ with the 3-cocycle $\gamma = p^* \wedge q^*
\wedge z^*$, where $p^*$, $q^*$, and $z^*$ is the basis dual to $p$, $q$ and
$z$.

It is easy to construct a Lie group $H$ with Lie algebra $\mathfrak{H}$. Just
take the group of $3 \times 3$ upper triangular matrices with units down the
diagonal:
\[ H = \left\{ \left(
		\begin{matrix} 
			1 & a & b \\ 
			0 & 1 & c \\
			0 & 0 & 1 \\
		\end{matrix} 
                \right)
		: a, b, c \in \R \right\} . 
\]
This is an exponential Lie group:
\[ \begin{array}{cccc} \exp \maps & \mathfrak{H} & \to     &  H \\
		                  & ap + cq + bz & \mapsto & \left( \begin{matrix} 1 & a & b \\ 0 & 1 & c \\ 0 & 0 & 1 \\ \end{matrix} \right) . 
   \end{array}
\]
So we can apply Proposition \ref{prop:standard} to construct the standard
left-invariant notion of simplices in $H$, and Proposition \ref{prop:integral}
to integrate the Lie algebra 3-cocycle $\gamma$ to a Lie group 3-cocycle
$\smallint \gamma$. We therefore get a Lie 2-group, the \define{Heisenberg Lie
2-group}:
\[ \mathrm{Heisenberg} = \String_{\smallint \gamma}(H) . \]

\chapter{Supergeometry and supergroups} \label{ch:supergeometry}

We would now like to generalize our work from Lie algebras and Lie groups to
Lie superalgebras and supergroups. Of course, this means that we need a way to
talk about Lie supergroups, their underlying supermanifolds, and the maps
between supermanifolds. This task is made easier because we do not need the
full machinery of supermanifold theory. Because our supergroups will be
exponential, we only need to work with supermanifolds that are diffeomorphic
to super vector spaces. Nonetheless, let us begin with a sketch of
supermanifold theory from the perspective that suits us best, which could
loosely be called the `functor of points' approach.

The rough geometric picture one should have of a supermanifold $M$ is that of
an ordinary manifold with infinitesimal `superfuzz', or `superdirections',
around each point. At the infinitesimal level, an ordinary manifold is merely a
vector space---its tangent space at a point. In contrast, the tangent space to
$M$ has a $\Z_2$-grading: tangent vectors which point along the underlying
manifold of $M$ are taken to be even, while tangent vectors which point along
the superdirections are taken to be odd.

At least infinitesimally, then, all supermanifolds look like super vector spaces, 
\[ \R^{p|q} : = \R^p \oplus \R^q , \] 
where $\R^p$ is even and $\R^q$ is odd.  And indeed, just as ordinary manifolds
are locally modeled on ordinary vector spaces, $\R^n$, supermanifolds are
locally modeled on super vector spaces, $\R^{p|q}$. But before we sketch how
this works, let us introduce our main tool: the so-called `functor of points'.

The basis for the functor of points is the Yoneda Lemma, a very general and
fundamental fact from category theory:
\begin{YL} 
	Let $C$ be a category. The functor  
	\[ \begin{array}{ccc} 
		 C & \to & \Fun(C^{\op}, \Set) \\
		 x & \mapsto & \Hom(-,x) \\
	\end{array} \]
	is a full and faithful embedding of $C$ into the category
	$\Fun(C^{\op}, \Set)$ of contravariant functors from $C$ to $\Set$.
	This embedding is called the \define{Yoneda embedding}. 
\end{YL}
The upshot of this lemma is that, without losing any information, we can
replace an object $x$ by a functor $\Hom(-,x)$, and a morphism $f \maps x
\to y$ by a natural transformation 
\[ \Hom(-,f) \maps \Hom(-,x) \Rightarrow \Hom(-,y) \] 
of functors. Each component of this natural transformation is the `obvious'
thing: for an object $z$, the function
\[ \Hom(z,f) \maps \Hom(z,x) \to \Hom(z,y) \]
just takes the morphism $g \maps z \to x$ to the morphism $fg \maps z \to y$.

On a more intuitive level, the functor of points tells us how to reconstruct a
`space' $x \in C$ by probing it with \emph{every other space} $z \in C$---that
is, by looking at all the ways in which $z$ maps into $x$, which forms the set
$\Hom(z,x)$. The true power of the functor of points, however, arises when we
can reconstruct $x$ without having to probe it will \emph{every} $z$, but with
$z$ from a manageable subcategory of $C$. And while it deviates slightly from the
spirit of the Yoneda Lemma, we can shrink this subcategory still further if we
allow $\Hom(z,x)$ to have more structure than that of a mere set. In fact, when
$M$ is a supermanifold, we will consider probes $z$ for which $\Hom(z,M)$ is an
ordinary manifold. 

For what $z$ is $\Hom(z,M)$ a manifold? One clue is that when $M$ is an
ordinary manifold, there is a manifold of ways to map a point into $M$:
\[ M \iso \Hom(\R^0, M) , \]
but the space of maps from any higher-dimensional manifold to $M$ is generally
not a finite-dimensional manifold in its own right. Similarly, when $M$ is a
supermanifold, there is an ordinary manifold of ways to map a point into $M$:
\[ M_{\R^{0|0}} = \Hom(\R^{0|0}, M) . \]
One should think of this as the ordinary manifold one gets from $M$ by
forgetting about the superdirections. But thanks to the superdirections, we now we have
more ways of obtaining a manifold of maps to $M$: there is an ordinary manifold
of ways to map a point with $q$ superdirections into $M$:
\[ M_{\R^{0|q}} = \Hom(\R^{0|q}, M) . \]
So, for every supermanifold $M$, we get a functor:
\[ \begin{array}{cccc} 
	\Hom(-,M) \maps & \SuperPoints^\op & \to     & \Man \\
		        & \R^{0|q}         & \mapsto & \Hom(\R^{0|q}, M) 
\end{array}
\]
where $\SuperPoints$ is the category consisting of supermanifolds of the form
$\R^{0|q}$ and smooth maps between them. Of course, we have not yet said what
this category is precisely, but one should think of $\R^{0|q}$ as a
supermanifold whose underlying manifold consists of one point, with $q$
infinitesimal superdirections---a `superpoint'. Because this lets us
probe the superdirections of $M$, this functor has enough information to completely
reconstruct $M$. We will go further, however, and sketch how to define $M$ as a
certain kind of functor from $\SuperPoints^\op$ to $\Man$.

This approach goes back to Schwarz \cite{Schwarz} and Voronov \cite{Voronov},
who used it to formalize the idea of `anticommuting coordinates' used in the
physics literature. Since Schwarz, a number of other authors have developed the
functor of points approach to supermanifolds, most recently Sachse
\cite{Sachse} and Balduzzi, Carmeli and Fioresi \cite{BCF}. We will follow
Sachse, who defines supermanifolds entirely in terms of their functors of
points, rather than using sheaves.

\section{Supermanifolds} \label{sec:supermanifolds}

\subsection{Super vector spaces as supermanifolds}

Let us now dive into supermathematics. Our main need is to define smooth maps
between super vector spaces, but we will sketch the full definition of
supermanifolds and the smooth maps between them. Just as an ordinary manifold
is a space that is locally modeled on a vector space, a supermanifold is
locally modeled on a super vector space. Since we will define a supermanifold
$M$ as a functor
\[ M \maps \SuperPoints^\op \to \Man , \]
we first need to say how to think of the simplest kind of supermanifold, a
super vector space $V$, as such a functor:
\[ V \maps \SuperPoints^\op \to \Man . \]
But first we owe the reader a definition of the category of superpoints.

Recall from Section \ref{sec:superalgebra} that a \define{super vector
space} is a $\Z_2$-graded vector space $V = V_0 \oplus V_1$ where $V_0$ is
called the \define{even} part, and $V_1$ is called the \define{odd} part.
There is a symmetric monoidal category $\SuperVect$ which has:
\begin{itemize}
	\item $\Z_2$-graded vector spaces as objects;
	\item Grade-preserving linear maps as morphisms;
	\item A tensor product $\tensor$ that has the following grading: if $V
		= V_0 \oplus V_1$ and $W = W_0 \oplus W_1$, then $(V \tensor
		W)_0 = (V_0 \tensor W_0) \oplus (V_1 \tensor W_1)$ and 
              $(V \tensor W)_1 = (V_0 \tensor W_1) \oplus (V_1 \tensor W_0)$;
	\item A braiding
		\[ B_{V,W} \maps V \tensor W \to W \tensor V \]
		defined as follows: $v \in V$ and $w \in W$ are of grade $|v|$
		and $|w|$, then
		\[ B_{V,W}(v \tensor w) = (-1)^{|v||w|} w \tensor v. \]
\end{itemize}
The braiding encodes the `the rule of signs': in any calculation, when two odd
elements are interchanged, we introduce a minus sign. We write $\R^{p|q}$ for
the super vector space with even part $\R^p$ and odd part $\R^q$.

We define a \define{supercommutative superalgebra} to be a commutative algebra $A$
in the category $\SuperVect$.  More concretely, it is a real, associative
algebra $A$ with unit which is $\Z_2$-graded:
\[ A = A_0 \oplus A_1, \]
and is graded-commutative. That is:
\[ ab = (-1)^{|a||b|} ba, \]
for all homogeneous elements $a, b \in A$, as required by the rule of signs. We
define a \define{homomorphism of superalgebras} $f \maps A \to B$ to be an
algebra homomorphism that respects the grading.  So, there is a category
$\SuperAlg$ with supercommutative superalgebras as objects, and homomorphisms
of superalgebras as morphisms. Henceforth, we will assume all our superalgebras
to be supercommutative unless otherwise indicated.

A particularly important example of a supercommutative superalgebra is a
\define{Grassmann algebra}: a finite-dimensional exterior algebra
\[ A = \Lambda \R^n, \]
equipped with the grading:
\[ A_0 = \Lambda^0 \R^n \oplus \Lambda^2 \R^n \oplus \cdots, \quad A_1 = \Lambda^1 \R^n \oplus \Lambda^3 \R^n \oplus \cdots . \]
Let us write $\GrAlg$ for the category with Grassmann algebras as objects and
homomorphisms of superalgebras as morphisms. 

In fact, the Grassmann algebras are essential for our approach to supermanifold
theory, because:
\[ \GrAlg = \SuperPoints^\op \]
so rather than thinking of a supermanifold $M$ as a contravariant functor from
$\SuperPoints$ to $\Man$, we can view a supermanifold as a covariant functor:
\[ M \maps \GrAlg \to \Man \]
To see why this is sensible, recall that a smooth map between ordinary manifolds
\[ \varphi \maps M \to N \]
is the same as a homomorphism between their algebras of smooth functions which
goes the other way:
\[ \varphi^* \maps C^\infty(N) \to C^\infty(M) \]
By analogy, we expect something similar to hold for supermanifolds. In particular, 
a smooth map from a superpoint:
\[ \varphi \maps \R^{0|q} \to M \]
ought to be to the same as a homomorphism of their `superalgebras of smooth
functions' which points the other way:
\[ \varphi^* \maps C^{\infty}(M) \to C^{\infty}(\R^{0|q}) . \]
But since $\R^{0|q}$ is a purely odd super vector space, we \emph{define} its
algebra of smooth functions to be $\Lambda(\R^q)^*$.  Intuitively, this is
because $\R^{0|q}$ is a supermanifold with $q$ `odd, anticommuting
coordinates', given by the standard projections:
\[ \theta^1, \dots, \theta^q \maps \R^q \to \R , \] 
so a `smooth function' $f$ on $\R^{0|q}$ should have a `power series expansion'
that looks like:
\[ f = \sum_{i_1 < i_2 < \cdots < i_k} f_{i_1 i_2 \dots i_k} \theta^{i_1} \wedge \theta^{i_2} \wedge \cdots \wedge \theta^{i_k} . \]
where the coefficients $f_{i_1 i_2 \dots i_k}$ are real. Thus $f$ is precisely
an element of $\Lambda(\R^q)^*$. Thus, we \emph{define}
\[ \Hom(\R^{0|q}, M) = \Hom(C^\infty(M), \Lambda(\R^q)^*) . \]
In this way, rather than thinking of $M$ as a functor:
\[ \begin{array}{cccc} 
	\Hom(-,M) \maps & \SuperPoints^\op & \to     & \Man \\
		        & \R^{0|q}         & \mapsto & \Hom(\R^{0|q}, M) 
\end{array}
\]
where $\Hom$ is in the category of supermanifolds (though we have not defined
this), we think of $M$ as a functor:
\[ \begin{array}{cccc} 
	\Hom(C^{\infty}(M),-) \maps & \GrAlg  & \to     & \Man \\
	                            & A & \mapsto & \Hom(C^\infty(M),A)
\end{array}
\]
where $\Hom$ is in the category of superalgebras (which we have defined, though
we have not defined $C^\infty(M)$).

Since we have just given a slew of definitions, let us bring the discussion
back down to earth with a concise summary:
\begin{itemize}
	\item Every supermanifold is a functor: 
		\[ M \maps \GrAlg \to \Man , \]
		though not every such functor is a supermanifold.
	\item Every smooth map of supermanifolds is a natural transformation: 
	       \[ \varphi \maps M \to N , \]
	      though not every such natural transformation is a smooth map of
	      supermanifolds. 
\end{itemize}
Next, let us introduce some concise notation:
\begin{itemize}
	\item Let us write $M_A$ for the value of $M$ on the Grassmann
	       algebra $A$, and call this the \define{$A$-points}
	       of $M$.
	\item Let us write $M_f \maps M_A \to M_{B}$ for the smooth
	       map induced by a homomorphism $f \maps A \to B$.
	\item Finally, we write $\varphi_A \maps M_A \to N_A$ for the smooth
		map which the natural transformation $\varphi$ gives between
		the $A$-points. We call $\varphi_A$ a \define{component} of the
		natural transformation $\varphi$.
\end{itemize}

With this background, we can now build up the theory of supermanifolds in
perfect analogy to the theory of manifolds. First, we need to say how to think
of our model spaces, the super vector spaces, as supermanifolds. 

Indeed, given a finite-dimensional super vector space $V$, the \define{
supermanifold associated to $V$}, or just the \define{supermanifold $V$} to be
the functor:
\[ V \maps \GrAlg \to \Man \]
which takes:
\begin{itemize}
	\item each Grassmann algebra $A$ to the vector space:
		\[ V_A = (A \tensor V)_0 = A_0 \tensor V_0 \, \oplus \, A_1 \tensor V_1 \]
		regarded as a manifold in the usual way;
	\item each homomorphism $f \maps A \to B$ of Grassmann algebras to the
		linear map $V_f \maps V_A \to V_B$ that is the identity on
		$V$ and $f$ on $A$:
		\[ V_f = (f \tensor 1)_0 \maps (A \tensor V)_0 \to (B \tensor V)_0 . \]
		This map, being linear, is also smooth.
\end{itemize}
We take this definition because, roughly speaking, the set of $A$-points is
the set of homomorphisms of superalgebras, $\Hom(C^\infty(V),
A)$. By analogy with the ordinary manifold case, we expect that any such
homomorphism is determined by its restriction to the `dense subalgebra' of
polynomials: 
\[ \Hom(C^\infty(V), A) \iso \Hom(\Sym(V^*), A) , \]
though here we are being very rough, because we have not assumed any topology
on our superalgebras, so the term `dense subalgebra' is not meaningful.
Since $\Sym(V^*)$ is the free supercommutative superalgebra on $V^*$, a
homomorphism out of it is the same as a linear map of super vector spaces:
\[ \Hom(\Sym(V^*), A) \iso \Hom(V^*, A) , \]
where the first $\Hom$ is in $\SuperAlg$ and the second $\Hom$ is in
$\SuperVect$.  Finally, because $V$ is finite-dimensional and linear maps of
super vector spaces preserve grading, this last $\Hom$ is just:
\[ \Hom(V^*, A) \iso V_0 \tensor A_0 \, \oplus \, V_1 \tensor A_1 . \]
which, up to a change of order in the factors, is how we defined $V_A$. This
last set is a manifold in an obvious way: it is an ordinary,
finite-dimensional, real vector space. In fact, it is just the even part of the
super vector space $A \tensor V$:
\[ V_{A} = (A \tensor V)_0 , \]
as we have noted in our definition. 

In fact, $V_A = A_0 \tensor V_0 \, \oplus \, A_1 \tensor V_1$ is more than a
mere vector space---it is an $A_0$-module. Moreover, given any linear map of
super vector spaces:
\[ L \maps V \to W \]
we get an $A_0$-module map between the $A$-points in a natural way:
\[ L_A = (1 \tensor L)_0 \maps (A \tensor V)_0 \to (A \tensor W)_0 . \]
Indeed, $L$ induces a natural transformation between the supermanifold $V$ and
the supermanifold $W$. That is, given any homomorphism $f \maps A \to B$ of
Grassmann algebras, the following square commutes:
\[ \xymatrix{   V_A \ar[r]^{L_A} \ar[d]_{V_f} & W_A \ar[d]^{W_f} \\
		V_B \ar[r]_{L_B}              & W_B } 
\]
We therefore have a functor 
\[ \SuperVect \to \Fun(\GrAlg, \Man) \]
which takes super vector spaces to their associated supermanifolds, and linear
transformations to natural transformations between supermanifolds. For future
reference, we note this fact in a proposition:
\begin{prop} \label{prop:supervectormanifold}
	There is a faithful functor:
	\[ \SuperVect \to \Fun(\GrAlg, \Man) \]
	that takes a super vector space $V$ to the supermanifold $V$ whose
	$A$-points are:
	\[ V_A = (A \tensor V)_0, \]
	and takes a linear map of super vector spaces:
	\[ L \maps V \to W \]
	to the natural transformation whose components are:
	\[ L_A = (1 \tensor L)_0 \maps (A \tensor V)_0 \to (A \tensor W)_0 . \]
	In the above, $A$ is a Grassmann algebra and the tensor product takes
	place in $\SuperVect$. 
\end{prop}
\begin{proof}
	It is easy to check that this defines a functor. Faithfulness follows
	from a more general result in Sachse \cite{Sachse}, c.f.\ Proposition
	3.1.
\end{proof}
\noindent While this functor is faithful, it is far from full; in particular,
it misses all of the `smooth maps' between super vector spaces which do not
come from a linear map. We define these additional maps now.

Infinitesimally, all smooth maps should be like a linear map $L \maps V \to
W$, so given two finite-d imensional super vector spaces $V$ and $W$,
we define a \define{smooth map between super vector spaces}:
\[ \varphi \maps V \to W , \]
to be a natural transformation between the supermanifolds $V$ and $W$ such that
the derivative
\[ (\varphi_A)_* \maps T_x V_A \to T_{\varphi_A(x)} W_A \]
is $A_0$-linear at each $A$-point $x \in V_A$, where the $A_0$-module structure
on each tangent space comes from the canonical identification of a vector space
with its tangent space:
\[ T_x V_A \iso V_A, \quad T_{\varphi(x)} W_A \iso W_A . \]
Note that each component $\varphi_A \maps V_A \to W_A$ is smooth in the
ordinary sense, because by virtue of living in the category of smooth
manifolds. We say that a smooth map $\varphi_A \maps V_A \to W_A$ whose
derivative is $A_0$-linear at each point is \define{$A_0$-smooth} for short.

Finally, note that there is a supermanifold:
\[ 1 \maps \GrAlg \to \Man , \]
which takes each Grassmann algebra to the one-point manifold. We call this the
\define{one-point supermanifold}, and note that it is the supermanifold
associated to the super vector space $\R^{0|0}$. Just as $\R^{0|0}$ is the
terminal object in $\SuperVect$, 1 is the terminal object in
$\Fun(\GrAlg,\Man)$, as well as the smaller category of supermanifolds, whose
definition we now describe.  

\subsection{Supermanifolds in general}

The last section treated the special kind of a supermanifold of greatest
interest to us: the supermanifold associated to a super vector space.

Nonetheless, we now sketch how to define a general supermanifold, $M$. Since
$M$ will be locally isomorphic to a super vector space $V$, it helps to have
local pieces of $V$ to play the same role that open subsets of $\R^n$ play for
ordinary manifolds.  So, fix a super vector space $V$, and let $U \subseteq
V_0$ be open. The \define{superdomain} over $U$ is the functor: 
\[ \mathcal{U} \maps \GrAlg \to \Man \]
that takes each Grassmann algebra $A$ to
\[ \mathcal{U}_A = V_{\epsilon_A}^{-1}(U) \]
where $\epsilon_A \maps A \to \R$ the projection of the Grassmann algebra
$A$ that kills all nilpotent elements. We say that $\mathcal{U}$ is a
\define{superdomain in $V$}, and write $\mathcal{U} \subseteq V$.

If $\mathcal{U} \subseteq V$ and $\mathcal{U}' \subseteq W$ are two
superdomains in super vector spaces $V$ and $W$, a \define{smooth map of
superdomains} is a natural transformation:
\[ \varphi \maps \mathcal{U} \to \mathcal{U}' \]
such that for each Grassmann algebra $A$, the component on
$A$-points is smooth:
\[ \varphi_{A} \maps \mathcal{U}_A \to {\mathcal{U}'}_A . \]
and the derivative:
\[ (\varphi_A)_* \maps T_x \mathcal{U}_A \to T_{\varphi_A(x)} {\mathcal{U}'}_A \]
is $A_0$-linear at each $A$-point $x \in \mathcal{U}_A$, where the  
$A_0$-module structure on each tangent space comes from the canonical
identification with the ambient vector spaces:
\[ T_x \mathcal{U}_A \iso V_A, \quad T_{\varphi(x)} \mathcal{U}'_A \iso W_A . \]
Again, we say that a smooth map $\varphi_{A} \maps \mathcal{U}_A \to
{\mathcal{U}'}_A$ whose derivative is $A_0$-linear at each point is
\define{$A_0$-smooth} for short.

At long last, a \define{supermanifold} is a functor
\[ M \maps \GrAlg \to \Man \]
equipped with an atlas 
\[ (\mathcal{U}_\alpha, \varphi_\alpha \maps \mathcal{U} \to M) , \]
where each $\mathcal{U}_\alpha$ is a superdomain, each $\varphi_\alpha$ is a
natural transformation, and one can define transition functions that are smooth
maps of superdomains. 

Finally, a \define{smooth map of supermanifolds} is a natural
transformation:
\[ \psi \maps M \to N \]
which induces smooth maps between the superdomains in the atlases.
Equivalently, each component 
\[ \varphi_A \maps M_A \to N_A \] 
is \define{$A_0$-smooth}: it is smooth and its derivative
\[ (\varphi_A)_* \maps T_x M_A \to T_{\varphi_A(x)} N_A \]
is $A_0$-linear at each $A$-point $x \in M_A$, where the $A_0$-module structure
on each tangent space comes from the superdomains in the atlases. Thus, there
is a category $\SuperMan$ of supermanifolds. See Sachse \cite{Sachse} for more
details. 

\section{Supergroups from nilpotent Lie superalgebras} \label{sec:supergroups}

We now describe a procedure to integrate a nilpotent Lie superalgebra to a Lie
supergroup. This is a partial generalization of Lie's Third Theorem, which
describes how any Lie algebra can be integrated to a Lie group. In fact, the
full theorem generalizes to Lie supergroups \cite{Tuynman}, but we do not need
it here.

Recall from Section \ref{sec:cohomology} that a \define{Lie superalgebra} $\g$
is a Lie algebra in the category of super vector spaces. More concretely, it is
a super vector space $\g = \g_0 \oplus \g_1$, equipped with a
graded-antisymmetric bracket:
\[ [-,-] \maps \Lambda^2 \g \to \g , \]
which satisfies the Jacobi identity up to signs:
\[ [X, [Y,Z]] = [ [X,Y], Z] + (-1)^{|X||Y|} [Y, [X, Z]]. \]
for all homogeneous $X, Y, Z \in \g$.  A Lie superalgebra $\n$ is called
\define{k-step nilpotent} is any $k$ nested brackets vanish, and it is called
\define{nilpotent} if it is $k$-step nilpotent for some $k$. Nilpotent Lie
superalgebras can be integrated to a unique supergroup $N$ defined on the same
underlying super vector space $\n$.

A \define{Lie supergroup}, or \define{supergroup}, is a group object in the
category of supermanifolds. That is, it is a supermanifold $G$ equipped with
the following maps of supermanifolds:
\begin{itemize}
	\item \define{multiplication}, $m \maps G \times G \to G$;
	\item \define{inverse}, $\inv \maps G \to G$;
	\item \define{identity}, $\id \maps 1 \to G$, where $1$ is the one-point
		supermanifold;
\end{itemize}
such that the following diagrams commute, encoding the usual group axioms:
\begin{itemize}
	\item the associative law:
	\[ \vcenter
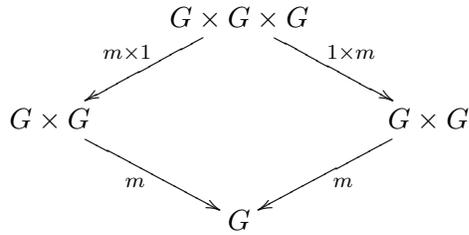
{
	\xymatrix{ &   G \times G \times G \ar[dr]^{1 \times m}
	   \ar[dl]_{m \times 1} \\
	 G \times G \ar[dr]_{m}
	&&  G \times G \ar[dl]^{m}  \\
	&  G }}
	\]
	\item the right and left unit laws:
	\[ \vcenter
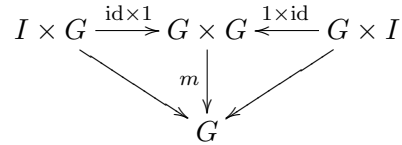
{
	\xymatrix{
	 I \times G \ar[r]^{\id \times 1} \ar[dr]
	& G \times G \ar[d]_{m}
	& G \times I \ar[l]_{1 \times \id} \ar[dl] \\
	& G }}
	\]
	\item the right and left inverse laws:
	
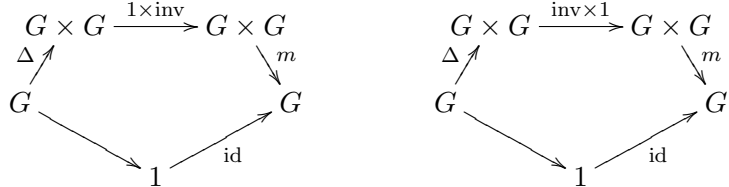
\[
	\xy (-12,10)*+{G \times G}="TL"; (12,10)*+{G \times G}="TR";
	(-18,0)*+{G}="ML"; (18,0)*+{G}="MR"; (0,-10)*+{1}="B";
	     {\ar_{} "ML";"B"};
	     {\ar^{\Delta} "ML";"TL"};
	     {\ar_{\id} "B";"MR"};
	     {\ar^{m} "TR";"MR"};
	     {\ar^{1 \times \inv } "TL";"TR"};
	\endxy
	\qquad \qquad \xy (-12,10)*+{G \times G}="TL"; (12,10)*+{G \times
	G}="TR"; (-18,0)*+{G}="ML"; (18,0)*+{G}="MR"; (0,-10)*+{1}="B";
	     {\ar_{} "ML";"B"};
	     {\ar^{\Delta} "ML";"TL"};
	     {\ar_{\id} "B";"MR"};
	     {\ar^{m} "TR";"MR"};
	     {\ar^{\inv \times 1} "TL";"TR"};
	\endxy
	\]
\end{itemize}
where $\Delta \maps G \to G \times G$ is the diagonal map. In addition, a
supergroup is \define{abelian} if the following diagram commutes:
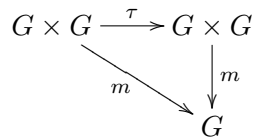
\[ \xymatrix{ 
G \times G \ar[r]^\tau \ar[dr]_m & G \times G \ar[d]^m \\
                                 & G
}
\]
where $\tau \maps G \times G \to G \times G$ is the \define{twist map}. Using
$A$-points, it is defined to be:
\[ \tau_A(x,y) = (y,x), \]
for $(x,y) \in G_A \times G_A$.

Examples of supergroups arise easily from Lie groups: if $G$ is an
Lie group, it is also a Lie group defined on the supermanifold whose
$A$-points are:
\[ G_A = \Hom(C^\infty(G), A) , \]
where $C^\infty(G)$ is the ordinary algebra of smooth functions on $G$,
regarded as a purely even superalgebra.  In this way, any classical Lie group,
such as $\SO(n)$, $\SU(n)$ and $\Sp(n)$, becomes a supergroup.

To obtain more interesting examples, we will integrate a nilpotent Lie
superalgebra, $\n$ to a supergroup $N$. For any Grassmann algebra $A$, the bracket 
\[ [-,-] \maps \Lambda^2 \n \to \n \]
induces an $A_0$-linear map between the $A$-points:
\[ [-,-]_A \maps \Lambda^2 \n_A \to \n_A, \]
where $\Lambda^2 \n_A$ denotes the exterior square of the $A_0$-module $\n_A$.
Thus $[-,-]_A$ is antisymmetric, and it easy to check that it makes $\n_A$ into
a Lie algebra which is also nilpotent.

On each such $A_0$-module $\n_A$, we can thus define a Lie group $N_A$ where
the multplication is given by the Baker--Campbell--Hausdorff formula, inversion
by negation, and the identity is $0$. Because we want to write the group $N_A$
multiplicatively, we write $\exp_A \maps n_A \to N_A$ for the identity map, and
then define the multiplication, inverse and identity maps:
\[ m_A \maps N_A \times N_A \to N_A, \quad \inv_A \maps N_A \to N_A, \quad \id_A \maps 1_A \to N_A, \]
as follows:
\[ m_A(\exp_A(X), \exp_A(Y)) = \exp_A(X) \exp_A(Y) = \exp_A(X + Y + \half[X,Y]_A + \cdots ) \]
\[ \inv_A(\exp_A(X)) = \exp_A(-X) = \exp_A(X)^{-1}, \]
\[ \id_A(1) = \exp_A(0) = 1, \]
for any $A$-points $X, Y \in \n_A$, and the first 1 in the last equation refers to
the single element of $1_A$. But it is clear that all of these maps are natural in
$A$. Furthermore, they are all $A_0$-smooth, because as polynomials with
coefficients in $A_0$, they are smooth with derivatives that are $A_0$-linear.
They thus define smooth maps of supermanifolds:
\[ m \maps N \times N \to N, \quad \inv \maps N \to N, \quad \id \maps 1 \to N, \]
where $N$ is the supermanifold $\n$. And because each of the $N_A$ is a group,
$N$ is a supergroup. We have thus proved:

\begin{prop} \label{prop:nilpotentsupergroup}
	Let $\n$ be a nilpotent Lie superalgebra. Then there is a supergroup
	$N$ defined on the supermanifold $\n$, obtained by integrating the
	nilpotent Lie algebra $\n_A$ with the Baker--Campbell--Hausdorff
	formula for all Grassmann algebras $A$. More precisely, we define the maps:
	\[ m \maps N \times N \to N, \quad \inv \maps N \to N, \quad \id \maps 1 \to N, \]
	by defining them on $A$-points as follows:
	\[ m_A(\exp_A(X), \exp_A(Y)) = \exp_A(Z(X,Y)), \]
	\[ \inv_A(\exp_A(X)) = \exp_A(-X), \]
	\[ \id_A(1) = \exp_A(0), \]
	where 
	\[ \exp \maps \n \to N \]
	is the identity map of supermanifolds, and:
	\[ Z(X,Y) = X + Y + \half[X,Y]_A + \cdots \]
	denotes the Baker--Campbell--Hausdorff series on $\n_A$, which
	terminates because $\n_A$ is nilpotent.
\end{prop}
\noindent
Experience with ordinary Lie theory suggests that, in general, there will be
more than one supergroup which has Lie superalgebra $\n$. To distinguish the
one above, we call $N$ the \define{exponential supergroup} of $\n$. 

\chapter{Lie \emph{n}-supergroups from supergroup cohomology} \label{ch:Lie-n-supergroups}

We saw in Chapter \ref{ch:Lie-n-groups} that 3-cocycles in Lie group cohomology
allow us to construct Lie 2-groups. We now generalize this to supergroups. The
most significant barrier is that we now work internally to the category of
supermanifolds instead of the much more familiar category of smooth manifolds.
Our task is to show that this change of categories does not present a problem.
The main obstacle is that the category of supermanifolds is not a concrete
category: morphisms are determined not by their value on the underlying set of
a supermanifold, but by their value on $A$-points for all Grassmann algebras
$A$.

The most common approach is to define morphisms without reference to elements,
and to define equations between morphisms using commutative diagrams. This is
how we gave the definition of smooth bicategory, except that we found it
convenient to state the pentagon and triangle identities using elements.  As an
alternative to commutative diagrams, for supermanifolds, one can use $A$-points
to define morphisms and specify equations between them. This tends to make
equations look friendlier, because they look like equations between functions.
We shall use this approach.

First, let us define the cohomology of a supergroup $G$ with coefficients in an
abelian supergroup $H$, on which $G$ \define{acts by automorphism}. This means
that we have a morphism of supermanifolds:
\[ \alpha \maps G \times H \to H, \]
which, for any Grassmann algebra $A$, induces an action of the group
$G_A$ on the abelian group $H_A$:
\[ \alpha_A \maps G_A \times H_A \to H_A. \]
For this action to be by automorphism, we require:
\[ \alpha_A(g)(h + h') = \alpha_A(g)(h) + \alpha_A(g)(h'), \]
for all $A$-points $g \in G_A$ and $h, h' \in H_A$.

We define supergroup cohomology using \define{the supergroup cochain complex},
$C^\bullet(G,H)$, which at level $p$ just consists of the set of maps
from $G^p$ to $H$ as supermanifolds:
\[ C^p(G,H) = \left\{ f \maps G^p \to H \right\} . \]
Addition on $H$ makes $C^p(G,H)$ into an abelian group for all $p$.  The
differential is given by the usual formula, but using $A$-points:
\begin{eqnarray*} 
	df_A(g_1, \dots, g_{p+1}) & = & g_1 f_A(g_2, \dots, g_{p+1}) \\
	                          &   & + \sum_{i=1}^p (-1)^i f_A(g_1, \dots, g_i g_{i+1}, \dots, g_{p+1}) \\
				  &   & + (-1)^{p+1} f_A(g_1, \dots, g_p) , \\
\end{eqnarray*}
where $g_1, \dots, g_{p+1} \in G_A$ and the action of $g_1$ is given by
$\alpha_A$. Noting that $f_A$, $\alpha_A$, multiplication and $+$ are all:
\begin{itemize}
	\item natural in $A$;
	\item $A_0$-smooth: smooth with derivatives which are $A_0$-linear;
\end{itemize}
we see that $df_A$ is:
\begin{itemize}
	\item natural in $A$; 
	\item $A_0$-smooth: smooth with a derivative which is $A_0$-linear;
\end{itemize}
so it indeed defines a map of supermanifolds:
\[ df \maps G^{p+1} \to H. \]
Furthermore, it is immediate that:
\[ d^2 f_A = 0 \]
for all $A$, and thus
\[ d^2 f = 0. \]
So $C^\bullet(G,H)$ is truly a cochain complex. Its cohomology $H^\bullet(G,H)$
is the \define{supergroup cohomology of $G$ with coefficients in $H$}. Of
course, if $df = 0$, $f$ is called a \define{cocycle}, and $f$ is \define{normalized} if
\[ f_A(g_1, \dots, g_p) = 0 \]
for any Grassmann algebra $A$, whenever one of the $A$-points $g_1, \dots, g_p$
is 1. When $H = \R$, we omit reference to it, and write $C^\bullet(G,\R)$ as
$C^\bullet(G)$.

A \define{super bicategory} $B$ has 
\begin{itemize}
	\item a \define{supermanifold of objects} $B_0$;
	\item a \define{supermanifold of morphisms} $B_1$;
	\item a \define{supermanifold of 2-morphisms} $B_2$;
\end{itemize}
equipped with maps of supermanifolds as described in Definition
\ref{def:smoothbicat}: source, target, identity-assigning, horizontal
composition, vertical composition, associator and left and right unitors all
maps of supermanifolds, and satisfying the same axioms as a smooth bicategory.
The associator satisfies the pentagon identity, which we state in terms of
$A$-points: the following pentagon commutes:
\[
\xy
 (0,20)*+{(f g) (h k)}="1";
 (40,0)*+{f (g (h k))}="2";
 (25,-20)*{ \quad f ((g h) k)}="3";
 (-25,-20)*+{(f (g h)) k}="4";
 (-40,0)*+{((f g) h) k}="5";
 {\ar@{=>}^{a(f,g,h k)}     "1";"2"}
 {\ar@{=>}_{1_f \cdot a_(g,h,k)}  "3";"2"}
 {\ar@{=>}^{a(f,g h,k)}    "4";"3"}
 {\ar@{=>}_{a(f,g,h) \cdot 1_k}  "5";"4"}
 {\ar@{=>}^{a(fg,h,k)}    "5";"1"}
\endxy
\]
for any `composable quadruple of morphisms':
\[ (f,g,h,k) \in (B_1 \times_{B_0} B_1 \times_{B_0} B_1 \times_{B_0} B_1)_A . \]
Similarly, the associator and left and right unitors satisfy the triangle
identity, which we state in terms of $A$-points: the following triangle
commutes:
\[ 
\xy
(-20,10)*+{(f 1) g}="1";
(20,10)*+{f (1 g)}="2";
(0,-10)*+{f g}="3";
{\ar@{=>}^{a(f,1,g)}	"1";"2"}
{\ar@{=>}_{r(f) \cdot 1_g}	"1";"3"}
{\ar@{=>}^{1_f \cdot l(g)} "2";"3"}
\endxy
\]
for any `composable pair of morphisms':
\[ (f,g) \in (B_1 \times_{B_0} B_1)_A . \]

A \define{2-supergroup} is a super bicategory with one object (more precisely,
the one-point supermanifold), and all morphisms and 2-morphisms weakly
invertible.  Given a normalized $H$-valued 3-cocycle $a$ on $G$, we can
construct a 2-supergroup $\String_a(G,H)$ in the same way we constructed
the Lie 2-group $\String_a(G,H)$ when $G$ and $H$ were Lie groups, by just
deleting every reference to elements of $G$ or $H$:
\begin{itemize}
	\item The supermanifold of objects is the one-point supermanifold, 1.
	\item The supermanifold of morphisms is the supergroup $G$, with
		composition given by the multiplication:
		\[ \cdot \maps G \times G \to G. \]
		The source and target maps are the unique maps to the one-point
		supermanifold. The identity-assigning map is the
		identity-assigning map for $G$:
		\[ \id \maps 1 \to G. \]
	\item The supermanifold of 2-morphisms is $G \times H$. The source and
		target maps are both the projection map to $G$. The identity
		assigning map comes from the identity-assigning map for $H$:
		\[ 1 \times \id \maps G \times 1 \to G \times H. \]
	\item Vertical composition of 2-morphisms is given by addition in $H$:
		\[ 1 \times + \maps G \times H \times H \to G \times H, \]
		where we have used the fact that the pullback of 2-morphisms
		over objects is trivially:
		\[ (G \times H) \times_1  (G \times H) \iso G \times H \times H. \]
		Horizontal composition, $\cdot$, given by the multiplication
		on the semidirect product:
		\[ \cdot \maps ( G \ltimes H ) \times ( G \ltimes H ) \to G \ltimes H. \]
	\item The left and right unitors are trivial.
	\item The associator is given by the 3-cocycle $a \maps G^3 \to H$, where
		the source (and target) is understood to come from
		multiplication on $G$.
\end{itemize}
A \define{slim 2-supergroup} is one of this form. It remains to check that it
is, indeed, a 2-supergroup.

\begin{prop} \label{prop:2supergroup}
	$\String_a(G,H)$ is a 2-supergroup: a super bicategory
	with one object and all morphisms and 2-morphisms weakly invertible.
\end{prop}
\begin{proof}
	This proof is a duplicate of the proof of Proposition
	\ref{prop:Lie2group}, but with $A$-points instead of elements.
\end{proof}

In a similar way, we can generalize our construction of Lie 3-groups to
`3-supergroups'.  A \define{super tricategory} $T$ has 
\begin{itemize}
	\item a \define{supermanifold of objects} $T_0$;
	\item a \define{supermanifold of morphisms} $T_1$;
	\item a \define{supermanifold of 2-morphisms} $T_2$;
	\item a \define{supermanifold of 3-morphisms} $T_3$;
\end{itemize}
equipped with maps of supermanifolds as described in Definition
\ref{def:smoothtricat}: source, target, identity-assigning, composition at
0-cells, 1-cells and 2-cells, associator and left and right unitors,
pentagonator and triangulators all maps of supermanifolds, and satisfying the
same axioms as a smooth tricategory. As in the case of the pentagon identity
above, we express the pentagonator identity in terms of $A$-points: the
following equation holds:

\newpage
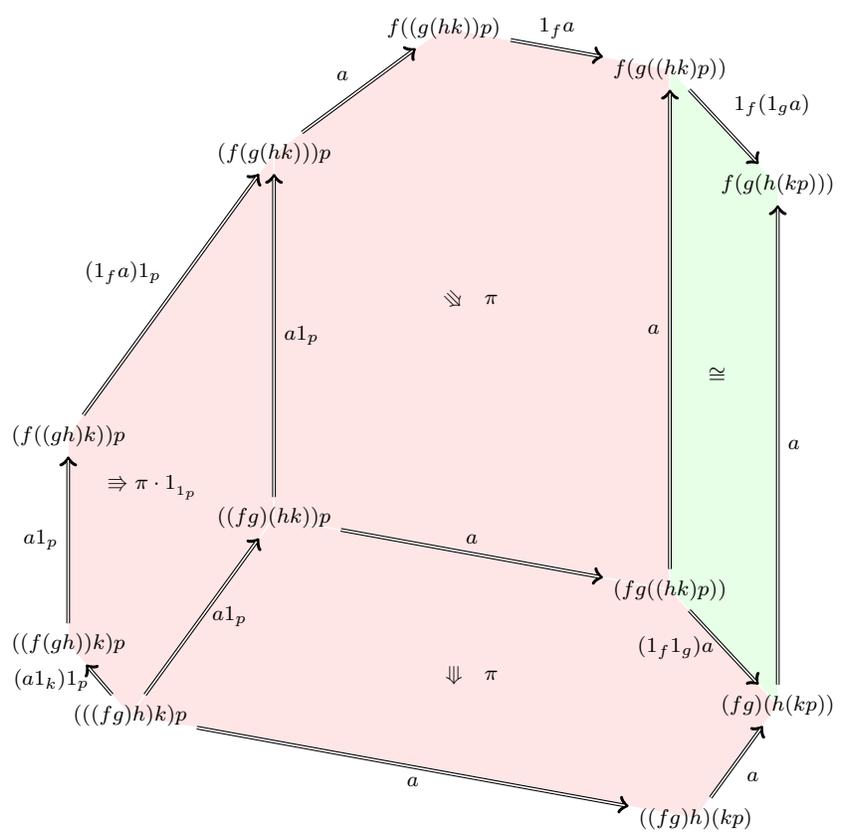
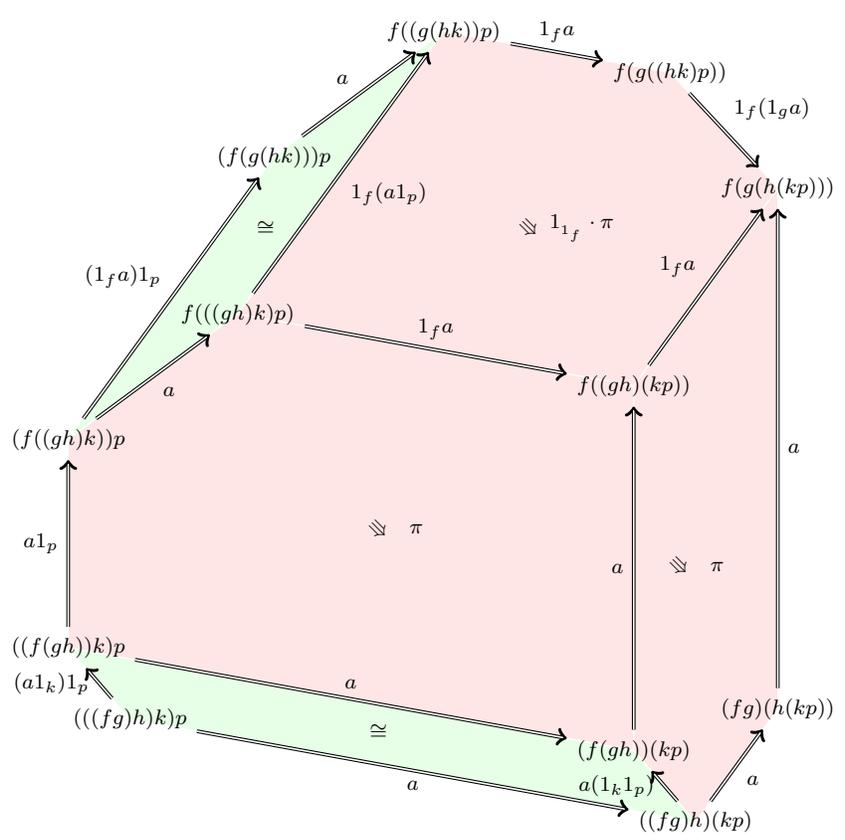
\thispagestyle{empty}

\begin{figure}[H]
    \begin{center}
      \begin{tikzpicture}[line join=round]
        \filldraw[white,fill=red,fill opacity=0.1](-4.306,-3.532)--(-2.391,-.901)--(-2.391,3.949)--(-5.127,.19)--(-5.127,-2.581)--cycle;
        \filldraw[white,fill=red,fill opacity=0.1](-4.306,-3.532)--(-2.391,-.901)--(2.872,-1.858)--(4.306,-3.396)--(3.212,-4.9)--cycle;
        \filldraw[white,fill=red,fill opacity=0.1](2.872,-1.858)--(2.872,5.07)--(-.135,5.617)--(-2.391,3.949)--(-2.391,-.901)--cycle;
        \filldraw[white,fill=green,fill opacity=0.1](4.306,-3.396)--(4.306,3.532)--(2.872,5.07)--(2.872,-1.858)--cycle;
        \begin{scope}[font=\fontsize{8}{8}\selectfont]
          \node (A) at (-2.391,3.949) {$(f(g(hk)))p$};
          \node (B) at (-5.127,.19) {$(f((gh)k))p$}
	  edge [->, double] node [l, above left] {$(1_{f}  a)1_{p}$} (A);
          \node (C) at (-5.127,-2.581) {$((f(gh))k)p$}
	  edge [->, double] node [l, left] {$a  1_{p}$} (B);
          \node (D) at (-4.306,-3.532) {$(((fg)h)k)p$}
	  edge [->, double] node [l, left] {$(a  1_{k})  1_{p}$} (C);
          \node (E) at (3.212,-4.9) {$((fg)h)(kp)$}
            edge [<-, double] node [l, below] {$a$} (D);
          \node (F) at (4.306,-3.396) {$(fg)(h(kp))$}
            edge [<-, double] node [l, below right] {$a$} (E);
          \node (G) at (4.306,3.532) {$f(g(h(kp)))$}
            edge [<-, double] node [l, right] {$a$} (F);
          \node (H) at (2.872,5.07) {$f(g((hk)p))$}
            edge [->, double] node [l, above right] {$1_{f}(1_{g}a)$} (G);
          \node (I) at (-.135,5.617) {$f((g(hk))p)$}
            edge [->, double] node [l, above] {$1_{f}a$} (H)
            edge [<-, double] node [l, above left] {$a$} (A);
          \node (M) at (-2.391,-.901) {$((fg)(hk))p$}
            edge [<-, double] node [l, right] {$a1_{p}$} (D)
            edge [->, double] node [l, right] {$a1_{p}$} (A);
          \node (N) at (2.872,-1.858) {$(fg((hk)p))$}
            edge [<-, double] node [l, above] {$a$} (M)
            edge [->, double] node [l, left] {$a$} (H)
            edge [->, double] node [l, left] {$(1_{f}1_{g})a$} (F);
	    \node at (-4,-.5) {$\Rrightarrow \pi \cdot 1_{_{1_{p}}}$};
          \node at (0,-3) {\tikz\node [rotate=-90] {$\Rrightarrow$};};
          \node at (0.5,-3) {$\pi$};
          \node at (0,2) {\tikz\node [rotate=-45] {$\Rrightarrow$};};
          \node at (0.5,2) {$\pi$};
          \node at (3.5,1) {$\cong$};
        \end{scope}
      \end{tikzpicture}
      \[ = \]
      \begin{tikzpicture}[line join=round]
        \filldraw[white,fill=green,fill opacity=0.1](-2.872,1.858)--(-.135,5.617)--(-2.391,3.949)--(-5.127,.19)--cycle;
        \filldraw[white,fill=red,fill opacity=0.1](3.212,-4.9)--(4.306,-3.396)--(4.306,3.532)--(2.391,.901)--(2.391,-3.949)--cycle;
        \filldraw[white,fill=green,fill opacity=0.1](-4.306,-3.532)--(3.212,-4.9)--(2.391,-3.949)--(-5.127,-2.581)--cycle;
        \filldraw[white,fill=red,fill opacity=0.1](-2.872,1.858)--(2.391,.901)--(4.306,3.532)--(2.872,5.07)--(-.135,5.617)--cycle;
        \filldraw[white,fill=red,fill opacity=0.1](-5.127,-2.581)--(-5.127,.19)--(-2.872,1.858)--(2.391,.901)--(2.391,-3.949)--cycle;
        \begin{scope}[font=\fontsize{8}{8}\selectfont]
          \node (A) at (-2.391,3.949) {$(f(g(h k)))p$};
          \node (B) at (-5.127,.19) {$(f((gh)k))p$}
            edge [->, double] node [l, above left] {$(1_{f}a)1_{p}$} (A);
          \node (C) at (-5.127,-2.581) {$((f(gh))k)p$}
            edge [->, double] node [l, left] {$a1_{p}$} (B);
          \node (D) at (-4.306,-3.532) {$(((fg)h)k)p$}
            edge [->, double] node [l, left] {$(a1_{k})1_{p}$} (C);
          \node (E) at (3.212,-4.9) {$((fg)h)(kp)$}
            edge [<-, double] node [l, below] {$a$} (D);
          \node (F) at (4.306,-3.396) {$(fg)(h(kp))$}
            edge [<-, double] node [l, below right] {$a$} (E);
          \node (G) at (4.306,3.532) {$f(g(h(kp)))$}
            edge [<-, double] node [l, right] {$a$} (F);
          \node (H) at (2.872,5.07) {$f(g((hk)p))$}
            edge [->, double] node [l, above right] {$1_{f}(1_{g}a)$} (G);
          \node (I) at (-.135,5.617) {$f((g(hk))p)$}
            edge [->, double] node [l, above] {$1_{f}a$} (H)
            edge [<-, double] node [l, above left] {$a$} (A);
          \node (J) at (-2.872,1.858) {$f(((gh)k)p)$}
            edge [->, double] node [l, below right] {$1_{f}(a1_{p})$} (I)
            edge [<-, double] node [l, below right] {$a$} (B);
          \node (K) at (2.391,-3.949) {$(f(gh))(kp)$}
            edge [<-, double] node [l, left] {$a(1_{k}1_{p})$} (E)
            edge [<-, double] node [l, above] {$a$} (C);
          \node (L) at (2.391,.901) {$f((gh)(kp))$}
            edge [<-, double] node [l, left] {$a$} (K)
            edge [<-, double] node [l, above] {$1_{f}a$} (J)
            edge [->, double] node [l, above left] {$1_{f}a$} (G);
          \node at (-1,-1) {\tikz\node [rotate=-45] {$\Rrightarrow$};};
	  \node at (-.5,-1) {$\pi$};
          \node at (1,3) {\tikz\node [rotate=-45] {$\Rrightarrow$};};
	  \node at (1.7,3) {$1_{_{1_{f}}} \cdot \pi$};
          \node at (3,-1.5) {\tikz\node [rotate=-45] {$\Rrightarrow$};};
          \node at (3.5,-1.5) {$\pi$};
          \node at (-1,-3.7) {$\cong$};
          \node at (-2.5,3) {$\cong$};
        \end{scope}
      \end{tikzpicture}
    \end{center}
\end{figure}
\clearpage

for any `composable quintet of morphisms':
\[ (f,g,h,k,p) \in (T_1 \times_{T_0} T_1 \times_{T_0} T_1 \times_{T_0} T_1 \times_{T_0} T_1)_A \]

A \define{3-supergroup} is a super tricategory with one object (more precisely,
the one-point supermanifold) and all morphisms, 2-morphisms and 3-morphisms
weakly invertible. Given a normalized $H$-valued 4-cocycle $\pi$ on $G$, we can
construct a 3-supergroup $\Brane_\pi(G,H)$ in the same way we constructed the
Lie 3-group $\Brane_\pi(G,H)$ when $G$ and $H$ were Lie groups, but deleting
every reference to elements of $G$ or $H$:
\begin{itemize}
	\item The supermanifold of objects is the one-point supermanifold, $1$.

	\item The supermanifold of morphisms is the supergroup $G$. 
		Composition at a 0-cell is given by multiplication in the group:
		\[ \cdot \maps G \times G \to G. \]
		The source and target maps are the unique maps to $1$. The
		identity-assigning map is the identity-assigning map for $G$:
		\[ \id \maps 1 \to G. \]

	\item The supermanifold of 2-morphisms is again $G$. The source, target
		and identity-assigning maps are all the identity on $G$.
		Composition at a 1-cell is the identity on $G$, while
		composition at a 0-cell is again multiplication in $G$.
		This encodes the idea that all 2-morphisms are trivial.

	\item The supermanifold of 3-morphisms is $G \times H$. The source and
		target maps are projection onto $G$. The identity-assigning map
		is the inclusion:
		\[ G \to G \times H \]
		that takes $A$-points $g \in G_A$ to $(g,0) \in G_A \times
		H_A$, for all $A$.

	\item Three kinds of composition of 3-morphisms: composition at a
		2-cell and at a 3-cell are both given by addition on $H$:
		\[ 1 \times + \maps G \times H \times H \to G \times H. \]
		While composition at a 0-cell is just given by multiplication on
		the semidirect product:
		\[ \cdot \maps (G \ltimes H) \times (G \ltimes H) \to G \ltimes H. \]

	\item The associator, left and right unitors are automatically
		trivial, because all 2-morphisms are trivial.

	\item The triangulators are trivial.

	\item The \define{2-associator} or \define{pentagonator} is given by the
		4-cocycle $\pi \maps G^4 \to H$, where the source (and target)
		is understood to come from multiplication on $G$.
\end{itemize}
A \define{slim 3-supergroup} is one of this form. It remains to check that it
is, indeed, a 3-supergroup.

\begin{prop} \label{prop:3supergroup}
	$\Brane_a(G,H)$ is a 3-supergroup: a super tricategory with one
	object and all morphisms, 2-morphisms and 3-morphisms weakly
	invertible.
\end{prop}
\begin{proof}
	This proof is a duplicate of the proof of Proposition
	\ref{prop:Lie3group}, but with $A$-points instead of elements.
\end{proof}

\chapter{Integrating nilpotent Lie \emph{n}-superalgebras} \label{ch:integrating2}

We now generalize our technique for integrating cocycles from nilpotent Lie
algebras to nilpotent Lie \emph{superalgebras}. Those familiar with
supermanifold theory may find it surprising that this is possible---the theory
of differential forms is very different for supermanifolds than for manifolds,
and integrating differential forms on a manifold was crucial to our method in
Section \ref{sec:integratingcochains}. But we can sidestep this issue on a
supergroup $N$ by considering $A$-points for any Grassmann algebra $A$. Then
$N_A$ is a manifold, so the usual theory of differential forms applies.

Here is how we will proceed. Fixing a nilpotent Lie superalgebra $\n$ with
exponential supergroup $N$, we can use Proposition \ref{prop:supervectormanifold}
turn any Lie superalgebra cochain $\omega$ on $\n$ into a Lie algebra cochain $\omega_A$
on $\n_A$. We then use the techniques in Section \ref{sec:integratingcochains} to turn
$\omega_A$ into a Lie group cochain $\smallint \omega_A$ on $N_A$. Checking
that $\smallint \omega_A$ is natural in $A$ and $A_0$-smooth, this defines a
supergroup cochain $\smallint \omega$ on $N$.

As we saw in Proposition \ref{prop:supervectormanifold}, any map of super vector
spaces becomes an $A_0$-linear map on $A$-points. We have already touched on
the way this interacts with symmetry: for a Lie superalgebra $\g$, the
graded-antisymmetric bracket
\[ [-,-] \maps \Lambda^2 \g \to \g \]
becomes an honest antisymmetric bracket on $A$-points:
\[ [-,-]_A \maps \Lambda^2 \g_A \to \g_A. \]
More generally, we have:
\begin{lem}
	Graded-symmetric maps of super vector spaces:
	\[ f \maps \Sym^p V \to W \]
	induce symmetric maps on $A$-points:
	\[ f_A \maps \Sym^p V_A \to W_A, \]
	defined by:
	\[ f_A(a_1 v_1, \dots, a_p v_p) = a_p \cdots a_1 f(v_1, \dots, v_p), \]
	where $\Sym^p V_A$ is the symmetric power of $V_A$ as an $A_0$-module
	and $a_i \in A$, $v_i \in V$ are of matching parity.
	Similarly, graded-antisymmetric maps of super vector spaces:
	\[ f \maps \Lambda^p V \to W \]
	induce antisymmetric maps on $A$-points:
	\[ f_A \maps \Lambda^p V_A \to W_A, \]
	defined by:
	\[ f_A(a_1 v_1, \dots, a_p v_p) = a_p \cdots a_1 f(v_1, \dots, v_p), \]
	where $\Lambda^p V_A$ is the exterior power $V_A$ as an $A_0$-module
	and $a_i \in A$, $v_i \in V$ are of matching parity.
\end{lem}
\begin{proof}
	This is straightforward and we leave it to the reader.
\end{proof}

Next, we need to show that Lie superalgebra cochains $\omega$ on $\n$ give rise
to Lie algebra cochains $\omega_A$ on the $A$-points $\n_A$. In fact, this
works for any Lie superalgebra, but there is one twist: because $\n_A$ is an
$A_0$-module, $\omega \maps \Lambda^p \n \to \R$ gives rise to an $A_0$-linear
map:
\[ \omega_A \maps \Lambda^p \n_A \to A_0, \]
using the fact that $\R_A = A_0$. So, we need to say how to do Lie algebra
cohomology with coefficients in $A_0$. It is just a straightforward
generalization of cohomology with coefficients in $\R$.

Indeed, any Lie superalgebra $\g$ induces a Lie algebra structure on $\g_A$
where the bracket is $A_0$-bilinear. We say that $\g_A$ is an \define{$A_0$-Lie
algebra}. Given any $A_0$-Lie algebra $\g_A$, we define its cohomology with the
\define{$A_0$-Lie algebra cochain complex}, which at level $p$ consists of
antisymmetric $A_0$-multilinear maps:
\[ C^p(\g_A) = \left\{ \omega \maps \Lambda^p \g_A \to A_0 \right\}. \]
We define $d$ on this complex in exactly the same way we define $d$ for
$\R$-valued Lie algebra cochains. This makes $C^\bullet(\g_A)$ into a cochain
complex, and the \define{cohomology of an $A_0$-Lie algebra with coefficients
in $A_0$} is the cohomology of this complex.

\begin{prop}
	Let $\g$ be a Lie superalgebra, and let $\g_A$ be the $A_0$-Lie algebra
	of its $A$-points. Then there is a cochain map:
	\[ C^\bullet(\g) \to C^\bullet(\g_A) \]
	given by taking the $p$-cochain $\omega$
	\[ \omega \maps \Lambda^p \g \to \R \]
	to the induced $A_0$-linear map $\omega_A$:
	\[ \omega_A \maps \Lambda^p \g_A \to A_0, \]
	where $\Lambda^p \g_A$ denotes the $p$th exterior power of $\g_A$ as
	an $A_0$-module.
\end{prop}
\begin{proof}
	We need to show:
	\[ d(\omega_A) = (d\omega)_A. \]
	Since these are both linear maps on $\Lambda^{p+1}(\g_A)$, it suffices
	to check that they agree on generators, which are of the form:
	\[ a_1 X_1 \wedge a_2 X_2 \wedge \cdots \wedge a_{p+1} X_{p+1} \]
	for $a_i \in A$ and $X_i \in \g$ of matching parity. By definition:
	\[ (d\omega)_A(a_1 X_1 \wedge a_2 X_2 \wedge \cdots \wedge a_{p+1} X_{p+1}) = a_{p+1} a_p \cdots a_{1} d\omega(X_1 \wedge X_2 \wedge \cdots \wedge X_{p+1}). \]

	On the other hand, to compute $d(\omega_A)$, we need to apply the
	formula for $d$ to obtain the intimidating expression:
	\begin{eqnarray*}
		&   & d(\omega_A)(a_1 X_1, \dots, a_{p+1} X_{p+1}) \\
		& = & \sum_{i < j} (-1)^{i+j} \omega_A([a_i X_i, a_j X_j]_A, a_1 X_1, \dots, \widehat{a_i X_i}, \dots, \widehat{a_j X_j}, \dots, a_{p+1} X_{p+1}) \\
		& = & \sum_{i < j} (-1)^{i+j} a_{p+1} \cdots \hat{a}_j \cdots \hat{a}_i \cdots a_{1} a_j a_i \omega([X_i, X_j], X_1, \dots, \hat{X}_i, \dots, \hat{X}_j, \dots, X_{p+1}). \\
	\end{eqnarray*}
	If we reorder the each of the coefficients $a_{p+1} \cdots \hat{a}_j
	\cdots \hat{a}_i \cdots a_{1} a_j a_i$ to $a_{p+1} \cdots a_2 a_{1}$ at
	the cost of introducing still more signs, we can factor all of the
	$a_i$s out of the summation to obtain:
	\begin{eqnarray*}
		&   & a_{p+1} \cdots a_2 a_1  \\
		& \times &  \sum_{i < j} (-1)^{i+j} (-1)^{|X_i||X_j|} \epsilon_1^{i-1}(i) \epsilon_1^{j-1}(j) \omega([X_i, X_j], X_1, \dots, \hat{X}_i, \dots, \hat{X}_j, \dots, X_{p+1}) \\
		& = & a_{p+1} \cdots a_2 a_1 d\omega(X_1 \wedge X_2 \wedge \dots \wedge X_{p+1}).
	\end{eqnarray*}
	Note that the first two lines are a single quantity, the product of
	$a_{p+1} \cdots a_1$ and a large summation. The last line is
	$(d\omega)_A(a_1 X_1 \wedge \dots \wedge a_{p+1} X_{p+1})$, as desired.
\end{proof}
This proposition says that from any Lie superalgebra cocycle on $\n$ we obtain
a Lie algebra cocycle on $\n_A$, albeit now valued in $A_0$. Since $N_A$ is an
exponential Lie group with Lie algebra $\n_A$, we can apply the techniques we
developed in Section \ref{sec:integratingcochains} to integrate $\omega_A$ to a group
cocycle, $\smallint \omega_A$, on $N_A$. 

First, however, we must pause to give some preliminary definitions concerning
calculus on $N_A$, which is diffeomorphic to the $A_0$-module $\n_A$. Recall
from Section \ref{sec:supermanifolds} that a map 
\[ \varphi \maps V \to W \]
between two $A_0$-modules said to be \define{$A_0$-smooth} if it is smooth in
the ordinary sense and its derivative
\[ \varphi_* \maps T_x V \to T_{\varphi(x)} W \]
is $A_{0}$-linear at each point $x \in V$. Here, the $A_0$-module structure on
each tangent space comes from the canonical identification with the ambient
vector space:
\[ T_x V \iso V, \quad T_{\varphi(x)} W \iso W. \]
It is clear that the identity is $A_0$-smooth and the composite of any two
$A_0$-smooth maps is $A_0$-smooth. A vector field $X$ on $V$ is
\define{$A_0$-smooth} if $Xf$ is an $A_0$-smooth function for all $f \maps V
\to A_0$ that are $A_0$-smooth. An $A_0$-valued differential $p$-form $\omega$
on $V$ is \define{$A_0$-smooth} if $\omega(X_1, \dots, X_p)$ is an $A_0$-smooth
function for all $A_0$-smooth vector fields $X_1, \dots, X_p$.

Now, we return to integrating $\omega$. As a first step, because $\n_A = T_1
N_A$, we can view $\omega_A$ as an $A_0$-valued $p$-form on $T_1 N_A$. Using
left translation, we can extend this to a left-invariant $A_0$-valued $p$-form
on $N_A$. Indeed, we can do this for any $A_0$-valued $p$-cochain on $\n_A$:
\[ C^p(\n_A) \iso \left\{ \mbox{left-invariant $A_0$-valued $p$-forms on $N_A$} \right\}. \]
Note that any left-invariant $A_0$-valued form on $N_A$ is automatically
$A_0$-smooth: multiplication on $N_A$ is given the $A_0$-smooth operation
induced from multiplication on $N$, and so left translation on $N_A$ is
$A_0$-smooth. We can differentiate and integrate $A_0$-valued $p$-forms in just
the same way as we would real-valued $p$-forms, and the de Rham differential
$d$ of left-invariant $p$-forms coincides with the usual differential of Lie
algebra $p$-cochains.

As before, we need a notion of simplices in $N$. Since $N$ is a supermanifold,
the vertices of a simplex should not be points of $N$, but rather $A$-points
for arbitrary Grassmann algebras $A$. This means that for any $(p+1)$-tuple of
$A$-points, we want to get a $p$-simplex:
\[ [n_0, n_1, \dots, n_p] \maps \Delta^p \to N_A, \]
where, once again, $\Delta^p$ is the standard $p$-simplex in $\R^{p+1}$, and
this map is required to be smooth. But this only defines a $p$-simplex in
$N_A$. To really get our hands on a $p$-simplex in $N$, we need it to depend
functorially on the choice of Grassmann algebra $A$ we use to probe $N$. So if $f
\maps A \to B$ is a homomorphism between Grassmann algebras and $N_f \maps N_A \to
N_B$ is the induced map between $A$-points and $B$-points, we require:
\[ N_f \circ [n_0, n_1, \dots, n_p] = [ N_f(n_0), N_f(n_1), \dots, N_f(n_p) ] \]
Thus given a collection of maps:
\[ (\varphi_p)_A \maps \Delta^p \times (N_A)^{p+1} \to N_A \]
for all $A$ and $p \geq 0$, we say this collection defines a
\define{left-invariant notion of simplices} in $N$ if 
\begin{itemize}

	\item each $(\varphi_p)_A$ is smooth, and for each $x \in \Delta^p$,
		the restriction: 
		\[ (\varphi_p)_A \maps \left\{x\right\} \times N_A^{p+1} \to N_A \] 
		is $A_0$-smooth;

	\item it defines a left-invariant notion of simplices in $N_A$ for each
		$A$, as in Definition \ref{def:simplices}; 

	\item the following diagram commutes for all homomorphisms $f \maps A
		\to B$:
		\[ \xymatrix{
		\Delta^p \times N_A^{p+1} \ar[r]^>>>>>{(\varphi_p)_A} \ar[d]_{1 \times N_f^{p+1}}  & N_A \ar[d]^{N_f} \\
		\Delta^p \times N_B^{p+1} \ar[r]_>>>>>{(\varphi_p)_B} & N_B 
		}
		\]

\end{itemize}
We can use a left-invariant notion of simplices to define a cochain map
$\smallint \maps C^\bullet(\n) \to C^\bullet(N)$:
\begin{prop} \label{prop:superintegrating}
	Let $\n$ be a nilpotent Lie superalgebra, and let $N$ be the
	exponential supergroup which integrates $\n$. If $N$ is equipped with a
	left-invariant notion of simplices, then there is a cochain map:
	\[ \smallint \maps C^\bullet(\n) \to C^\bullet(N) \]
	which sends the Lie superalgebra $p$-cochain $\omega$ to the supergroup
	$p$-cochain $\smallint \omega$, given on $A$-points by:
	\[ (\smallint \omega)_A(n_1, \dots, n_p) = \int_{[1, n_1, n_1 n_2, \dots, n_1 n_2 \dots n_p ] } \omega_A \]
	for $n_1, \dots, n_p \in N_A$.
\end{prop}
\begin{proof}
	First, we must check that $\smallint \omega_A \maps N_A^p \to A_0$ is
	natural in $A$ and $A_0$-smooth, and hence defines a map of
	supermanifolds:
	\[ \smallint \omega \maps N^p \to \R . \]
	Smoothness is clear, so we check naturality and the $A_0$-linearity of
	the derivative.

	To check naturality, let $f \maps A \to B$ be a homomorphism, and $N_f
	\maps N_A \to N_B$ be the induced map from $A$-points to $B$-points. We
	wish to show the following square commutes:
	\[ \xymatrix{
	N_A^p \ar[r]^{\smallint \omega_A} \ar[d]_{N_f^p} & A_0 \ar[d]^{f_0} \\
	N_B^p \ar[r]_{\smallint \omega_B}                & B_0 \\
	} \]
	For $A$-points $n_1, \dots, n_p \in N_A$, we have:
	\[ f_0 \int_{[1, n_1, n_1 n_2, \dots, n_1 n_2 \dots n_p]} \omega_A = \int_{[1,n_1, n_1 n_2, \dots, n_1 \dots n_p]}  f_0 \omega_A. \]
	Since $\omega_A \maps \Lambda^p \n_A \to A_0$ is natural itself, we have:
	\[ f_0 \omega_A(X_1, \dots, X_p) = \omega_B(\n_f(X_1), \dots, \n_f(X_p)), \]
	for all $X_1, \dots, X_p \in \n_A$. Now, under the identification $\n_A \iso T_1 N_A$, the linear map:
	\[ \n_f \maps \n_A \to \n_B \]
	is the derivative of the linear map $N_f \maps N_A \to N_B$, so we get the pullback of $\omega_A$ along $N_f$:
	\[ \omega_B(\n_f(X_1), \dots, \n_f(X_p)) = \omega_B((N_f)_*(X_1), \dots, (N_f)_*(X_p)) = N^*_f \omega_B(X_1, \dots, X_p). \]
	Finally:
	\begin{eqnarray*} 
		f_0 \int_{[1, n_1, n_1 n_2, \dots, n_1 n_2 \dots n_p]} \omega_A & = & \int_{[1, n_1, n_1 n_2, \dots, n_1 n_2 \dots n_p]} N^*_f \omega_B \\ 
		& = & \int_{N_f \circ [1, n_1, n_1 n_2, \dots , n_1 n_2 \dots n_p]} \omega_B \\
		& = & \int_{[1, N_f(n_1), N_f(n_1) N_f(n_2), \dots , N_f(n_1) N_f(n_2) \dots N_f(n_p)]} \omega_B \\
	\end{eqnarray*}
	where in the last step we have used the fact that $[1, n_1, n_1 n_2,
	\dots, n_1 \dots n_p]$ is a left-invariant simplex in $N$, as well as
	the fact that $N_f$ is a group homomorphism. But this says exactly that
	$\smallint \omega_A$ is natural in $A$. 

	Next, we check that $\smallint \omega_A$ has a derivative that is
	$A_0$-linear. Briefly, this holds because the derivative of
	$(\varphi_p)_A$ with respect to $N_A$ is $A_0$-linear.  The
	$A_0$-linearity of the derivative of $\smallint \omega_A$ then follows
	from the elementary analytic fact that integration with respect to one
	variable and differentiation respect to another commute with each
	other, at least when the integration is performed over a compact set.

	In detail, let us write $\psi$ for the function $\smallint \omega_A
	\maps N_A^p \to A_0$.  Let $v \in T_n N_A^{p}$ be a tangent vector, and let $a \in A_0$.
	Tedious as it may seem, we will show directly that the derivative of
	$\psi$ at $n$ is $A_0$-linear by computing its value on
	$av$. Denoting the derivative of $\psi$ at $n$ by $\psi_*$, we want to
	show that:
	\[ \psi_*(av) = a \psi_*(v) . \]
	Take $\gamma$ to be a path through $n$ with tangent $v$:
	\[ \gamma(0) = n, \quad \dot{\gamma}(0) = v , \]
	and $\delta$ to be a path through $n$ with tangent $av$:
	\[ \delta(0) = n, \quad \dot{\delta}(0) = av , \]
	With wish to show that:
	\[ \frac{d}{dt} \psi(\delta(t)) \, \mid_{t=0} \, \, = \, \,  a \frac{d}{dt} \psi(\gamma(t)) \, \mid_{t=0} . \]

	Now, by definition, 
	\[ \psi(n) = \int_{\Delta^p} g(x, n) \, dx , \]
	where $n \in N_A^p$, and $g$ denotes the pullback of $\omega_A$ along
	the function:
	\[ \begin{array}{ccc} 
			\Delta^p \times N_A^p & \to     & N_A \\
			(x, n_1, \dots, n_p)  & \mapsto & (\varphi_p)_A (x, 1, n_1, n_1 n_2, \dots, n_1 n_2 \cdots n_p) , 
	\end{array} \]
	where $x \in \Delta^p$ and $n_1, \dots, n_p \in N_A$. So: 
	\[ \psi(\delta(t)) = \int_{\Delta^p} g(x, \delta(t)) \, dx, \quad \psi(\gamma(t)) = \int_{\Delta^p} g(x, \gamma(t)) \, dx . \]
	And thus, by differentiating and commuting with integration, we get:
	\[ \psi_*(av) = \int_{\Delta^p} \frac{\partial}{\partial t} g(x,\delta(t)) \mid_{t=0} \, dx , \quad \psi_*(v) = \int_{\Delta^p} \frac{\partial}{\partial t} g(x,\gamma(t)) \mid_{t=0} \, dx \]
	By hypothesis, $(\varphi_p)_A$ and $\omega_A$ are $A_0$-smooth in
	$N_A$, therefore so is the pullback $h$. Thus:
	\[ \frac{\partial}{\partial t} h(x,\delta(t)) \mid_{t=0} = a \frac{\partial}{\partial t} h(x,\gamma(t)) \mid_{t=0} . \]
	Using the $A_0$-linearity of integration, it follows that:
	\[ \psi_*(av) = a \psi_*(v) \]
	as desired.

	Thus, $\smallint \omega \maps N_A^p \to A_0$, being natural in $A$ and
	$A_0$-smooth, defines a map of supermanifolds: 
	\[ \smallint \omega \maps N^p \to \R . \] 
	We therefore have a map:
	\[ \smallint \maps C^\bullet(\n) \to C^\bullet(N) \]
	It remains check that it is a cochain map. Indeed, $\smallint \omega_A$ is the composite
	of the cochain maps:
	\[ \omega \mapsto \omega_A \mapsto \smallint \omega_A, \]
	So on $A$-points, $\smallint d\omega$ is:
	\begin{eqnarray*}
	(\smallint d \omega)_A & = & \smallint (d\omega)_A \\
	& = & d( \smallint \omega_A) \\
	& = & (d \smallint \omega)_A ,
	\end{eqnarray*}
	where in the last step we have used the fact that $d(f_A) = (df)_A$ by
	definition of $d$. In brief, for every Grassmann algebra $A$:
	\[ (\smallint d \omega)_A = (d \smallint \omega)_A . \]
	We therefore conclude:
	\[ \smallint d \omega = d \smallint \omega,  \]
	and so $\smallint$ is a cochain map.
	
\end{proof}

Finally, we shall prove that there is a left-invariant notion of simplices with
which we can equip $N$. For a fixed Grassmann algebra $A$, the Lie group $N_A$ is
exponential. We shall show that if we take:
\[ (\varphi_p)_A \maps \Delta^p \times N^{p+1}_A \to N_A \]
to be the standard notion of left-invariant simplices in Proposition
\ref{prop:standard}, then this defines a left-invariant notion of simplices in
$N$. The key is to note that each stage of the inductive definition of
$(\varphi_p)_A$ we get maps that are natural in $A$.

\begin{prop}
	Let $N$ be the exponential supergroup of the nilpotent Lie superalgebra
	$\n$. Fix a smoothing factor $\ell \maps [0,1] \to [0,1]$.  For each
	Grassmann algebra $A$ and $p \geq 0$, define:
	\[ (\varphi_p)_A \maps \Delta^p \times N^{p+1}_A \to N_A \]
	to be the standard left-invariant notion of simplices with smoothing
	factor $\ell$. Then this defines a left-invariant notion of simplices
	in $N$.
\end{prop}
\begin{proof}
	Fix Grassmann algebras $A$ and $B$ and a map $f \maps A \to B$. We proceed
	by induction on $p$. For $p = 0$, the maps:
	\[ (\varphi_0)_A \maps \Delta^0 \times N_A \to N_A, \]
	\[ (\varphi_0)_B \maps \Delta^0 \times N_B \to N_B, \]
	are the obvious projections. The fact that:
	\[ \xymatrix{
	\Delta^0 \times N_A \ar[r]^>>>>>{(\varphi_0)_A} \ar[d]_{1 \times N_f}  & N_A \ar[d]^{N_f} \\
	\Delta^0 \times N_B \ar[r]_>>>>>{(\varphi_0)_B} & N_B 
	}
	\]
	commutes is then automatic.

	For arbitrary $p$, suppose that the following square commutes:
	\[ \xymatrix{
	\Delta^{p-1} \times N_A^p \ar[r]^>>>>>{(\varphi_{p-1})_A} \ar[d]_{1 \times N_f^p}  & N_A \ar[d]^{N_f} \\
	\Delta^{p-1} \times N_B^p \ar[r]_>>>>>{(\varphi_{p-1})_B} & N_B 
	}
	\]
	and that $(\varphi_{p-1})_A$ and  $(\varphi_{p-1})_B$ are $A_0$- and
	$B_0$-smooth. In other words, the above square says that for any
	$p$-tuple of $A$-points, we have:
	\[ N_f \circ [ n_1, \dots, n_p ] = [N_f(n_1), \dots, N_f(n_p)]. \]
	We construct $(\varphi_p)_A$ and $(\varphi_p)_B$ from
	$(\varphi_{p-1})_A$ and $(\varphi_{p-1})_B$, respectively, using the
	apex-base construction. That is, given the $(p-1)$-simplex $[n_1,
	\dots, n_p]$ given by $(\varphi_{p-1})_A$ for the $A$-points $n_1,
	\dots, n_p \in N_A$, we define the based $p$-simplex:
	\[ [1,n_1, \dots, n_p] \]
	in $N_A$ by using the exponential map $\exp_A$ to sweep out a path from the apex $1$ to
	each point of the base $[n_1, \dots, n_p]$. Similarly, we define the based $p$-simplex:
	\[ [1,N_f(n_1), \dots, N_f(n_p)] \]
	in $N_B$ by using the exponential map $\exp_B$ to sweep out a path from
	the apex $1$ to each point of the base $[N_f(n_1), \dots, N_f(n_p)]$.
	From the naturality of $\exp$, we will establish that:
	\[ N_f \circ [1, n_1, \dots, n_p] = [1, N_f(n_1), \dots, N_f[n_p)] . \]

	To verify this claim, let 
	\[ \exp_A(X) = [n_1, \dots, n_p](x), \mbox{ for some } x \in \Delta^{p-1} \]
	be a point of the base in $N_A$. By the inductive hypothesis, $N_f(\exp_A(X)) =
	\exp_B(\n_f(X))$ is the corresponding point of the base in $N_B$. We
	wish to see that points of the path $\exp_A(\ell(t)X)$ connecting $1$
	to $\exp_A(X)$ in $N_A$ correspond via $N_f$ to points on the path
	$\exp_B(\ell(t)\n_f(X))$ connecting $1$ to $\exp_B(\n_f(X))$ in $N_B$.
	But this is automatic, because:
	\[ N_f(\exp_A(\ell(t)X) = \exp_B(\n_f(\ell(t)X)) = \exp_B(\ell(t) \n_f(X)) , \]
	where in the last step we use the fact that $\n_f \maps \n_A \to \n_B$
	is linear. Thus, it is true that:
	\[ N_f \circ [1, n_1, \dots, n_p] = [1, N_f(n_1), \dots, N_f[n_p)] , \]
	for based $p$-simplices. 
	
	Using left translation, we can show that:
	\[ N_f \circ [n_0, n_1, \dots, n_p] = [N_f(n_0), N_f(n_1), \dots, N_f[n_p)] . \]
	for all $p$-simplices. In other words, the following diagram commutes:
	\[ \xymatrix{
	\Delta^p \times N_A^{p+1} \ar[r]^>>>>>{(\varphi_p)_A} \ar[d]_{1 \times N_f^{p+1}}  & N_A \ar[d]^{N_f} \\
	\Delta^p \times N_B^{p+1} \ar[r]_>>>>>{(\varphi_p)_B} & N_B 
	}
	\]
	Because each step in the apex-base construction preserves $A_0$- or
	$B_0$-smoothness, we note that $(\varphi_p)_A$ and $(\varphi_p)_B$ are
	$A_0$- and $B_0$-smooth, respectively. The result now follows for all
	$p$ by induction.
\end{proof}

\chapter{Superstring Lie 2-supergroups, 2-brane Lie 3-supergroups} \label{ch:finale}

We are now ready to unveil the Lie $n$-super\-groups which integrate our
favorite Lie $n$-super\-algebras, $\superstring(k+1,1)$ and $\twobrane(k+2,1)$.
Remember, these are the Lie $n$-superalgebras which occur only in the
dimensions for which string theory and 2-brane theory make sense. They are
\emph{not} nilpotent, simply because the Poincar\'e superalgebras
$\siso(k+1,1)$ and $\siso(k+2,1)$ that form degree 0 of $\superstring(k+1,1)$
and $\twobrane(k+2,1)$ are not nilpotent. Nonetheless, we are equipped to
integrate them using only the tools we have built to perform this task for
nilpotent Lie $n$-superalgebras.

The road to this result has been a long one, and there is yet some ground to
cover before we are finished. So, let us take stock of our progress before we
move ahead:

\begin{itemize}
	\item In spacetime dimensions $k+2 = 3$, 4, 6 and 10, we used division
		algebras to construct a 3-cocycle $\alpha$ on the
		supertranslation algebra:
		\[ \T = V \oplus S \]
		which is nonzero only when it  eats a vector and two spinors:
		\[ \alpha(A, \psi, \phi) = \langle \psi, A \phi \rangle . \]

	\item In spacetime dimensions $k+3 = 4$, 5, 7 and 11, we used division
		algebras to construct a 4-cocycle $\alpha$ on the
		supertranslation algebra:
		\[ \T = \V \oplus \S \]
		which is nonzero only when it eats two vectors and two spinors:
		\[ \beta(\A,\B,\Psi,\Phi) = \langle \Psi, (\A\B - \B\A) \Phi \rangle . \]

	\item Because $\alpha$ is invariant under the action of $\so(k+1,1)$,
		it can be extended to a 3-cocycle on the Poincar\'e
		superalgebra:
		\[ \siso(k+1,1) = \so(k+1,1) \ltimes \T. \]
		The extension is just defined to vanish outside of $\T$, and we
		call it $\alpha$ as well.

	\item Because $\beta$ is invariant under the action of $\so(k+2,1)$,
		it can be extended to a 3-cocycle on the Poincar\'e
		superalgebra:
		\[ \siso(k+2,1) = \so(k+2,1) \ltimes \T. \]
		The extension is just defined to vanish outside of $\T$, and we
		call it $\beta$ as well.

	\item Therefore, in spacetime dimensions $k+2$, we get a Lie
		2-superalgebra $\superstring(k+1,1)$ by extending
		$\siso(k+1,1)$ by the 3-cocycle $\alpha$.

	\item Likewise, in spacetime dimensions $k+3$, we get a Lie
		3-superalgebra $\twobrane(k+2,1)$ by extending
		$\siso(k+2,1)$ by the 4-cocycle $\beta$.
\end{itemize}

In the last chapter, we built the technology necessary to integrate Lie
superalgebra cocycles to supergroup cocycles, \emph{provided} the Lie
superalgebra in question is nilpotent. This allows us to integrate nilpotent
Lie $n$-superalgebras to $n$-supergroups. But $\superstring(k+1,1)$ and
$\twobrane(k+2,1)$ are not nilpotent, so we cannot use this directly here.

However, the cocycles $\alpha$ and $\beta$ are supported on a nilpotent
subalgebra: the supertranslation algebra, $\T$, for the appropriate dimension.
This saves the day: we can integrate $\alpha$ and $\beta$ as cocycles on $\T$.
This gives us cocycles $\smallint \alpha$ and $\smallint \beta$ on the
supertranslation supergroup, $T$, for the appropriate dimension. We will then
be able to extend these cocycles to the Poincar\'e supergroup, thanks to their
invariance under Lorentz transformations.

The following proposition helps us to accomplish this, but takes its most
beautiful form when we work with `homogeneous supergroup cochains', which we
have not actually defined. Rest assured---they are exactly what you expect. If
$G$ is a supergroup that acts on the abelian supergroup $M$ by automorphism, a
\define{homogeneous $M$-valued $p$-cochain} on $G$ is a smooth map:
\[ F \maps G^{p+1} \to M \]
such that, for any Grassmann algebra $A$ and $A$-points $g, g_0, \dots, g_p \in
G_A$:
\[ F_A(gg_0, g g_1, \dots, g g_p) = g F_A(g_1, \dots, g_p) . \]
We can define the supergroup cohomology of $G$ using homogeneous or
inhomogeneous cochains, just as was the case with Lie group cohomology.

\begin{prop} Let $G$ and $H$ be Lie supergroups such that $G$ acts on $H$, and
	let $M$ be an abelian supergroup on which $G \ltimes H$ acts by
	automorphism.  Given a homogeneous $M$-valued $p$-cochain $F$ on $H$:
	\[ F \maps H^{p+1} \to M, \]
       	we can extend it to a map of supermanifolds:
	\[ \tilde{F} \maps (G \ltimes H)^{p+1} \to M \]
	by pulling back along the projection $(G \ltimes H)^{p+1} \to H^{p+1}$.
	In terms of $A$-points 
	\[ (g_0,h_0), \ldots, (g_p,h_p) \in G_A \ltimes H_A , \] 
	this means $\tilde{F}$ is defined by:
	\[ \tilde{F}_A((g_0,h_0), \ldots, (g_p,h_p)) = F_A(h_0, \ldots, h_p), \]
	Then $\tilde{F}$ is a homogeneous $p$-cochain on $G \ltimes H$ if and only if $F$ is
	$G$-equivariant, and in this case $d\tilde{F} = \widetilde{dF}$.
\end{prop}
\begin{proof}
	We work over $A$-points, $G_A \ltimes H_A$. Denoting the action of $g
	\in G_A$ on $h \in H_A$ by $h^g$, recall that multiplication in the
	semidirect product $G_A \ltimes H_A$ is given by:
	\[ (g_1, h_1) (g_2,h_2) = (g_1 g_2, h_1 h_2^{g_1}). \]
	
	Now suppose $\tilde{F}$ is homogeneous. By definition of homogeneity, we have:
	\[ \tilde{F}_A((g,h)(g_0,h_0), \ldots, (g,h)(g_p,h_p)) = (g,h) \tilde{F}_A((g_0,h_0), \ldots, (g_p,h_p)). \]
	Multiplying out each pair on the left and using the definition of
	$\tilde{F}$ on both sides, we get:
	\[ F_A(h h_0^{g}, \ldots, h h_p^{g}) = (g,h) F_A(h_0, \ldots, h_p). \]
	Writing $(g,h)$ as $(1,h)(g,1)$, and pulling $h$ out on the left-hand
	side, we now obtain:
	\[ (1,h) F_A(h_0^{g}, \ldots, h_p^{g}) = (1,h)(g,1) F_A(h_0, \ldots, h_p). \]
	Cancelling $(1,h)$ from both sides, this last equation just says that
	$F_A$ is $G_A$-equivariant. The converse follows from reversing this
	calculation. Since this holds for all $A$-points, we conclude that
	$\tilde{F}$ is homogeneous if and only if $F$ is $G$-equivariant.

	When $F$ is $G$-equivariant, it is easy to see that $dF$ is also, and
	that $d \tilde{F} = \widetilde{dF}$, so we are done.
\end{proof}

\noindent 
Now, at long last, we are ready to integrate $\alpha$ and $\beta$. In the
following proposition, $T$ denotes the \define{supertranslation group}, the
exponential supergroup of the supertranslation algebra $\T$.

\begin{prop}
	In dimensions 3, 4, 6 and 10, the Lie supergroup 3-cocycle $\smallint
	\alpha$ on the supertranslation group $T$ is invariant under the action
	of $\Spin(k+1,1)$.  Similarly, in dimensions 4, 5, 7 and 11, the Lie
	supergroup 4-cocycle $\smallint \beta$ on the supertranslation group
	$T$ is invariant under the action of $\Spin(k+2,1)$.
\end{prop}

\noindent This is an immediate consequence of the following:

\begin{prop} 
	Let $H$ be a nilpotent Lie supergroup with Lie superalgebra $\h$.
	Assume $H$ is equipped with its standard left-invariant notion of 
	simplices, and let $G$ be a Lie supergroup that acts on $H$ by
	automorphism. If $\omega \in C^p(\h)$ is an even Lie superalgebra
	$p$-cochain which is invariant under the induced action of $G$ on $\h$,
	then $\smallint \omega \in C^p(H)$ is a Lie supergroup $p$-cochain
	which is invariant under the action of $G$ on $H$.
\end{prop}

\begin{proof}
	Fixing a Grassmann algebra $A$, we must prove that
	\[ \int_{[h_0^g, h_1^g, \dots, h_p^g]} \omega_A = \int_{[h_0, h_1, \dots, h_p]} \omega_A , \]
	for all $A$-points $g \in G_A$ and $h_0, h_1, \dots, h_p \in H_A$. We
	shall see this follows from the fact that the $p$-simplices in $H$ are
	themselves $G$-equivariant, in the sense that:
	\[ [h_0^g, h_1^g, \dots, h_p^g] = [h_0, h_1, \dots, h_p]^g . \]
	Assuming this for the moment, let us check that our result follows.
	Indeed, applying the above equation, we get:
	\begin{eqnarray*}
		\int_{[h_0^g, h_1^g, \dots, h_p^g]} \omega_A & = & \int_{[h_0, h_1, \dots, h_p]^g} \omega_A \\
		                                             & = & \int_{[h_0, h_1, \dots, h_p]} \Ad(g)^* \omega_A \\
							     & = & \int_{[h_0, h_1, \dots, h_p]} \omega_A ,
	\end{eqnarray*}
	where the final step uses $\Ad(g)^* \omega_A = \omega_A$, which is just
	the $G$-invariance of $\omega$.

	It therefore remains to prove the equation $[h_0^g, h_1^g, \dots, h_p^g] =
	[h_0, h_1, \dots, h_p]^g$ actually holds. Note that this is the same as
	saying that the map
	\[ (\varphi_{p})_A \maps \Delta^p \times H_A^{p+1} \to H_A \]
	is $G_A$-equivariant. We check it by induction on $p$.

	For $p = 0$, the map:
	\[ (\varphi_0)_A \maps \Delta^0 \times H_A \to H_A  \]
	is just the projection, and $G_A$-equivariance is obvious. So fix some $p
	\geq 0$ and suppose that $(\varphi_{p-1})_A$ is $G_A$-equivariant. We now
	construct $(\varphi_p)_A$ from $(\varphi_{p-1})_A$ using the apex-base
	construction, and show that equivariance is preserved.

	So, given the $(p-1)$-simplex $[h_1, \dots, h_p]$ given by
	$(\varphi_{p-1})_A$ for the $A$-points $h_1, \dots, h_p \in H_A$, we
	define the based $p$-simplex:
	\[ [1,h_1, \dots, h_p] \]
	in $H_A$ by using the exponential map to sweep out a path from the apex
	$1$ to each point of the base $[h_1, \dots, h_p]$. In a similar way, we
	define the based $p$-simplex:
	\[ [1,h_1^g, \dots, h_p^g] \]
	By hypothesis, $[h_1^g, \dots, h_p^g] = [h_1, \dots, h_p]^g$, and since
	the exponential map $\exp \maps \h_A \to H_A$ is itself
	$G_A$-equivariant, it follows for based $p$-simplices that:
	\[ [1, h_1^g, \dots, h_p^g] = [1, h_1, \dots, h_p]^g . \]
	The result now follows for all $p$-simplices by left translation. This
	completes the proof.
\end{proof}

It thus follows that in dimensions 3, 4, 6 and 10, the cocycle $\smallint
\alpha$ on the supertranslations can be extended to a 3-cocycle on the full
Poincar\'e supergroup:
\[ \SISO(k+1, 1) = \Spin(k+1, 1) \ltimes T , \] 
while in dimensions 4, 5, 7 and 11, the cocycle $\smallint \beta$ can be
extended to the Poincar\'e supergroup:
\[ \SISO(k+2,1) = \Spin(k+2,1) \ltimes T . \] 
By a slight abuse of notation, we continue to denote these extensions by
$\smallint \alpha$ and $\smallint \beta$, respectively. As an immediate
consequence, we have:

\begin{thm} \label{thm:superstringgroup}
	In dimensions 3, 4, 6 and 10, there exists a slim Lie 2-supergroup
	formed by extending the Poincar\'e supergroup $\SISO(k+1,1)$ by the
	3-cocycle $\smallint \alpha$, which we call we the \define{superstring
	Lie 2-supergroup}, \define{\boldmath{$\Superstring(k+1,1)$}}.
\end{thm}

\begin{thm} \label{thm:2branegroup}
	In dimensions 4, 5, 7 and 11, there exists a slim Lie 3-supergroup
	formed by extending the Poincar\'e supergroup $\SISO(k+2,1)$ by the
	4-cocycle $\smallint \beta$, which we call the \define{2-brane Lie
	3-supergroup}, \define{\boldmath{$\Twobrane(k+2,1)$}}.
\end{thm}

\section{Outlook}

In this thesis we have seen a number of clues that a categorified geometry is
relevant to superstrings, M-theory, and supergravity.  Categorifying gauge
theory to obtain higher gauge theory boosts the dimension of objects which we
can parallel transport. The special identities which make supersymmetry work
allow us to categorify the spacetime symmetries of superstrings. We propose
there is a simple underlying reason for all of this: strings are extended
objects, not point particles, so we need a geometry in which extended objects
can play a role as fundamental as points. It is precisely this kind of geometry
which we are now ready to explore, now that we have the $n$-supergroups
$\Superstring(k+1,1)$ and $\Twobrane(k+2,1)$ in our hands.

\singlespacing

\bibliographystyle{plain}

\end{document}